\documentclass[conference]{IEEEtran}

\usepackage{ifthen}
\usepackage{xspace}

\newboolean{talk}
\setboolean{talk}{false}
\newboolean{paper}
\setboolean{paper}{false}

\newboolean{IEEEstyle}
\setboolean{IEEEstyle}{false}
\newboolean{lipicsstyle}
\setboolean{lipicsstyle}{false}
\newboolean{eptcsstyle}
\setboolean{eptcsstyle}{false}
\newboolean{entcsstyle}
\setboolean{entcsstyle}{false}
\newboolean{sigplanstyle}
\setboolean{sigplanstyle}{false}
\newboolean{easychairstyle}
\setboolean{easychairstyle}{false}
\newboolean{scrartclstyle}
\setboolean{scrartclstyle}{false}
\newboolean{lncsstyle}
\setboolean{lncsstyle}{false}

\newboolean{acmmart}
\setboolean{acmmart}{false}

\newboolean{needstheorems}
\setboolean{needstheorems}{false}
\newboolean{withimages}
\setboolean{withimages}{false}
\newboolean{withproofs}
\setboolean{withproofs}{true}

\newboolean{french}
\setboolean{french}{false}

\setboolean{IEEEstyle}{true}
\usepackage{booktabs}   
\usepackage{subcaption} 
\usepackage[all]{xy}

\usepackage{ebproof}
\ebproofset{label separation=0.3em}

\usepackage{booktabs}   
\usepackage{subcaption} 
\usepackage{graphicx}
\usepackage[normalem]{ulem}
\usepackage{multirow}

\usepackage[utf8]{inputenc}

\usepackage{multirow}
\usepackage{cmll}
\usepackage{listings}

\usepackage{amsmath, amsfonts, 
	amsthm, mathtools}
\usepackage{stmaryrd}
\usepackage{mathtools}
\usepackage{enumerate}
\usepackage{listings}

\usepackage{fancybox}
\usepackage{hhline}
\usepackage{marginnote}
\usepackage{soul}
\usepackage{xcolor}

\usepackage[hidelinks]{hyperref}
\usepackage{cleveref}
\usepackage{complexity}

\newcommand{\sem}[1]{\llbracket#1\rrbracket}

\newcommand{\semfull}[1]{\sem{#1}^\fullsym}

\newcommand{\ignore}[1]{}

\newcommand{\colspace}{@{\hspace{.5cm}}}
\newcommand{\myinput}[1]{\ifthenelse{\boolean{withimages}}{\input{#1}}{}}

\newcommand{\reflemma}[1]{Lemma~\ref{l:#1}}
\newcommand{\reflemmap}[2]{Lemma~\ref{l:#1}.\ref{p:#1-#2}}

\newcommand{\refpoint}[1]{Point~\ref{p:#1}}

\newcommand{\refthm}[1]{Thm.~\ref{thm:#1}}

\newcommand{\refprop}[1]{Prop.~\ref{prop:#1}}

\newcommand{\refpropp}[2]{Prop.~\ref{prop:#1}.\ref{p:#1-#2}} 

\newcommand{\reffig}[1]{Fig.~\ref{fig:#1}}

\newcommand{\refrmk}[1]{Remark~\ref{rmk:#1}} 

\newcommand{\ie}{\textit{i.e.}\xspace}
\newcommand{\eg}{\textit{e.g.}\xspace}
\newcommand{\ih}{\textit{i.h.}\xspace}
\newcommand{\resp}{\textnormal{resp.}\xspace}

\newcommand{\ES}{\text{ES}\xspace}
\newcommand{\leftsh}{\text{left}\xspace}
\newcommand{\rightsh}{\text{right}\xspace}
\newcommand{\full}{\text{full}\xspace}
\newcommand{\Full}{\text{Full}\xspace}
\renewcommand{\full}{\text{strong}\xspace}
\renewcommand{\Full}{\text{Strong}\xspace}

\newcommand{\defeq}{\coloneqq} 
\newcommand{\eqdef}{\eqqcolon} 
\newcommand{\grameq}{\Coloneqq} 
\newcommand{\set}[1]{\{#1\}}

\newcommand{\card}[1]{\# #1}

\newcommand{\nat}{\mathbb{N}}
\newcommand{\size}[1]{|#1|}
\newcommand{\sizectx}[1]{\size{\typelist{#1}}}
\newcommand{\sizectxm}[1]{\sizem{\typelist{#1}}}
\newcommand{\sizetyp}[1]{\size{#1}}
\newcommand{\sizetypm}[1]{\sizem{#1}}

\renewcommand{\l}{\lambda}
\newcommand{\isub}[2]{\{#1/#2\}}
\newcommand{\replace}[2]{#1{\shortleftarrow}#2}
\renewcommand{\isub}[2]{\{\replace{#1}{#2}\}}
\newcommand{\esub}[2]{[\replace{#1}{#2}]}
\renewcommand{\esub}[2]{[#1{\shortleftarrow}#2]}

\newcommand{\fv}[1]{{\sf fv}(#1)}
\newcommand{\gt}[1]{{\sf gt}(#1)}

\newcommand{\nf}{\mathsf{nf}} 

\newcommand{\rootRew}[1]{\mapsto_{#1}}
\newcommand{\rootlRew}[1]{\; \mbox{}_{#1}{\mapsfrom}\ }
\newcommand{\Rew}[1]{\rightarrow_{#1}}

\newcommand{\lRew}[1]{\; \mbox{}_{#1}{\leftarrow}\ }

\newcommand{\lto}{\lRew{}}

\newcommand{\rtom}{\rootRew{\msym}}
\newcommand{\rtoe}{\rootRew{\esym}}
\newcommand{\reversertoe}{\rootlRew{\esym}} 

\newcommand{\rtoevar}{\rootRew{\expovar}} 

\newcommand{\betaplot}{\beta_v}
 
\newcommand{\tobvplot}{\Rew{\betaplot}} 
\newcommand{\rtobvplot}{\rootRew{\betaplot}} 

\newcommand{\esym}{{\mathsf e}}

\newcommand{\msym}{\mathsf{m}}

\newcommand{\ssym}{{\mathsf s}}
\newcommand{\wsym}{{\mathsf{o}}} 
\newcommand{\wmsym}{{\wsym\msym}} 
\newcommand{\wesym}{{\wsym\esym}} 
\newcommand{\omsym}{{\wmsym}} 
\newcommand{\oesym}{{\wesym}} 
\newcommand{\smsym}{\esssym{\msym}} 
\newcommand{\sesym}{\esssym{\esym}} 
\newcommand{\vmsym}{\mathsf{shuf}} 
\newcommand{\shuf}{\vmsym} 

\newcommand{\shufeqext}{\shufeqext} 
\newcommand{\tsym}{{\mathsf t}} 

 \newcommand{\tom}{\Rew{\msym}}
 \newcommand{\toe}{\Rew{\esym}}

\newcommand{\toevar}{\Rew{\expovar}}

\newcommand{\vsubval}{\vsub_\val}
\renewcommand{\vsubval}{\vsub_\l}

\newcommand{\tovsubval}{\Rew{\vsubval}}

\newcommand{\tovsubo}{\Rew{\wsym}}

\newcommand{\tomo}{\Rew{\omsym}}

\newcommand{\toeo}{\Rew{\wsym{\esym}}}
\newcommand{\reversetoeo}{\lRew{\oesym}}

\newcommand{\tovsubs}{\Rew{\esssym}}

\newcommand{\toms}{\Rew{\esssym{\msym}}}
\newcommand{\reversetoms}{\lRew{\smsym}}
\newcommand{\toes}{\Rew{\esssym{\esym}}}
\newcommand{\reversetoes}{\lRew{\sesym}}

\newcommand{\tm}{t}
\newcommand{\tmtwo}{u}
\newcommand{\tmthree}{r}
\newcommand{\tmfour}{q}
\newcommand{\tmfive}{p}
\newcommand{\tmsix}{s}

\newcommand{\tmp}{\tm'}
\newcommand{\tmtwop}{\tmtwo'}
\newcommand{\tmthreep}{\tmthree'}
\newcommand{\tmfourp}{\tmfour'}
\newcommand{\tmfivep}{\tmfive'}

\newcommand{\tmfourpp}{\tmfour''}
\newcommand{\tmfivepp}{\tmfive''}

\newcommand{\var}{x}
\newcommand{\vartwo}{y}
\newcommand{\varthree}{z}

\newcommand{\val}{v}

\newcommand{\tval}{\val_\tsym}

\newcommand{\ctxholep}[1]{\langle #1\rangle}
\newcommand{\ctxhole}{\ctxholep{\cdot}}

\newcommand{\ctx}{C}
\newcommand{\ctxtwo}{\ctx'}

\newcommand{\ctxp}[1]{\ctx\ctxholep{#1}}
\newcommand{\ctxtwop}[1]{\ctxtwo\ctxholep{#1}}

\newcommand{\sctx}{L}
\newcommand{\sctxtwo}{\sctx'}

\newcommand{\sctxp}[1]{\sctx\ctxholep{#1}}
\newcommand{\sctxtwop}[1]{\sctxtwo\!\ctxholep{#1}}

\newcommand{\arbctxp}[1]{\arbctxp{#1}}
\newcommand{\arbctxtwop}[1]{\arbctxtwop{#1}}

\newcommand{\evsctx}{S}
\newcommand{\evsctxtwo}{\evsctx'}

\newcommand{\evsctxp}[1]{\evsctx\ctxholep{#1}}
\newcommand{\evsctxtwop}[1]{\evsctxtwo\ctxholep{#1}}

\newcommand{\fctx}{F}

\newcommand{\fctxp}[1]{\fctx\ctxholep{#1}}

\ifthenelse{\boolean{acmmart}
}{
  
  }{
  
  }

\newcommand{\sizefu}[1]{\size{#1}_{\fullsym}}

\newcommand{\deriv}{d}
\newcommand{\derivp}{d'} 

\newcommand{\sizehole}[2]{|#2|_{#1}}

\newcommand{\sizem}[1]{\sizehole{\msym}{#1}} 
\newcommand{\sizes}[1]{\sizehole{\ssym}{#1}} 
\newcommand{\sizeo}[1]{\sizehole{\osym}{#1}} 

\ifthenelse{\boolean{IEEEstyle}}{
\usepackage{amssymb}
\usepackage{amsthm}
    \newtheorem{theorem}{Theorem}[section]
    \newtheorem{lemma}[theorem]{Lemma}
    \newtheorem{corollary}[theorem]{Corollary}
    \newtheorem{proposition}[theorem]{Proposition}
    \newtheorem{definition}[theorem]{Definition}
}{}

\ifthenelse{\boolean{lncsstyle}}{}{\newtheorem{remark}{Remark}[section]}

\newcommand{\itm}{i}
\newcommand{\itmtwo}{\itm'} 
\newcommand{\itmthree}{\itm''}

\newcommand{\fire}{f}
\newcommand{\firetwo}{\fire'}

\newcommand{\sfire}{\fire_\fullsym}
\newcommand{\sfiretwo}{\firetwo_\fullsym}

\newcommand{\vsub}{\mathsf{vsc}} 
\newcommand{\VSC}{\textnormal{VSC}\xspace}

\newcommand{\tovsub}{\Rew{\vsub}}

\newcommand{\tow}{\Rew{\wsym}} 
\newcommand{\towm}{\Rew{\wmsym}} 
\newcommand{\towe}{\Rew{\wsym{\esym}}} 

\newcommand{\osym}{{\mathsf o}}

\newcommand{\la}[1]{\lambda #1.}

\newcommand{\ictx}{R}

\newcommand{\ictxp}[1]{\ictx\ctxholep{#1}}

\newcommand{\rctx}{R}
\newcommand{\rctxtwo}{\rctx'}
\newcommand{\rctxthree}{\rctx''}
\newcommand{\rctxp}[1]{\rctx\ctxholep{#1}}
\newcommand{\rctxtwop}[1]{\rctxtwo\ctxholep{#1}}
\newcommand{\rctxthreep}[1]{\rctxthree\ctxholep{#1}}

\newcommand{\myproof}[1]{
\ifthenelse{\boolean{omitproofs}}{\begin{IEEEproof} Proof available but omitted for readability. \end{IEEEproof}}{#1}}

\newcommand{\valES}{\text{answer}\xspace}
\newcommand{\valESs}{\text{answers}\xspace}

\newcommand{\gregoire}{Gr{\'{e}}goire\xspace}

\newcommand{\withproofs}[1]{\ifthenelse{\boolean{withproofs}}{#1}{}}

\newcommand{\withoutproofs}[1]{\ifthenelse{\boolean{withproofs}}{}{#1}}

\newcommand{\NoteProof}[1]{
	\marginnote{{\normalfont\scriptsize{Proof\,p.\,{\pageref{#1}}\,}}}}
\newcommand{\NoteState}[1]{
	\marginnote{{\normalfont\scriptsize{See p.\,{\pageref{#1}}}}}}
\renewcommand{\NoteProof}[1]{\marginnote{{Proof\,p.\,{\pageref{#1}}}}}
\renewcommand{\NoteState}[1]{\marginnote{{See\,p.\,{\pageref{#1}}\\\cref{#1}}}}

\crefname{proposition}{Prop.}{Props.}
\crefname{theorem}{Thm.}{Thms.}
\crefname{lemma}{Lemma}{Lemmas}
\crefname{corollary}{Cor.}{Cors.}
\crefname{section}{Sect.}{Sects.}
\Crefname{section}{Section}{Sections}

\newcommand{\vsubcalc}{\lambda_\vsub}

\newcommand{\shufcalc}{\lambda_\shuf}

\newcommand{\plotsym}{\mathsf{Plot}}
\newcommand{\plotcalc}{\lambda_{\plotsym}}

\newcommand{\doubt}[1]{}

\newcommand{\letexp}{\mathsf{let}}

\newcommand{\lambdamucalc}{\overline\lambda\mu\tilde{\mu}}

\newcommand{\Rule}{\mathsf{r}}

\newcounter{numberone}
\newcounter{numberoneroman}
\newcounter{numberonealph}

\newcommand{\cbn}{CbN\xspace}

\newcommand{\cbv}{CbV\xspace}

\newcommand{\ocbv}{Open \cbv}

\newcommand{\ccbv}{Closed \cbv}
\newcommand{\scbv}{Strong \cbv}

\newcommand{\expr}{e}
\newcommand{\mset}[1]{[#1]}
\newcommand{\emptymset}{\mset{\,}}
\renewcommand{\emptymset}{\zero}
\newcommand{\zero}{\mathbf{0}}

\newcommand{\ltype}{\typefont{A}}
\newcommand{\ltypetwo}{\typefont{B}}

\newcommand{\typctx}{\Gamma}
\newcommand{\typctxtwo}{\Delta}
\newcommand{\typctxthree}{\Sigma}

\newcommand{\tder}{\pi}

\newcommand{\tderthree}{\rho}

\newcommand{\hastype}{\!:\!}

\newcommand{\domain}[1]{\mathsf{dom}(#1)}
\newcommand{\typelist}[1]{\hat{#1}}

\newcommand{\mytr}[1]{\underline{#1}}
\newcommand{\auxtr}[1]{\overline{#1}}

\makeatletter
\newcommand\Copy[2]{
        \marginpar{\scriptsize \ \ \hyperlink{hl-appendix-#1}{Proof p.\,{\pageref*{appendix-#1}}}}
	\immediate\write\@auxout{\unexpanded{\global\long\@namedef{mytext@#1}{#2}
  }}%
	#2%
}

\newcommand\Paste[1]{%
        \hypertarget{hl-appendix-#1}{}\label{appendix-#1}
	\renewcommand{\inappendix}[1]{}
	\ifcsname mytext@#1\endcsname
	\@nameuse{mytext@#1}%
	\else
	``??''
	\fi
	\renewcommand{\inappendix}[1]{#1}
}
\makeatother

\newcommand{\inappendix}[1]{#1}

\newcommand{\weakctx}{O}
\newcommand{\weakctxtwo}{\weakctx'}
\newcommand{\weakctxp}[1]{\weakctx\ctxholep{#1}}
\newcommand{\weakctxtwop}[1]{\weakctxtwo\ctxholep{#1}}

\newcommand{\openctx}{O}

\newcommand{\openctxp}[1]{\openctx\ctxholep{#1}}

\newcommand{\subctx}{\sctx}

\newcommand{\strongctx}{\extctx}
\newcommand{\strongctxtwo}{\strongctx'}

\newcommand{\strongctxp}[1]{\strongctx\ctxholep{#1}}
\newcommand{\strongctxtwop}[1]{\strongctxtwo\!\ctxholep{#1}}

\newcommand{\extctx}{S}

\newcommand{\extctxp}[1]{\extctx\ctxholep{#1}}

\newcommand{\larrow}[2]{#1 \multimap #2}
\newcommand{\ground}{\typefont{G}}

\newcommand{\Ax}{\mathsf{ax}}
\newcommand{\Es}{\mathsf{es}}

\newcommand{\derive}[2]{#1 \vartriangleright #2}
\newcommand{\concl}[4]{\derive{#1}{#2 \vdash #3 \hastype #4}}
\newcommand{\conclin}[4]{#1 \vartriangleright #2 \vdash^\infty #3 \hastype #4}

\newcommand{\subctxp}[1]{\subctx\ctxholep{#1}}

\newcommand{\Pointed}{Rigid\xspace}
\newcommand{\pointed}{rigid\xspace}
\newcommand{\ptm}{r}
\newcommand{\ptmtwo}{\ptm'}
\newcommand{\ptmthree}{\ptm''}

\newcommand{\rtm}{r}

\newcommand{\esssym}{\mathsf{x}}
\newcommand{\essesym}{\esssym\esym}
\newcommand{\essmsym}{\esssym\msym}

\newcommand{\toesse}{\Rew{\essesym}}
\newcommand{\toessm}{\Rew{\essmsym}}
\newcommand{\toess}{\Rew{\esssym}}

\newcommand{\fullsym}{{\mathsf{f}}}
\renewcommand{\fullsym}{{\mathsf{s}}}
\newcommand{\sitm}{\itm_\fullsym}
\newcommand{\sitmtwo}{\itmtwo_\fullsym}
\newcommand{\sitmthree}{\itmthree_\fullsym}

\newcommand{\sval}{\val_\fullsym}

\newcommand\Crumb\mytr
\newcommand\CrumbAux\auxtr

\newcommand{\evarsym}{\esym_{\mathsf{var}}}
\renewcommand{\rtoevar}{\rootRew\evarsym}
\renewcommand{\toevar}{\Rew\evarsym}

\makeatletter
\newcommand{\exder}{%
  \def\exderW[##1]{\triangleright_{##1}\ }%
  \def\exderWO{\triangleright\ }%
  \@ifnextchar[\exderW\exderWO%
  }
\makeatother

\newcommand{\tderiv}{\Phi}
\newcommand{\tderivtwo}{\Psi}
\newcommand{\tderivtwop}{\tderivtwo'} 
\newcommand{\tderivthree}{\Theta}

\newcommand{\typefont}[1]{{\mathsf{#1}}}
\newcommand{\mtype}{\typefont{M}}
\newcommand{\mtypetwo}{\typefont{N}}
\newcommand{\mtypethree}{\typefont{O}}

\newcommand\emptytype{\mathbf{0}}
\newcommand{\type}{\typefont{T}}
\newcommand{\imtype}{\inertop{\mtype}}
\newcommand{\imtypetwo}{\inertop{\mtypetwo}}
\newcommand{\inltype}{\inertop{\ltype}}
\newcommand\mplus{\uplus}

\newcommand{\tyjp}[4]{{#3} \vdash^{#1} #2 \hastype #4}
\newcommand{\namedtyjp}[5]{#1 \vartriangleright \tyjp{#2}{#3}{#4}{#5}}

\newcommand{\dom}[1]{\mathsf{dom}(#1)}

\newcommand{\ruleApp}{@}
\newcommand{\ruleFun}{\lambda}
\newcommand{\ruleES}{\mathsf{es}}
\newcommand{\ruleMany}{\mathsf{many}}
\newcommand{\ruleManyVar}{\ruleMany}
\newcommand{\ruleManyVal}{\ruleMany}
\newcommand{\ruleAp}{@}
\newcommand{\ruleAx}{\mathsf{ax}}

\newcommand{\I}{I}
\newcommand{\J}{J}
\renewcommand{\K}{K}
\newcommand{\iI}{{i \in \I}}
\newcommand{\jJ}{{j \in \J}}
\newcommand{\kK}{{k \in \K}}
\newcommand{\iN}{{i \in \set{1,\mydots,n}}}

\newcommand{\tarrow}[2]{#1 \multimap #2}
\newcommand{\ty}[2]{\tarrow{#1}{#2}}

\newcommand{\mult}[1]{[ #1 ] }

\newcommand{\bigmplus}{\biguplus}

\newcommand{\Id}{{\mathsf{I}}}

\newcommand{\unitarysym}{{\textsc{u}}}

\newcommand{\inertop}[1]{#1^{\mathsf{i}}}

\newcommand{\rightsym}{{\textsc{r}}}
\newcommand{\rltype}{\ltype^{\rightsym}}

\newcommand{\rmtype}{\mtype^{\rightsym}}

\newcommand{\leftsym}{{\textsc{l}}}
\newcommand{\lltype}{\ltype^\leftsym}

\newcommand{\lmtype}{\mtype^\leftsym}

\newcommand{\urltype}{\ltype^{\unitarysym\rightsym}}

\newcommand{\urmtype}{\mtype^{\unitarysym\rightsym}}

\newcommand{\ulltype}{\ltype^{\unitarysym\leftsym}}

\newcommand{\ulmtype}{\mtype^{\unitarysym\leftsym}}

\newcommand{\myparagraph}[1]{\emph{#1.}}
\newcommand{\myparagraphsp}[1]{\medskip

\myparagraph{#1}}

\newcommand\mydots{\hbox to .6em{.\hss.}}
\newcommand\comp{\circ}

\newcommand{\tderiveq}{\sim}
\newcommand{\groundset}{S}
\newcommand{\groundsettwo}{\groundset'}
\newcommand{\groundsetthree}{\groundset''}

\renewcommand{\tmthree}{s}
\renewcommand{\extctx}{X}

\begin{document}
%
\title{Semantic Bounds and\\ Strong Call-by-Value Normalization}

\author{\IEEEauthorblockN{Beniamino Accattoli}
\IEEEauthorblockA{Inria \& \'Ecole Polytechnique}
\and
\IEEEauthorblockN{Giulio Guerrieri}
\IEEEauthorblockA{University of Bath}
\and
\IEEEauthorblockN{Maico Leberle}
\IEEEauthorblockA{Inria \& \'Ecole Polytechnique}}

\onecolumn
\maketitle

\begin{abstract}
This paper explores two topics at once: 
the use of denotational semantics to bound the evaluation length of functional programs, and the semantics of strong (that is, possibly under abstractions) 
call-by-value evaluation.

About the first, we analyze de Carvalho's seminal use of relational semantics for bounding the evaluation length of $\l$-terms, starting from the presentation of the semantics as an intersection types system. We focus on the part of his work which is usually neglected in its many recent adaptations, despite being probably the conceptually deeper one: how to transfer the bounding power from the type system to the relational semantics itself. We dissect this result and re-understand it via the isolation of a simpler \emph{size representation property}. 

About the second, we use relational semantics to develop a semantical study of strong call-by-value evaluation, which is both a delicate and neglected topic. We give a semantic characterization of terms normalizable with respect to strong evaluation, providing in particular the first result of adequacy with respect to strong call-by-value. Moreover, we extract bounds about strong evaluation from both the type systems and the relational semantics.

Essentially, we use strong call-by-value to revisit de Carvalho's semantic bounds, and de Carvalho's technique to provide semantical foundations for strong call-by-value.

%
\end{abstract}


\section{Introduction}
\label{sect:intro}
In this paper we are concerned with two topics of the theory of $\l$-calculus, in which linear logic plays a prominent role. The first one is the use of denotational semantics to obtain bounds on operational aspects of $\l$-terms. The second one is the semantics of call-by-value (shortened to \cbv) evaluation.

Both these topics have a literature of their own. For different reasons, and in different ways, however, neither of them has found clean and solid grounds. 
We aim to use the two topics, together with linear logic technology, to clarify each other.

\subsection{De Carvalho's Semantical Bounds} 
Denotational semantics studies invariants of evaluation. The typical way in which it is connected to the operational semantics of $\l$-calculi is at the \emph{qualitative} level, via \emph{adequacy}: the denotational interpretation $\sem\tm$ of a $\l$-term $\tm$ is non-trivial (typically non-empty) if and only if the evaluation~of~$\tm$~terminates.

At first sight, denotational semantics cannot provide \emph{quantitative} operational insights such as evaluation lengths, because of its invariance by evaluation. The question is in fact subtler. Being invariant by evaluation, denotational semantics models normal forms, 
and in a \emph{compositional} way: 
by composing the interpretations of two terms one can obtain the interpretation of the result of their application---therefore, denotational semantics does reflect the evaluation process~\emph{somehow}.

This intuition guided de Carvalho in its use of linear logic relational semantics---probably the simplest denotational model of linear logic---to not only characterize semantically termination, but also obtain exact bounds on the evaluation lengths of $\l$-terms \cite{Carvalho07,deCarvalho18}. 
A key aspect of his work is that he uses a type system to give a syntactic presentation of relational semantics. The types are the \emph{non-idempotent} variant of intersection types, also known as \emph{multi types}, because non-idempotent intersections are multi sets. The interpretation $\sem\tm$ of a $\l$-term $\tm$ in relational semantics is given by the types (including the typing context $\typctx$) in all the judgments $\typctx \vdash \tm \hastype \ltype$ that can be derived for $\tm$. It is important to stress that relational semantics can also be presented independently of multi types. 

\paragraph*{Three Kinds of Bounds} In \cite{Carvalho07,deCarvalho18}, de Carvalho measures two forms of strong evaluations realized by abstract machines implementing the head and leftmost call-by-name strategies, organizing his results in two parts. In this overview, we prefer to slightly depart from the details of his work, forgetting about abstract machines, focussing on leftmost evaluation, and isolating three (rather than two) kinds of bounds:
\begin{enumerate}
\item \emph{Bounds from type derivations}: given a type derivation $\concl{\tderiv}{\typctx}{\tm}{\ltype}$, he shows that the size $\size\tderiv$ of the derivation bounds the number of leftmost steps from $\tm$ to its normal form $\nf{(\tm)}$ \emph{plus} the size $\size{\nf{(\tm)}}$ of the normal form. Moreover, minimal derivations provide exact bounds. 

\item \emph{Size bounds from types}: the types in the final judgment---
a point of the relational interpretation $\sem\tm$---also provide a bound, independently of the derivation $\tderiv$. Being invariant by evaluation, they cannot bound evaluation lengths. They do however bound the size $\size{\nf{(\tm)}}$ of the normal form, and bounds are exact when types are minimal.

\item \emph{Bounds from composable types}: de Carvalho shows that types can be used to bound evaluation lengths \emph{compositionally}, \ie, from the types in $\sem\tm$ and $\sem\tmtwo$ it is possible to extract bounds about the leftmost evaluation and the normal form of $\tm\tmtwo$, with \emph{no reference to type derivations}.
\end{enumerate}
The third kind is where the bounding power of the type system 
(just one possible way of defining relational semantics) is fully \emph{lifted} to relational semantics. Therefore, the lifting guarantees that the bounding power is an inherent feature of the relational model---multi types are just a handy tool to show it.

Of the three results, the third one is also the more technical one. In particular, it requires to enrich the type system with an infinity of ground types and work up to type substitutions. The somewhat puzzling fact is that  these extra features play no role in the  two previous points.

\paragraph*{De Carvalho's Legacy} De Carvalho developed his results in his PhD defended in 2007 \cite{Carvalho07}, known by the community thanks to a technical report that was eventually published much later \cite{deCarvalho18}.
Soon after his PhD, he adapted his work to linear logic, 
with Pagani and Tortora de Falco \cite{DBLP:journals/tcs/CarvalhoPF11,DBLP:journals/iandc/CarvalhoF16}. 
A few years later, Bernadet and Graham-Lengrand adapted his work to measure the longest evaluation in the $\l$-calculus \cite{DBLP:journals/corr/BernadetL13}, but they only provided bounds of kind 1 (\emph{from type~derivations}).

At the time, it was not known whether it would make sense to count the number of $\beta$-steps (or linear logic cut-elimination steps) as a reasonable measure of complexity. After this was clarified (in the positive) by Accattoli and Dal Lago \cite{DBLP:journals/corr/AccattoliL16}, de Carvalho's work has been revisited by Accattoli, Graham-Lengrand, and Kesner in 2018 \cite{DBLP:journals/pacmpl/AccattoliGK18}. The revisitation started a new wave of works adapting de Carvalho's study to many evaluation strategies and extensions of the $\l$-calculus, including call-by-value \cite{DBLP:conf/aplas/AccattoliG18}, call-by-need \cite{DBLP:conf/esop/AccattoliGL19}, a linear logic presentation of call-by-push-value \cite{DBLP:conf/flops/BucciarelliKRV20}, the $\lambda\mu$-calculus \cite{DBLP:conf/lics/KesnerV20}, the $\l$-calculus with pattern matching \cite{DBLP:conf/types/AlvesKV19}, the probabilistic $\l$-calculus \cite{DBLP:journals/pacmpl/LagoFR21}, and the abstract machine underlying the geometry of interaction \cite{DBLP:journals/pacmpl/AccattoliLV21}. All these works provide bounds of kind 1,  and some of them also of kind 2, but none of them deals with those of kind 3 (\emph{bounds from composable types}).

\paragraph*{Contributions of This Paper for Semantic Bounds} We adapt de Carvalho's study to 
another setting, Strong \cbv (surveyed below). 
Contrary to the new wave of similar works, we provide all three kinds of bounds, including \emph{bounds from composable types}, 
appeared only in \cite{deCarvalho18,DBLP:journals/tcs/CarvalhoPF11}. Beyond proving the results, we have a very close look at this part, isolating the subtleties and decomposing 
it in smaller steps. In particular: 
\begin{itemize}
	\item \emph{Subtlety 1, minimality does not work}: when bounding a single term, both derivations and types provide exact bounds when they are minimal. When dealing with the application of $\tm$ to $\tmtwo$, every pair of composable types for them provides bounds. The minimal composable pair, however, does not provide exact bounds.
	\item \emph{Subtlety 2, the gap between derivations and types}: the previous subtlety stems from the fact that, for a normal term $\tm$, both the derivation $\concl{\tderiv}{\typctx}{\tm}{\ltype}$ and the types in $\typctx$ and $\ltype$ provide bounds for $\size\tm$, but they may not coincide. In general, the bound from types is laxer. The bounds gap hinders the possibility of lifting the bounding power from derivations to types, if bounds from some derivations are not reflected by any type in the interpretation $\sem\tm$. 
	\item \emph{Size representation property}: luckily, given $\concl{\tderiv}{\typctx}{\tm}{\ltype}$ for which there is a gap between the bound from $\tderiv$ and the bound from $\typctx$ and $\ltype$, there always is a second derivation $\concl{\tderivtwo}{\typctxtwo}{\tm}{\ltypetwo}$, whose types $\typctxtwo$ and $\ltypetwo$ give the same bound as the first derivation $\tderiv$. Then, all bounds coming from derivations can also be seen as coming from types, potentially from the types of other derivations.
\end{itemize}
Size representation is the key property where the lifting of the bounding power from type derivations to types---and thus to relational semantics---takes place. On the good side, it does not need an infinity of ground types nor type substitutions, so it is simpler than what done in \cite{deCarvalho18,DBLP:journals/tcs/CarvalhoPF11}. On the bad side, it only implies a weaker form of de Carvalho's final results. 

We then show that, by enriching the type system as de Carvalho does, one can reinforce the size representation property into a \emph{dissection property} describing how to obtain $\tderiv$ from $\tderivtwo$ via a type substitution. Finally, the dissection property implies de Carvalho's \emph{bounds from composable types}. 

The refinement, however, is of a technical nature, and does not rest nor introduce any deep concept. Therefore, we claim that the key point in lifting the bounds from the type system to the relational model is the size representation property.

\subsection{The Semantics of Strong Call-by-Value}
\label{subsect:intro-cbv}

Plotkin's call-by-value $\l$-calculus \cite{DBLP:journals/tcs/Plotkin75} is at 
the heart  of 
programming languages such as OCaml and proof assistants such as Coq. In the study of programming 
languages, call-by-value (\cbv) evaluation is usually \emph{weak}, that is, it does not reduce under 
abstractions, and terms are assumed to be \emph{closed}, \ie , without free variables.
These constraints give rise to an elegant framework---we call it \emph{Closed \cbv},  following 
\cite{DBLP:conf/aplas/AccattoliG16}.

It often happens, however, that one needs to go beyond Closed \cbv  by 
considering \emph{Strong \cbv}, which is the extended setting where reduction under abstractions is allowed 
and terms may be open, or the intermediate framework of \emph{Open \cbv}, where evaluation is weak 
but terms are not necessarily closed. 
The need arises, most notably, when describing the implementation model of Coq, as done by \gregoire and Leroy
\cite{DBLP:conf/icfp/GregoireL02}, to realize the essential conversion test for dependent types. Other 
motivations lie in the study of bisimulations by Lassen
\cite{DBLP:conf/lics/Lassen05}, partial evaluation \cite{Jones:1993:PEA:153676}, or various topics 
of a semantical or logical nature, recalled below.

\myparagraphsp{Na\"ive Extension of \cbv} In call-by-name (\cbn) turning to open terms  or strong 
evaluation is harmless because \cbn does not impose any special form to the arguments of 
$\beta$-redexes. 
On the contrary, turning to Open  or Strong \cbv is delicate. While some fundamental properties such 
as confluence and standardization hold also in such cases, as showed by Plotkin's himself \cite{DBLP:journals/tcs/Plotkin75}, 
others---typically of a semantical nature---break as soon as one considers open terms. 

The problems of \scbv can be traced back to Plotkin's seminal paper,  
where he points out the 
incompleteness of \cbv with respect to CPS translations, an issue later solved with categorical 
tools by Moggi \cite{DBLP:conf/lics/Moggi89}. This 
led to a number of studies, among others
\cite{DBLP:journals/lisp/SabryF93,DBLP:journals/toplas/SabryW97,DBLP:journals/tcs/MaraistOTW99,
DBLP:conf/icfp/CurienH00,DBLP:journals/logcom/DyckhoffL07,DBLP:conf/tlca/HerbelinZ09}, that 
introduced many proposals of 
improved calculi~for~\cbv.

The relationship with denotational semantics is also problematic, as  first shown by Paolini and Ronchi della Rocca
\cite{DBLP:journals/ita/PaoliniR99,DBLP:conf/ictcs/Paolini01,parametricBook}. There are two subtle 
points: 
\begin{enumerate}
\item \emph{Solvability}: the adaptation of the notion of solvability---roughly, 
a form of 
meaningfulness for terms---to \cbv;
\item \emph{Adequacy}: denotational semantics that are 
\emph{adequate} for Closed \cbv \cite{DBLP:conf/csl/AbramskyM97,DBLP:journals/fuin/EgidiHR92,DBLP:journals/tcs/HondaY99,DBLP:journals/mscs/PravatoRR99} are no longer adequate for the extended settings. Roughly, 
there are terms that are semantically divergent, that is, with 
trivial semantics (or unsolvable), 
while they are normal forms with respect to Plotkin's rules, and 
so are expected to have non-trivial 
\end{enumerate}
Let us stress that, for a \cbv model, adequacy\footnote{In the literature, \eg~\cite{DBLP:journals/fuin/ManzonettoPR19}, the word \emph{adequacy} can have different meanings.} is somewhat mandatory, because any model of \cbn provides a non-adequate model of \cbv that does not model the \cbv behavior.

The problem with adequacy can be seen also from a logical point of view. The terms pointed out by Paolini and Ronchi Della Rocca diverge also if seen as 
 linear logic proof nets, as pointed out by Accattoli
\cite{DBLP:journals/tcs/Accattoli15}, or as terms in the computational interpretation of sequent 
calculus by Curien and Herbelin \cite{DBLP:conf/icfp/CurienH00}. 

\paragraph*{Linear Logic and the Semantical Issues of Open/Strong \cbv} Both semantical issues of Open/Strong \cbv have been addressed in the literature, relying of linear logic tools. 

About solvability, Accattoli and Paolini \cite{AccattoliPaolini12} 
characterize operationally solvable terms using a calculus isomorphic to the proof-net \cbv representation of $\l$-calculus, the \emph{value substitution calculus} (shortened to \VSC).
Namely, they introduce a \emph{solvable evaluation strategy} (called stratified-weak in \cite{AccattoliPaolini12}) in the \VSC that terminates if and only if the term is solvable. 
This is akin to what happens in \cbn, where solvable term are 
those for which head evaluation terminates; 
such a characterization is impossible in Plotkin's original formulation of \cbv. 

About adequacy, Accattoli and Guerrieri \cite{DBLP:conf/aplas/AccattoliG18} showed that (a calculus equivalent to) the \emph{open} sub-calculus of the \VSC is adequate with respect to Ehrhard's \cbv relational semantics \cite{DBLP:conf/csl/Ehrhard12}, the \cbv analogous of the system used by de Carvalho. Moreover, they adapted de Carvalho's quantitative bounds (namely, of kind 1 and 2) to \ocbv.

\myparagraphsp{Asymmetry Between \cbn and \cbv} The described results---the solvable strategy characterizing \cbv solvability, and adequacy of relational semantics for \ocbv---show that a semantical theory of \cbv beyond the closed case is possible. Much still needs to be done, though, 
to have a solid theory and to close the gap with the very well studied \cbn case.


In particular, concerning adequacy one would expect a similar result also for evaluation in \scbv, as it is the case for strong \cbn---there are however no such results in the literature. Even worse about the quantitative refinement: while there are strategies in the literature that are fully normalizing for \scbv---essentially the leftmost-outermost strategy---they are formulated in calculi 
for \scbv (\eg~\cite{Guerrieri15,DBLP:journals/lmcs/GuerrieriPR17,DBLP:conf/csl/Santo20}) 
whose number of steps cannot be measured via multi types nor can be taken as a reasonable time cost model (see below).

\paragraph*{Subtleties of Strong \cbv}
Let us clarify a basic point of Strong \cbv. 
	A \emph{value} is a
variable or an abstraction. In \cbv, $\beta$-redexes can be fired only when the argument is a value (and so only values can be duplicated or erased).
In \ccbv, values are normal forms---it can be thought as a \emph{call-by-normal-form} calculus. 
A
similar property can be recovered also in \ocbv (by slightly extending the notion of value), see \cite{DBLP:conf/aplas/AccattoliG16}. 
In Strong \cbv such an essence is lost. To give an idea, if $\Omega$ denotes
the usual diverging term, then in Strong \cbv $\la{\var}.\Omega$ diverges (evaluation under abstraction) while
$(\la{\vartwo}\varthree) (\la{\var}.\Omega)$ may both diverge and reduce to $\varthree$---the argument need not be fully normalized, but
 only reduced to a (weak) value. 
This is mandatory, to be conservative over Closed and Open CbV.
Thus, there is an inherent underlying tension in Strong CbV: on the one hand, values keep
playing an essential role to identify which $\beta$-redexes can be fired, on the other hand they no longer
suffice to describe evaluation and normal forms, that become more sophisticated notions.

\smallskip
\paragraph*{Contributions of This Paper for the Semantics of Strong \cbv} We provide a quantitative study of \scbv via multi types. Our contributions are:
\begin{itemize}
\item \emph{Normalizing strong strategy}: we define a strong \VSC strategy generalizing leftmost-outermost evaluation, named \emph{external strategy}, and prove it normalizing, that is, it reaches a \VSC normal form whenever there is one. 
\item \emph{Normalization, multi types, and adequacy}: we give a characterization of normalizable terms via multi types, obtaining the adequacy theorem.

\item \emph{Semantical bounds}:  we show the three kinds of de Carvalho's bounds with respect to the external strategy. In particular, we dissect the \emph{bounds from composable types} kind, as explained in the first part of the introduction.
\end{itemize}


\paragraph*{Reasonable Cost Models} In a companion paper by Accattoli, Condoluci, and Sacerdoti Coen \cite{LICS2021machine}, it is proved that the number of steps of the external strategy is a reasonable time cost model for \scbv (where \emph{reasonable} means polynomially related to the time cost model of Turing machines). The proof is based on a novel abstract machine that implements the external strategy within a overhead which is bilinear (that is, linear in the number of $\beta$/multiplicative steps and in the size of the initial term). 
When combined with the study of this paper, we obtain that multi types provide quantitative bounds that are meaningful from a computational complexity point of view.

\smallskip
\paragraph*{Syntactic Variants} We study \scbv via the VSC, to stress the linear logic background and foundation of our study. Everything could equivalently be  easily  reformulated inside (the intuitionistic and \cbv fragment of) Curien and Herbelin's
$\lambdamucalc$-calculus, following the isomorphism with the VSC developed in \cite{DBLP:conf/aplas/AccattoliG16} 
and preserving the number of steps/cost model. The external strategy can also be reformulated using the \emph{shuffling calculus} $\shufcalc$ by Guerrieri and Carraro
\cite{DBLP:conf/fossacs/CarraroG14}, a calculus with commuting conversions used recently also by Manzonetto, Pagani, and Ronchi della Rocca
\cite{DBLP:journals/fuin/ManzonettoPR19}. The cost model of $\shufcalc$, however, is unclear (see 
\cite{DBLP:conf/aplas/AccattoliG16}) and commuting conversion steps cannot be counted via multi 
types (but they should)---$\shufcalc$ provides a qualitative foundation for \scbv, but it cannot be used for a 
quantitative one. Calculi with $\letexp$-commutation rules such as the one by Herbelin and Zimmerman \cite{DBLP:conf/tlca/HerbelinZ09} simply 
can be seen as subcalculi of the VSC (up to structural equivalence, see \cite{AccattoliPaolini12}). Similar remarks apply to many other \cbv calculi
\cite{DBLP:conf/lics/Moggi89,DBLP:journals/lisp/SabryF93,DBLP:journals/toplas/SabryW97,DBLP:journals/tcs/MaraistOTW99, 
DBLP:journals/logcom/DyckhoffL07,DBLP:conf/csl/Santo20}. The key point of the VSC (valid also 
in $\lambdamucalc$) is that it does not need any commuting rule. 

\smallskip
\paragraph*{Summing Up} We study \scbv using a calculus, the VSC, rooted in linear logic. We provide a solid foundation of \scbv with respect to multi types and relational semantics, analogous to the one for \cbn, by characterizing qualitatively and quantitatively \scbv normalization. Moreover, we fully express the quantitative aspect of relational semantics by adapting de Carvalho's semantic bounds and re-understanding them via the isolation of the key \emph{size representation} property. 

\smallskip
\paragraph*{Proofs} Proofs are in the Appendix. 

\section{Value Substitution Calculus}
\label{sect:vsc}

Here we present the \emph{value substitution calculus} (\VSC for short) introduced by Accattoli and Paolini \cite{AccattoliPaolini12}, and we recall some properties.
The \VSC is a $\lambda$-calculus with let-expressions whose reduction rules mimic cut-elimination on proof-nets, via Girard's \cbv translation $(A \Rightarrow B)^v = \oc (A^v \multimap B^v)$ of intuitionistic logic into linear logic, as explained in \cite{DBLP:journals/tcs/Accattoli15}.  

In \VSC, $\beta$-redexes are decomposed via let-expressions, and the
\emph{by-value} restriction on evaluation is on the let-substitution rule, not on $\beta$-redexes, because only values
can be substituted.
A let-expression is formulated as an \emph{explicit substitution} or \emph{sharing} (\ES for short) $\tm\esub{\var}{\tmtwo}$ which binds $\var$ in $\tm$.
All along the
paper we use (many notions of) \emph{contexts}, \ie terms with a hole, noted $\ctxhole$. For now, we need \emph{substitution contexts} $\sctx$, which are simply lists of \ES. The grammars are:
\begin{center}
$\arraycolsep=3pt\begin{array}{rrl}
\textsc{Values } & \val & \grameq \la\var\tm 
\\
\textsc{Terms } & \tm,\tmtwo, \tmthree & \grameq \var \mid \val \mid \tm\tmtwo 
\mid \tm \esub\var\tmtwo 
\\
\textsc{Substitution Ctxs } &\subctx & \grameq \ctxhole \mid \subctx \esub\var\tm
\end{array}$
\end{center}

The set of free variables of term $\tm$ is denoted by $\fv{\tm}$. 
Plugging a term $\tm$ in a context $\ctx$ is noted $\ctxp{\tm}$, possibly~capturing variables.
An \emph{answer} is a term of the shape $\sctxp\val$, where $\val$ is a value (\ie an abstraction) and $\sctx$ is a substitution context.
We use $\tm \isub{\var}{\tmtwo}$ for the capture-avoiding substitution of
$\tm$ for each free occurrence of $\var$ in $\tm$.
There are two kinds of rewrite rules, both work \emph{at a distance}, that is, up to a substitution context. 
\begin{center}
$\begin{array}{rr@{\ }l@{\ }l}
    \textsc{Multiplicative rule} & \subctxp{\la\var\tm}\tmtwo &  \rtom  & \subctxp{\tm\esub{\var}{\tmtwo}} \\
    \textsc{Exponential rule}  & \tm\esub\var{\subctxp{\val}} &  \rtoe  & \subctxp{\tm\isub{\var}{\val}} 
\end{array}$    
\end{center}

We shall consider three fragments of the  \VSC. 
They all contain the terms of \VSC, they differ only in the choice of evaluation contexts for the rewrite rules.

\smallskip
\paragraph*{The Open VSC} We first focus on the open fragment of the \VSC, where rewriting is forbidden under abstraction and terms are possibly open (but not necessarily). This fragment has a nice inductive description of its normal forms, called fireballs, that is the starting point for many other definitions in the paper. Open contexts and rules are defined as follows.
\begin{center}
		$\begin{aligned}
		\textsc{Open ctxs} && \weakctx & \grameq \ctxhole \mid \weakctx \tm \mid \tm \weakctx \mid \weakctx \esub\var\tm \mid 
\tm\esub\var \weakctx 
		\end{aligned}$
		\smallskip
		
		\begin{tabular}{ccc}
\textsc{Open rewrite rules}:
			&
			\multirow{2}{*}{\begin{prooftree}
					\hypo{\tm \rootRew{a} \tm'}		
					\infer1{\weakctxp{\tm} \Rew{\wsym a} \weakctxp{\tm'}}
			\end{prooftree}}
		\\
		($a \in \set{\msym,\esym}$)
	\end{tabular}\medskip

		$\begin{array}{cccc}
\textsc{Open reduction}:& \tovsubo  \, \defeq \, \tomo \cup \toeo
	\end{array}$
\end{center}

\begin{proposition}[Properties of the open reduction]
	\label{prop:ovsc-diamond}\label{prop:properties-open-reduction}
	\NoteProof{propappendix:properties-open-reduction}
	\begin{enumerate}
		\item \label{p:properties-open-reduction-diamond} $\tovsubo$ is diamond; $\tomo$ and $\toeo$ strongly commute.
		\item \label{p:properties-open-redction-harmony} A term is $\osym$-normal if and only if it is a fireball, where \emph{fireballs} (and \emph{inert terms}) are defined by:
	\end{enumerate}
\begin{center}
	$\begin{array}{r r@{\ } l}
	\textsc{Inert term} & \itm &\grameq \var \mid \itm \fire \mid \itm \esub{\var}{\itmtwo}
	\\
	\textsc{Fireball} & \fire &\grameq \val \mid \itm \mid \fire \esub{\var}{\itm}
	\end{array}$
\end{center}
\end{proposition}

Diamond of $\tovsubo$ and strong commutation of $\tomo$ and $\toeo$ are technical facts (their definitions are in \Cref{sect:preliminaries} in the Appendix) with relevant consequences:
$\tovsubo$ is confluent and its non-determinism is only apparent, because if an $\osym$-evaluation from $\tm$ reaches
a $\osym$-normal form $\tmtwo$, then \emph{every} $\osym$-evaluation from $\tm$ eventually ends in $\tmtwo$. And all these $\osym$-evaluations have the same length and same number of $\msym$-steps and $\esym$-steps.  
Same properties hold for other strategies we shall consider. 


%

For our quantitative study we need a notion of size specific for every fragment of the VSC. The \emph{open size} $\sizeo{\tm}$ of a term $\tm$ is its number of applications out of abstractions, \ie 
	\begin{align*}
	\sizeo{\var} &\defeq 0 
	& 
	\sizeo{\tm\tmtwo} & \defeq  \sizeo{\tm} + \sizeo{\tmtwo} + 1 
	\\
	\sizeo{\la{\var}{\tm}} & \defeq  0
	&
	\sizeo{\tm \esub{\var}{\tmtwo}} & \defeq  \sizeo{\tm} + \sizeo{\tmtwo}.
	\end{align*}

The \emph{Closed \VSC} is the restriction of the Open \VSC to closed terms. Its normal forms are the closed values.

	\smallskip
\paragraph*{The Strong VSC} The \Full \VSC is obtained by allowing rewriting rules  everywhere. 
\begin{center}
		$\begin{array}{l r rclccc}
		\textsc{\Full ctxs} & \fctx  \grameq \ctxhole \mid \fctx \tm \mid \tm \fctx \mid \la{\var}{\fctx} \mid 
\fctx \esub\var\tm \mid \tm \esub \var \fctx 
		\end{array}$\medskip

		\begin{tabular}{ccc}
\textsc{\Full rewrite rules:}
			&
			\multirow{2}{*}
			{\begin{prooftree}
    \hypo{\tm \rootRew{a} \tm'}		
    \infer1{\fctxp\tm \Rew{a} \fctxp{\tm'}}
			\end{prooftree}}
		\\
		($a \in \set{\msym,\esym}$)
	\end{tabular}\medskip

		$\begin{array}{r l@{\ } l@{\ } lllllll}
  \textsc{\Full reduction:} & \tovsub  & \defeq  &\tom \cup \toe
\end{array}$
\end{center}
Unlike the previous cases, $\tovsub$ is not diamond: consider all the $\vsub$-evaluations of $(\var\var) \esub{\var}{\la{\vartwo} \Id \Id}$, with $\Id \defeq \la{\varthree}\varthree$.  

\begin{proposition}[Properties of the \full reduction] 
\label{thm:vsc-confluence}\label{prop:properties-full-reduction}
	\NoteProof{propappendix:properties-full-reduction}
\hfill
\begin{enumerate}
	\item \label{p:properties-full-reduction-confluence} The reduction $\tovsub$ is confluent.
	
	\item \label{p:properties-full-reduction-harmony} A term is $\vsub$-normal if and only if it is a \full fireball, where \emph{\full fireballs} (and \emph{\full inert terms}, \emph{\full values}) are:
\end{enumerate}
\begin{center}
	$\begin{array}{rl}
	\textsc{\Full inert terms } 
	&
	\sitm \grameq \var \mid \sitm \sfire \mid \sitm \esub\var{\sitmtwo}
	\\
	\textsc{\Full values } 
	&
	\sval  \grameq  \la\var\sfire
	\\
	\textsc{\Full fireballs } 
	&
	\sfire \grameq \sitm \mid \sval \mid \sfire \esub\var\sitm
	\end{array}$
\end{center}
\end{proposition}
The notions of \full inert terms and \full fireballs are a generalization of inert terms and fireballs, respectively, by simply iterating the construction under all abstractions.
%
%
Note that they are similar to normal forms of the (\cbn) $\l$-calculus, but they can have \ES containing \full inert terms. 


The \emph{\full size} $\sizefu{\tm}$---or simply size---of a term $\tm$ is its number of applications and abstractions, \ie 
	\begin{align*}
		\sizefu{\var} &\defeq 0 
		& 
		\sizefu{\tm\tmtwo} & \defeq  \sizefu{\tm} + \sizefu{\tmtwo} + 1 
		\\
		\sizefu{\la{\var}{\tm}} & \defeq  \sizefu{\tm} + 1 
		&
		\sizefu{\tm \esub{\var}{\tmtwo}} & \defeq  \sizefu{\tm} + \sizefu{\tmtwo}.
	\end{align*}

\paragraph*{On the Value of Variables}
In Plotkin's original work \cite{DBLP:journals/tcs/Plotkin75}, and in works of a theoretical nature, values are abstractions or variables.
In implementative works, variables are excluded from values. 
The difference is not only terminological, but also semantic because it affects whether exponential step can fire $\tm\esub{\var}{\vartwo}$ or not.
Accattoli and Sacerdoti Coen \cite{DBLP:journals/iandc/AccattoliC17} showed that excluding variables from values brings a speed up in implementations, which explains the discrepancy. 
Actually, from a theoretical viewpoint, variables have a double nature as both values and (\full) inert terms, in that all the statements concerning types for values hold with or without variables and the same for statements concerning types for inert terms.
We chose to define values as only abstractions and to include variables in the (\full) inert terms essentially for two reasons.

First, for (\full) inert terms we shall prove stronger claims about typability to have the right inductive hypothesis.
So, seeing variables as (\full) inert terms is somehow \emph{more general}. 

Second, considering variables as values so that exponential steps can also fire \ES with variables is \emph{irrelevant} with respect to evaluation.
Let us see that more in detail.
We set:
\begin{center}
	$\begin{array}{l r rclccc}
	& 
	\!\!\tm\esub\var{\subctxp{\vartwo}}  \rtoevar\!  \subctxp{\tm\isub{\var}{\vartwo}} 
	& \ 
	{\begin{prooftree}
		\hypo{\tm \rootRew{\evarsym} \tm'}		
		\infer1{\fctxp\tm \Rew{\evarsym}\! \fctxp{\tm'}}
		\end{prooftree}}
\end{array}$
	
%
\end{center}

\begin{proposition}[Irrelevance of $\toevar$]
	\label{prop:irrelevance}
	\NoteProof{propappendix:irrelevance}
	\hfill
	\begin{enumerate}
		\item\label{p:irrelevance-postponement} \emph{Postponement}: whenever $\deriv \colon \tm \, (\tovsub \cup \toevar)^*\, \tmtwo$ then there is $\derivp \colon \tm \tovsub^* \!\cdot \toevar^* \, \tmtwo$ with $\sizem\derivp = \sizem \deriv$;
		\item\label{p:irrelevance-termination} \emph{Termination}: a term is weakly (\resp strongly) normalizing for $\tovsub$ if and only if so is it for $\tovsub \!\cup \toevar$.
	\end{enumerate}
\end{proposition}

\paragraph*{Plotkin \textit{vs} \VSC}

Plotkin's original \cbv $\lambda$-calculus $\plotcalc$ \cite{DBLP:journals/tcs/Plotkin75} can be easily simulated in the \VSC.
The syntax of $\plotcalc$ is simply the same as in the \VSC but without \ES.
The reduction in $\plotcalc$ is the closure under \full contexts (without \ES) of the rule $(\la{\var}\tm)\tval \rtobvplot \tm\isub{\var}{\tval}$,
where $\tval$ is a \emph{theoretical value}, \ie a variable or an abstraction.
Since $\tobvplot$ fires also variables, the simulation requires also $\evarsym$-steps, but we have seen that this extension is harmless in \VSC (\Cref{prop:irrelevance}).
\begin{proposition}[Simulation]
	\label{prop:plotkin-vsc}
	\NoteProof{propappendix:plotkin-vsc}
	Let $\tm$ be a term without \ES. 
	\begin{enumerate}
		\item\label{p:plotkin-vsc-step}  	If $\tm \tobvplot \tm'$ then $\tm \tom \!\cdot\, (\toe \cup \toevar) \ \tm'$.
		\item\label{p:plotkin-vsc-steps}  If $\tm \tobvplot^* \tm'$ then $\tm \tovsub^* \tmtwo  \toevar^* \tm'$ for some term $\tmtwo$.
	\end{enumerate}
\end{proposition}

There is no sensible way to simulate \VSC into $\plotcalc$.
Indeed \VSC is a proper extension of $\plotcalc$: \VSC makes divergent terms such as $(\la{\var}\delta)(\vartwo\vartwo)\delta$ and $\delta ((\la{\var}\delta) (\vartwo\vartwo))$ that are $\betaplot$-normal.


\section{The External Strategy}
\label{sect:external}

Here we define the \full strategy, called \emph{external}, that we shall show to be $\vsub$-normalizing, as a corollary of our type theoretical study. 
Its role is analogous to the leftmost-outermost 
strategy of the $\l$-calculus. A notable difference, however, is that the external strategy is itself non-deterministic, but in a 
harmless way, as it is diamond. 

We need a few notions. First, \emph{rigid terms}, 
\ie the variation over \full inert terms where the arguments of the head variable can be whatever term (note that \valESs are not rigid terms). 
\begin{center}
	$\textsc{\Pointed terms} \quad \ptm, \ptmtwo  \grameq  \var \mid \ptm\tm \mid \ptm\esub\var \ptmtwo$
\end{center}
Every (\full) inert term is a \pointed term, but the converse does not hold---consider $\vartwo({\delta\Id})$ where $\delta \defeq \la{\var}{\var\var}$ and $\Id = \la{\varthree}{\varthree}$.

Then we need evaluation contexts for the external strategy $\toess$. The base case is given by the open rewriting rules (themselves defined via a closure by open contexts, see \Cref{sect:vsc}), which are then closed by \emph{external (evaluation) contexts}, 
defined by mutual induction with \emph{rigid contexts}. 
\begin{center}
		$\begin{array}{r@{\quad} r@{\ } llcc}
		\textsc{External ctxs} & \extctx & \grameq \ctxhole \mid \la\var \extctx \mid \tm\esub\var \ictx \mid \extctx \esub\var \ptm \mid \ictx 
		\\
		\textsc{Rigid ctxs} & \ictx & \grameq \ptm \extctx \mid \ictx \tm \mid \ictx\esub\var \ptm \mid \ptm \esub\var\ictx
		\end{array}$
\end{center}
\begin{center}{
		\begin{tabular}{ccc}
\textsc{External rewrite rules:}
			&
			\multirow{2}{*}
			{\begin{prooftree}
					\hypo{\tm \Rew{\wsym a} \tm'}		
					\infer1{\extctxp\tm \Rew{\esssym a} \extctxp{\tm'}}
			\end{prooftree}}
		\\
		($a \in \set{\msym,\esym}$)
	\end{tabular}}
\end{center}

\begin{center}{
		$\begin{array}{cccccc}
\textsc{External reduction:} &
		\tovsubs \, \defeq \, \toms \cup \toes
	\end{array}$}
\end{center}
The strategy 
diverges on $\vartwo (\la\varthree\delta\delta)$ and normalizes the potentially diverging term $(\la\var\vartwo) (\la\varthree\delta\delta) \tovsubs^* \vartwo$. 
The grammars of external and rigid contexts allow evaluation
to enter only inside non-applied abstractions, \eg $(\la\var(\Id\Id)) \val \not\tovsubs 
(\la\var(\vartwo\esub\vartwo\Id)) \val$. 
This is a sort of outside-in order\footnote{More precisely, the strategy 
respects the box order of \cite{DBLP:conf/popl/AccattoliBKL14}, but according to how boxes are used by the \cbv 
translation into linear logic, which is different from the one underlying \cite{DBLP:conf/popl/AccattoliBKL14}. Also, 
the box order is incomparable with respect to the \emph{least level} order of the linear logic literature.}
which 
is neither left-to-right nor right-to-left---we have both $(\Id\Id) (\Id\Id)\toms (\vartwo\esub\vartwo\Id) (\Id\Id)$ 
and $(\Id\Id) (\Id\Id)\toms (\Id\Id)(\vartwo\esub\vartwo\Id)$---since open contexts do not impose an order on applications. As just showed the strategy is non-deterministic: another example is given by $\tm = \var 
(\la\vartwo(\Id \Id)) \esub\var{\var 
(\Id\Id)} \toms \var (\la\vartwo \varthree\esub\varthree \Id) \esub\var{\var (\Id\Id)}$, and $\tm \toms \var 
(\la\vartwo(\Id \Id)) \esub\var{\var(\varthree\esub\varthree{\Id})}$. 
Such a behavior however is only a relaxed form of determinism, as it satisfies the \emph{diamond 
property} (see also the comment after \Cref{prop:properties-open-reduction}).

\begin{proposition}[Properties of $\tovsubs$]
	\label{prop:external-properties}
	\label{prop:vsc-diamond}
	\NoteProof{propappendix:vsc-diamond}
	\label{l:fullness}
	Let $\tm$ be a VSC term.
	\begin{enumerate}
		\item\label{p:external-properties-diamond} 
		$\tovsubs$ is diamond; $\toms$ and $\toes$ strongly commute.	
		\item\label{p:external-properties-fullness} \emph{Fullness}: 
		$\tm$ is $\esssym$-normal if and only if $\tm$ is $\vsub$-normal.
	\end{enumerate}
\end{proposition}

%
%

\section{Multi Types by Value}
\label{sect:types}

We present a (linear logic based)  \emph{multi type system} for \cbv, inducing the relational model for \cbv.
For Plotkin's \cbv $\l$-calculus, it has been introduced by Ehrhard \cite{DBLP:conf/csl/Ehrhard12}, as the \cbv 
version of de Carvalho's System $\mathsf{R}$ for \cbn \cite{Carvalho07,deCarvalho18}. 

The study of multi types in this section takes the shrinking technique as presented for \cbn in 
\cite{DBLP:journals/pacmpl/AccattoliGK18}, and uses it to lift the semantical study of \ocbv by 
\cite{Guerrieri18,DBLP:conf/aplas/AccattoliG18} to \scbv\footnote{A notable difference: in 
\cite{Guerrieri18,DBLP:conf/aplas/AccattoliG18} there are no base types and termination is characterized using $\emptymset$; here this is not possible because 1) it is incompatible with the shrinking technique, and 2) in \ocbv terminating and erasable terms coincide, while in \scbv they do not (see \Cref{subsect:intro-cbv}.)}.

\smallskip
\paragraph*{Multi Types} There are two layers of types, \emph{linear} and \emph{multi types}, mutually defined by: 
\begin{center}
$\begin{array}{r\colspace r@{\ } ll}
\textsc{Linear types} & \ltype, \ltypetwo &\grameq \ground \mid \larrow{\mtype}{\mtypetwo}
\\
\textsc{Multi types} & \mtype, \mtypetwo &\grameq \mset{\ltype_1, \dots, \ltype_n} \quad n \in \nat
\end{array}$
\end{center}
where $\ground$ is an unspecified ground type and $\mset{\ltype_1, \dots, \ltype_n}$ is our notation for finite 
multisets.
The \emph{empty} multi type $\mset{\,}$ (obtained taking $n = 0$) is also denoted by $\emptymset$. 
A generic (multi or linear) type is denoted by $\type$.
A multi type $\mset{\ltype_1, \dots, \ltype_n}$ has to be intended as a conjunction $\ltype_1 \land \dots \land 
\ltype_n$ of linear types $\ltype_1, \dots, \ltype_n$, for a commutative, associative, non-idempotent conjunction 
$\land$ (morally a tensor $\otimes$), whose neutral element~is~$\emptymset$.

The intuition is that a linear type corresponds to a single use of a term $\tm$, and that $\tm$ is typed with a 
multiset 
$\mtype$ of $n$ linear types if it is going to be used (at most) $n$ times. The  meaning of \emph{using a term} is not 
easy to define precisely. Roughly, it means that if $\tm$ is part of a larger term $\tmtwo$, then (at most) $n$ copies 
of $\tm$ shall end up in evaluation position during the evaluation of $\tmtwo$. More precisely, the $n$ copies shall 
end 
up in evaluation positions where \mbox{they are applied to some terms.}

\begin{figure*}[t!]
	\begin{center}
	\scalebox{0.95}{
	$
			{\begin{prooftree}[label separation = .2em]
					\infer0
					[\scriptsize$\ruleAx$]
					{\var \hastype [\ltype] \vdash \var \hastype \ltype}
			\end{prooftree}}
			\quad \
			{\begin{prooftree}[separation=1em, label separation = .2em]
					\hypo{\typctx \vdash \tm \hastype \mset{ \larrow{\mtype\!}{\!\mtypetwo} }}
					\hypo{\typctxtwo \vdash \tmtwo \hastype \mtype}
					\infer2[\scriptsize$\ruleAp$]
					{\typctx \uplus \typctxtwo \vdash \tm\tmtwo \hastype \mtypetwo}
			\end{prooftree}}
			\quad	\
			{\begin{prooftree}[label separation = .2em]
					\hypo{\tyjp{}{\tm}{\typctx, \var \hastype \mtype}{\mtypetwo}}
					\infer1[\scriptsize$\ruleFun$]
					{\tyjp{}{\la{\var}{\tm}}{\typctx}{\ty{\mtype}{\mtypetwo}}}
			\end{prooftree}}
			\quad \
			{\begin{prooftree}[separation=1em, label separation = .2em]
					\hypo{\typctx, \var \hastype \mtype \vdash \tm \hastype \mtypetwo}
					\hypo{\typctxtwo \vdash \tmtwo \hastype \mtype}
					\infer2
					[\scriptsize$\ruleES$]
					{\typctx \uplus \typctxtwo \vdash \tm \esub\var\tmtwo \hastype \mtypetwo}
			\end{prooftree}}
			\quad \
			{\begin{prooftree}[separation=1em, label separation = .2em]
				\hypo{\left[\tyjp{}{\tval}{\typctx_{i}}{\ltype_{i}}\right]_{\iI}}
				\infer1
				[\scriptsize$\ruleMany$]
				{\tyjp{}{\tval}{\biguplus_{\iI}\typctx_{i} }{\mult{\ltype_{i}}_{\iI}}}
				\end{prooftree}}
$
	}
\end{center}
	\caption{Call-by-Value Multi Type System. In the rule $\ruleMany$, $\tval$ is a \emph{theoretical value}, \ie a variable or an abstraction.}
	\label{fig:cbvtypes}
\end{figure*}		
The derivation rules for the multi types system are in \Cref{fig:cbvtypes} (explanation follows).  The rules are the same as in Ehrhard \cite{DBLP:conf/csl/Ehrhard12}, up to the fact that they are extended to \ES. 

A \emph{multi} (\resp \emph{linear}) \emph{judgment} has the shape $\typctx \vdash \tm \hastype \type$ where $\tm$ is a term, $\type$ is a 
multi (\resp linear) type and $\typctx$ is a \emph{type context}, that is, a total function from variables to multi types 
such that  the set $\domain{\typctx} \defeq \{\var \mid \typctx(\var) \neq \emptymset\}$ is finite.

\paragraph*{Technicalities about Types} The type context $\typctx$ is \emph{empty} if $\dom{\typctx} = \emptyset$.  
\emph{Multi-set sum} $\mplus$ is extended to type contexts point-wise,
\ie\  $(\typctx \mplus \typctxtwo)(\var) \defeq \typctx(\var) \mplus \typctxtwo(\var)$ for each variable $\var$.
This notion is extended to a finite family of type contexts as expected, 
in particular $\bigmplus_{i \in J\!} \typctx_i$ is the empty context  when $J = \emptyset$.
A type context $\typctx$ is denoted by $\var_1 \hastype \mtype_1, \dots, \var_n \hastype \mtype_n$ (for some $n \in 
\nat$) if $\dom{\typctx} \subseteq \{\var_1, \dots, \var_n\}$ and $\typctx(\var_i) = \mtype_{i}$ for all $1 \leq i \leq 
n$.
Given two type contexts $\typctx$ and $\typctxtwo$ such that $\dom{\typctx} \cap \dom{\typctxtwo} = \emptyset$, the 
type 
context $\typctx, \typctxtwo$ is defined by $(\typctx, \typctxtwo)(\var) \defeq \typctx(\var)$ if $\var \in 
\dom{\typctx}$, $(\typctx, \typctxtwo)(\var) \defeq \typctxtwo(\var)$ if $\var \in \dom{\typctxtwo}$, and $(\typctx, 
\typctxtwo)(\var) \defeq \emptymset$ otherwise.
Note that $\typctx, \var \hastype \emptymset = \typctx$, where we implicitly assume $\var \notin \dom{\typctx}$. 

We write $\concl{\tderiv}{\typctx}{\tm}{\mtype}$ if $\tderiv$ is a (\emph{type}) \emph{derivation} (\ie a tree constructed using the rules in \Cref{fig:cbvtypes}) with conclusion the multi judgment $\typctx \vdash \tm \hastype \mtype$.
In particular,  we write $\concl{\tderiv}{\,}{\tm}{\mtype}$ when $\typctx$ is empty.
We write $\derive{\tderiv}{\tm}$ if $\concl{\tderiv}{\typctx}{\tm}{\mtype}$ for some type context $\typctx$ and some multi type $\mtype$. 
Our study being quantitative, we use two notions of size of type derivations.

\begin{definition}[Derivation size(s)]
\label{def:two-sizes}
	Let $\tderiv$ be a derivation. 
	The \emph{(general) size} $\size{\tderiv}$ of $\tderiv$ is the number of 
	rule occurrences in $\tderiv$ except for the rule $\ruleMany$.
	The \emph{multiplicative size} $\sizem{\tderiv}$ of $\tderiv$ is the number of occurrences of the rules $\lambda$ and 
$\ruleAp$ in $\tderiv$.
\end{definition}
The two sizes for a derivation play different qualitative and quantitative roles. \emph{Qualitative}: to have a 
combinatorial proof of the characterization of normalizable terms, we need a measure that 
decreases for some kinds of $\tovsub$ steps; 
this role is played by the general size $\size{\!\cdot\!}$.
	\emph{Quantitative}: to count the number of $\tom$ step in 
	some kinds of $\vsub$-normalizing evaluation, \ie our cost model;
	his role is played by the multiplicative size $\sizem{\!\cdot\!}$.

\paragraph*{Multisets and the $\ruleMany$ rule}
The rule $\ruleMany$ plays a crucial role.
It can be applied to theoretical values only (\ie variables and abstractions) 
and it says how many ``copies'' (possibly none) of a theoretical value occurring in a term $\tm$ are needed to evaluate $\tm$.
It corresponds to the promotion rule of linear logic,
which, in the \cbv representation of the $\lambda$-calculus, is indeed used for typing
abstractions and variables.

\paragraph*{Some preliminary results}
The two next lemmas establish a key feature of this type system: in a typed term $\tm$,
substituting a value for
a variable (as in the exponential step) preserves the type of $\tm$ and ``consumes'' the multi type of the variable.
Besides, it also provides  quantitative information about the type derivation for $\tm$ before and after the substitution.

\begin{lemma}[Substitution]
	\label{l:substitution}	
	\NoteProof{lappendix:substitution}
	Let $\tm$ be a term, $\val$ be a value and $\namedtyjp{\tderiv}{}{\tm}{\typctx, \var \hastype 
		\mtypetwo}{\mtype}$ and $\namedtyjp{\tderivtwo}{}{\val}{\typctxtwo}{\mtypetwo}$ be derivations.
	Then there is a derivation $\namedtyjp{\tderivthree}{}{\tm \isub{\var}{\val}}{\typctx \mplus \typctxtwo}{\mtype}$ 
	with $\sizem{\tderivthree} = \sizem{\tderiv} + \sizem{\tderivtwo}$ and $\size{\tderivthree} \leq \size{\tderiv} + 
	\size{\tderivtwo}$. 
\end{lemma}

\begin{lemma}[Anti-substitution]
	\label{l:anti-substitution}
	\NoteProof{lappendix:anti-substitution}
	Let $\tm$ be a term, $\val$ be a value, and 
	$\namedtyjp{\tderiv}{}{\tm\isub{\var}{\val}}{\typctx}{\mtype}$
	be a derivation. 
	Then there are two  derivations $\namedtyjp{\tderivtwo}{}{\tm}{\typctxtwo, \var \hastype 
		\mtypetwo}{\mtype}$ 
	and $\namedtyjp{\tderivthree}{}{\val}{\typctxthree}{\mtypetwo}$ such that $\typctx = \typctxtwo \mplus \typctxthree$ 
	with $\sizem{\tderiv} = \sizem{\tderivtwo} + \sizem{\tderivthree}$ and $\size{\tderiv} \leq \size{\tderivtwo} + 
	\size{\tderivthree}$.
\end{lemma}

\cref{l:substitution,l:anti-substitution} are needed to prove subject reduction and expansion, respectively, which mean that the type can be preserved after 
any reduction and expansion step. 

\begin{proposition}[Qualitative subjects]
	\label{prop:qual-subject}
	\NoteProof{propappendix:qual-subject}
	Let $\tm \tovsub \tm'$.
	\begin{enumerate}
		\item \label{p:qual-subject-reduction}
		\emph{Reduction}: 	if $\namedtyjp{\tderiv}{}{\tm}{\typctx}{\mtype}$ then there is a derivation $\namedtyjp{\tderiv'}{}{\tm'}{\typctx}{\mtype}$ such that
		$\sizem{\tderiv'} \leq \sizem{\tderiv}$ and $\size{\tderiv'} \leq \size{\tderiv} $.
		
		\item \label{p:qual-subject-expansion}
		\emph{Expansion}: if $\namedtyjp{\tderiv'}{}{\tm'}{\typctx}{\mtype}$ then there is a derivation $\namedtyjp{\tderiv}{}{\tm}{\typctx}{\mtype}$ such that
		$\sizem{\tderiv'} \leq \sizem{\tderiv}$ and $\size{\tderiv'} \leq \size{\tderiv} $.
	\end{enumerate}
\end{proposition}

\Cref{prop:qual-subject} says also that the derivation sizes cannot increase (\resp decrease) after a reduction (\resp expansion) step.
However, they are not precise enough to be used for a combinatorial proof of normalization, or to extract quantitative information about the evaluation length.
There are two ways to make them precise, 
using always the \emph{same} type system: we restrict either the reductions, or the kind of types taken into account.
We shall explore both directions in the next sections.

\paragraph*{Denotational Semantics} 
The multi type system induces a denotational model for \VSC, called \emph{relational semantics}.
The semantics of a term is simply the set of its derivable type judgments. 
More precisely,
let $\tm$ be a term and $\var_1, \dots , \var_n$ (with $n \geq 0$) be pairwise distinct variables. 
If
$\fv{\tm} \subseteq \{\var_1, \dots, \var_n\}$, the list $\vec{\var} = (\var_1 , \dots , \var_n)$ is \emph{suitable} for $\tm$. 
If
$\vec{\var} = (\var_1,\dots, \var_n)$ is suitable for $\tm$, the (\emph{relational}) \emph{semantics} of $\tm$ for $\vec{\var}$ is:
\begin{align*}
\sem{\tm}_{\vec{\var}} = \{((\mtypetwo_1, \dots, \mtypetwo_n), \mtype) \mid \exists \, \concl{\tderiv}{\var_1 \hastype \mtypetwo_1, \dots, \var_n \hastype \mtypetwo_n}{\tm}{\mtype} \}.
\end{align*}

Subject reduction and expansion (\Cref{prop:qual-subject}) guarantee that the semantics of a term is invariant under evaluation, \ie relational semantics is a denotational model for \VSC.

\begin{proposition}[Invariance under evaluation]
	\label{prop:invariance}
	\NoteProof{propappendix:invariance}
	Let $\vec{\var}$ be suitable for $\tm$. 
	If $\tm \tovsub^* \tm'$ then $\sem{\tm}_{\vec{\var}} = \sem{\tm'}_{\vec{\var}}$.
\end{proposition}

\paragraph*{A Uniform Approach}
In \Cref{sect:open,sect:shrinking}, we use the \emph{same} given multi type system to characterize normalization of two different evaluation strategies in \VSC, namely open  $\tovsubo$ (defined in \Cref{sect:vsc}) and external $\tovsubs$ (defined in \Cref{sect:external}), verifying $\tovsubo \,\subsetneq\,  \tovsubs$. In both case proving that, given an evaluation $\to_\alpha$, a term is $\alpha$-normalizing if and only if is typable, and extracting bounds on the evaluation length from type derivations (kind 1). 

The two cases follow the same pattern, whose blueprint is given by the simpler open case. 
Each case shall use some predicates (inert, tight, unitary, shrinking)  to isolate the derivations characterizing termination or providing the right bounds for that case, often needing more than one predicate, for different things. 
The predicates shall work by restricting the types appearing in the final judgment of the derivation---this way isolating fragments of the denotational semantics. 
Sometimes, often for the induction to go through, the predicate concerns only the typing context, or only the right-hand term, and not the whole final judgment.

\paragraph*{The Special Role of Inert Terms} Another recurrent ingredient are (\full) inert terms. In statements about normal forms, they usually satisfy stronger properties, essential for the induction to go through. Predicates shall also \emph{spread} on inert terms: if assumed on the type context, they transfer to the right-hand type, which is in turn the key step to propagate the predicate on sub-derivations. 
For bounds of kind 1, we give the full 
statements. 
For bounds of kind 2 and 3, 
more technical, the full statements are only in the Appendix, for~readability.

\section{Multi Types for Open \cbv}
\label{sect:open}

Here we recall the relationship between \cbv multi types and \ocbv developed by Accattoli and Guerrieri in 
\cite{DBLP:conf/aplas/AccattoliG18}. The reason is threefold. First, the  external case relies on the open one. 
Second, it provides the blueprint for the development of the strong case, allowing us to stress similarities and differences. 
Last, the development in \cite{DBLP:conf/aplas/AccattoliG18} needs to be slightly adapted to our present framework. 
Namely, here we use the Open VSC 
instead of the split fireball calculus used in \cite{DBLP:conf/aplas/AccattoliG18} (another formalism for \ocbv),
and we include a ground type $\ground$, necessary to later deal with the external strategy (for \ocbv the empty multi type $\zero$ would suffice).

Qualitatively, the open evaluation of $\tm$ terminates if and only if $\tm$ is typable.
Since $\tovsubo$ does not reduce under abstractions, every abstraction is $\osym$-normal and hence must be typable: for this reason, $\la{\var}\delta\delta$ is typable with $\emptytype$ (take the derivation only made of one rule $\ruleManyVal$ with $0$ premises), though $\delta\delta$ is not.

Quantitatively, the multiplicative size 
$\sizem{\tderiv}$ of the derivation $\tderiv$ bounds the sum of the length of the open evaluation plus 
the open size of its open normal form. The bound is exact when $\tderiv$ satisfies the \emph{tight} predicate, defined below.

\paragraph*{Inert and Tight Derivations} 
A multi type is \emph{ground} if it is $n\mset{\ground} \defeq \mset{ \ground, \dots, \ground}$  ($n$  times $\ground$) for any $n \geq 0$ (so, $0\mset{\ground} = \emptytype$).
\begin{center}
	$\begin{array}{r r@{\ } ll}
\textsc{Inert multi type}& \imtype &\grameq 	\mset{ \inltype_1, \dots, \inltype_n} & n \geq 0
\\
\textsc{Inert linear type}& \inltype &\grameq \ground \mid \larrow{n\mset{\ground}}{\imtype} & n \geq 0 
\end{array}$
\end{center}
A type context $\var_1 \hastype \mtype_1, \dots, \var_n \hastype \mtype_n$ is \emph{inert} if $\mtype_1, \dots, 
\mtype_n$ are inert multi types. 
Note that 
every ground multi type is inert.

A derivation $\namedtyjp{\tderiv}{}{\tm}{\typctx}{\mtype}$ is \emph{inert} if $\typctx$ is an inert type context and $\mtype$ is an inert multi type. If, moreover, $\mtype$ is a ground multi type, then $\tderiv$ is \emph{tight}.
Note that the definitions of inert and tight derivation depend only on its final judgment.

Tight
derivations are those we are actually interested in, but often, for the induction to go through, we have to consider the wider class of inert derivations.
For instance, tight derivations may have a complex
structure, having sub-derivations for inert terms that might
not be tight, but only inert.
Also, inert types spread on inert terms---the first key property of inert terms.

\begin{lemma}
	[Spreading of inertness on judgments]
	\label{l:spread-inert}
	\NoteProof{lappendix:spread-inert}
	Let $\concl{\tderiv}{\typctx}{\itm}{\mtype}$ be a derivation and $\itm$ be an inert term. 
	If $\typctx$ is a inert type context, then $\mtype$ is a inert multi type (and so $\tderiv$ is inert).
\end{lemma}

\paragraph*{Correctness}
Open correctness establishes that all typable terms $\osym$-normalize and the multiplicative size of the derivation bounds the number of $\tomo$ steps plus the open size of the $\osym$-normal form; this bound is exact if the derivation is tight. 
Open correctness is proved following a standard scheme in two stages: $(a)$ \emph{quantitative} subject reduction states that every $\tovsubo$ step preserves types and decreases the general size of a derivation, and that any $\tomo$ step decreases by an exact quantity the multiplicative size of a derivation;
$(b)$ a lemma states that the multiplicative size of any derivation typing a $\osym$-normal form $\tm$ provides an \emph{upper bound} to the open size of $\tm$, and if moreover the derivation is tight then the bound~is~\emph{exact}. The lemma follows---its unusual statement puts forward 
\newcounter{l:size-fireballs}
\addtocounter{l:size-fireballs}{\value{theorem}}
\begin{lemma}[Size of fireballs]
	\label{l:size-fireballs}
	\NoteProof{lappendix:size-fireballs}
	Let $\fire$ be a fireball. 
	If $\namedtyjp{\tderiv}{}{\fire}{\typctx}{\mtype}$
	then $\sizem{\tderiv} \geq \sizeo{\fire}$.
	If, moreover, $\typctx$ is inert and ($\mtype$ is ground inert or $\fire$ is inert), then 
$\sizem{\tderiv} = \sizeo{\fire}$.
\end{lemma}

\emph{Example}. The non-inert fireball $\delta \defeq \la{\var}{\var\var}$ is typable and the last rule of any derivation $\concl{\tderiv}{\typctx}{\delta}{\mtype}$ is 
$\ruleManyVal$.
If $\mtype$ is not ground then $\ruleManyVal$ has at least one premise that types the subterm $\var\var$, so $\sizem{\tderiv} > \sizeo{\delta}$.
If $\mtype$ is ground, then $\mtype = \emptytype$ and $\ruleManyVal$ has no premises, with $\sizem{\tderiv} = 0 = \sizeo{\delta}$.
For inert terms, \Cref{l:size-fireballs} says that their open size is the multiplicative size of their \emph{inert} derivations; 
by spreading of inertness (\Cref{l:spread-inert}), it amounts to say that~the~type~context~is~inert.

\begin{proposition}[Open quantitative subject reduction]
	\label{prop:weak-subject-reduction}
	\NoteProof{propappendix:weak-subject-reduction}
	Let $\namedtyjp{\tderiv}{}{\tm}{\typctx}{\mtype}$ be a derivation.
	\begin{enumerate}
		\item\emph{Multiplicative step:} if $\tm \towm \tm'$ then there is a derivation 
$\namedtyjp{\tderiv'}{}{\tm'}{\typctx}{\mtype}$ with
		$\sizem{\tderiv'} = \sizem{\tderiv} - 2$ and $\size{\tderiv'} = \size{\tderiv} - 1$; 
		\item\emph{Exponential step:} if $\tm \towe \tm'$ then there is a derivation 
$\namedtyjp{\tderiv'}{}{\tm'}{\typctx}{\mtype}$ such that
		$\sizem{\tderiv'} = \sizem{\tderiv}$ and $\size{\tderiv'} < \size{\tderiv}$.
	\end{enumerate}
\end{proposition}

As the general size of derivations decreases after any $\tovsubo$ step, 
we have a combinatorial proof of open correctness.

\begin{theorem}[Open correctness]
	\label{thm:open-correctness}
	\NoteProof{thmappendix:open-correctness}
	Let 
	$\derive{\tderiv}{\tm}$ be a derivation.
	Then there is a $\osym$-normalizing evaluation $\deriv \colon \tm \tovsubo^* \tmtwo$ with $2\sizem{\deriv} + 
	\sizeo{\tmtwo} \leq \sizem{\tderiv}$.
	And if $\tderiv$ is tight, then $2\sizem{\deriv} + \sizeo{\tmtwo} = \sizem{\tderiv}$.
\end{theorem}

\paragraph*{Completeness}
Open completeness establishes that every $\osym$-normalizing term is typable, and with a tight derivation $\tderiv$ such that $\sizem\tderiv$ is exactly the number of $\tomo$ steps plus the open size of the $\osym$-normal form. 
Again, the proof technique is standard: $(a)$ \emph{quantitative} subject expansion states that typability can be pulled back along $\tovsubo$ steps, increasing $\size\tderiv$;
$(b)$ a lemma states that every $\osym$-normal form is typable with a \emph{tight} derivation---inert terms verify a stronger statement. 

\newcounter{prop:precise-open-typability-nf}
\addtocounter{prop:precise-open-typability-nf}{\value{theorem}}
\begin{lemma}[Tight typability of open normal forms]
	\label{prop:precise-open-typability-nf} 
	\NoteProof{propappendix:precise-open-typability-nf}
	\begin{enumerate}
		\item \emph{Inert:}\label{p:precise-open-typability-nf-inert} if $\tm$ is an inert term then, for any inert multi 
type $\mtype$, there is an inert derivation $\concl{\tderiv}{\typctx}{\tm}{\mtype}$.
		\item \emph{Fireball:}\label{p:precise-open-typability-nf-fireball} if $\tm$ is a fireball then there is a tight derivation $\concl{\tderiv}{\typctx}{\tm}{\emptytype}$.		
	\end{enumerate}
\end{lemma}

\begin{proposition}[Open quantitative subject expansion]
	\label{prop:weak-subject-expansion}
	\NoteProof{propappendix:weak-subject-expansion}
	Let $\namedtyjp{\tderiv'}{}{\tm'}{\typctx}{\mtype}$ be a derivation.
	\begin{enumerate}
		\item\emph{Multiplicative step:} if $\tm \towm \tm'$ then there is a derivation 
$\namedtyjp{\tderiv}{}{\tm}{\typctx}{\mtype}$ with
		$\sizem{\tderiv'} = \sizem{\tderiv} - 2$ and $\size{\tderiv'} = \size{\tderiv} - 1$; 
		\item\emph{Exponential step:} if $\tm \towe \tm'$ then there is a derivation 
$\namedtyjp{\tderiv}{}{\tm}{\typctx}{\mtype}$ such that
		$\sizem{\tderiv'} = \sizem{\tderiv}$ and $\size{\tderiv'} < \size{\tderiv}$.
	\end{enumerate}
\end{proposition}

As the general size of derivations increases 
after a backward $\osym$-step,
a combinatorial proof of open completeness~follows.

\begin{theorem}[Open completeness]
	\label{thm:open-completeness}
	\NoteProof{thmappendix:open-completeness}
	Let $\deriv \colon \tm \tovsubo^* \tmtwo$ be an $\osym$-normalizing evaluation. 
	Then there is a tight derivation $\concl{\tderiv}{\typctx}{\tm}{\emptytype}$ such that $2\sizem{\deriv} + \sizes{\fire} = \sizem{\tderiv}$.
\end{theorem}

\paragraph*{Adequate denotational semantics for $\tovsubo$}
Relational semantics is a denotational model for $\tovsub$ (\Cref{prop:invariance}).
Moreover, it is \emph{adequate} for $\tovsubo$: it provides a semantic characterization of all and only the terminating terms for $\tovsubo$, thanks to correctness (\Cref{thm:correctness}) and completeness~(\Cref{thm:open-completeness}).

\begin{theorem}[Open adequacy]
	\label{thm:open-adequacy}
	\NoteProof{thmappendix:open-adequacy}
	Let $\tm$ be a term, with $\vec{\var}$ suitable for $\tm$.
	Then, $\tm$ is $\osym$-normalizing if and only if $\sem{\tm}_{\vec{\var}} \neq \emptyset$. 
\end{theorem} 

\section{Shrinking Multi Types for the External Strategy}
\label{sect:shrinking}

\begin{figure*}
	\begin{center}
			\scalebox{.95}{
		$\begin{array}{rrcl \colspace rrcl}
		\text{Right  multi type } & \rmtype &\grameq &	\mset{\rltype_1, \dots, \rltype_n} \ \ n \geq 1
		&
		\text{Right  linear type } & \rltype &\grameq & \ground \mid  \larrow{\lmtype}{\rltype}
		\\
		\text{Left  multi type } & \lmtype &\grameq &	\mset{\lltype_1, \dots, \lltype_n} \ \ n \geq 0
		&
		\text{Left  linear type } & \lltype &\grameq & \ground \mid  \larrow{\rmtype}{\lmtype}
		\\[5pt]
		\text{Unitary right multi type } & \urmtype &\grameq &	\mset{\urltype} 
		&
		\text{Unitary right linear type } & \urltype &\grameq & \ground \mid  \larrow{\ulmtype}{\urmtype}
		\\
		\text{Unitary left multi type } & \ulmtype &\grameq &	\mset{\ulltype_1, \dots, \ulltype_n} \ \ n \geq 0
		&
		\text{Unitary left linear type } & \ulltype &\grameq & \ground \mid  \larrow{\urmtype}{\ulmtype}
		\end{array}$
	}
	\end{center}
	\caption{Kinds of right and left (shrinking) types.}
	\label{fig:shrinking}
\end{figure*}

In \scbv, diverging terms may be typable. 
For instance, $\var(\la\vartwo\delta\delta)$ (with $\delta \defeq \la{\varthree}\varthree\varthree$) can be typed by assigning $\mset{\larrow\emptymset{\mset\ground}}$ to $\var$, and typing the diverging argument $(\la\vartwo\delta\delta)$ with $\emptymset$ (via a $\ruleManyVal$ rule with 0 premises). 
The point is that, since $\emptymset$ is on the left of $\multimap$, the argument is meant to be erased, but $\var$ cannot actually 
erase it. This is a problem typical of strong settings, as it occurs also in Strong \cbn. The solution is to restrict 
to type derivations satisfying a predicate that forbids types where $\emptymset$ plays these dangerous tricks. This is 
unavoidable and standard, see 
\cite{DBLP:books/daglib/0071545,DBLP:journals/pacmpl/AccattoliGK18}. Following \cite{DBLP:journals/pacmpl/AccattoliGK18}, the predicate is here called \emph{shrinkingness} 
because it ensures that the size of type derivations shrinks at each $\tovsubs$ step (see \Cref{prop:shrinking-subject-reduction,prop:shrinking-subject-expansion}).

Since the typing rule $\ruleFun$ shifts a multi type from the left-hand side of a judgment to the left of $\multimap$ on the right-hand side of a judgment, we actually need two notions of shrinking types, \emph{left} and \emph{right}. 
%
%
Indeed, a judgment is shrinking (and so it ensures that the size of the derivation shrinks at each $\tovsubs$ step) when the typing context is made of left shrinking multi types, and the type on the right-hand side is right shrinking.
Formally, \emph{left} and \emph{right shrinking} (multi or linear) \emph{types} are defined in \Cref{fig:shrinking} (we often omit ``shrinking'' when referring to left or right types, for brevity), by mutual induction: they are somehow dual notions.
Additionally, we need \emph{unitary (left/right) shrinkingness} (see \Cref{fig:shrinking}) to 
characterize those minimal derivations that measure 
\emph{exactly} the number of $\toms$ steps to normal form and the size of the normal form:
indeed, $\esssym$-evaluation reduces under abstractions and unitality requires to take just one ``copy'' of the body of the abstraction. 
The notions extend to type contexts and to derivations as follows:
	\begin{itemize}

	\item A type context $\var_1 \hastype \mtype_1, \dots, \var_n \hastype \mtype_n$ is \emph{\leftsh} (\resp \emph{unitary \leftsh}) if each $\mtype_i$ is \leftsh (\resp unitary \leftsh).
	\item A derivation $\concl{\tderiv}{\typctx}{\tm}{\mtype}$ is \emph{shrinking} (\resp~\emph{unitary shrinking}) if $\typctx$ is \leftsh (\resp~unitary \leftsh) and $\mtype$ is \rightsh \mbox{(\resp unitary \rightsh).}
\end{itemize}
Of course, (right/left) unitary shrinkingness implies shrinkingness, but the converse is false---just consider $\mset{\ground,\ground}$.
Note that $\mset{\ground}$ is both unitary \leftsh and unitary \rightsh, while $\zero$ is unitary \leftsh but not \rightsh, and $\mset{\larrow{\emptymset}{\mset{\ground}}}$ is unitary \rightsh but not \leftsh.

The important notion is the one of (unitary) shrinking derivation, but often, for the induction to go through, statements have a less intuitive form, typically concerning inertness or type contexts alone.
Shrinkingness is a predicate of derivations 
defined by looking only at the final judgment. For proving  
properties of shrinking derivations, we have to analyze how shrinking propagates to sub-derivations. 
The following lemma is specific to \emph{\leftsh} shrinkingness, on which the propagation of shrinkingness 
then build. 
%
It says that for a specific kind of terms---typable \emph{rigid} terms---\leftsh shrinkingess spreads from the type context to the right-hand multi type~in~a~judgment. 

\begin{lemma}
	[Spreading of \leftsh shrinkingness on judgments]
	\label{l:spread-shrinking}
	\NoteProof{lappendix:spread-shrinking}
	Let $\concl{\tderiv}{\typctx}{\ptm}{\mtype}$ be a derivation and $\ptm$ a \pointed term. 
	If $\typctx$ is a \leftsh (\resp~unitary \leftsh) type context, then $\mtype$ is a \leftsh (\resp~unitary \leftsh) multi type.
\end{lemma}

To prove correctness (\Cref{thm:correctness}) and completeness (\Cref{thm:completeness}) of shrinking derivation with respect to the external strategy, we follow the same pattern as for the open case. 

\paragraph*{Correctness}
Shrinking correctness establishes that all typable terms with a shrinking derivation $\esssym$-normalize and the multiplicative size of the derivation bounds the number of $\toessm$ steps plus the \full size of the $\vsub$-normal form; this bound is exact if the derivation is unitary shrinking. 

\newcounter{l:size-strong-fireballs}
\addtocounter{l:size-strong-fireballs}{\value{theorem}}
\begin{lemma}[Size of \full fireballs]
	\label{l:size-strong-fireballs}
	\NoteProof{lappendix:size-strong-fireballs}
	Let $\typctx$ be a \leftsh (\resp unitary \leftsh) type context, $\sfire$ be a \full fireball and $\namedtyjp{\tderiv}{}{\sfire}{\typctx}{\mtype}$ be a derivation such that if $\sfire$ is an \valES then $\mtype$ is \rightsh (\resp~unitary \rightsh).
	Then $\sizefu{\sfire} \leq \sizem{\tderiv}$ (resp.~$\sizefu{\sfire} = \sizem{\tderiv}$).
\end{lemma}

Note that \Cref{l:size-strong-fireballs} applies in particular to any shrinking and unitary shrinking derivation of a \full fireball.
The more involved statement is required because when $\sfire = \sitm \sfiretwo$ (which is not an \valES) its multi type $\mtype$ need not be \rightsh. Indeed, the subterm $\sitm$ (which is inert and so rigid) has type $\mset{\larrow{\mtypetwo}{\mtype}}$ and the spreading of \leftsh shrinkingness (\Cref{l:spread-shrinking}) ensures that $\mset{\larrow{\mtypetwo}{\mtype}}$ is \leftsh, but it does not imply that $\mtype$ is \rightsh.

\newcounter{prop:shrinking-subject-reduction}
\addtocounter{prop:shrinking-subject-reduction}{\value{theorem}}
\begin{proposition}[Shrinking quantitative subject reduction]
	\label{prop:shrinking-subject-reduction}
	\NoteProof{propappendix:shrinking-subject-reduction}
	Let $\concl{\tderiv}{\typctx}{\tm}{\mtype}$ with $\typctx$  \leftsh (\resp unitary \leftsh).
	Assume that if $\tm$ is an \valES, then $\mtype$ is \rightsh (\resp unitary \rightsh).
	\begin{enumerate}
		\item \emph{Multiplicative step:} if $\tm \toessm \tm'$ then there is a derivation $\concl{\tderiv'}{\typctx}{\tm'}{\mtype}$ such 
that 
$\sizem{\tderiv'} \leq \sizem{\tderiv}-2$ and $\size{\tderiv'} < \size{\tderiv}$
		(\resp $\sizem{\tderiv'} = \sizem{\tderiv}-2$ and $\size{\tderiv'} = \size{\tderiv}-1$).
		\item \emph{Exponential step:} if $\tm \toesse \tm'$ then there is a derivation $\concl{\tderiv'}{\typctx}{\tm'}{\mtype}$ such that
		$\sizem{\tderiv'} = \sizem{\tderiv}$ and $\size{\tderiv'} < \size{\tderiv}$.
	\end{enumerate}
\end{proposition}

\begin{theorem}[Shrinking correctness]
	\label{thm:correctness}
	\NoteProof{thmappendix:correctness}
	Let $\concl{\tderiv}{\typctx}{\tm}{\mtype}$ be a shrinking (\resp unitary shrinking) derivation.
	Then there are a $\vsub$-normal form $\tm'$ and an evaluation $\deriv \colon \tm \tovsubs^* \tm'$ with $2\sizem{\deriv} + 
\sizefu{\tm'} \leq \sizem{\tderiv}$ (\resp $2\sizem{\deriv} + \sizefu{\tm'} = \sizem{\tderiv}$).
\end{theorem}

\myparagraph{Completeness} 
Shrinking completeness establishes that every $\esssym$-normalizing term is typable with a unitary shrinking derivation whose multiplicative size is exactly the number of $\toms$ steps plus the \full size of the $\vsub$-normal form. 

\begin{lemma}[Shrinking typability of normal forms]
	\label{prop:typability-normal}
	\NoteProof{propappendix:typability-normal}
	\begin{enumerate}
		\item\label{p:typability-normal-inert}
		\emph{Inert:}
		for any \full inert term $\sitm$ and \leftsh multi type $\mtype$,
		there is a \leftsh type context $\typctx$ 
		and a derivation $\concl{\tderiv}{\typctx}{\sitm}{\mtype}$.
		If moreover $\mtype$ is unitary, then so is $\typctx$.
		\item\label{p:typability-normal-fireball}
		\emph{Fireball:}
		for every \full fireball $\sfire$ 
		there exists a unitary shrinking derivation $\concl{\tderiv}{\typctx}{\sfire}{\mtype}$.
	\end{enumerate}
\end{lemma}

The proof of \Cref{prop:typability-normal} relies on 
the ground type $\ground$, which is both unitary \leftsh and unitary \rightsh. 
In this way, \refpoint{typability-normal-inert} (needed to have a stronger \ih) can be used to~prove~\refpoint{typability-normal-fireball}.

\begin{proposition}[Shrinking quantitative subject expansion]
	\label{prop:shrinking-subject-expansion}
	\NoteProof{propappendix:shrinking-subject-expansion}
	Let $\concl{\tderiv'}{\typctx}{\tm'}{\mtype}$ with $\typctx$ unitary \leftsh.
	Assume that if $\tm'$ is an \valES then $\mtype$ is unitary \rightsh.
	\begin{enumerate}
		\item \emph{Multiplicative step:} if $\tm \toessm \tm'$ then there is a derivation $\concl{\tderiv}{\typctx}{\tm}{\mtype}$ such that
		$\sizem{\tderiv'} \leq \sizem{\tderiv}-2$ and $\size{\tderiv'} < \size{\tderiv}$
		(\resp $\sizem{\tderiv'} = \sizem{\tderiv}-2$ and $\size{\tderiv'} = \size{\tderiv} - 1$);
		\item \emph{Exponential step:} if $\tm \toesse \tm'$ then there is a derivation $\concl{\tderiv}{\typctx}{\tm}{\mtype}$ such that
		$\sizem{\tderiv'} = \sizem{\tderiv}$ and $\size{\tderiv'} < \size{\tderiv}$.
	\end{enumerate}
\end{proposition}

\begin{theorem}[Shrinking completeness]
	\label{thm:completeness}
	\NoteProof{thmappendix:completeness}
	Let $\deriv \colon \tm \tovsubs^* \tm'$ be an evaluation with $\tm'$ $\vsub$-normal. 
	Then there is a unitary shrinking derivation $\concl{\tderiv}{\typctx}{\tm}{\mtype}$ 
	with \mbox{$2\sizem{\deriv} + \sizefu{\tm'} = \sizem{\tderiv}$}.
\end{theorem}

\section{Normalization and Adequacy}
\label{sect:normalization}

In our multi type system, subject reduction and expansion hold for the whole reduction $\tovsub$, without quantitative information (\Cref{prop:qual-subject}). 
Combining this with shrinking correctness and completeness---which hold only for the external strategy---we obtain two notable results: 1) a \emph{normalization} theorem for \VSC, by exploiting an elegant proof technique used in \cite{DBLP:journals/tcs/CarvalhoPF11,DBLP:journals/pacmpl/MazzaPV18}; 2) we provide an \emph{adequate} model for the whole \VSC. 

\begin{theorem}[Operational and semantic characterizations of normalization]
	\label{thm:normalization}
	\NoteProof{thmappendix:normalization}
	Let $\tm$ be a term. The following are equivalent:
	\begin{enumerate}
		\item there is a $\vsub$-normal form $\tmtwo$ such that $\tm \tovsub^* \tmtwo$;
		\item $\tm$ is $\esssym$-normalizing;
		\item $\semfull{\tm}_{\vec{\var}} \neq \emptyset$, where $\vec{\var} = (\var_1, \dots, \var_n)$ is suitable for $\tm$ and 
	\end{enumerate}
\begin{equation*}
\begin{split}
\semfull{\tm}_{\vec{\var}} \defeq \{((\mtypetwo_1&,\dots, \mtypetwo_n),\mtype) \mid \exists 
\, \concl{\tderiv}{\var_1 \hastype \mtypetwo_1, \dots, \var_n \hastype \mtypetwo_n}{\tm}{\mtype} 
\mbox{ such that $\mtypetwo_1, \dots, \mtypetwo_n$ are \leftsh and $\mtype$ is \rightsh} \} .
\end{split}
\end{equation*}
\end{theorem}


This is our main qualitative result for \Full \cbv. 
Quantitatively, we provided bounds from type derivation (kind 1). 

%

\begin{figure*}[t!]
\begin{equation*}
	\resizebox{0.85\hsize}{!}{
			\begin{prooftree}[separation=1em]
			\hypo{}
			\infer1[\footnotesize$\ruleAx$]{\var\hastype \mset{\larrow{\mset{\ground^2}}{\mset{\ground^2}}} \vdash \var\hastype \larrow{\mset{\ground^2}}{\mset{\ground^2}} }
			\infer1[\footnotesize$\ruleManyVar$]{\var\hastype \mset{\larrow{\mset{\ground^2}}{\mset{\ground^2}}} \vdash \var\hastype \mset{\larrow{\mset{\ground^2}}{\mset{\ground^2}}} }	
			\hypo{}
			\infer1[\footnotesize$\ruleAx$]{\var\hastype \mset{\ground^2} \vdash \var\hastype {\ground^2} }
			\infer1[\footnotesize$\ruleManyVar$]{\var\hastype \mset{\ground^2} \vdash \var\hastype \mset{\ground^2} }	
			\infer2[\footnotesize$\ruleAp$]{\var\hastype \mset{\mset{\larrow{\mset{\ground^2}}{\mset{\ground^2}}, {\ground^2}}} \vdash \var\var \hastype \mset{\ground^2}}
			\infer1[\footnotesize$\ruleFun$]{\vdash \la{\var}\var\var \hastype \larrow{\mset{\mset{\larrow{\mset{\ground^2}}{\mset{\ground^2}}, {\ground^2}}}}{\mset{\ground^2}}}
			\infer1[\footnotesize$\ruleManyVal$]{\vdash \la{\var}\var\var \hastype \mset{\larrow{\mset{\mset{\larrow{\mset{\ground^2}}{\mset{\ground^2}}, {\ground^2}}}}{\mset{\ground^2}}}}
			
						\hypo{}
			\infer1[\footnotesize$\ruleAx$]{\vartwo\hastype \mset{\ground^2} \vdash \vartwo\hastype {\ground^2} }
			\infer1[\footnotesize$\ruleManyVar$]{\vartwo\hastype \mset{\ground^2} \vdash \vartwo\hastype \mset{\ground^2} }	
			\infer1[\footnotesize$\ruleFun$]{\vdash \la{\vartwo}\vartwo \hastype \larrow{\mset{\ground^2}}{\mset{\ground^2}}}			
			
			\hypo{}
			\infer1[\footnotesize$\ruleAx$]{\vartwo\hastype \mset\ground \vdash \vartwo\hastype \ground }
			\infer1[\footnotesize$\ruleManyVar$]{\vartwo\hastype \mset\ground \vdash \vartwo\hastype \mset\ground }	
			\infer1[\footnotesize$\ruleFun$]{\vdash \la{\vartwo}\vartwo \hastype \larrow{\mset\ground}{\mset\ground}}
			\infer2[\footnotesize$\ruleManyVal$]{\vdash \la{\vartwo}\vartwo \hastype \mset{\larrow{\mset{\ground^2}}{\mset{\ground^2}}, \ground^2}}
			\infer2[\footnotesize$\ruleAp$]{\vdash \delta\Id \hastype \mset{\ground^2}}
			\end{prooftree}
		}
			\end{equation*}
\caption{Unitary shrinking derivation $\tderivtwo_{\delta\Id}$ of minimal type for $\delta\Id$, where $\ground^2 \defeq \larrow{\mset\ground}{\mset\ground}$.}
\label{fig:composed-deriv}
\end{figure*}

\section{Bounds of Kind 2 and 3, Revisited}
\label{sect:semantic-bounds}

Up to now, we studied bounds of kind 1 (bounds from type derivations), now we turn to bounds of kind 2 and 3 for the external strategy. 
We start with kind 2 (bounds from types), use them to understand kind 3 (bounds from composable types), and then use kind 3 to obtain exact bounds of kind 2.

\paragraph*{Types Size and Bounds from Types} With multi types, not only type derivations but also multi types provide quantitative information, in this case on the size of normal forms. 
To see it, we have to define a size for types and type contexts, which is simply the number of occurrences of $\larrow{}{}$. 
Formally, the \emph{size} of linear and multi types is defined by mutual induction by:
\begin{center}
$\begin{array}{r@{\,}l@{\,} l@{\quad} r@{\,}l@{\,} l@{\quad} c@{\,}c@{\,}c}
\size{\ground} &\defeq &0 
& 
\size{\larrow{\mtype}{\mtypetwo}} &\defeq &1 \!+\! \size\mtype \!+\! \size\mtypetwo
&
\size{\mset{\ltype_1, \dots, \ltype_n}} &\defeq &{\textstyle{\sum_{i=1}^n}} \size{\ltype_i}
\end{array}$
\end{center}
Clearly, $\size{\type} \geq 0$ and $\size{\mtype} = 0$ if and only if $\mtype = n\mset{\ground}$ ($n \geq 0)$.

Given a type context $\typctx = \var_1 \hastype \mtype_1, \dots, \var_n \hastype \mtype_n$ we often consider the list of 
its types, 
noted  $\typelist\typctx \defeq (\mtype_1, \dots, \mtype_n)$.  Since any list of multi types $(\mtype_1, \dots, 
\mtype_n)$ can be seen as extracted from a type context $\typctx$, we 
use the notation $\typelist\typctx$ for lists of multi types.
The \emph{size} of a list of multi types is 
$\size{(\mtype_1, \dots, \mtype_n)} \allowbreak\defeq \allowbreak\sum_{i=1}^n \size{\mtype_i}$, and of the conclusion 
of a derivation $\concl{\tder}{\typctx}{\expr}{\mtype}$ is $\size{(\typelist\typctx, \mtype)} \defeq 
\size{\typelist\typctx} + \size\mtype$.
Clearly, $\domain{\typctx} = \emptyset$ implies $\size{\typelist\typctx} = 0$.

\smallskip
\paragraph*{Bounds of Kind 2} Types bound the size of normal forms---this is kind 2. It is however instructive to decompose this fact, by showing that more generally, types bound the size of type derivations for normal terms (which in turn bound the size of normal forms, as we already know).
\begin{proposition}[Types bound the size of derivations for normal forms]
	\label{prop:types-bound-normal-derivations}
	\NoteProof{propappendix:types-bound-normal-derivations}
	Let $\tm$ be a \full fireball and $\concl{\tderiv}{\typctx}{\tm}{\mtype}$ be a derivation. Then 
$\sizem{\tderiv} \leq \sizectx{\typctx} + \sizetyp{\mtype}$.
\end{proposition}
Note that we are not restricting to shrinking derivations.

Perhaps surprisingly, the bound may not be exact, as a forthcoming example shows---that is, there can be a  gap. As expected, we shall see that minimal types give exact bounds.

Composing with the bound from shrinking derivations (\reflemma{size-strong-fireballs}), we obtain bounds of kind 2.

\begin{corollary}[Types bound the size of \full fireballs]
	\label{cor:type-bound-size-normal}
	\NoteProof{corappendix:type-bound-size-normal}
Let $\tm$ be a \full fireball and $\concl{\tderiv}{\typctx}{\tm}{\mtype}$ be a shrinking derivation. Then $\sizefu{\tm} 
\leq \sizetyp{\mtype} + \sizectx{\typctx}$.
\end{corollary}

\paragraph*{Lax Bounds of Kind 3} De Carvalho's idea is that, given two normal forms $\tm$ and $\tmtwo$, one can extract bounds for the evaluation of $\tm\tmtwo$ by truly \emph{semantic means}, by looking only at the types of $\tm$ and $\tmtwo$---that is, at $\sem\tm$ and $\sem\tmtwo$---because a derivation for $\tm\tmtwo$ is just the application of a derivation for $\tm$ and one for $\tmtwo$.  
\emph{Remark}: in formulating bounds of kind 3 we restrict to closed terms for the sake of simplicity (it allows to ignore the type context)---de Carvalho does the same. There are however no issues in dealing with open terms. 

\begin{definition}[Composable pairs]
Let $\tm$ and $\tmtwo$ be closed normal terms. The set 
of composable pairs of $\tm$ and $\tmtwo$ is:
\begin{center}$ U(\tm,\tmtwo) \defeq \set{(\larrow{\mtype}{\mtypetwo}, \mtype)\ |\ \larrow{\mtype}{\mtypetwo} \in 
\sem\tm\mbox{ and } \mtype\in \sem\tmtwo\mbox{ s.t. }\mtypetwo\mbox{ is right}}  $.
\end{center}
\end{definition}

\begin{theorem}[Lax bounds of kind 3]
	\label{thm:lax-bounds-3}
	\NoteProof{thmappendix:lax-bounds-3}
Let $\tm$ and $\tmtwo$ be closed normalizing terms. Then $\tm\tmtwo$ normalizes if and only if $U(\tm,\tmtwo) \neq \emptyset$. Moreover, if $\tm$ and $\tmtwo$ are normal, $\deriv:\tm\tmtwo \tovsubs^* \tmthree$, and $\tmthree$ is normal then $2\sizem{\deriv} + \sizefu{\tmthree} \leq \size\mtype+\size\mtypetwo+1$ for every composable pair $(\mtype, \mtypetwo)\in U(\tm,\tmtwo)$.
\end{theorem}

Now, the tricky point is how to obtain \emph{exact} bounds. For a term in isolation, exact bounds of kind 2 are given by minimal types, as we shall prove. The problem is that, instead, for the application of two terms (kind 3) the minimal composable pair does not provide exact bounds. This happens because minimal types do not compose---said differently, the minimal composable pair is not made out of two minimal types for $\tm$ and $\tmtwo$. The following example pinpoints the subtleties.

\smallskip
\paragraph*{Key Example} 
Consider the unitary shrinking derivations $\tderiv_\delta$ and $\tderiv_\Id$ of minimal types for  $\delta= \la\var\var\var$ and for $\Id = \la\vartwo\vartwo$:
\begin{center}
	\small
			$\tderiv_\delta =
			\begin{prooftree}[separation=.7em]
			\hypo{}
			\infer1[\footnotesize$\ruleAx$]{\var\hastype \mset{\larrow{\mset\ground}{\mset\ground}} \vdash \var\hastype \larrow{\mset\ground}{\mset\ground} }
			\infer1[\footnotesize$\ruleManyVar$]{\var\hastype \mset{\larrow{\mset\ground}{\mset\ground}} \vdash \var\hastype \mset{\larrow{\mset\ground}{\mset\ground}} }	
			\hypo{}
			\infer1[\footnotesize$\ruleAx$]{\var\hastype \mset\ground \vdash \var\hastype \ground }
			\infer1[\footnotesize$\ruleManyVar$]{\var\hastype \mset\ground \vdash \var\hastype \mset\ground }	
			\infer2[\footnotesize$\ruleAp$]{\var\hastype \mset{\mset{\larrow{\mset\ground}{\mset\ground}, \ground}} \vdash \var\var \hastype \mset\ground}
			\infer1[\footnotesize$\ruleFun$]{\vdash \la{\var}\var\var \hastype \larrow{\mset{\mset{\larrow{\mset\ground}{\mset\ground}, \ground}}}{\mset\ground}}
			\infer1[\footnotesize$\ruleManyVal$]{\vdash \la{\var}\var\var \hastype \mset{\larrow{\mset{\mset{\larrow{\mset\ground}{\mset\ground}, \ground}}}{\mset\ground}}}
			\end{prooftree}$
\end{center}
and
\begin{center}
	\small
			$\tderiv_\Id = 
			\begin{prooftree}[separation=1em]
			\hypo{}
			\infer1[\footnotesize$\ruleAx$]{\vartwo\hastype \mset\ground \vdash \vartwo\hastype \ground }
			\infer1[\footnotesize$\ruleManyVar$]{\vartwo\hastype \mset\ground \vdash \vartwo\hastype \mset\ground }	
			\infer1[\footnotesize$\ruleFun$]{\vdash \la{\vartwo}\vartwo \hastype \larrow{\mset\ground}{\mset\ground}}
			\infer1[\footnotesize$\ruleManyVal$]{\vdash \la{\vartwo}\vartwo \hastype \mset{\larrow{\mset\ground}{\mset\ground}}}
			\end{prooftree}$
\end{center}
			Note that, pleasantly, $\sizem{\tderiv_\delta} = \sizefu{\delta} = \sizetyp{\mset{\larrow{\mset{\mset{\larrow{\mset\ground}{\mset\ground}, \ground}}}{\mset\ground}}} = 2$, and $\sizem{\tderiv_\Id} = \sizefu{\Id} = \sizetyp{\mset{\larrow{\mset\ground}{\mset\ground}}} = 1$.
Let's now consider the application $\delta\Id$. Unfortunately, the two obtained minimal types do not compose. Now, consider the unitary shrinking derivation $\tderivtwo_{\delta\Id}$ for $\delta\Id$ with minimal types (which provides exact information for $\delta\Id$) in \reffig{composed-deriv}. Note that its sub-derivations for $\tderivtwo_\delta$ and $\tderivtwo_\Id$ for $\delta$ and $\Id$ do not derive minimal types. The derivation  $\tderivtwo_{\delta\Id}$ indeed  is obtained by composing 
\begin{enumerate}
\item the variant $\tderivtwo_\delta$ of $\tderiv_\delta$ which has the same exact structure of $\tderiv_\delta$ and where every occurrence of $\ground$ has been replaced with $\ground^2 \defeq \larrow{\mset\ground}{\mset\ground}$, obtaining the type $\mset{\larrow{\mset{\mset{\larrow{\mset{\ground^2}}{\mset{\ground^2}}, {\ground^2}}}}{\mset{\ground^2}}}$.
\item with two derivations for $\Id$, one being $\tderiv_\Id$ (which has type $\mset{\ground^2}$), and one being the variant $\tderiv_\Id'$ of $\tderiv_\Id$ where $\ground$ has been replaced with $\larrow{\mset\ground}{\mset\ground}$ (of type $\mset{\larrow{\mset{\ground^2}}{\mset{\ground^2}}}$).
\end{enumerate}
Note that there is a mismatch: since the types derived by $\tderivtwo_\delta$ and $\tderivtwo_\Id$ are not minimal, their sizes are bigger than $\sizem{\tderivtwo_\delta}$ and $\sizem{\tderivtwo_\Id}$, and so they do not provide exact bounds for $\delta\Id$. Namely, the size of $\mset{\larrow{\mset{\mset{\larrow{\mset{\ground^2}}{\mset{\ground^2}}, {\ground^2}}}}{\mset{\ground^2}}}$, which is the type of $\tderivtwo_\delta$, is 6, while $\sizem{\tderivtwo_\delta}=2$---this is an instance of the mentioned gap. 
Summing up, \emph{minimal types do not compose}, and \emph{composable types do not give exact bounds}. 

\smallskip
\paragraph*{Out of the Impasse}De Carvalho solves this apparent \emph{cul-de-sac} via a technical study relating composable pairs of types, on one side, and minimal types for the derivations of composable pairs, on the other side, via type substitutions. An especially puzzling point is that his approach requires the type system to have an infinity of ground types, while the quantitative properties of multi types never do. 

The \emph{composable \textit{vs} minimal} tension stems from a fact about derivations in isolation, namely the mentioned \emph{gap}:  given a derivation $\concl{\tderivtwo}{\typctx}{\tm}{\mtype}$ for a normal term $\tm$, in general $\sizem\tderivtwo$ is bound by $\size{(\typelist\typctx, \mtype)}$, not equal to it. 
The key point behind de Carvalho's solution is that there always exists a re-typing $\tderiv$ of $\tderivtwo$ (more precisely, a derivation with the same structure and different types) whose types have the same size of $\tderivtwo$ (and $\tderiv$)---the derivation $\tderiv$ is obtained from $\tderivtwo$ by minimizing the types introduced by axioms---in the example, this is the case for $\tderivtwo_\delta$ and $\tderiv_\delta$. We deem this fact the \emph{size representation property}. It holds independently of the number of ground types, and we claim it being the key property at work in semantical bounds---its isolation is a contribution of this work. 

\smallskip
\paragraph*{Skeletons} For the property, we need to formalize the notion of derivations having the same structure. The 
equivalence $\tderiveq$
relates derivations having the same rules arranged in the same way, but not necessarily having the same types.

\begin{definition}[Skeleton equivalence]
Let $\tm$ be a term. 
Two derivations $\derive{\tderiv}{\tm}$ and $\derive{\tderivtwo}{\tm}$ 
are \emph{skeleton equivalent}, noted $\tderiv \tderiveq \tderivtwo$, if they end with the same kind of rule and the derivations on the 
premises are $\tderiveq$-equivalent, namely:
\begin{itemize}
	\item Both $\tderiv$ and $\tderivtwo$ are axioms.
	\item Both $\tderiv$ and $\tderivtwo$ end with rule $\ruleAp$, their two left premises $\tderiv_l$ and $\tderivtwo_l$ 
satisfy $\tderiv_l \tderiveq\tderivtwo_l$, and their right premises $\tderiv_r$ and $\tderivtwo_r$ satisfy $\tderiv_r 
\tderiveq\tderivtwo_r$---similarly for rules $\ruleFun$ and $\ruleES$.
	\item Both $\tderiv$ and $\tderivtwo$ end with a rule $\ruleManyVal$ with $n$ premises and there is a permutation 
$\sigma$ of $\set{1,\dots,n}$ such that the $i$-th premise $\tderiv_i$ of $\tderiv$ and the $\sigma(i)$-th premise 
$\tderivtwo_{\sigma(i)}$ of $\tderivtwo$ satisfy $\tderiv_i \tderiveq\tderivtwo_{\sigma(i)}$ for 
$i\in\set{1,\dots,n}$. 
\end{itemize} 
\end{definition}

The next lemma shows that skeleton equivalence preserves more or less everything you can imagine, but types. 
It is used to prove the size representation property (\Cref{prop:size-representation}).
We denote by $\card{m}$ the cardinality of a multiset~$m$.

\begin{lemma}[$\tderiveq$-Invariants]
\label{l:skel-equiv-properties}
\NoteProof{lappendix:skel-equiv-properties}
Let $\concl{\tderiv}{\typctx}{\tm}{\mtype}$ and $\concl{\tderivtwo}{\typctxtwo}{\tm}{\mtypetwo}$ be two derivations such 
that $\tderiv \tderiveq \tderivtwo$. Then $\size\tderiv = \size\tderivtwo$, $\sizem\tderiv = \sizem\tderivtwo$, 
$\card\mtype = \card\mtypetwo$, $\dom\typctx = \dom\typctxtwo$, and $\card{(\typctx(\var))} = \card{(\typctxtwo(\var))}$ 
for every variable $\var$. Moreover, $\tderiv$ is shrinking (resp. unitary shrinking) if and only if $\tderivtwo$ is.
\end{lemma}


\begin{proposition}[Size representation]
	\label{prop:size-representation} 
	\NoteProof{propappendix:size-representation}
	Let $\tm$ be a \full fireball and $\concl{\tderiv}{\typctx}{\tm}{\mtype}$ be a derivation. Then there is a derivation $\concl{\tderivtwo}{\typctxtwo}{\tm}{\mtypetwo}$ such that $\tderivtwo \tderiveq \tderiv$ and 
		$\sizem{\tderiv} = \sizectx{\typctxtwo} + \sizetyp{\mtypetwo}$.
\end{proposition}

Note that, since in general types bound the size of the derivation (\Cref{prop:types-bound-normal-derivations}), the size representation property is stating that the types of $\tderivtwo$ are minimal.

\smallskip
\paragraph*{Weak Exact Bounds of Kind 3} The size representation property induces bounds of kind 3, but only of a weak form.
\begin{theorem}[Weak exact bounds of kind 3]
	\label{thm:weak-exact-bounds-3}
	\NoteProof{thmappendix:weak-exact-bounds-3}
 Let $\tm$ and $\tmtwo$ be normal. 
If $\deriv:\tm\tmtwo \tovsubs^* \tmthree$ and $\tmthree$ is normal then there 
exist $\mtype\in \sem{\tm}$ and $\mtypetwo\in\sem{\tmtwo}$ such that $2\sizem{\deriv} + \sizefu\tmthree = 
\size\mtype+\size\mtypetwo+1$. 
\end{theorem}
We obtained exact bounds, but we have lost the fact that the involved types form a composable pair, because in general the types providing the bound do not compose, as the example about $\delta\Id$ shows. This is the weakness of \Cref{thm:weak-exact-bounds-3}.

We can now prove de Carvalho's original statement, connecting exactness with composability. The refinement however amounts to quite mechanical technicalities, 
which is why we claim that the essence is in the size representation property.

\paragraph*{Dissecting Size Representations} We now provide a refinement of size representation. The refinement 
requires the variant of the type system where the grammar of types has an infinity of ground types 
$\{\ground_i\}_{i\in\nat}$ (instead of having just $\ground$), whose judgments are noted $\vdash^\infty$. The 
improvement is that given $\conclin{\tderiv}{\typctx}{\tm}{\mtype}$ we can find a derivation 
$\conclin{\tderivtwo}{\typctxtwo}{\tm}{\mtypetwo}$ and a substitution $\sigma$ such that $\sigma(\mtypetwo) = \mtype$ 
and $\sigma(\typctxtwo) = \typctx$---we call the pair $(\sigma,\tderivtwo)$ a \emph{dissection}. 

\paragraph*{Technicalities about Substitutions} A substitution $\sigma$ is a function from ground types to 
linear types that is the identity but a for finite number of ground types. It is extended to act on types, multi types, 
type contexts, and derivations as expected. We use  $\gt\tderiv$ for 
the set of ground types occurring in $\tderiv$.

\begin{definition}[Dissections]
Let $\conclin{\tderiv}{\typctx}{\tm}{\mtypetwo}$ be a derivation. A \emph{dissection} of $\tderiv$ is a pair $(\tderivtwo, 
\sigma)$ where $\conclin{\tderivtwo}{\typctx'}{\tm}{\mtypetwo'}$ is a derivation and $\sigma$ is a substitution such that:
\begin{enumerate}
	\item \emph{Skeletons}: $\tderivtwo \tderiveq \tderiv$;
	\item \emph{Representation}: $\sigma(\tderivtwo) = \tderiv$;
	\item \emph{Disjoint names}: $\gt{\tderivtwo}\cap \gt{\tderiv} = \emptyset$.
\end{enumerate}
And if $\sizem{\tderiv} = \sizectx{\typctx'} + \sizetyp{\mtypetwo'}$ then $(\tderivtwo, \sigma)$ is a \emph{size 
dissection} of $\tderiv$.
\end{definition}
Any derivation for a \full fireball admits a size dissection.

\begin{lemma}[Size dissection]
	\label{l:size-dissection} 
	\NoteProof{lappendix:size-dissection}
	Let $\tm$ be a \full fireball. 	For any $\conclin{\tderiv}{\typctx}{\tm}{\mtype}$
	there are a 
substitution $\sigma$ and a derivation $\conclin{\tderivtwo}{\typctx'}{\tm}{\mtype'}$ forming a size dissection of 
$\tderiv$.
\end{lemma}

Finally, we use size dissections to retrieve the \cbv analogous de Carvalho's result, stating that the types providing the bound are the smallest ones which can be turned into a composable pair via a type substitution. 

\begin{theorem}[Exact bounds of kind 3]
	\label{thm:exact-bounds-3}
	\NoteProof{thmappendix:exact-bounds-3}
Let $\tm$ and $\tmtwo$ be closed \full fireballs. If $\deriv \colon \tm\tmtwo \tovsubs^* \tmthree$ and $\tmthree$ is normal 
then $2\sizem{\deriv} + \sizefu{\tmthree} =  \inf\set{\size{\larrow{\mtypetwo}{\mtype}}+\size\mtypethree+1\ |\ \exists 
\sigma\mbox{ such that }(\larrow{\sigma(\mtypetwo)}{\sigma(\mtype)}, \sigma(\mtypethree))\in U(\tm,\tmtwo)}$.
\end{theorem}


\paragraph*{Exact Bounds of Kind 2} As for kind 2, we obtain exact bounds as a corollary of size representation (\Cref{prop:size-representation}), by composing it with the exact bounds from derivations. Here, exactness comes from plain minimality, as usual.

\begin{corollary}[Minimal types catch the size of normal forms]
	\label{cor:minimal-type-normal}
	\NoteProof{corappendix:minimal-type-normal}
Let $\tm$ be a \full fireball. There exists a unitary shrinking derivation $\concl{\tderiv}{\typctx}{\tm}{\mtype}$ such 
that $\sizefu{\tm} = \sizetyp{\mtype} + \sizectx{\typctx}$.
\end{corollary}

\IEEEpeerreviewmaketitle

\section*{Acknowledgments}
To Andrea Condoluci and Claudio Sacerdoti Coen for discussions about \scbv.

\bibliographystyle{IEEEtranS}
\bibliography{main.bbl}

\clearpage
\appendices

\section{Preliminaries and Notations in Rewriting}
\label{sect:preliminaries}
For a relation $R$ on a set of terms, $R^*$ is its reflexive-transitive closure. 
Given a relation $\Rew{\Rule}$, an $\Rule$-\emph{evaluation}
(or simply evaluation if unambiguous) $\deriv$ is a finite sequence of terms $(\tm_i)_{0 \leq i \leq n}$ (for some $n \geq 0$) such that $\tm_i \Rew{\Rule} \tm_{i+1}$ for all $1 \leq i < n$, and we write $\deriv \colon \tm \Rew{\Rule}^* \tmtwo$ if $\tm_0 = \tm$ and $\tm_n = \tmtwo$. The
\emph{length} $n$ of $\deriv$ is denoted by $\size{\deriv}$, and $\size{\deriv}_a$ is the number of $a$-\emph{steps} (\ie the number of $\tm_i \Rew{a} \tm_{i+1}$ for some $1 \leq i \leq n$) in $\deriv$, for a given subrelation $\Rew{a}$ of $\Rew{\Rule}$.

A term $\tm$ is $\Rule$-\emph{normal} if there is no $\tmtwo$ such that $\tm \Rew{\Rule} \tmtwo$.
An evaluation $\deriv \colon \tm \Rew{\Rule}^* \tmtwo$ is \emph{$\Rule$-normalizing} if $\tmtwo$ is $\Rule$-normal.
A term $\tm$ is \emph{weakly $\Rule$-normalizing} if there is a $\Rule$-normalizing evaluation $\deriv \colon \tm \Rew{\Rule}^* \tmtwo$; and $\tm$ is \emph{strongly $\Rule$-normalizing} if there no infinite sequence $(\tm_i)_{i \in \nat}$  such that $\tm_0  = \tm$ and $\tm_i \Rew{\Rule} \tm_{i+1}$ for all $i \in \nat$.
Clearly, strong $\Rule$-normalization implies weak $\Rule$-normalization.

A relation $\Rew{\Rule}$ is \emph{diamond} if $\tmtwo_1 \,{}_\Rule\!\!\lto \tm \Rew{\Rule} \tmtwo_2$ and $\tmtwo_1 \neq \tmtwo_2$ imply $\tmtwo_1 \Rew{\Rule} \tmthree \, {}_\Rule\!\!\lto \tmtwo_2$ for some $\tmthree$. 
As a consequence:
\begin{enumerate}
	\item $\Rew{\Rule}$ is confluent (\ie $\tmtwo_1 \,{}_\Rule^*\!\!\lto \tm \Rew{\Rule}^* \tmtwo_2$  implies $\tmtwo_1 \Rew{\Rule}^* \tmthree \, {}_\Rule^*\!\!\lto \tmtwo_2$ for some $\tmthree$); \item any term $\tm$ has at most one normal form (\ie if $\tm \Rew{\Rule}^* \tmtwo$  and $\tm \Rew{\Rule}^* \tmthree$ with $\tmtwo$ and $\tmthree$ $\Rule$-normal, then $\tmtwo = \tmthree$);
	\item all $\Rule$-evaluations with the same start and end terms have the same length (\ie if $\deriv \colon \tm \Rew{\Rule}^* \tmtwo$  and $\deriv' \colon \tm \Rew{\Rule}^* \tmtwo$ then $\size{\deriv} = \size{\deriv'}$);
	\item $\tm$ is weakly $\Rule$-normalizing iff it is strongly $\Rule$-normalizing.
\end{enumerate}

Two relations $\Rew{\Rule_1}$ and $\Rew{\Rule_2}$ \emph{strongly commute} if $\tmtwo_1 \,{}_{\Rule_1}\!\!\!\lto \tm \Rew{\Rule_2} \tmtwo_2$ implies $\tmtwo_1 \Rew{\Rule_2} \tmthree \, {}_{\Rule_2}\!\!\!\lto \tmtwo_2$ for some $\tmthree$. 
If $\Rew{\Rule_1}$ and $\Rew{\Rule_2}$ strongly commute and are diamond, then 
\begin{enumerate}
	\item $\Rew{\Rule} \, = \, \Rew{\Rule_1} \!\cup \Rew{\Rule_2}$ is diamond,
	\item all $\Rule$-evaluations with the same start and end terms have the same number of any kind of steps (\ie if $\deriv \colon \tm \Rew{\Rule}^* \tmtwo$  and $\deriv' \colon \tm \Rew{\Rule}^* \tmtwo$ then $\size{\deriv}_{\Rule_1} = \size{\deriv'}_{\Rule_1}$ and $\size{\deriv}_{\Rule_2} = \size{\deriv'}_{\Rule_2}$).
\end{enumerate}

It is a strong form of confluence and implies \emph{uniform normalization} (if there is a normalizing sequence from $\tm$ then there are no diverging sequences from $\tm$) and the 
\emph{random descent property} (all normalizing sequences from $\tm$ have the same length)

%
\section{Proofs of \Cref{sect:vsc}}

\begin{lemma}[Basic Properties of $\vsubcalc$]
	\label{l:basic-value-substitution}
	\hfill
	\begin{enumerate}
		\item\label{p:basic-value-substitution-tom-toe-terminates} $\tom$ and $\toe$ are strongly normalizing (separately).
		\item\label{p:basic-value-substitution-tom-toe-diamond-open} $\tomo$ and $\toeo$ are diamond  (separately).
		
		\item\label{p:basic-value-substitution-tom-toe-commute-open}  $\tomo$ and $\toeo$ strongly commute.
	\end{enumerate}
\end{lemma}

\begin{proof}
	The statements of \reflemma{basic-value-substitution} are a refinement of some results proved in \cite{AccattoliPaolini12}, where $\tovsubo$ is denoted by $\to_\mathsf{w}$.
	\begin{enumerate}
		\item See \cite[Lemma~3]{AccattoliPaolini12}.
		
		\item We prove that $\tomo$ is diamond, \ie if $\tmtwo \lRew{\wmsym} \tm \tomo \tmthree$ with $\tmtwo \neq \tmthree$ then there exists $\tmp \in \Lambda_\vsub$ such that $\tmtwo \tomo \tmp \lRew{\wmsym} \tmthree$.
		The proof is by induction on the definition of $\tomo$. 
		Since there $\tm \tomo \tmthree \neq \tmtwo$ and the reduction $\tomo$ is weak, there are only eight cases:
		\begin{itemize}
			\item \emph{Step at the Root for $\tm \!\tomo\! \tmtwo$ and Application Right for $\tm \!\tomo\! \tmthree$}, \ie $\tm \defeq \sctxp{\la\var\tmfive}\tmfour \rtom \sctxp{\tmfive\esub{\var}{\tmfour}} \eqdef \tmtwo$ and $\tm \!\rtom\! \sctxp{\la\var\tmfive}\tmfourp\! \eqdef \tmthree$ with $\tmfour \!\tomo\! \tmfourp$: then, $\tmtwo \!\tomo\! \sctxp{\tmfive\esub{\var}{\tmfourp}} \!\lRew{\wmsym}\! \tmthree$;
			
			\item \emph{Step at the Root for $\tm \tomo \tmtwo$ and Application Left for $\tm \tomo \tmthree$}, \ie, for some $n > 0$, 
			$$
			\tm \defeq (\la\var\tmfive)\esub{\var_1}{\tm_1}\dots\esub{\var_n}{\tm_n}\tmfour \allowbreak\rtom \tmfive\esub{\var}{\tmfour}\esub{\var_1}{\tm_1}\dots\esub{\var_n}{\tm_n} \eqdef \tmtwo
			$$
			whereas $\tm \tomo \allowbreak (\la\var\tmfive)\esub{\var_1}{\tm_1}\dots\esub{\var_j}{\tmp_j}\dots\esub{\var_n}{\tm_n}\tmfour \eqdef \tmthree$ with $\tm_j \tomo \tmp_j$ for some $1 \leq j \leq n$: then, 
			\begin{align*}
			\tmtwo \tomo \allowbreak \tmfive\esub{\var}{\tmfour}\esub{\var_1}{\tm_1}\dots\esub{\var_j}{\tmp_j}\dots\esub{\var_n}{\tm_n} \lRew{\wmsym} \tmthree;
			\end{align*}
			\item \emph{Application Left for $\tm \tomo \tmtwo$ and Application Right for $\tm \tomo \tmthree$}, \ie $\tm \defeq \tmfour\tmfive \tomo \tmfourp\tmfive \eqdef \tmtwo$ and $\tm \tomo \tmfour\tmfivep \eqdef \tmthree$ with $\tmfour \tomo \tmfourp$ and $\tmfive \tomo \tmfivep$: then, $\tmtwo \tomo \tmfourp\tmfivep\! \lRew{\wmsym} \tmthree$;
			\item \emph{Application Left for both $\tm \tomo \tmtwo$ and $\tm \tomo \tmthree$}, \ie $\tm \defeq \tmfour\tmfive \tomo \tmfourp\tmfive \eqdef \tmtwo$ and $\tm \tomo \tmfour''\tmfive \eqdef \tmthree$ with $\tmfourp \lRew{\wmsym} \tmfour \tomo \tmfour''$: by \ih, there exists $\tmfour_0 \in \Lambda_\vsub$ such that $\tmfourp \tomo \tmfour_0 \lRew{\msym} \tmfour''$, hence $\tmtwo \tomo \tmfour_0\tmfive \lRew{\msym} \tmthree$;
			\item \emph{Application Right for both $\tm \tomo \tmtwo$ and $\tm \tomo \tmthree$}, \ie $\tm \defeq \tmfive\tmfour \tomo \tmfive\tmfourp \eqdef \tmtwo$ and $\tm \tomo \tmfive\tmfour'' \eqdef \tmthree$ with $\tmfourp \lRew{\wmsym} \tmfour \tomo \tmfour''$: by \ih, there exists $\tmfour_0 \in \Lambda_\vsub$ such that $\tmfourp \tomo \tmfour_0 \lRew{\wmsym} \tmfour''$, hence $\tmtwo \tomo \tmfive\tmfour_0 \lRew{\wmsym} \tmthree$;
			\item \emph{$\mathsf{ES}$ left for $\tm \tomo \tmtwo$ and $\mathsf{ES}$ right for $\tm \tomo \tmthree$}, \ie $\tm \defeq \tmfour\esub\var\tmfive \tomo \tmfourp\esub\var\tmfive \eqdef \tmtwo$ and $\tm \tomo \tmfour\esub\var\tmfivep \eqdef \tmthree$ with $\tmfour \tomo \tmfourp$ and $\tmfive \tomo \tmfivep$: then, 
			$$
			\tmtwo \tomo \tmfourp\esub\var\tmfivep\! \lRew{\wmsym} \tmthree
			$$
			\item \emph{$\mathsf{ES}$ left for both $\tm \tomo \tmtwo$ and $\tm \tomo \tmthree$}, \ie $\tm \defeq \tmfour\esub\var\tmfive \tomo \tmfourp\esub\var\tmfive \eqdef \tmtwo$ and $\tm \tomo \tmfour''\esub\var\tmfive \eqdef \tmthree$ with $\tmfourp \lRew{\wmsym} \tmfour \tomo \tmfour''$: by \ih, there exists $\tmfour_0 \in \Lambda_\vsub$ such that $\tmfourp \tomo \tmfour_0 \lRew{\wmsym} \tmfour''$, hence $\tmtwo \tom \tmfour_0\esub\var\tmfive \lRew{\wmsym} \tmthree$;
			\item \emph{$\mathsf{ES}$ right for both $\tm \tomo \tmtwo$ and $\tm \tomo \tmthree$}, \ie $\tm \defeq \tmfive\esub\var\tmfour \tomo \tmfive\esub\var\tmfourp \eqdef \tmtwo$ and $\tm \tomo \tmfive\esub\var{\tmfour''} \eqdef \tmthree$ with $\tmfourp \lRew{\wmsym} \tmfour \tom \tmfour''$: by \ih, there exists $\tmfour_0 \in \Lambda_\vsub$ such that $\tmfourp \tomo \tmfour_0 \lRew{\wmsym} \tmfour''$, hence $\tmtwo \tom \tmfive\esub\var{\tmfour_0} \lRew{\wmsym} \tmthree$.
		\end{itemize}
		
		\smallskip
		We prove that $\toeo$ is diamond, \ie if $\tmtwo \lRew{\wesym} \tm \toeo \tmthree$ with $\tmtwo \neq \tmthree$ then there exists $\tmfour \in \Lambda_\vsub$ such that $\tmtwo \toeo \tmp \lRew{\esym} \tmthree$.
		The proof is by induction on the definition of $\toeo$. 
		Since there $\tm \toeo \tmthree \neq \tmtwo$ and the reduction $\toeo$ is weak, there are only eight cases:
		\begin{itemize}
			\item \emph{Step at the Root for $\tm \!\toeo\! \tmtwo$} and \emph{$\mathsf{ES}$ left for $\tm \!\toeo\! \tmthree$}, \ie $\tm \defeq \tmfour\esub\var{\sctxp{\val}} \rtoe \sctxp{\tmfour\isub{\var}{\val}} \eqdef \tmtwo$ and $\tm \!\rtoe\! \tmfourp\esub\var{\sctxp{\val}}\! \eqdef \tmthree$ with $\tmfour \!\toeo\! \tmfourp$: then, 
			$$
			\tmtwo \!\toeo\! \sctxp{\tmfourp\esub{\var}{\val}} \!\lRew{\wesym}\! \tmthree
			$$
			
			\item \emph{Step at the Root for $\tm \toeo \tmtwo$} and \emph{$\mathsf{ES}$ right for $\tm \toeo \tmthree$}, \ie, for some $n > 0$, $\tm \defeq \tmfour\esub\var{\val\esub{\var_1}{\tm_1}\dots\esub{\var_n}{\tm_n}} \allowbreak\rtoe \tmfour\isub{\var}{\val}\esub{\var_1}{\tm_1}\dots\esub{\var_n}{\tm_n} \eqdef \tmtwo$ whereas $\tm \toeo \allowbreak \tmfour\esub{\var}{\val\esub{\var_1}{\tm_1}\dots\esub{\var_j}{\tmp_j}\dots\esub{\var_n}{\tm_n}} \eqdef \tmthree$ with $\tm_j \toeo \tmp_j$ for some $1 \leq j \leq n$: then, 
			\begin{align*}
			\tmtwo \toeo \allowbreak \tmfour\isub{\var}{\val}\esub{\var_1}{\tm_1}\dots\esub{\var_j}{\tmp_j}\dots\esub{\var_n}{\tm_n} \lRew{\wesym} \tmthree;
			\end{align*}
			\item \emph{Application Left for $\tm \toeo \tmtwo$} and \emph{Application Right for $\tm \toeo \tmthree$}, \ie $\tm \defeq \tmfour\tmfive \toeo \tmfourp\tmfive \eqdef \tmtwo$ and $\tm \toeo \tmfour\tmfivep \eqdef \tmthree$ with $\tmfour \toeo \tmfourp$ and $\tmfive \toeo \tmfivep$: then, $\tmtwo \toeo \tmfourp\tmfivep\! \lRew{\wesym} \tmthree$;
			\item \emph{Application Left for both $\tm \toeo \tmtwo$ and $\tm \toeo \tmthree$}, \ie $\tm \defeq \tmfour\tmfive \toeo \tmfourp\tmfive \eqdef \tmtwo$ and $\tm \toeo \tmfour''\tmfive \eqdef \tmthree$ with $\tmfourp \lRew{\wesym} \tmfour \toeo \tmfour''$: by \ih, there exists $\tmfour_0 \in \Lambda_\vsub$ such that $\tmfourp \toeo \tmfour_0 \lRew{\wesym} \tmfour''$, hence $\tmtwo \toeo \tmfour_0\tmfive \lRew{\wesym} \tmthree$;
			\item \emph{Application Right for both $\tm \toeo \tmtwo$ and $\tm \toeo \tmthree$}, \ie $\tm \defeq \tmfive\tmfour \toeo \tmfive\tmfourp \eqdef \tmtwo$ and $\tm \toeo \tmfive\tmfour'' \eqdef \tmthree$ with $\tmfourp \lRew{\wesym} \tmfour \toeo \tmfour''$: by \ih, there exists $\tmfour_0 \in \Lambda_\vsub$ such that $\tmfourp \toeo \tmfour_0 \lRew{\wesym} \tmfour''$, hence $\tmtwo \toeo \tmfive\tmfour_0 \lRew{\wesym} \tmthree$;
			\item \emph{$\mathsf{ES}$ left for $\tm \toeo \tmtwo$} and \emph{$\mathsf{ES}$ right for $\tm \toeo \tmthree$}, \ie $\tm \defeq \tmfour\esub\var\tmfive \toeo \tmfourp\esub\var\tmfive \eqdef \tmtwo$ and $\tm \toeo \tmfour\esub\var\tmfivep \eqdef \tmthree$ with $\tmfour \toeo \tmfourp$ and $\tmfive \toeo \tmfivep$: then, $\tmtwo \toeo \tmfourp\esub\var\tmfivep\! \lRew{\wesym} \tmthree$;
			\item \emph{$\mathsf{ES}$ left for both $\tm \toe \tmtwo$ and $\tm \toe \tmthree$}, \ie $\tm \defeq \tmfour\esub\var\tmfive \toe \tmfourp\esub\var\tmfive \eqdef \tmtwo$ and $\tm \toe \tmfour''\esub\var\tmfive \eqdef \tmthree$ with $\tmfourp \lRew{\esym} \tmfour \toe \tmfour''$: by \ih, there exists $\tmfour_0 \in \Lambda_\vsub$ such that $\tmfourp \toe \tmfour_0 \lRew{\esym} \tmfour''$, hence $\tmtwo \toe \tmfour_0\esub\var\tmfive \lRew{\esym} \tmthree$;
			\item \emph{$\mathsf{ES}$ right for both $\tm \toe \tmtwo$ and $\tm \toe \tmthree$}, \ie $\tm \defeq \tmfive\esub\var\tmfour \toe \tmfive\esub\var\tmfourp \eqdef \tmtwo$ and $\tm \toe \tmfive\esub\var{\tmfour''} \eqdef \tmthree$ with $\tmfourp \lRew{\esym} \tmfour \toe \tmfour''$: by \ih, there exists $\tmfour_0 \in \Lambda_\vsub$ such that $\tmfourp \toe \tmfour_0 \lRew{\esym} \tmfour''$, hence $\tmtwo \toe \tmfive\esub\var{\tmfour_0} \lRew{\esym} \tmthree$.
		\end{itemize}
		
		\smallskip
		
		Note that in \cite[Lemma~11]{AccattoliPaolini12} it has just been proved the strong confluence of $\tovsub$, not of $\tom$ or $\toe$.
		
		\item We show that $\toeo$ and $\tomo$ strongly commute, \ie if $\tmtwo \lRew{\wesym} \tm \tomo \tmthree$, then $\tmtwo \neq \tmthree$ and there is $\tmp \in \Lambda_\vsub$ such that $\tmtwo \tomo \tmp \lRew{\wesym} \tmthree$. 
		The proof is by induction on the definition of $\tm \toeo \tmtwo$. 
		The proof that $\tmtwo \neq \tmthree$ is left to the reader.
		Since the $\toe$ and $\tom$ cannot reduce under $\l$'s, all values are $\omsym$-normal and $\oesym$-normal. So, there are the following cases.
		\begin{itemize}
			\item \emph{Step at the Root for $\tm \toeo \tmtwo$} and \emph{$\mathsf{ES}$ left for $\tm \tomo \tmthree$}, \ie $\tm \defeq \tmfour\esub\varthree{\sctxp\val} \toeo \sctxp{\tmfour\isub\varthree{\val}} \eqdef \tmtwo$ and $\tm \tomo \tmfourp\esub\varthree{\sctxp{\val}} \eqdef \tmthree$ with $\tmfour \tomo \tmfourp$: then 
			$$
			\tmtwo \tomo \sctxp{\tmfourp\isub\varthree{\val}} \lRew{\wesym} \tmtwo
			$$
			\item \emph{Step at the Root for $\tm \toeo \tmtwo$} and \emph{$\mathsf{ES}$ right for $\tm \tomo \tmthree$}, \ie 
			\begin{align*}
			\tm &\defeq \tmfour\esub\varthree{\val\esub{\var_1}{\tm_1}\dots\esub{\var_n}{\tm_n}} \\
			&\toeo \tmfour\isub\varthree{\val}\esub{\var_1}{\tm_1}\dots\esub{\var_n}{\tm_n} \eqdef \tmtwo
			\end{align*}
			and $\tm \tomo \tmfour\esub\varthree{\val\esub{\var_1}{\tm_1}\dots\esub{\var_j}{\tmp_j}\dots\esub{\var_n}{\tm_n}} \eqdef \tmthree$
			for some $n > 0$, and $\tm_j \tomo \tmp_j$ for some $1 \leq j \leq n$: then, $\tmtwo \tomo \tmfour\isub\varthree{\val}\esub{\var_1}{\tm_1}\dots\esub{\var_j}{\tmp_j}\dots\esub{\var_n}{\tm_n} \lRew{\wesym} \tmthree$;
			
			\item \emph{Application Left for $\tm \toe \tmtwo$} and \emph{Application Right for $\tm \tom \tmthree$}, \ie $\tm \defeq \tmfour\tmfive \toeo \tmfourp\tmfive \eqdef \tmtwo$ and $\tm \tomo \tmfour\tmfivep \eqdef \tmthree$ with $\tmfour \toe \tmfourp$ and $\tmfive \tomo \tmfivep$: then, $\tm \tomo \tmfourp\tmfivep \lRew{\wesym} \tmtwo$;
			\item \emph{Application Left for both $\tm \toeo \tmtwo$ and $\tm \tomo \tmthree$}, \ie $\tm \defeq \tmfour\tmfive \toeo \tmfourp\tmfive \eqdef \tmtwo$ and $\tm \tomo \tmfour''\tmfive \eqdef \tmthree$ with $\tmfourp \lRew{\wesym} \tmfour \tomo \tmfour''$: by \ih, there exists $\tmsix \in \Lambda_\vsub$ such that $\tmfourp \tomo \tmsix \lRew{\wesym} \tmfour''$, hence $\tmtwo \tomo \tmsix\tmfive \lRew{\wesym} \tmthree$;
			\item \emph{Application Left for $\tm \toeo \tmtwo$} and \emph{Step at the Root for $\tm \tomo \tmthree$}, \ie $\tm \defeq (\la\var\tmfive){\esub{\var_1}{\tm_1}\dots\esub{\var_n}{\tm_n}}\tmfour \toeo (\la\var\tmfive){\esub{\var_1}{\tm_1}\dots\esub{\var_j}{\tmp_j}\dots\esub{\var_n}{\tm_n}}\tmfour \eqdef \tmtwo$ with $n > 0$ and $\tm_j \toeo \tmp_j$ for some $1 \leq j \leq n$, and 
			$$
			\tm \tomo \allowbreak \tmfive\esub\var\tmfour{\esub{\var_1}{\tm_1}\dots\esub{\var_n}{\tm_n}} \eqdef \tmthree
			$$ 
			Then,
			\begin{equation*}
			\tmtwo \tomo \tmfive\esub\var\tmfour{\esub{\var_1}{\tm_1}\dots\esub{\var_j}{\tmp_j}\dots\esub{\var_n}{\tm_n}} \lRew{\wesym} \tmthree;
			\end{equation*}
			\item \emph{Application Right for $\tm \toeo \tmtwo$} and \emph{Application Left for $\tm \tom \tmthree$}, \ie $\tm \defeq \tmfive\tmfour \toeo \tmfive\tmfourp \eqdef \tmtwo$ and $\tm \tomo \tmfivep\tmfour \eqdef \tmthree$ with $\tmfour \toeo \tmfourp$ and $\tmfive \tomo \tmfivep$: then, $\tmtwo \tomo \tmfivep\tmfourp \lRew{\wesym} \tmthree$;
			\item \emph{Application Right for both $\tm \toeo \tmtwo$ and $\tm \tomo \tmthree$}, \ie $\tm \defeq \tmfive\tmfour \toeo \tmfive\tmfourp \eqdef \tmtwo$ and $\tm \tomo \tmfive\tmfour'' \eqdef \tmthree$ with $\tmfourp \lRew{\wesym} \tmfour \tomo \tmfour''$: by \ih, there exists $\tmsix \in \Lambda_\vsub$ such that $\tmfourp \tomo \tmsix \lRew{\wesym} \tmfour''$, hence $\tmtwo \tomo \tmfive\tmsix \lRew{\wesym} \tmthree$;
			\item \emph{Application Right for $\tm \toeo \tmtwo$} and \emph{Step at the Root for $\tm \tomo \tmthree$}, \ie $\tm \defeq \sctxp{\la\var\tmfive}\tmfour \toeo \sctxp{\la\var\tmfive}\tmfourp \eqdef \tmtwo$ with $\tmfour \toeo \tmfourp$\!, and $\tm \tomo \allowbreak \sctxp{\tmfive\esub\var\tmfour} \eqdef \tmthree$: then, $\tmtwo \tomo \sctxp{\tmfive\esub\var\tmfourp} \lRew{\wesym} \tmthree$;
			\item \emph{$\mathsf{ES}$ left for $\tm \toeo \tmtwo$} and \emph{$\mathsf{ES}$ right for $\tm \tomo \tmthree$}, \ie $\tm \defeq \tmfour\esub\var\tmfive \toeo \tmfourp\esub\var\tmfive \eqdef \tmtwo$ and $\tm \tomo \tmfour\esub\var\tmfivep \eqdef \tmthree$ with $\tmfour \toe \tmfourp$ and $\tmfive \tomo \tmfivep$: then, $\tmtwo \tomo \tmfourp\esub\var\tmfivep \lRew{\wesym} \tmthree$;
			\item \emph{$\mathsf{ES}$ left for both $\tm \toeo \tmtwo$ and $\tm \tomo \tmthree$}, \ie $\tm \defeq \tmfour\esub\var\tmfive \toe \tmfourp\esub\var\tmfive \eqdef \tmtwo$ and $\tm \tomo \tmfour''\esub\var\tmfive \eqdef \tmthree$ with $\tmfourp \lRew{\wesym} \tmfour \tomo \tmfour''$: by \ih,there is $\tmsix \in \Lambda_\vsub$ such that $\tmfourp \tomo \tmsix \lRew{\wesym} \tmfour''$, hence $\tmtwo \tomo \tmsix\esub\var\tmfive \lRew{\wesym} \tmthree$;
			\item \emph{$\mathsf{ES}$ right for $\tm \toeo \tmtwo$} and \emph{$\mathsf{ES}$ left for $\tm \tomo \tmthree$}, \ie $\tm \defeq \tmfive\esub\var\tmfour \toeo \tmfive\esub\var\tmfourp \eqdef \tmtwo$ and $\tm \tomo \tmfivep\esub\var\tmfour \eqdef \tmthree$ with $\tmfour \toeo \tmfourp$ and $\tmfive \tomo \tmfivep$: then, $\tmtwo \tomo \tmfivep\esub\var\tmfourp \lRew{\wesym} \tmthree$;
			\item \emph{$\mathsf{ES}$ right for both $\tm \toeo \tmtwo$ and $\tm \tomo \tmthree$}, \ie $\tm \defeq \tmfive\esub\var\tmfour \toeo \tmfive\esub\var{\tmfourp} \eqdef \tmtwo$ and $\tm \tomo \tmfive\esub\var{\tmfour''} \eqdef \tmthree$ with $\tmfour \lRew{\esym} \tmfourp \tomo \tmfour''$: by \ih, there is $\tmsix \in \Lambda_\vsub$ such that $\tmfour \tomo \tmsix \lRew{\wesym} \tmfour''$, so $\tmtwo \tomo \tmfive\esub\var\tmsix \lRew{\wesym} \tmthree$.
			\qedhere
		\end{itemize}
	\end{enumerate}
\end{proof}

\begin{proposition}[Properties of the open reduction]
	\label{propappendix:properties-open-reduction}
	\NoteState{prop:properties-open-reduction}
	\begin{enumerate}
		\item \label{pappendix:properties-open-reduction-diamond} The reduction $\tovsubo$ is diamond.
		\item \label{pappendix:properties-open-redction-harmony} A term is $\osym$-normal if and only if it is a fireball, where \emph{fireballs} (and \emph{inert terms}) are defined by:
	\end{enumerate}
	\begin{alignat*}{5}
	\textsc{Inert term} \ & \itm &\grameq \var \mid \itm \fire \mid \itm \esub{\var}{\itmtwo}
	& \quad
	\textsc{Fireball} \ & \fire &\grameq \val \mid \itm \mid \fire \esub{\var}{\itm}
	\end{alignat*}
\end{proposition}

\begin{proof}
	\begin{enumerate}
		\item Strong commutation of $\tomo$ and $\toeo$ is proven in \reflemmap{basic-value-substitution}{tom-toe-commute-open}.
		Diamond of $\tovsubo$ follows from that, from diamond for $\tomo$ and $\toeo$ (\reflemmap{basic-value-substitution}{tom-toe-diamond-open}), 
		and from Hindley-Rosen lemma (\cite[Prop. 3.3.5]{Barendregt84}).
		
		\item See \cite[Lemma~5]{AccattoliPaolini12}, where $\tovsubo$ is denoted by $\to_\mathsf{w}$.
		\qedhere
	\end{enumerate}
	
\end{proof}

\begin{lemma}[Sizes for inert terms]
	\label{l:sizes-inert}
	For every inert term $\itm$, $\sizeo{\itm} = \sizes{\itm}$.
\end{lemma}

\begin{proof}
	By induction of the inert term $\itm$.
	Cases:
	\begin{itemize}
		\item \emph{Variable}, \ie $\itm = \var$. 
		Then, $\sizeo{\itm} = 0 = \sizes{\itm}$.
		
		\item \emph{Application}, \ie $\itm = \itmtwo \fire$. 
		By \ih, $\sizeo{\itmtwo} = \sizes{\itmtwo}$.
		Hence, $\sizeo{\itm} = 1 + \sizeo{\itmtwo} + \sizeo{\fire} = 1 + \sizes{\itmtwo} + \sizeo{\fire} = \sizes{\itm}$.
		
		\item \emph{Explicit substitution}, \ie $\itm = \itmtwo \esub{\var}{\itmthree}$. 
		By \ih, $\sizeo{\itmtwo} = \sizes{\itmtwo}$ and $\sizeo{\itmthree} = \sizes{\itmthree}$.
		Thus, $\sizeo{\itm} = \sizeo{\itmtwo} + \sizeo{\itmthree} = \sizes{\itmtwo} + \sizeo{\itmthree} = \sizes{\itm}$.
		\qedhere
	\end{itemize}
\end{proof}

\begin{lemma}[Shape of strong fireballs]
	\label{l:shape-of-strong-fireballs}
	Let $\tm$ be a \full fireball. Then exactly one of the following holds:
	\begin{itemize}
		\item either $\tm$ is a \full inert term,
		
		\item or $\tm$ is a \full value.
		
	\end{itemize}
\end{lemma}

\begin{proof}
	Proving that at least one of the two holds is left to the reader. 
	We now prove that only one of them holds:
	
	\begin{itemize}
		\item Let $\tm$ be a \full inert term. We prove that $\tm$ is not a \full value by structural induction on $\tm$:
		\begin{itemize}
			\item \emph{Variable}: Trivial.
			
			\item \emph{Application}: Trivial.
			
			\item \emph{$\mathsf{ES}$}; \ie, $\tm = \sitm \esub{\var}{\sitmtwo}$: Then $\sitm$ is not a \full value---by \ih---, and so neither is $\tm$.
			
		\end{itemize}
		
		\item Let $\tm = \sctxp{\la{\var}{\sfire}}$, with $\sctx = \esub{\var_{1}}{{\sitm}_{1}} \dots \esub{\var_{n}}{{\sitm}_{n}}$, with $n \geq 0$. 
		We prove that $\tm$ is not a \full inert term by induction on $n$:
		\begin{itemize}
			\item \emph{Empty substitution context}; \ie, $\sctx = \ctxhole$: Trivial.
			
			\item \emph{Non-empty substitution context}; \ie, $\sctx = \sctxtwo \esub{\var}{{\sitm}_{n+1}}$: As $\sctxtwop{\la{\var}{\fire}}$ is not a \full inert term by \ih, then neither is $\sctxtwop{\la{\var}{\fire}} \esub{\var}{{\sitm}_{n+1}} = \tm$.
			\qedhere
		\end{itemize}
		
	\end{itemize}
\end{proof}

\begin{proposition}[Properties of the \full reduction] 
	\label{propappendix:properties-full-reduction}
	\NoteState{prop:properties-full-reduction}
	\hfill
	\begin{enumerate}
		\item \label{pappendix:properties-full-reduction-confluence} The reduction $\tovsub$ is confluent.
		
		\item \label{pappendix:properties-full-reduction-harmony} A term is $\vsub$-normal if and only if it is a \full fireball, where \emph{\full fireballs} (and \emph{\full inert terms}, \emph{\full values}) are:
	\end{enumerate}
\end{proposition}

\begin{proof}
	\begin{enumerate}
		\item See \cite[Corollary 1]{AccattoliPaolini12}.
		
		\item We prove the two directions of the equivalence separately.
		\begin{enumerate}
		\item \label{p:asfa} Let $\tm$ be $\vsub$-normal. We shall prove that $\tm$ is a \full fireball by induction on $\tm$:
		\begin{itemize}
			\item \emph{Variable}. Trivial.
			
			\item \emph{Abstraction}; \ie, $\tm = \la{\var}{\tmtwo}$. 
			As $\tm$ is $\tovsubval$-normal, so is $\tmtwo$.
			By \ih, $\tmtwo$ is a \full fireball, and then so~is~$\tm$.
			
			\item \emph{Application}; \ie, $\tm = \tm_{1} \tm_{2}$. 
			Since $\tm$ is $\tovsubval$-normal,  so are $\tm_1$ and $\tm_2$.
			By \ih, $\tm_{1}$ and $\tm_{2}$ are \full fireballs. 
			Note that  $\tm_{1} \neq \sctxp{\la{\var}{\tmtwo}}$, otherwise $\tm \rtom \sctxp{\tmtwo \esub{\var}{\tm_{2}}}$ which contradicts $\vsub$-normality of $\tm$. 
			So, $\tm_{1}$ is a \full inert term (by \Cref{l:shape-of-strong-fireballs}), thus $\tm$ is a \full fireball.
			
			\item \emph{Explicit substitution}; \ie, $\tm = \tm_{1} \esub{\var}{\tm_{2}}$. 
			Since $\tm$ is $\vsub$-normal, then so are $\tm_1$ and $\tm_2$.
			By \ih, $\tm_{1}$ and $\tm_{2}$ are \full fireballs. 
			Note that $\tm_{2} \neq \sctxp{\val}$, otherwise $\tm \rtoe \sctxp{\tm_{1} \isub{\var}{\val}}$ which contradicts the $\vsub$-normality of $\tm$. 
			Thus $\tm_{2}$ is a \full inert term (by \Cref{l:shape-of-strong-fireballs}), and so $\tm$ is a \full fireball.
			
		\end{itemize}
		
		\item\label{p:fireball-to-normal} Let $\tm$ be a \full fireball. 
		We prove that $\tm$ is $\vsub$-normal by induction on the definition of \full fireball.
		\begin{itemize}
			\item \emph{Variable}. Trivial.
			
			\item \emph{Abstraction}; \ie, $\tm = \la{\var}{\sfire}$. 
			By \ih, $\sfire$ is $\vsub$-normal, and hence so is $\tm$.
			
			\item \emph{Application}; \ie, $\tm = \sitm \sfire$.
			By \ih, $\sitm $ and $\sfire$  are $\vsub$-normal. 
			Since $\sitm$ is not of the of the form $\sctxp{\la{\var}{\tmtwo}}$ (by \Cref{l:shape-of-strong-fireballs}), then $\tm$ is also $\vsub$-normal.
			
			\item \emph{Explicit substitution}; \ie, $\tm = \sfire \esub{\var}{\sitm}$ (it includes the case when $\sfire$ is a \full inert term). 
			By \ih, both $\sfire$ and $\sitm$ are $\vsub$-normal. 
			Since $\sitm$ is not of the form $\sctxp{\val}$ (by \reflemma{shape-of-strong-fireballs}), then $\tm$ is also $\vsub$-normal.
			\qedhere
		\end{itemize}
	\end{enumerate}
	\end{enumerate}
\end{proof}

\begin{lemma}[Strong normalization of $\toevar$]
	\label{l:strongly-normalizing}
	The reduction $\toe \cup \toevar$ is strongly normalizing.
\end{lemma}

\begin{proof}
	Each $\toevar$ step strictly decreases the number of \ES. This shows that $\toevar$ is strongly normalizing.
	A proof that $\toe \cup \toevar$ is strongly normalizing is in \cite[Lemma~3]{AccattoliPaolini12}, where it is denoted $\to_\mathsf{vs}$.
\end{proof}

\begin{lemma}[Preservation of normal forms]
	\label{l:preservation-normal}
	Let $\tm \toevar \tmtwo$. Then, $\tm$ is $\vsub$-normal if and only if $\tmtwo$ is $\vsub$-normal.
\end{lemma}

\begin{proof}
	By induction on the definition of $\tm \toevar \tmtwo$. 
\end{proof}

\begin{lemma}[Swaps]
	\label{l:swaps}
	\hfill
	\begin{enumerate}
		\item $\toevar\tom \,\subseteq\, \tom\toevar$;
		\item $\toevar\toe \,\subseteq\, \toe\toevar^* \!\cup \toe\toe$;
		\item $\toevar^*\toe^+ \,\subseteq\, \toe^+\toevar^*$.
	\end{enumerate}
\end{lemma}

\begin{proof}
	The only non trivial point is the third one, that we now prove. The hypothesis is $\tm \toevar^k\toe^h\tmtwo$ for some $k,h$. 
	The proof is by lexicographic induction on $(n,k)$, where $n$ is the length of the longest $\toe \cup\toevar$ sequence from $\tm$ (it exists because $\toe \cup \toevar$ is strongly normalizing, \Cref{l:strongly-normalizing}). Now the case $k=h=1$ is given by the second point. Then consider $\tm\toevar^k\toevar\toe\toe^h\tmtwo$, and apply the swap of the second point to the central pair. There are two cases:
	\begin{enumerate}
		\item $\toevar\toe \subseteq \toe\toevar^*$, and so  $\tm\toevar^k\toe\toevar^*\toe^h\tmtwo$. 
		Now by \ih (on $k$) applied to the prefix $\toevar^k\toe$ we obtain $\tm\toe^+\toevar^*\toevar^*\toe^h\tmtwo$. 
		By \ih (on the longest sequence), the suffix $\toevar^*\toevar^*\toe^h$ turns into $\toe^+ \toevar^*$, giving $\tm\toe^+\toe^+ \toevar^*\tmtwo$, that is, the statement holds.
		\item $\toevar\toe \subseteq \toe\toe$, and so $\tm\toevar^k\toe\toe\toe^h\tmtwo$. 
		By \ih (on $k$) applied to the whole sequence, $\tm\toe^+ \toevar^*\tmtwo$.
		\qedhere
	\end{enumerate}
\end{proof}

\begin{proposition}[Irrelevance of $\toevar$]
	\label{propappendix:irrelevance}
	\NoteState{prop:irrelevance}
	\hfill
	\begin{enumerate}
		\item\label{pappendix:irrelevance-postponement} \emph{Postponement}: whenever $\deriv \colon \tm \, (\tovsub \cup \toevar)^*\, \tmtwo$ then there is $\derivp \colon \tm \tovsub^* \!\cdot \toevar^* \, \tmtwo$ with $\sizem\derivp = \sizem \deriv$;
		\item\label{pappendix:irrelevance-termination} \emph{Termination}: a term is weakly (\resp strongly) normalizing for $\tovsub$ if and only if so is it for $\tovsub \!\cup \toevar$.
	\end{enumerate}
\end{proposition}

\begin{proof}\hfill
	\begin{itemize}
		\item \emph{Postponement}: By induction on the length $\size{\deriv}$ of the evaluation $\deriv:\tm\tovsub^*\tmtwo$, using \reflemma{swaps}.
		
		\item \emph{Termination:} Clearly, if $\tovsub \cup \toevar$ is strongly normalizing on $\tm$ then $\tovsub$ is strongly normalizing on $\tm$, since $\tovsub \,\subseteq\, \tovsub \cup \toevar$. 
		
		Conversely, suppose that $\tovsub \cup \toevar$ is not strongly normalizing on $\tm$. 
		Then, there is an infinite sequence of terms $(\tm_i)_{i \in \nat}$ such that $\tm_i \,(\tovsub \cup \toevar)\, \tm_{i+1}$ and $\tm = \tm_0$. 
		Hence, for every $n \in \nat$ there is an evaluation $\deriv_n \colon \tm \,(\tovsub \cup \toevar)^*\, \tm_{n}$ with $\sizem{\deriv_n} = n$ and, by postponement irrelevance (\Cref{p:irrelevance-postponement}), there is an evaluation $\deriv'_n \colon \tm \tovsub^* \tmtwo_{n}$ with $\sizem{\deriv'_n} = n$.
		Thus $\tovsub$ is not strongly normalizing on $\tm$ (by Konig's lemma). 
		
		If $\tovsub \cup \toevar$ is weakly normalizing on $\tm$ then $\tm \,(\tovsub \cup \toevar)^*\, \tmtwo$ for some $\tmtwo$ normal for $\tovsub \cup \toevar$.
		By postponement irrelevance (\Cref{p:irrelevance-postponement}), $\tm \tovsub^* \tmthree \toevar^* \tmtwo$ where $\tmthree$ is $\vsub$-normal by \cref{l:preservation-normal}.
		Therefore, $\tovsub$ is weakly normalizing on $\tm$.
		
		If $\tovsub$ is weakly normalizing on $\tm$ then $\tm \tovsub^* \tmthree$ for some $\tmthree$ that is $\vsub$-normal.
		Since $\toevar$ is strongly normalizing (\cref{l:strongly-normalizing}), $\tmthree \toevar^* \tmtwo$ for some $\tmtwo$ that is $\toevar$-normal.
		By \cref{l:preservation-normal}, $\tmtwo$ is also $\vsub$-normal. 
		Therefore, $\tovsub \cup \toevar$ is weakly normalizing on $\tm$.
		\qedhere 	
	\end{itemize}	
\end{proof}

\begin{proposition}[Simulation]
	\label{propappendix:plotkin-vsc}
	\NoteState{prop:plotkin-vsc}
	Let $\tm$ be a term without \ES. 
	\begin{enumerate}
		\item\label{pappendix:plotkin-vsc-step} 	If $\tm \tobvplot \tm'$ then $\tm \tom \!\cdot\, (\toe \cup \toevar) \ \tm'$.
		\item\label{pappendix:plotkin-vsc-steps} 	If $\tm \tobvplot^* \tm'$ then $\tm \tovsub^* \tmtwo  \toevar^* \tm'$ for some term $\tmtwo$.
	\end{enumerate}
\end{proposition}

\begin{proof}
	\begin{enumerate}
		\item By induction on the context $\ctx$ such that $\tm = \ctxp{\tmthree} \tobvplot \ctxp{\tmthree'} = \tm'$ with $\tmthree \rtobvplot \tmthree'$.
		
		If $\ctx = \ctxhole$, then we have the root-step $\tm = (\la{\var}\tmthree)\tval \rtobvplot \tmthree \isub{\var}{\tval} = \tm'$.  
		Then, $\tm \rtom \tmthree \esub{\var}{\tval} \rtoe \tmthree \isub{\var}{\tval} = \tm'$.
		
		The cases where $\ctx \neq \ctxhole$ follow  from the \ih easily.
		
		\item Let $\deriv \colon \tm \tobvplot^* \tm'$. We prove that $\tm \tovsub^* \tmtwo  \toevar^* \tm'$ for some term $\tmtwo$, by induction on $\size{\deriv}_{\betaplot} \!\in \nat$.
		
		If $n = 0$ then $\tm = \tm'$ and we are done taking $\tmtwo \defeq \tm$.
		
		Otherwise, $n > 0$ and $\tm \tobvplot \tmthree \tobvplot^* \tm'$ for some term $\tmthree$ without \ES. 
		By \ih, $\tmthree \tovsub^* \cdot \toevar^* \tm'$. 
		According to \refpoint{plotkin-vsc-step},  $\tm \tom \!\cdot\, (\toe \cup \toevar) \ \tmthree$.
		Therefore, $\tm \ (\tovsub \cup \toevar)^* \ \tm'$ and so, by \Cref{prop:irrelevance}.\ref{p:irrelevance-postponement}, $\tm \tovsub^* \tmtwo  \toevar^* \tm'$ for some term $\tmtwo$.
		\qedhere
	\end{enumerate}
	
\end{proof}

\section{Proofs of \Cref{sect:external}}

Note that \pointed terms are not abstractions.

\begin{remark}
	\label{rmk:rigid} Every \full inert term is a \pointed term, but the converse does not hold (take $\vartwo\,\la{\var}{\delta\delta}$).
\end{remark}

\begin{lemma}[Properties of \pointed terms]
\label{l:properties-of-rigid-terms}
\hfill
\begin{enumerate}
	\item \label{p:properties-of-rigid-terms-plugging-is-rigid} 
	For any term $\tm$ and any rigid context $\rctx$, $\rctxp{\tm}$ is a \pointed term.
	
	\item \label{p:properties-of-rigid-terms-open-strategy-preserves-rigid} Let $\rtm$ be a \pointed term and 
	$\tm$ be a term such that $\rtm \tovsubo \tm$. Then $\tm$ is \pointed.
	
	\item \label{p:properties-of-rigid-terms-strong-strategy-preserves-rigid} Let $\rtm$ be a \pointed term and 
	$\tm$ be a term such that $\rtm \tovsubs \tm$. Then $\tm$ is \pointed.

	\item \label{p:properties-of-rigid-terms-from-non-abstractions-in-open-nf} Let 
	$\tm$ be a term that is $\osym$-normal and is not an \valES. Then $\tm$ is \pointed.
	
	\item \label{p:properties-of-rigid-terms-rigid-not-in-strong-nf-has-plugging} 
	Let $\rtm$ be a \pointed term that is not $\esssym$-normal. 
	Then there exist a rigid context $\rctx$ and terms $\tmtwo, \tmtwop$ such that $\rtm = \rctxp{\tmtwo} \tovsubs \rctxp{\tmtwop}$ with $\tmtwo \tovsubo \tmtwop$.
\end{enumerate}
\end{lemma}

\begin{proof}
\begin{enumerate}
	\item By induction on the definition of $\rctx$.
	
	\item Let $\weakctx$ be an open evaluation context such that $\rtm = \weakctxp{\tmtwo} \tovsubo \weakctxp{\tmtwop} = \tm$, with $\tmtwo \rtom \tmtwop$ or $\tmtwo \rtoe \tmtwop$. 
	We proceed by induction on the open context $\weakctx$.
	\begin{itemize}
		\item \emph{Empty context}; \ie, $\weakctx = \ctxhole$. 
		This case is not possible, because it would imply that 
		\begin{itemize}
			\item either $\rtm = \tmtwo = \sctxp{\la{\var}{\tmthree}} \tmfour$ (if $\tmtwo \rtom \tmtwop$),
			\item or $\rtm = \tmtwo = {\tmthree}\esub{\var}{\sctxp{\val}}$ (if $\tmtwo \rtoe \tmtwop$),
		\end{itemize}
		and in both cases $\rtm$ is not a \pointed term.
		
		\item \emph{Application right}; \ie, $\weakctx = \tmfour \weakctxtwo$. 
		As $\rtm = \tmfour \weakctxtwop{\tmtwo}$, then $\tmfour$ is a \pointed term, and so $\tmfour \weakctxtwop{\tmtwop} = \tm$ is \pointed.
		
		\item \emph{Application left}; \ie, $\weakctx = \weakctxtwo \tmfour$. 
		As $\rtm = \weakctxtwop{\tmtwo} \tmfour$ and $\weakctxtwop{\tmtwop}$ is \pointed by \ih, then $\weakctxtwop{\tmtwop} \tmfour = \tm$ is \pointed too.
		
		\item \emph{$\mathsf{ES}$ left}; \ie, $\weakctx = \weakctxtwo \esub{\var}{\tmfour}$. Since $\rtm = \weakctxtwop{\tmtwo} \esub{\var}{\tmfour}$, then both $\weakctxtwop{\tmtwo}$ and $\tmfour$ are \pointed. Moreover, $\weakctxtwop{\tmtwop}$ is \pointed by \ih, and so $\weakctxtwop{\tmtwop} \esub{\var}{\tmfour} = \tm$ is \pointed too.
		
		\item \emph{$\mathsf{ES}$ right}; \ie, $\weakctx = \tmfour \esub{\var}{\weakctxtwo}$. Since $\rtm = \tmfour \esub{\var}{\weakctxtwop{\tmtwo}}$, then both $\tmfour$ and $\weakctxtwop{\tmtwo}$ are \pointed. Moreover,  $\weakctxtwop{\tmtwop}$ is \pointed by \ih, and so $\tmfour \esub{\var}{\weakctxtwop{\tmtwop}} = \tm$ is \pointed too.
		
	\end{itemize}
	
	\item Let $\evsctx$ be a strong evaluation context such that $\rtm = \evsctxp{\tmtwo} \tovsubs \evsctxp{\tmtwop} = \tm$, with $\tmtwo \tomo \tmtwop$ or $\tmtwo \toeo \tmtwop$. 
	We proceed by induction on the external context $\evsctx$.
	\begin{itemize}
		\item \emph{Empty context}; \ie, $\evsctx = \ctxhole$. Then $\rtm = \tmtwo \tomo \tmtwop = \tm$ or $\rtm = \tmtwo \toeo \tmtwop = \tm$, and the statement holds by \reflemmap{properties-of-rigid-terms}{open-strategy-preserves-rigid}
		
		\item \emph{Under $\lambda$-abstraction right}; \ie, $\evsctx = \la{\var}{\evsctxtwo}$. This case is not possible, because it would imply that $\rtm = \la{\var}{\evsctxtwop{\tmtwo}}$, which is not a \pointed term.
		
		\item \emph{Strong context, $\mathsf{ES}$ right}; \ie, $\evsctx = \tmfour \esub{\var}{\rctx}$. Since $\rtm = \tmfour \esub{\var}{\rctxp{\tmtwo}}$, then both $\tmfour$ and $\rctxp{\tmtwo}$ are \pointed terms. Moreover, $\rctxp{\tmtwop}$ is \pointed by \ih, and so $\tmfour \esub{\var}{\rctxp{\tmtwop}} = \tmthree$ is \pointed too.
		
		\item \emph{Strong context, $\mathsf{ES}$ left}; \ie, $\evsctx = \evsctxtwo \esub{\var}{\tmfour}$, with $\tmfour$ a \pointed term. Since $\rtm = \evsctxtwop{\tmtwo} \esub{\var}{\tmfour}$, then $\evsctxtwop{\tmtwo}$ is \pointed. Moreover, $\evsctxtwop{\tmtwop}$ is \pointed by \ih, and so $\evsctxtwop{\tmtwop} \esub{\var}{\tmfour} = \tmthree$ is \pointed too.
		
		\item \emph{Rigid context, application right}; \ie, $\evsctx = \tmfour \evsctxtwo$, with $\tmfour$ a \pointed term. Then $\tmfour \evsctxtwop{\tmtwop} = \tmthree$ is \pointed too.
		
		\item \emph{Rigid context, application left}; \ie, $\evsctx = \rctx \tmfour$. Since $\rtm = \rctxp{\tmtwo} \tmfour$, then $\rctxp{\tmtwo}$ is \pointed. Moreover, $\rctxp{\tmtwop}$ is \pointed by \ih, and so $\rctxp{\tmtwop} \tmfour$ is \pointed too.
		
		\item \emph{Rigid context, $\mathsf{ES}$ left}; \ie, $\evsctx = \rctx \esub{\var}{\tmfour}$, with $\tmfour$ a \pointed term. Since $\rtm = \rctxp{\tmtwo} \esub{\var}{\tmfour}$, then $\rctxp{\tmtwo}$ is \pointed. Moreover, $\rctxp{\tmtwop}$ is \pointed by \ih, and so $\rctxp{\tmtwop} \esub{\var}{\tmfour} = \tmthree$ is \pointed too.
		
		\item \emph{Rigid context, $\mathsf{ES}$ right}; \ie, $\evsctx = \tmfour \esub{\var}{\rctx}$, with $\tmfour$ a \pointed term. Since $\rtm = \tmfour \esub{\var}{\rctxp{\tmtwo}}$, then $\rctxp{\tmtwo}$ is \pointed. Moreover,  $\rctxp{\tmtwop}$ is \pointed by \ih, and so $\tmfour \esub{\var}{\rctxp{\tmtwop}} = \tmthree$ is \pointed too.
	\end{itemize}
	
	\item By structural induction on the term $\tm$. Cases:
		\begin{itemize}
			\item \emph{Variable}. Trivial.
			
			\item \emph{Abstraction}. The statement is vacuously true, because if $\tm$ is an abstraction then it is an \valES.
			
			\item \emph{Application}; \ie, $\tm = \tm_{1} \tm_{2}$. If $\tm_{1}$ is not in $\tovsubo$-normal form, then neither is $\tm$. Moreover, if $\tm_{1}$ is a $\lambda$-abstraction in a substitution context, then $\tm$ is not in $\rtom$-normal form, which is absurd. Thus, $\tm_{1}$ is \pointed---by \ih---and then so is $\tm$.
			
			\item \emph{$\mathsf{ES}$}; \ie, $\tm = \tm_{1} \esub{\var}{\tm_{2}}$. First of all, if $\tm_{1}$ is a $\lambda$-abstraction in a substitution context then so is $\tm$ ---which is absurd---, and if $\tm_{1}$ is not in $\tovsubo$-normal form then neither is $\tm$. Therefore, $\tm_{1}$ is a \pointed term ---by \ih. Second, if $\tm_{2}$ is an abstraction in a substitution context then $\tm$ is not in $\rtoe$-normal form ---which is absurd---, and if $\tm_{2}$ is not in $\tovsubo$-normal form then neither is $\tm$. Therefore, $\tm_{2}$ is a \pointed term ---by \ih. Thus, $\tm$ is a \pointed term as well.
			
		\end{itemize}
	
	\item By induction on the rigid term $\rtm$. Cases:
		\begin{itemize}
			\item \emph{Variable}. The statement is vacuously true, because if $\rtm$ is a variable then it is $\esssym$-normal.
			
			\item \emph{Application}; \ie, $\tm = \tm_{1} \tm_{2}$, with $\tm_{1}$ a \pointed term. Note that $\tm_{1}$ is not an \valES. 
			
			Let $\tm_{1}$ be not in $\esssym$-normal form. Then by \ih there exist a rigid context $\rctxtwo$ and terms $\tmthree, \tmthreep$ such that $\tm_{1} = \rctxthreep{\tmtwo} \tovsubs \rctxtwop{\tmthreep}$. Thus, the statement holds by taking $\rctx \defeq \rctxtwo \tm_{2}$, $\tmtwo \defeq \tmthree$ and $\tmtwop \defeq \tmthreep$.
			
			Let $\tm_{1}$ be in $\tovsubs$-normal form. Then $\tm_{2}$ is not $\esssym$-normal, which implies the existence of a external context $\evsctx$ and terms $\tmthree, \tmthreep$ such that $\tm_{2} = \evsctxp{\tmthree} \tovsubs \evsctxp{\tmthreep}$ and $\tmthree \tovsubo \tmthreep$, and so the statement holds by taking $\rctx \defeq \tm_{1} \evsctx$.
			
			\item \emph{$\mathsf{ES}$}; \ie, $\tm = \tm_{1} \esub{\var}{\tm_{2}}$, with both $\tm_{1}$ and $\tm_{2}$ \pointed terms. If $\tm_{1}$ is not $\esssym$-normal, then there exist---by \ih---a rigid context $\rctxtwo$ and terms $\tmthree, \tmthreep $ such that $\tm_{1} = \rctxtwop{\tmthree} \tovsubs \rctxtwop{\tmthreep}$. Thus, the statement holds by taking $\rctx \defeq \rctxtwo \esub{\var}{\tm_{2}}$, $\tmtwo \defeq \tmthree$ and $\tmtwop \defeq \tmthreep$.
			
			If $\tm_{1}$ is in $\esssym$-normal form, then $\tm_{2}$ is not in $\esssym$-normal form and so there exist rigid context $\rctxtwo$ and terms $\tmthree, \tmthreep$ such that $\tm_{2} = \rctxtwop{\tmthree} \tovsubs \rctxtwop{\tmthreep}$. Thus, the statement holds by taking $\rctx \defeq \tm_{1} \esub{\var}{\rctxtwo}$, $\tmtwo \defeq \tmthree$ and $\tmtwop \defeq \tmthreep$.
			\qedhere
		\end{itemize}
	
\end{enumerate}
\end{proof}

\begin{lemma}[Basic Properties of $\tovsubs$]
	\label{l:basic-value-substitution-external}
	\hfill
	\begin{enumerate}
%
		\item\label{p:basic-value-substitution-external-tom-toe-diamond-strong} $\toms$ and $\toes$ are diamond  (separately).
		
		\item\label{p:basic-value-substitution-external-tom-toe-commute-strong}  $\toms$ and $\toes$ strongly commute.
	\end{enumerate}
\end{lemma}

\begin{proof} 		
	\begin{enumerate}

		\item We prove that $\toms$ is diamond, \ie if $\tmtwo \reversetoms \tm \toms \tmthree$ with $\tmtwo \neq \tmthree$ then there exists a term $\tmp$ such that $\tmtwo \toms \tmp \lRew{\smsym} \tmthree$ (the proof that $\toes$ is diamond is analogue). 
		The proof is by structural induction on $\tm$, doing case analysis on $\tm \toms \tmtwo$ and $\tm \toms \tmthree$: 
		\begin{itemize}
			\item \emph{Under $\lambda$-abstraction for both $\tm \toms \tmtwo$ and $\tm \toms \tmthree$}; \ie, $\tm = \la{\var}{\tmfour} \toms \la{\var}{\tmfive} = \tmtwo$ and $\tm = \la{\var}{\tmfour} \toms \la{\var}{\tmsix} = \tmthree$, 
			with $\tmfour \toms \tmfive$ and $\tmfour \toms \tmsix$. 
			By \ih there exists $\tmfourp$ such that $\tmfive \toms \tmfourp \reversetoms \tmsix$ and so $\tmtwo = \la{\var}{\tmfive} \toms \la{\var}{\tmfourp} \reversetoms \la{\var}{\tmsix} = \tmthree$.
			
			\item \emph{Application right for $\tm \toms \tmtwo$ and application left for $\tm \toms \tmthree$}; \ie, $\tm = \tmfour \tmfive \toms \tmfour \tmfivep = \tmtwo$ and $\tm = \tmfour \tmfive \toms \tmfourp \tmfive = \tmthree$. There are several sub-cases to this:
			\begin{itemize}
				\item Let $\tm = \tmfour \tmfive = \tmfour \openctxp{\tilde{\tmfive}} \tomo \tmfour \openctxp{\tilde{\tmfivep}} = \tmtwo$, with $\tilde{\tmfive} \rtom \tilde{\tmfivep}$, and $\tm = \rctxp{\tilde{\tmfour}} \tmfive \toms \rctxp{\tilde{\tmfourp}} \tmfive$, with $\tilde{\tmfour} \tomo \tilde{\tmfourp}$. Let $\tmp \defeq \rctxp{\tilde{\tmfourp}} \openctxp{\tilde{\tmfivep}}$, having that 
				$$
				\tmtwo = \rctxp{\tilde{\tmfour}} \openctxp{\tilde{\tmfivep}} \toms \tmp \reversetoms \rctxp{\tilde{\tmfourp}} \openctxp{\tilde{\tmfive}} = \tmthree
				$$
				Note that $\tmtwo \toms \tmp$ holds because every rigid context is an open context.

				\item Let $\tm = \tmfour \tmfive = \tmfour \openctx_{1} \ctxholep{\tilde{\tmfive}} \tomo \tmfour \openctx_{1} \ctxholep{\tilde{\tmfivep}} = \tmtwo$, with $\tilde{\tmfive} \rtom \tilde{\tmfivep}$, and $\tm = \openctx_{2} \ctxholep{\tilde{\tmfour}} \tmfive \tomo \openctx_{2} \ctxholep{\tilde{\tmfourp}} \tmfive$, with $\tilde{\tmfour} \rtom \tilde{\tmfourp}$. 
				Then the statement holds by \reflemmap{basic-value-substitution}{tom-toe-diamond-open}
				
				\item Let $\tm = \tmfour \evsctxp{\tilde{\tmfive}} \toms \tmfour \evsctxp{\tilde{\tmfivep}} = \tmtwo$, with $\tmfour$ a rigid term and $\tilde{\tmfive} \tomo \tilde{\tmfivep}$, and $\tm = \rctxp{\tilde{\tmfour}} \tmfive \toms \rctxp{\tilde{\tmfourp}} \tmfive = \tmthree$, with $\tilde{\tmfour} \tomo \tilde{\tmfourp}$. Let $\tmp \defeq \rctxp{\tilde{\tmfourp}} \evsctxp{\tilde{\tmfivep}}$, having that 
				$$
				\tmtwo = \rctxp{\tilde{\tmfour}} \evsctxp{\tilde{\tmfivep}} \toms\tmp \reversetoms \rctxp{\tilde{\tmfourp}} \evsctxp{\tilde{\tmfive}} = \tmthree
				$$
				
				Note that $\tmp \reversetoms \tmthree$ holds because $\rctxp{\tilde{\tmfourp}}$ is a \pointed term ---by \reflemmap{properties-of-rigid-terms}{plugging-is-rigid}.
				
				\item Let $\tm = \tmfour \evsctxp{\tilde{\tmfive}} \toms \tmfour \evsctxp{\tilde{\tmfivep}} = \tmtwo$, with $\tmfour$ a rigid term and $\tilde{\tmfive} \tomo \tilde{\tmfivep}$, and $\tm = \openctxp{\tilde{\tmfour}} \tmfive \toms \openctxp{\tilde{\tmfourp}} \tmfive$, with $\tilde{\tmfour} \rtom \tilde{\tmfourp}$. Let $\tmp \defeq \openctxp{\tilde{\tmfourp}} \evsctxp{\tilde{\tmfivep}}$, having that
				$$
				\tmtwo = \openctxp{\tilde{\tmfour}} \evsctxp{\tilde{\tmfivep}} \toms \tmp \reversetoms \openctxp{\tilde{\tmfourp}} \evsctxp{\tilde{\tmfive}} = \tmthree
				$$
				
				Note that $\tmp \reversetoms \tmthree$ holds because the fact that $\tmfour$ is a \pointed term and that $\tmfour = \openctxp{\tilde{\tmfour}} \tomo \openctxp{\tilde{\tmfourp}}$ imply that $\openctxp{\tilde{\tmfourp}}$ is a \pointed term ---by \reflemmap{properties-of-rigid-terms}{open-strategy-preserves-rigid}.
				
			\end{itemize}
			
			\item \emph{Application right for both $\tm \toms \tmtwo$ and $\tm \toms \tmthree$}; \ie, $\tm = \tmfour \tmfive \toms \tmfour \tmfivep = \tmtwo$ and $\tm = \tmfour \tmfive \toms \tmfour \tmfivepp = \tmthree$. By \ih there exists a term $\tmsix$ such that $\tmfivep \toms \tmsix \reversetoms \tmfivepp$. The analysis of the sub-cases, depending on the open/strong/rigid type contexts involved in $\tm \toms \tmtwo$ and $\tm \toms \tmthree$, follows the same schema as for the previous item, all showing that
			$$
			\tmtwo = \tmfour \tmfivep \toms \tmfour \tmsix \reversetoms \tmfour \tmfivepp = \tmthree
			$$

			\item \emph{Application left for both $\tm \toms \tmtwo$ and $\tm \toms \tmthree$}; \ie, $\tm = \tmfour \tmfive \toms \tmfourp \tmfive = \tmtwo$ and $\tm = \tmfour \tmfive \toms \tmfourpp \tmfive = \tmthree$. By \ih there exists a term $\tmsix$ such that $\tmfourp \toms \tmsix \reversetoms \tmfourpp$. The analysis of the sub-cases, depending on the open/strong/rigid type contexts involved in $\tm \toms \tmtwo$ and $\tm \toms \tmthree$, follows the same schema as for the previous item, all showing that
			$$
			\tmtwo = \tmfourp \tmfive \toms \tmsix \tmfive \reversetoms \tmfourpp \tmfive = \tmthree
			$$
			
			\item \emph{$\mathsf{ES}$ right for $\tm \toms \tmtwo$ and $\mathsf{ES}$ left for $\tm \toms \tmthree$}; \ie, $\tm = \tmfour \esub{\var}{\tmfive} \toms \tmfour \esub{\var}{\tmfivep} = \tmtwo$ and $\tm = \tmfour \esub{\var}{\tmfive} \toms \tmfourp \esub{\var}{\tmfive}$. There are several sub-cases to this:
			\begin{itemize}
				\item Let $\tm = \tmfour \esub{\var}{\openctxp{\tilde{\tmfive}}} \toms \tmfour \esub{\var}{\openctxp{\tilde{\tmfivep}}} = \tmtwo$, with $\tilde{\tmfive} \rtom \tilde{\tmfivep}$, and $\tm = \openctxp{\tilde{\tmfour}} \esub{\var}{\tmfive} \toms \openctxp{\tilde{\tmfourp}} \esub{\var}{\tmfive} = \tmthree$, with $\tilde{\tmfour} \rtom \tilde{\tmfourp}$. Then the statement holds by \reflemmap{basic-value-substitution}{tom-toe-diamond-open}.
				
				\item Let $\tm = \tmfour \esub{\var}{\openctxp{\tilde{\tmfive}}} \toms \tmfour \esub{\var}{\openctxp{\tilde{\tmfivep}}} = \tmtwo$, with $\tilde{\tmfive} \rtom \tilde{\tmfivep}$, and $\tm = \evsctxp{\tilde{\tmfour}} \esub{\var}{\tmfive} \toms \evsctxp{\tilde{\tmfourp}} \esub{\var}{\tmfive} = \tmthree$, with $\tilde{\tmfour} \tomo \tilde{\tmfourp}$ and $\tmfive$ is a \pointed term. Let $\tmp \defeq \evsctxp{\tilde{\tmfourp}} \esub{\var}{\openctxp{\tilde{\tmfivep}}}$, having that 
				$$
				\tmtwo = \evsctxp{\tilde{\tmfour}} \esub{\var}{\openctxp{\tilde{\tmfivep}}} \toms \tmp \reversetoms \evsctxp{\tilde{\tmfourp}} \esub{\var}{\openctxp{\tilde{\tmfive}}} = \tmthree
				$$
				
				Note that $\tmtwo \toms \tmp$ holds because the fact that$\tmfive$ is a \pointed term and that $\tmfive = \openctxp{\tilde{\tmfive}} \toms \openctxp{\tilde{\tmfivep}}$ imply that  $\openctxp{\tilde{\tmfivep}}$ is a \pointed term ---by \reflemmap{properties-of-rigid-terms}{strong-strategy-preserves-rigid}.
				
				\item Let $\tm = \tmfour \esub{\var}{\openctxp{\tilde{\tmfive}}} \toms \tmfour \esub{\var}{\openctxp{\tilde{\tmfivep}}} = \tmtwo$, with $\tilde{\tmfive} \rtom \tilde{\tmfivep}$, and $\tm = \rctxp{\tilde{\tmfour}} \esub{\var}{\tmfive} \toms \rctxp{\tilde{\tmfourp}} \esub{\var}{\tmfive} = \tmthree$, with $\tilde{\tmfour} \tomo \tilde{\tmfourp}$ and $\tmfive$ is a \pointed term. Let $\tmp \defeq \rctxp{\tilde{\tmfourp}} \esub{\var}{\openctxp{\tilde{\tmfivep}}}$, having that 
				$$
				\tmtwo = \rctxp{\tilde{\tmfour}} \esub{\var}{\openctxp{\tilde{\tmfivep}}} \toms \tmp \reversetoms \rctxp{\tilde{\tmfourp}} \esub{\var}{\openctxp{\tilde{\tmfive}}} = \tmthree
				$$
				
				Note that $\tmtwo \toms \tmp$ holds because the fact that $\tmfive$ is a \pointed term and that $\tmfive = \openctxp{\tilde{\tmfive}} \toms \openctxp{\tilde{\tmfivep}}$ imply $\openctxp{\tilde{\tmfivep}}$ is a \pointed term ---by \reflemmap{properties-of-rigid-terms}{open-strategy-preserves-rigid}.
				
				\item Let $\tm = \tmfour \esub{\var}{\rctxp{\tilde{\tmfive}}} \toms \tmfour \esub{\var}{\rctxp{\tilde{\tmfivep}}} = \tmtwo$, with $\tilde{\tmfive} \tomo \tilde{\tmfivep}$, and $\tm = \openctxp{\tilde{\tmfour}} \esub{\var}{\tmfive} \toms \openctxp{\tilde{\tmfourp}} \esub{\var}{\tmfive} = \tmthree$, with $\tilde{\tmfour} \rtom \tilde{\tmfourp}$. Let $\tmp \defeq \openctxp{\tilde{\tmfourp}} \esub{\var}{\rctxp{\tilde{\tmfivep}}}$, having that 
				$$
				\tmtwo = \openctxp{\tilde{\tmfour}} \esub{\var}{\rctxp{\tilde{\tmfivep}}} \toms \tmp \reversetoms \openctxp{\tilde{\tmfourp}} \esub{\var}{\rctxp{\tilde{\tmfive}}} = \tmthree
				$$
				
				\item Let $\tm = \tmfour \esub{\var}{\rctxp{\tilde{\tmfive}}} \toms \tmfour \esub{\var}{\rctxp{\tilde{\tmfivep}}} = \tmtwo$, with $\tilde{\tmfive} \tomo \tilde{\tmfivep}$, and $\tm = \evsctxp{\tilde{\tmfour}} \esub{\var}{\tmfive} \toms \evsctxp{\tilde{\tmfourp}} \esub{\var}{\tmfive} = \tmthree$, with $\tilde{\tmfour} \tomo \tilde{\tmfourp}$ and $\tmfive$ is a \pointed term. Let $\tmp \defeq \evsctxp{\tilde{\tmfourp}} \esub{\var}{\rctxp{\tilde{\tmfivep}}}$, having that 
				$$
				\tmtwo = \evsctxp{\tilde{\tmfour}} \esub{\var}{\rctxp{\tilde{\tmfivep}}} \toms \tmp \reversetoms \evsctxp{\tilde{\tmfourp}} \esub{\var}{\rctxp{\tilde{\tmfive}}} = \tmthree
				$$
				Note that $\tmtwo \toms \tmp$ holds because the fact that $\tmfive$ is a \pointed term and that $\tmfive = \rctxp{\tilde{\tmfive}} \toms \rctxp{\tilde{\tmfivep}}$ imply that $\rctxp{\tilde{\tmfivep}}$ is a \pointed term ---by \reflemmap{properties-of-rigid-terms}{strong-strategy-preserves-rigid}.
				
				\item Let $\tm = \tmfour \esub{\var}{\rctx_{1} \ctxholep{\tilde{\tmfive}}} \toms \tmfour \esub{\var}{\rctx_{1} \ctxholep{\tilde{\tmfivep}}} = \tmtwo$, with $\tilde{\tmfive} \tomo \tilde{\tmfivep}$, and $\tm = \rctx_{2} \ctxholep{\tilde{\tmfour}} \esub{\var}{\tmfive} \toms \rctx_{2} \ctxholep{\tilde{\tmfourp}} \esub{\var}{\tmfive} = \tmthree$, with $\tilde{\tmfour} \tomo \tilde{\tmfourp}$ and $\tmfive$ is a \pointed term. Let $\tmp \defeq \rctx_{2} \ctxholep{\tilde{\tmfourp}} \esub{\var}{\rctx_{1} \ctxholep{\tilde{\tmfivep}}}$, having that 
				$$
				\tmtwo = \rctx_{2} \ctxholep{\tilde{\tmfour}} \esub{\var}{\rctx_{1} \ctxholep{\tilde{\tmfivep}}} \toms \tmp \reversetoms \rctx_{2} \ctxholep{\tilde{\tmfourp}} \esub{\var}{\rctx_{1} \ctxholep{\tilde{\tmfive}}} = \tmthree
				$$
				Note that $\tmtwo \toms \tmp$ holds because the fact that $\tmfive$ is a \pointed term and that $\tmfive = \rctx_{1} \ctxholep{\tilde{\tmfive}} \toms \rctx_{1} \ctxholep{\tilde{\tmfivep}}$ imply that $\rctx_{1} \ctxholep{\tilde{\tmfivep}}$ is a \pointed term ---by \reflemmap{properties-of-rigid-terms}{strong-strategy-preserves-rigid}.
				
				\item Let $\tm = \tmfour \esub{\var}{\rctxp{\tilde{\tmfive}}} \toms \tmfour \esub{\var}{\rctxp{\tilde{\tmfivep}}} = \tmtwo$, with $\tilde{\tmfive} \tomo \tilde{\tmfivep}$ and $\tmfour$ is a \pointed term, and $\tm = \openctxp{\tilde{\tmfour}} \esub{\var}{\tmfive} \toms \openctxp{\tilde{\tmfourp}} \esub{\var}{\tmfive} = \tmthree$, with $\tilde{\tmfour} \rtom \tilde{\tmfourp}$. Let $\tmp \defeq \openctxp{\tilde{\tmfourp}} \esub{\var}{\rctxp{\tilde{\tmfivep}}}$, having that
				$$
				\tmtwo = \openctxp{\tilde{\tmfour}} \esub{\var}{\rctxp{\tilde{\tmfivep}}} \toms \tmp \reversetoms \openctxp{\tilde{\tmfourp}} \esub{\var}{\rctxp{\tilde{\tmfive}}} = \tmthree
				$$
				Note that $\tmthree \reversetoms \tmp$ holds because the fact that $\tmfive$ is a \pointed term and that $\tmfive = \rctxp{\tilde{\tmfive}} \toms \rctxp{\tilde{\tmfivep}}$ imply that $\rctxp{\tilde{\tmfivep}}$ is a \pointed term ---by \reflemmap{properties-of-rigid-terms}{strong-strategy-preserves-rigid}.
				
				\item Let $\tm = \tmfour \esub{\var}{\rctxp{\tilde{\tmfive}}} \toms \tmfour \esub{\var}{\rctxp{\tilde{\tmfivep}}} = \tmtwo$, with $\tilde{\tmfive} \tomo \tilde{\tmfivep}$ and $\tmfour$ is a \pointed term, and $\tm = \evsctxp{\tilde{\tmfour}} \esub{\var}{\tmfive} \toms \evsctxp{\tilde{\tmfourp}} \esub{\var}{\tmfive} = \tmthree$, with $\tilde{\tmfour} \tomo \tilde{\tmfourp}$ and $\tmfive$ is a \pointed term. Let $\tmp \defeq \evsctxp{\tilde{\tmfourp}} \esub{\var}{\rctxp{\tilde{\tmfivep}}}$, having that 
				$$
				\tmtwo = \evsctxp{\tilde{\tmfour}} \esub{\var}{\rctxp{\tilde{\tmfivep}}} \toms \tmp \reversetoms \evsctxp{\tilde{\tmfourp}} \esub{\var}{\rctxp{\tilde{\tmfive}}} = \tmthree
				$$ 
				Note that $\tmtwo \toms \tmp$ holds because the fact that $\tmfive$ is a \pointed term and that $\tmfive = \rctxp{\tilde{\tmfive}} \toms \rctxp{\tilde{\tmfivep}}$ imply that $\rctxp{\tilde{\tmfivep}}$ is a \pointed term ---by \reflemmap{properties-of-rigid-terms}{strong-strategy-preserves-rigid}. Moreover, note that $\tmp \reversetoms \tmthree$ holds because the fact that $\tmfour$ is a \pointed term and that $\tmfour = \evsctxp{\tilde{\tmfour}} \toms \evsctxp{\tilde{\tmfourp}}$ imply that $\evsctxp{\tilde{\tmfourp}}$ is a \pointed term ---by \reflemmap{properties-of-rigid-terms}{strong-strategy-preserves-rigid}.
				
				\item Let $\tm = \tmfour \esub{\var}{\rctxp{\tilde{\tmfive}}} \toms \tmfour \esub{\var}{\rctxp{\tilde{\tmfivep}}} = \tmtwo$, with $\tilde{\tmfive} \tomo \tilde{\tmfivep}$ and $\tmfour$ is a \pointed term, and $\tm = \rctxp{\tilde{\tmfour}} \esub{\var}{\tmfive} \toms \rctxp{\tilde{\tmfourp}} \esub{\var}{\tmfive} = \tmthree$, with $\tilde{\tmfour} \tomo \tilde{\tmfourp}$ and $\tmfive$ is a \pointed term. Let $\tmp \defeq \rctxp{\tilde{\tmfourp}} \esub{\var}{\rctxp{\tilde{\tmfivep}}}$, having that 
				$$
				\tmtwo = \rctxp{\tilde{\tmfour}} \esub{\var}{\rctxp{\tilde{\tmfivep}}} \toms \tmp \reversetoms \rctxp{\tilde{\tmfourp}} \esub{\var}{\rctxp{\tilde{\tmfive}}} = \tmthree
				$$
				Note that $\tmtwo \toms \tmp$ holds because the fact that $\tmfive$ is a \pointed term and that $\tmfive = \rctxp{\tilde{\tmfive}} \toms \rctxp{\tilde{\tmfivep}}$ imply that $\rctxp{\tilde{\tmfivep}}$ is a \pointed term ---by \reflemmap{properties-of-rigid-terms}{strong-strategy-preserves-rigid}. Moreover, note that $\tmp \reversetoms \tmthree$ holds because the fact that $\tmfour$ is a \pointed term and that $\tmfour = \rctxp{\tilde{\tmfour}} \toms \rctxp{\tilde{\tmfourp}}$ imply that $\rctxp{\tilde{\tmfourp}}$ is a \pointed term ---by \reflemmap{properties-of-rigid-terms}{strong-strategy-preserves-rigid}.
				
			\end{itemize}
			
			\item \emph{$\mathsf{ES}$ right for both $\tm \toms \tmtwo$ and $\tm \toms \tmthree$}; \ie, $\tm = \tmfour \esub{\var}{\tmfive} \toms \tmfour \esub{\var}{\tmfivep} = \tmtwo$ and $\tm = \tmfour \esub{\var}{\tmfive} \toms \tmfour \esub{\var}{\tmfivepp} = \tmthree$. By \ih there exists a term $\tmsix$ such that $\tmfivep \toms \tmsix \reversetoms \tmfivepp$. The analysis of the sub-cases, depending on the open/strong/rigid type contexts involved in $\tm \toms \tmtwo$ and $\tm \toms \tmthree$, follows the same schema as for the previous item, all showing that
			$$
			\tmtwo = \tmfour \esub{\var}{\tmfivep} \toms \tmfour \esub{\var}{\tmsix} \reversetoms \tmfour \esub{\var}{\tmfivepp} = \tmthree
			$$
			
			\item \emph{$\mathsf{ES}$ left for both $\tm \toms \tmtwo$ and $\tm \toms \tmthree$}; \ie, $\tm = \tmfour \esub{\var}{\tmfive} \toms \tmfourp \esub{\var}{\tmfive} = \tmtwo$ and $\tm = \tmfourpp \esub{\var}{\tmfive}$. By \ih there exists a term $\tmsix$ such that $\tmfourp \toms \tmsix \reversetoms \tmfourpp$. The analysis of the sub-cases, depending on the open/strong/rigid type contexts involved in $\tm \toms \tmtwo$ and $\tm \toms \tmthree$, follows the same schema as for the previous item, all showing that
			$$
			\tmtwo = \tmfourp \esub{\var}{\tmfive} \toms \tmsix \esub{\var}{\tmfive} \reversetoms \tmfourpp \esub{\var}{\tmfive} = \tmthree
			$$
			
		\end{itemize}

		The proof that $\toes$ is diamond (\ie, if $\tmtwo \reversetoes \tm \toes \tmthree$ with $\tmtwo \neq \tmthree$ then there exists $\tmp \in \Lambda_\vsub$ such that $\tmtwo \toes \tmp \reversetoes \tmthree$) follows the same schema as for $\toms$.
		
		\item We show that $\toes$ and $\toms$ strongly commute; \ie, if $\tmtwo \reversetoes \tm \toms \tmthree$, then $\tmtwo \neq \tmthree$ and there is $\tmp \in \Lambda_\vsub$ such that $\tmtwo \toms \tmp \reversetoes \tmthree$. The proof is by structural induction on $\tm$, doing case analysis on $\tm \reversetoes \tmtwo$ and $\tm \toms \tmthree$:
		\begin{itemize}
			\item \emph{Under $\lambda$-abstraction for both $\tm \reversetoes \tmtwo$ and $\tm \toms \tmthree$}; \ie, $\tm = \la{\var}{\tmfive} \reversetoes \la{\var}{\tmfour} = \tmtwo$ and $\tm = \la{\var}{\tmfour} \toms \la{\var}{\tmsix} = \tmthree$, with $\tmfive \reversetoeo \tmfour \tomo \tmsix$. By \reflemmap{basic-value-substitution}{tom-toe-commute-open}, there exists $\tmfourp$ such that $\tmfive \toms \tmfourp \reversetoes \tmsix$, and so $\tmtwo = \la{\var}{\tmfive} \tomo \la{\var}{\tmfourp} \reversetoeo \la{\var}{\tmsix}$.
			
			\item \emph{Application right for $\tmtwo \reversetoes \tm$ and application left for $\tm \toms \tmthree$}; \ie, $\tmtwo = \tmfour \tmfivep \reversetoes \tmfour \tmfive = \tm$ and $\tm = \tmfour \tmfive \toms \tmfourp \tmfive = \tmthree$. There are several sub-cases to this:
			\begin{itemize}
				\item Let $\tmtwo = \tmfour \openctxp{\tilde{\tmfivep}} \reversetoes \tmfour \openctxp{\tilde{\tmfive}} = \tm$, with $\tilde{\tmfivep} \reversetoeo \tilde{\tmfive}$, and $\tm = \rctxp{\tilde{\tmfour}}  \tmfive \toms \rctxp{\tilde{\tmfourp}} \tmfive = \tmthree$, with $\tilde{\tmfour} \tomo \tilde{\tmfourp}$. Let $\tmp = \rctxp{\tilde{\tmfourp}} \openctxp{\tilde{\tmfivep}}$, having that 
				$$
				\tmtwo = \rctxp{\tilde{\tmfour}} \openctxp{\tilde{\tmfivep}} \toms \tmp \reversetoes \rctxp{\tilde{\tmfourp}} \openctxp{\tilde{\tmfive}} = \tmthree
				$$
				
				\item Let $\tmtwo = \tmfour \openctx_{1} \ctxholep{\tilde{\tmfivep}} \reversetoes \tmfour \openctx_{1} \ctxholep{\tilde{\tmfive}} = \tm$, and $\tm = \openctx_{2} \ctxholep{\tilde{\tmfour}} \tmfive \toms \openctx_{2} \ctxholep{\tilde{\tmfourp}} \tmfive = \tmthree$, with $\tilde{\tmfour} \toms \tilde{\tmfourp}$. Then the statement holds by \reflemmap{basic-value-substitution}{tom-toe-commute-open}
				
				\item Let $\tmtwo = \tmfour \evsctxp{\tilde{\tmfivep}} \reversetoes \tmfour \evsctxp{\tilde{\tmfive}} = \tm$, with $\tmfour$ a \pointed term and $\tilde{\tmfivep} \reversetoeo \tilde{\tmfive}$, and $\tm = \rctxp{\tilde{\tmfour}} \tmfive \toms \rctxp{\tilde{\tmfourp}} \tmfive = \tmthree$, with $\tilde{\tmfour} \tomo \tilde{\tmfourp}$. Let $\tmp = \rctxp{\tilde{\tmfourp}} \evsctxp{\tilde{\tmfivep}}$, having that
				$$
				\tmtwo = \rctxp{\tilde{\tmfour}} \evsctxp{\tilde{\tmfivep}} \toms \tmp \reversetoes \rctxp{\tilde{\tmfourp}} \evsctxp{\tilde{\tmfive}} = \tmthree
				$$
				Note that $\tmp \reversetoes \tmthree$ holds because $\rctxp{\tilde{\tmfourp}}$ is a \pointed term ---by \reflemmap{properties-of-rigid-terms}{strong-strategy-preserves-rigid}.
				
				\item Let $\tmtwo = \tmfour \evsctxp{\tilde{\tmfivep}} \reversetoes \tmfour \evsctxp{\tilde{\tmfive}} = \tm$, with $\tmfour$ a \pointed term and $\tilde{\tmfivep} \reversetoeo \tilde{\tmfive}$, and $\tm = \openctxp{\tilde{\tmfour}} \tmfive \toms \openctxp{\tilde{\tmfourp}} \tmfive = \tmthree$, with $\tilde{\tmfour} \rtom \tilde{\tmfourp}$. Let $\tmp = \openctxp{\tilde{\tmfourp}} \evsctxp{\tilde{\tmfivep}}$, having that
				$$
				\tmtwo = \openctxp{\tilde{\tmfour}} \evsctxp{\tilde{\tmfivep}} \toms \tmp \reversetoes \openctxp{\tilde{\tmfourp}} \evsctxp{\tilde{\tmfive}} = \tmthree
				$$
				Note that $\tmp \reversetoes \tmthree$ holds because $\openctxp{\tilde{\tmfourp}}$ is \pointed ---by \reflemmap{properties-of-rigid-terms}{open-strategy-preserves-rigid}.
				
			\end{itemize}
			
			\item \emph{Application left for $\tmtwo \reversetoes \tm$ and application right for $\tm \toms \tmthree$}; \ie, $\tmtwo = \tmfourp \tmfive \reversetoes \tmfour \tmfive = \tm$ and $\tm = \tmfour \tmfive \toms \tmfour \tmfivep = \tmthree$. There are several sub-cases to this:
			\begin{itemize}
				\item Let $\tmtwo = \rctxp{\tilde{\tmfourp}} \tmfive \reversetoes \rctxp{\tilde{\tmfour}} \tmfive = \tm$, with $\tilde{\tmfourp} \reversetoeo \tilde{\tmfour}$, and $\tm = \tmfour \openctxp{\tilde{\tmfive}} \toms \tmfour \openctxp{\tilde{\tmfivep}} = \tmthree$, with $\tilde{\tmfive} \rtom \tilde{\tmfivep}$. Let $\tmp = \rctxp{\tilde{\tmfourp}} \openctxp{\tilde{\tmfivep}}$, having that 
				$$
				\tmtwo = \rctxp{\tilde{\tmfourp}} \openctxp{\tilde{\tmfive}} \toms \tmp \reversetoes \rctxp{\tilde{\tmfour}} \openctxp{\tilde{\tmfivep}} = \tmthree
				$$
				
				\item Let $\tmtwo = \openctx_{1} \ctxholep{\tilde{\tmfourp}} \tmfive \reversetoes \openctx_{1} \ctxholep{\tilde{\tmfour}} \tmfive = \tm$, with $\tilde{\tmfourp} \reversetoeo \tilde{\tmfour}$, and $\tm = \tmfour \openctx_{2} \ctxholep{\tilde{\tmfive}} \toms \tmfour \openctx_{2} \ctxholep{\tilde{\tmfivep}} = \tmthree$, with $\tilde{\tmfive} \rtom \tilde{\tmfivep}$. Then the statement holds by \reflemmap{basic-value-substitution}{tom-toe-commute-open}
				
				\item Let $\tmtwo = \rctxp{\tilde{\tmfourp}} \tmfive \reversetoes \rctxp{\tilde{\tmfour}} \tmfive = \tm$, with $\tilde{\tmfourp} \reversetoeo \tilde{\tmfour}$, and $\tm = \tmfour \evsctxp{\tilde{\tmfive}} \toms \tmfour \evsctxp{\tilde{\tmfivep}} = \tmthree$, with $\tilde{\tmfive} \tomo \tilde{\tmfivep}$ and $\tmfour$ a \pointed term. Let $\tmp = \rctxp{\tilde{\tmfourp}} \evsctxp{\tilde{\tmfivep}}$, having that
				$$
				\tmtwo = \rctxp{\tilde{\tmfourp}} \evsctxp{\tilde{\tmfive}} \toms \tmp \reversetoes \rctxp{\tilde{\tmfour}} \evsctxp{\tilde{\tmfivep}} = \tmthree
				$$
				Note that $\tmtwo \toms \tmp$ holds because $\rctxp{\tilde{\tmfourp}}$ is a \pointed term ---by \reflemmap{properties-of-rigid-terms}{strong-strategy-preserves-rigid}.
				
				\item Let $\tmtwo = \openctxp{\tilde{\tmfourp}} \tmfive \reversetoes \openctxp{\tilde{\tmfour}} \tmfive = \tm$, with $\tilde{\tmfourp} \reversetoeo \tilde{\tmfour}$, and $\tm = \tmfour \evsctxp{\tilde{\tmfive}} \toms \tmfour \evsctxp{\tilde{\tmfivep}} = \tmthree$, with $\tmfour$ a \pointed term and $\tilde{\tmfive} \tomo \tilde{\tmfivep}$. Let $\tmp = \openctxp{\tilde{\tmfourp}} \evsctxp{\tilde{\tmfivep}}$, having that
				$$
				\tmtwo = \openctxp{\tilde{\tmfourp}} \evsctxp{\tilde{\tmfive}} \toms \tmp \reversetoes \openctxp{\tilde{\tmfour}} \evsctxp{\tilde{\tmfivep}} = \tmthree
				$$
				Note that $\tmtwo \toms \tmp$ holds because $\openctxp{\tilde{\tmfourp}}$ is \pointed ---by \reflemmap{properties-of-rigid-terms}{open-strategy-preserves-rigid}.
				
			\end{itemize}
			
			\item \emph{Application right for both $\tmtwo \reversetoes \tm$ and $\tm \toms \tmthree$}; \ie, $\tmtwo = \tmfour \tmfivep \reversetoes \tmfour \tmfive = \tm$ and $\tm = \tmfour \tmfive \toms \tmfour \tmfivepp = \tmthree$. By \ih, there exists a term $\tmsix$ such that $\tmfivep \toms \tmsix \reversetoes \tmfivepp$. The analysis of the sub-cases, depending on the open/strong/rigid type contexts involved in $\tmtwo \reversetoes \tm$ and $\tm \toms \tmthree$ follows the same schema as for the previous item, all showing that
			$$
			\tmtwo = \tmfour \tmfivep \toms \tmfour \tmsix \reversetoes \tmfour \tmfivepp = \tmthree
			$$
			
			\item \emph{Application left for both $\tmtwo \reversetoes \tm$ and $\tm \toms \tmthree$}; \ie, $\tmtwo = \tmfourp \tmfive \reversetoes \tmfour \tmfive = \tm$ and $\tm = \tmfour \tmfive \toms \tmfourpp \tmfive = \tmthree$. By \ih, there exists a term $\tmsix$ such that $\tmfourp \toms \tmsix \reversetoes \tmfourpp$. The analysis of the sub-cases, depending on the open/strong/rigid type contexts involved in $\tmtwo \reversetoes \tm$ and $\tm \toms \tmthree$ follows the same schema as for the previous item, all showing that
			$$
			\tmtwo = \tmfourp \tmfive \toms \tmsix \tmfive \reversetoes \tmfourpp \tmfive = \tmthree
			$$

\item \emph{$\mathsf{ES}$ right for $\tmtwo \reversetoes \tm$ and $\mathsf{ES}$ left for $\tm \toms \tmthree$}; \ie, $\tmtwo = \tmfour \esub{\var}{\tmfivep} \reversetoes \tmfour \esub{\var}{\tmfive} = \tm$ and $\tm = \tmfour \esub{\var}{\tmfive} \toms \tmfourp \esub{\var}{\tmfive} = \tmthree$. There are several sub-cases to this:
\begin{itemize}
	\item Let $\tmtwo = \tmfour \esub{\var}{\openctxp{\tilde{\tmfivep}}} \reversetoes \tmfour \esub{\var}{\openctxp{\tilde{\tmfive}}} = \tm$, with $\tilde{\tmfivep} \reversertoe \tilde{\tmfive}$, and $\tm = \openctxp{\tilde{\tmfour}} \esub{\var}{\tmfive} \toms \openctxp{\tilde{\tmfourp}} \esub{\var}{\tmfive} = \tmthree$, with $\tilde{\tmfour} \rtom \tilde{\tmfourp}$. Then the statement holds by \reflemmap{basic-value-substitution}{tom-toe-commute-open}.
	
	\item Let $\tmtwo = \tmfour \esub{\var}{\openctxp{\tilde{\tmfivep}}} \reversetoes \tmfour \esub{\var}{\openctxp{\tilde{\tmfive}}} = \tm$, with $\tilde{\tmfivep} \reversertoe \tilde{\tmfive}$, and $\tm = \evsctxp{\tilde{\tmfour}} \esub{\var}{\tmfive} \toms \evsctxp{\tilde{\tmfourp}} \esub{\var}{\tmfive} = \tmthree$, with $\tilde{\tmfour} \tomo \tilde{\tmfourp}$ and $\tmfive$ is a \pointed term. Let $\tmp \defeq \evsctxp{\tilde{\tmfourp}} \esub{\var}{\openctxp{\tilde{\tmfivep}}}$, having that 
	$$
	\tmtwo = \evsctxp{\tilde{\tmfour}} \esub{\var}{\openctxp{\tilde{\tmfivep}}} \toms \tmp \toms \evsctxp{\tilde{\tmfourp}} \esub{\var}{\openctxp{\tilde{\tmfive}}} = \tmthree
	$$
	Note that $\tmtwo \toms \tmp$ holds because $\openctxp{\tilde{\tmfivep}}$ is a \pointed term ---by \reflemmap{properties-of-rigid-terms}{open-strategy-preserves-rigid}
	
	\item Let $\tmtwo = \tmfour \esub{\var}{\openctxp{\tilde{\tmfivep}}} \reversetoes \tmfour \esub{\var}{\openctxp{\tilde{\tmfive}}} = \tm$, with $\tilde{\tmfivep} \reversertoe \tilde{\tmfive}$, and $\tm = \rctxp{\tilde{\tmfour}} \esub{\var}{\tmfive} \toms \rctxp{\tilde{\tmfourp}} \esub{\var}{\tmfive} = \tmthree$, with $\tilde{\tmfour} \tomo \tilde{\tmfourp}$ and $\tmfive$ is a \pointed term. Let $\tmp \defeq \rctxp{\tilde{\tmfourp}} \esub{\var}{\openctxp{\tilde{\tmfivep}}}$, having that 
	$$
	\tmtwo = \rctxp{\tilde{\tmfour}} \esub{\var}{\openctxp{\tilde{\tmfivep}}} \toms \tmp \toms \rctxp{\tilde{\tmfourp}} \esub{\var}{\openctxp{\tilde{\tmfive}}} = \tmthree
	$$
	Note that $\tmtwo \toms \tmp$ holds because $\openctxp{\tilde{\tmfivep}}$ is a \pointed term ---by \reflemmap{properties-of-rigid-terms}{open-strategy-preserves-rigid}
	
	\item Let $\tmtwo = \tmfour \esub{\var}{\rctxp{\tilde{\tmfivep}}} \reversetoes \tmfour \esub{\var}{\rctxp{\tilde{\tmfive}}} = \tm$, with $\tmfivep \reversetoeo \tmfive$, and $\tm = \openctxp{\tilde{\tmfour}} \esub{\var}{\tmfive} \toms \openctxp{\tilde{\tmfourp}} \esub{\var}{\tmfive} = \tmthree$, with $\tilde{\tmfive} \toms \tilde{\tmfivep}$. Let $\tmp \defeq \openctxp{\tilde{\tmfourp}} \esub{\var}{\rctxp{\tmfivep}}$, having that
	$$
	\tmtwo = \openctxp{\tilde{\tmfour}} \esub{\var}{\rctxp{\tilde{\tmfivep}}} \toms \tmp \reversetoes \openctxp{\tilde{\tmfourp}} \esub{\var}{\rctxp{\tilde{\tmfive}}} = \tmthree
	$$
	
	\item Let $\tmtwo = \tmfour \esub{\var}{\rctxp{\tilde{\tmfivep}}} \reversetoes \tmfour \esub{\var}{\rctxp{\tilde{\tmfive}}} = \tm$, with $\tilde{\tmfivep} \reversetoeo \tilde{\tmfive}$, and $\tm = \evsctxp{\tilde{\tmfour}} \esub{\var}{\tmfive} \toms \evsctxp{\tilde{\tmfourp}} \esub{\var}{\tmfive} = \tmthree$, with $\tilde{\tmfour} \tomo \tilde{\tmfourp}$ and $\tmfive$ is a \pointed term. Let $\tmp \defeq \evsctxp{\tilde{\tmfourp}} \esub{\var}{\rctxp{\tilde{\tmfivep}}}$, having that 
	$$
	\tmtwo = \evsctxp{\tilde{\tmfour}} \esub{\var}{\rctxp{\tilde{\tmfivep}}} \toms \tmp \reversetoes \evsctxp{\tilde{\tmfourp}} \esub{\var}{\rctxp{\tilde{\tmfive}}} = \tmthree
	$$
	Note that $\tmtwo \toms \tmp$ holds because $\rctxp{\tilde{\tmfivep}}$ is a \pointed term ---by \reflemmap{properties-of-rigid-terms}{strong-strategy-preserves-rigid}.
	
	\item Let $\tmtwo = \tmfour \esub{\var}{\rctx_{1} \ctxholep{\tilde{\tmfivep}}} \reversetoes \tmfour \esub{\var}{\rctx_{1} \ctxholep{\tilde{\tmfive}}} = \tm$, with $\tilde{\tmfivep} \reversetoeo \tilde{\tmfive}$, and $\tm = \rctx_{2} \ctxholep{\tilde{\tmfour}} \esub{\var}{\tmfive} \toms \rctx_{2} \ctxholep{\tilde{\tmfourp}} \esub{\var}{\tmfive}$, with $\tilde{\tmfour} \tomo \tilde{\tmfourp}$ and $\tmfive$ is a \pointed term. Let 
	
		$$
			\tmp \defeq \rctx_{2} \ctxholep{\tilde{\tmfourp}} \esub{\var}{\rctx_{1} \ctxholep{\tilde{\tmfivep}}}
		$$
	
	having that
	
	$
	\begin{array}{rcl}
		\tmtwo 
		& = & \rctx_{2} \ctxholep{\tilde{\tmfour}} \esub{\var}{\rctx_{1} \ctxholep{\tilde{\tmfivep}}} \\
		& \toms & \tmp \reversetoes \rctx_{2} \ctxholep{\tilde{\tmfourp}} \esub{\var}{\rctx_{1} \ctxholep{\tilde{\tmfive}}} \\
	\end{array}
	$
	
	Note that $\tmtwo \toms \tmp$ holds because $\rctx_{1} \ctxholep{\tilde{\tmfivep}}$ is a \pointed term ---by \reflemmap{properties-of-rigid-terms}{strong-strategy-preserves-rigid}.
	
	\item Let $\tmtwo = \tmfour \esub{\var}{\rctxp{\tilde{\tmfivep}}} \reversetoes \tmfour \esub{\var}{\rctxp{\tilde{\tmfive}}} = \tm$, with $\tilde{\tmfivep} \reversetoeo \tilde{\tmfive}$ and $\tmfour$ is a \pointed term, and $\tm = \openctxp{\tilde{\tmfour}} \esub{\var}{\tmfive} \toms \openctxp{\tilde{\tmfourp}} \esub{\var}{\tmfive} = \tmthree$, with $\tilde{\tmfour} \rtom \tilde{\tmfourp}$. Let $\tmp \defeq \openctxp{\tilde{\tmfourp}} \esub{\var}{\rctxp{\tilde{\tmfivep}}}$, having that 
	$$
	\tmtwo = \openctxp{\tilde{\tmfour}} \esub{\var}{\rctxp{\tilde{\tmfivep}}} \toms \tmp \reversetoes \openctxp{\tilde{\tmfourp}} \esub{\var}{\rctxp{\tilde{\tmfive}}} = \tmthree
	$$
	Note that $\tmp \reversetoes \tmthree$ holds because $\openctxp{\tilde{\tmfourp}}$ is a \pointed term ---by \reflemmap{properties-of-rigid-terms}{open-strategy-preserves-rigid}.
	
	\item Let $\tm = \tmfour \esub{\var}{\rctxp{\tilde{\tmfivep}}} \reversetoes \tmfour \esub{\var}{\rctxp{\tilde{\tmfive}}} = \tm$, with $\tilde{\tmfivep} \reversetoeo \tilde{\tmfive}$ and $\tmfour$ is a \pointed term, and $\tm = \evsctxp{\tilde{\tmfour}} \esub{\var}{\tmfive} \toms \evsctxp{\tilde{\tmfourp}} \esub{\var}{\tmfive}$, with $\tilde{\tmfour} \tomo \tilde{\tmfourp}$ and $\tmfive$ is a \pointed term. Let 
		$$
			\tmp \defeq \evsctxp{\tilde{\tmfourp}} \esub{\var}{\rctxp{\tilde{\tmfivep}}}
		$$ 
	having that 
		$$
			\tmtwo = \evsctxp{\tilde{\tmfour}} \esub{\var}{\rctxp{\tilde{\tmfivep}}} \toms \tmp \reversetoes \evsctxp{\tilde{\tmfourp}} \esub{\var}{\rctxp{\tilde{\tmfive}}} = \tmthree
		$$
	Note that $\tmtwo \toms \tmp$ holds because $\rctxp{\tilde{\tmfivep}}$ is a \pointed term ---by \reflemmap{properties-of-rigid-terms}{strong-strategy-preserves-rigid}---, and that $\tm \reversetoes \tmthree$ holds because $\evsctxp{\tilde{\tmfourp}}$ is a \pointed term ---by \reflemmap{properties-of-rigid-terms}{strong-strategy-preserves-rigid}.			
	
	\item Let $\tmtwo = \tmfour \esub{\var}{\rctx_{1} \ctxholep{\tilde{\tmfivep}}} \reversetoes \tmfour \esub{\var}{\rctx_{1} \ctxholep{\tilde{\tmfive}}} = \tm$, with $\tilde{\tmfivep} \reversetoeo \tilde{\tmfive}$ and $\tmfour$ is a \pointed term, and $\tm = \rctx_{2} \ctxholep{\tilde{\tmfour}} \esub{\var}{\tmfive} \toms \rctx_{2} \ctxholep{\tilde{\tmfourp}} \esub{\var}{\tmfive} = \tmthree$, with $\tilde{\tmfour} \tomo \tilde{\tmfourp}$ and $\tmfive$ is a \pointed term. Let $\tmp \defeq \rctx_{2} \ctxholep{\tilde{\tmfourp}} \esub{\var}{\rctx_{1} \ctxholep{\tilde{\tmfivep}}}$, having that 
	$$
	\tmtwo \rctx_{2} \ctxholep{\tilde{\tmfour}} \esub{\var}{\rctx_{1} \ctxholep{\tilde{\tmfivep}}} \toms \tmp \reversetoes \rctx_{2} \ctxholep{\tilde{\tmfourp}} \esub{\var}{\rctx_{1} \ctxholep{\tilde{\tmfive}}} = \tmthree
	$$
	Note that $\tmtwo \toms \tmp$ holds because $\rctx_{1} \ctxholep{\tilde{\tmfivep}}$ is a \pointed term ---by \reflemmap{properties-of-rigid-terms}{strong-strategy-preserves-rigid}---, and that $\tmp \reversetoes \tmthree$ because $\rctx_{2} \ctxholep{\tilde{\tmfourp}}$ is a \pointed term ---by \reflemmap{properties-of-rigid-terms}{strong-strategy-preserves-rigid}.
	
\end{itemize}

\item \emph{$\mathsf{ES}$ left for $\tmtwo \reversetoes \tm$ and $\mathsf{ES}$ right for $\tm \toms \tmthree$}; \ie, $\tmtwo = \tmfourp \esub{\var}{\tmfive} \reversetoes \tmfour \esub{\var}{\tmfive} = \tm$ and $\tm = \tmfour \esub{\var}{\tmfive} \toms \tmfour \esub{\var}{\tmfivep} = \tmthree$. There are several sub-cases to this, all of which follow the same kind of reasoning as for the case \emph{$\mathsf{ES}$ right for $\tmtwo \reversetoes \tm$ and $\mathsf{ES}$ left for $\tm \toms \tmthree$}. Therefore, we shall leave this case for the reader.

\item \emph{$\mathsf{ES}$ right for both $\tmtwo \reversetoes \tm$ and $\tm \toms \tmthree$}; \ie, 
		$$
			\tmtwo = \tmfour \esub{\var}{\tmfivep} \reversetoes \tmfour \esub{\var}{\tmfive} = \tm
		$$
	and 
		$$
			\tm = \tmfour \esub{\var}{\tmfive} \toms \tmfour \esub{\var}{\tmfivepp} = \tmthree
		$$
	
	By \ih there exists a term $\tmsix$ such that $\tmfivep \toms \tmfive \reversetoes \tmfivepp$. The analysis of the sub-cases, depending on the open/strong/rigid type contexts involved in $\tm \toms \tmtwo$ and $\tm \toms \tmthree$, follows the same schema as for the previous item, all showing that
$$
\tmtwo = \tmfour \esub{\var}{\tmfivep} \toms \tmfour \esub{\var}{\tmsix} \reversetoes \tmfour \esub{\var}{\tmfivepp} = \tmthree
$$

\item \emph{$\mathsf{ES}$ left for both $\tmtwo \reversetoes \tm$ and $\tm \toms \tmthree$}; \ie, 
		$$
			\tmtwo = \tmfourp \esub{\var}{\tmfive} \reversetoes \tmfour \esub{\var}{\tmfive} = \tm
		$$
	and 
		$$
			\tm = \tmfour \esub{\var}{\tmfive} \toms \tmfourpp \esub{\var}{\tmfive} = \tmthree
		$$
	By \ih there exists a term $\tmsix$ such that $\tmfivep \toms \tmfive \reversetoes \tmfivepp$. The analysis of the sub-cases, depending on the open / strong / rigid type contexts involved in $\tm \toms \tmtwo$ and $\tm \toms \tmthree$, follows the same schema as for the previous item, all showing that
$$
\tmtwo = \tmfourp \esub{\var}{\tmfive} \toms \tmsix \esub{\var}{\tmfive} \reversetoes \tmfourpp \esub{\var}{\tmfive} = \tmthree.
\qedhere
$$
		\end{itemize}
	\end{enumerate}
\end{proof}

\begin{proposition}[Properties of $\tovsubs$]
	\label{propappendix:vsc-diamond}
	\NoteState{prop:vsc-diamond}
	\label{lappendix:fullness}
	Let $\tm$ be a VSC term.
	\begin{enumerate}
		\item 
		$\tovsubs$ is diamond; $\toms$ and $\toes$ strongly commute.
		\item \emph{Fullness}: 
		$\tm$ is $\esssym$-normal if and only if $\tm$ is $\vsub$-normal.
	\end{enumerate}
\end{proposition}

\begin{proof}
	\begin{enumerate}
		\item Strong commutation of $\toms$ and $\toes$ is proven in \reflemmap{basic-value-substitution-external}{tom-toe-commute-strong}.
		Diamond of $\tovsubs$ follows immediately from that, from diamond for $\toms$ and $\toes$ (\reflemmap{basic-value-substitution-external}{tom-toe-diamond-strong}), 
		and from Hindley-Rosen lemma (\cite[Prop. 3.3.5]{Barendregt84}).

	\item If $\tm$ is $\vsub$-normal then it is $\esssym$-normal because $\tovsubs \, \subseteq \, \tovsub$.
	
	Conversely, suppose that $\tm$ is $\esssym$-normal. By \Cref{prop:properties-full-reduction}.\ref{p:properties-full-reduction-harmony}, it is enough to show that $\tm$ is a \full fireball.
	To have the right \ih, we prove simultaneously, by induction on $\tm$, the following stronger statements (we recall that all \full inert terms are \full fireballs): 
\begin{enumerate}
	\item \emph{Fireball property}: If $\tm$ is $\esssym$-normal, then $\tm$ is a \full fireball.
	
	\item \emph{Non-value property}: If $\tm$ is $\esssym$-normal and not an \valES, then $\tm$ is a \full inert term.
\end{enumerate}

Cases:

\begin{itemize}
	\item \emph{Variable}, \ie, $\tm = \var$: both properties trivially hold, since $\tm$ is a \full inert term and so a \full fireball.
	
	\item \emph{Abstraction}, \ie, $\tm = \la{\var}{\tmtwo}$:
	\begin{enumerate}
		\item \emph{Non-value property}: vacuously true, as $\tm$ is an abstraction and hence an \valES.			
		\item \emph{Fireball property}: Since $\tm$ is $\esssym$-normal, so is $\tmtwo$.
		By \ih applied to $\tmtwo$ (fireball property), $\tmtwo$ is a \full fireball and hence so is $\tm$ (as a \full value).
	\end{enumerate}
	
	\item \emph{Application}; \ie, $\tm = \tm_{1} \tm_{2}$ (which is not an \valES): 
	\begin{enumerate}
		\item \emph{Non-value property}: Since $\tm$ is $\esssym$-normal, so are $\tm_{1}$ and $\tm_{2}$.  
		Moreover, $\tm_{1}$ is not an \valES (otherwise $\tm$ would be a $\toms$-redex).
		By \ih applied to $\tm_{1}$ (non-value property) and to $\tm_{2}$ (fireball property),
		$\tm_{1}$ is a \full inert term and $\tm_{2}$ is a \full fireball.
		Thus, $\tm$ is a \full inert term.
		
		\item \emph{Fireball property}: We have just proved that $\tm$ is a \full inert term, and hence it is a \full fireball.
		
	\end{enumerate}
	
	\item \emph{Explicit substitutions}, \ie, $\tm = \tm_{1} \esub{\var}{\tm_{2}}$:
	\begin{enumerate}
		\item \emph{Fireball property}: Since $\tm$ is $\esssym$-normal, so are $\tm_{1}$ and $\tm_{2}$.  
		Moreover, $\tm_{2}$ is not an \valES (otherwise $\tm$ would be a $\toes$-redex).
		By \ih applied to $\tm_{1}$ (fireball property) and to $\tm_{2}$ (non-value property),
		$\tm_{1}$ is a \full fireball and $\tm_{2}$ is a \full inert term.
		Thus, $\tm$ is a \full fireball.
		
		\item \emph{Non-value property}: We have just proved that $\tm$ is a \full fireball.
		If moreover $\tm$ is not a \valES, then $\tm_{1}$ is not an \valES and hence, by \ih applied to $\tm_{1}$ (non-value property), $\tm_{1}$ is a \full inert term.
		So, $\tm$ is a \full inert term.
		\qedhere
	\end{enumerate}
	
\end{itemize}

\end{enumerate}
\end{proof}

\section{Proofs of \Cref{sect:types}}

First, we observe the following property: given a derivation for a term $\tm$, all variables associated with a non-empty multi type in the type context are free variables~of~ $\tm$.
\begin{remark}
	\label{rmk:free-variables}
	If $\namedtyjp{\tderiv}{}{\tm}{\typctx}{\mtype}$ then $\dom{\typctx} \subseteq \fv{\tm}$.
\end{remark}

\begin{lemma}[Typing of values: splitting]
	\label{l:typing-value-splitting}
	Let $\namedtyjp{\tderiv}{}{\tval}{\typctx}{\mtype}$ (for $\tval$ theoretical value).
	\begin{enumerate}
		\item \label{p:typing-value-splitting-one} If $\mtype = \emptytype$, then $\dom{\typctx} = \emptyset$ and 
		$\sizem{\tderiv} = 0 = \size{\tderiv}$. 
		
		\item \label{p:typing-value-splitting-two} For every splitting $\mtype = \mtype_{1} \mplus \mtype_{2}$, there exist 
		type derivations $\namedtyjp{\tderiv_{1}}{}{\tval}{\typctx_{1}}{\mtype_{1}}$ and 
		$\namedtyjp{\tderiv_{2}}{}{\tval}{\typctx_{2}}{\mtype_{2}}$ such that $\typctx = \typctx_{1} \mplus \typctx_{2}$, 
		$\sizem{\tderiv} = \sizem{\tderiv_{1}} + \sizem{\tderiv_{2}}$ and $\size{\tderiv} = \size{\tderiv_{1}} + 
		\size{\tderiv_{2}}$.
		
	\end{enumerate}
\end{lemma}

\begin{proof}\hfill
	\begin{enumerate}
		\item By a simple inspection of the typing rules, $\mtype = \emptytype$ and the fact that $\tval$ is a value imply 
		that 
		$\tderiv$ is of the form
		\begin{prooftree}
			\hypo{}
			\infer1[\footnotesize$\ruleMany$]{\tyjp{}{\tval}{}{\emptytype}}
		\end{prooftree}
		where $\dom{\typctx} = \emptyset$ and $\sizem{\tderiv} = 0 = \size{\tderiv}$.
		
		\item Let
		\begin{equation*}
		\tderiv = 
		\begin{prooftree}
		\hypo{}
		\ellipsis{$\tderiv_{i}$}{\tyjp{}{\tval}{\typctx_{i}}{\ltype_{i}}}
		\delims{\left(}{\right)_{\iI}}
		\infer1[\footnotesize$\ruleMany$]{\tyjp{}{\tval}{\bigmplus_{\iI} \typctx_{i}}{\mult{\ltype}_{\iI}}}
		\end{prooftree}
		\end{equation*}
		with $\bigmplus_{\iI} \typctx_{i} = \typctx$ and $\mult{\ltype}_{\iI} = \mtype = \mtype_{1} \mplus \mtype_{2}$. Let 
		$I_{1}$ and $I_{2}$ be sets of indices such that $I = I_{1} \cup I_{2}$, $\mtype_{1} = \mult{\ltype_{i}}_{i \in I_{1}}$ 
		and $\mtype_{2} = \mult{\ltype_{i}}_{i \in I_{2}}$. 
		As $\tval$ is a value, 	we can then derive
		\begin{equation*}
		\tderiv_{1} = 
		\begin{prooftree}
		\hypo{}
		\ellipsis{$\tderiv_{i}$}{\tyjp{}{\tval}{\typctx_{i}}{\ltype_{i}}}
		\delims{\left(}{\right)_{i \in I_1}}
		\infer1[\footnotesize$\ruleMany$]{\tyjp{}{\tval}{\bigmplus_{i \in I_{1}} \typctx_{i}}{\mult{\ltype}_{i \in I_{1}}}}
		\end{prooftree}
		\end{equation*}
		and 
		\begin{equation*}
		\tderiv_{2} = 
		\begin{prooftree}
		\hypo{}
		\ellipsis{$\tderiv_{i}$}{\tyjp{}{\tval}{\typctx_{i}}{\ltype_{i}}}
		\delims{\left(}{\right)_{i \in I_2}}
		\infer1[\footnotesize$\ruleMany$]{\tyjp{}{\tval}{\bigmplus_{i \in I_{2}} \typctx_{i}}{\mult{\ltype}_{i \in I_{2}}}}
		\end{prooftree}
		\end{equation*}
		noting that 
		$$
		\typctx = \bigmplus_{\iI} \typctx_{i} = \left( \bigmplus_{i \in I_{1}} \typctx_{i} \right) \mplus 
		\left(\bigmplus_{i \in I_{2}} \typctx_{i} \right)
		$$
		with 
		$$
		\sizem{\tderiv} = \sum_{\iI} \sizem{\tderiv_{i}} = \left( \sum_{i \in I_{1}} \sizem{\tderiv_{i}} \right) + \left( 
		\sum_{i \in I_{2}} \sizem{\tderiv_{i}} \right) = \sizem{\tderiv_{1}} + \sizem{\tderiv_{2}}
		$$
		and 
		$$
		\size{\tderiv} = \sum_{\iI} \size{\tderiv_{i}} = \left( \sum_{i \in I_{1}} \size{\tderiv_{i}} \right) + \left( 
		\sum_{i \in I_{2}} \size{\tderiv_{i}} \right) = \size{\tderiv_{1}} + \size{\tderiv_{2}}
		$$
		
	\end{enumerate}
\end{proof}

\begin{lemma}[Substitution]
	\label{lappendix:substitution}	
	\NoteProof{l:substitution}
	Let $\tm$ be a term, $\tval$ be a theoretical value and $\namedtyjp{\tderiv}{}{\tm}{\typctx, \var \hastype 
		\mtypetwo}{\mtype}$ and $\namedtyjp{\tderivtwo}{}{\tval}{\typctxtwo}{\mtypetwo}$ be derivations.
	Then there is a derivation $\namedtyjp{\tderivthree}{}{\tm \isub{\var}{\tval}}{\typctx \mplus \typctxtwo}{\mtype}$ 
	with $\sizem{\tderivthree} = \sizem{\tderiv} + \sizem{\tderivtwo}$ and $\size{\tderivthree} \leq \size{\tderiv} + 
	\size{\tderivtwo}$. 
\end{lemma}

\begin{proof}
	By induction on the term $\tm$.
	Cases:
	\begin{itemize}
		\item \emph{Variable}, then are two sub-cases:
		\begin{enumerate}
			\item $\tm = \var$, then $\tm \isub{\var}{\tval} = \var \isub{\var}{\tval} = \tval$ and 			
			$\sizem{\tderiv} = 0$ and $\size{\tderiv} = 1$.
			
			the derivation $\tderiv$ has necessarily the form (for some $n \in \nat$)
			\begin{equation*}
			\tderiv = 
			\begin{prooftree}
			\infer0[\footnotesize$\Ax$]{\tyjp{}{\var}{\var \hastype \mset{\ltype_1}}{\ltype_1}}
			\hypo{\overset{n \in \nat}{\ldots}}
			\infer0[\footnotesize$\Ax$]{\tyjp{}{\var}{\var \hastype \mset{\ltype_n}}{\ltype_n}}
			\infer3[\footnotesize$\ruleManyVar$]{\tyjp{}{\var}{\var \hastype 
					\mset{\ltype_1,\dots,\ltype_n}}{\mset{\ltype_1,\dots,\ltype_n}}}
			\end{prooftree}
			\end{equation*}
			with $\mtype = \mset{\ltype_1, \dots, \ltype_n} = \mtypetwo$ and $\dom{\typctx} = \emptyset$.
			Thus, $\sizem{\tderiv} = 0$ and $\size{\tderiv} = n$.
			Let $\tderivthree = \tderivtwo$: so, $\namedtyjp{\tderivthree}{}{\tm \isub{\var}{\tval}}{\typctx \mplus 
				\typctxtwo}{\mtype}$ (since $\typctx \mplus \typctxtwo = \typctxtwo$) with $\sizem{\tderivthree} = \sizem{\tderivtwo} = 
			\sizem{\tderivtwo} + \sizem{\tderiv}$ and $\size{\tderivthree} = \size{\tderivtwo} \leq \size{\tderivtwo} + 
			\size{\tderiv}$ (note that $\size{\tderivthree} = \size{\tderiv} + \size{\tderivtwo}$ if and only if $n=0$).
			
			\item $\tm = \varthree \neq \var$, then $\tm \isub{\var}{\tval} = \varthree$ and 
			$\sizem{\tderiv} = 0$, $\size{\tderiv} = 1$, 
			the derivation $\tderiv$ has necessarily the form (for some $n~\in~\nat$)
			\begin{equation*}
			\tderiv = 
			\begin{prooftree}
			\infer0[\footnotesize$\Ax$]{\tyjp{}{\varthree}{\varthree \hastype \mset{\ltype_1}}{\ltype_1}}
			\hypo{\overset{n \in \nat}{\ldots}}
			\infer0[\footnotesize$\Ax$]{\tyjp{}{\varthree}{\varthree \hastype \mset{\ltype_n}}{\ltype_n}}
			\infer3[\footnotesize$\ruleManyVar$]{\tyjp{}{\varthree}{\varthree \hastype 
					\mset{\ltype_1,\dots,\ltype_n}}{\mset{\ltype_1,\dots,\ltype_n}}}
			\end{prooftree}
			\end{equation*}
			where $\mtype = \mset{\ltype_1, \dots, \ltype_n}$ and  $\mtypetwo = \emptymset$ and $\typctx = \varthree \hastype 
			\mtype$ (while $\typctx(\var) = \emptymset$).
			Thus, $\sizem{\tderiv} = 0$ and $\size{\tderiv} = n$.
			By \reflemmap{typing-value-splitting}{one}, from $\namedtyjp{\tderivtwo}{}{\tval}{\typctxtwo}{\emptytype}$ it 
			follows that $\sizem{\tderivtwo} = 0 = \size{\tderivtwo}$ and $\dom{\typctxtwo} = \emptyset$.  
			Therefore, $\typctx \mplus \typctxtwo = \typctx$.
			Let $\tderivthree = \tderiv$: so, $\namedtyjp{\tderivthree}{}{\tm \isub{\var}{\tval}}{\typctx \mplus 
				\typctxtwo}{\mtype}$  with $\sizem{\tderivthree} = \sizem{\tderiv} = \sizem{\tderiv} + \sizem{\tderivtwo}$ and 
			$\size{\tderivthree} = \size{\tderiv} = \size{\tderiv} + \size{\tderivtwo}$.
		\end{enumerate}
		
		\item \emph{Application}, \ie $\tm = \tmtwo\tmthree$. 
		Then $\tm \isub{\var}{\tval} = \tmtwo \isub{\var}{\tval} \tmthree \isub{\var}{\tval}$ and necessarily
		\begin{equation*}
		\tderiv = 
		\begin{prooftree}
		\hypo{}
		\ellipsis{$\tderiv_{1}$}{\typctx_1, \var \hastype \mtypetwo_1 \vdash \tmtwo \hastype 
			\mset{\larrow{\mtypethree}{\mtype}}}
		\hypo{}
		\ellipsis{$\tderiv_{2}$}{\typctx_2, \var \hastype \mtypetwo_2 \vdash \tmthree \hastype \mtypethree}
		\infer2[\footnotesize$\ruleAp$]{\typctx, \var \hastype \mtypetwo \vdash \tmtwo \tmthree \hastype \mtype}
		\end{prooftree}
		\end{equation*}
		with $\sizem{\tderiv} = \sizem{\tderiv_{1}} + \sizem{\tderiv_{2}} + 1$, $\size{\tderiv} = \size{\tderiv_{1}} + 
		\size{\tderiv_{2}} + 1$, $\typctx = \typctx_1 \uplus \typctx_2$ and $\mtypetwo = \mtypetwo_2 \uplus \mtypetwo_2$. 
		According to \reflemmap{typing-value-splitting}{two} applied to $\tderivtwo$ and to the decomposition $\mtypetwo = 
		\mtypetwo_1 \uplus \mtypetwo_2$, there are contexts $\typctxtwo_1, \typctxtwo_2$ and derivations 
		$\namedtyjp{\tderivtwo_{1}}{}{\tval}{\typctxtwo_1}{\mtypetwo_1}$ and 
		$\namedtyjp{\tderivtwo_{2}}{}{\tval}{\typctxtwo_2}{\mtypetwo_2}$ such that $\typctxtwo = \typctxtwo_{1} \mplus 
		\typctxtwo_2$, $\sizem{\tderivtwo} = \sizem{\tderivtwo_1} + \sizem{\tderivtwo_2}$ and $\size{\tderivtwo} = 
		\size{\tderivtwo_{1}} + \size{\tderivtwo_{2}}$.
		
		By \ih, there are derivations $\namedtyjp{\tderivthree_1}{}{\tmtwo \isub{\var}{\tval}}{\typctx_1 \uplus 
			\typctxtwo_1}{\mult{\larrow{\mtypethree}{\mtype}}}$ and $\namedtyjp{\tderivthree_2}{}{\tmthree 
			\isub{\var}{\tval}}{\typctx_2 \uplus \typctxtwo_2}{\mtypethree}$ such that $\sizem{\tderivthree_{i}} = 
		\sizem{\tderiv_{i}} + \sizem{\tderivtwo_{i}}$ and $\size{\tderivthree_{i}} \leq \size{\tderiv_{i}} + 
		\size{\tderivtwo_{i}}$ for all $i \in \{1,2\}$.
		Since $\typctx \uplus \typctxtwo = \typctx_1 \uplus \typctxtwo_1 \uplus \typctx_2 \uplus \typctxtwo_2$, we can 
		build the derivation
		\begin{equation*}
		\tderivthree = 
		\begin{prooftree}
		\hypo{}
		\ellipsis{$\tderivthree_1$}{\typctx_1 \uplus \typctxtwo_1 \vdash \tmtwo\isub{\var}{\tval} \hastype 
			\mset{\larrow{\mtypethree}{\mtype}}}
		\hypo{}
		\ellipsis{$\tderivthree_2$}{\typctx_2 \uplus \typctxtwo_2 \vdash \tmthree\isub{\var}{\tval} \hastype \mtypethree}
		\infer2[\footnotesize$\ruleAp$]{\typctx \uplus \typctxtwo \vdash \tmtwo \isub{\var}{\tval} \tmthree 
			\isub{\var}{\tval} \hastype \mtype}
		\end{prooftree}
		\end{equation*}
		where $\sizem{\tderivthree} = \sizem{\tderivthree_{1}} + \size{\tderivthree_{2}} + 1 = \sizem{\tderiv_{1}} + 
		\sizem{\tderivtwo_{1}} + \sizem{\tderiv_{2}} + \sizem{\tderivtwo_{2}} + 1 = \sizem{\tderiv} + \sizem{\tderivtwo}$ and 
		$\size{\tderivthree} = \size{\tderivthree_{1}} + \size{\tderivthree_{2}} + 1 \leq \sizem{\tderiv_{1}} + 
		\sizem{\tderivtwo_{1}} + \sizem{\tderiv_{2}} + \size{\tderivtwo_{2}} + 1 = \size{\tderiv} + \size{\tderivtwo}$.
		
		\item \emph{Abstraction}, \ie $\tm = \la{\vartwo}{\tmtwo}$.
		We can suppose without loss of generality that $\vartwo \notin \fv{\tval} \cup \{\var \}$, hence $\tm 
		\isub{\var}{\tval} = \la{\vartwo}{\tmtwo\isub{\var}\tval}$ and  $\tderiv$ is necessarily of the form (for some $n \in 
		\nat$) 
		\begin{equation*}
		\begin{prooftree}[separation=1em]
		\hypo{}
		\ellipsis{$\tderiv_{i}$}{\typctx_{i}, \vartwo \hastype \mtypethree_{i}, \var \hastype \mtypetwo_{i} \vdash \tmtwo 
			\hastype \mtype_{i}}
		\infer1[\footnotesize$\ruleFun$]{\tyjp{}{\la{\vartwo}{\tmtwo}}{\typctx_{i}, \var \hastype 
				\mtypetwo_{i}}{\ty{\mtypethree_{i}\!}{\!\mtype_{i}}}}
		\delims{\left(}{\right)_{1\leq i \leq n}}
		\infer1[\footnotesize$\ruleManyVal$]{\tyjp{}{\la{\vartwo}{\tmtwo}}{\bigmplus_{i=1}^{n} \typctx_{i} ; \var \hastype 
				\mplus_{i=1}^{n} \mtypetwo_i}{\mplus_{i=1}^{n} \mset{\larrow{\mtypethree_i}{\mtype_i}}}}
		\end{prooftree}
		\end{equation*}
		with $\sizem{\tderiv} = \sum_{i=1}^{n} (\sizem{\tderiv_{i}} + 1)$ and $\size{\tderiv} = \sum_{i=1}^n 
		(\size{\tderiv_i} + 1)$.
		Since $\vartwo \notin \fv{\tval}$, then $\vartwo \notin \domain{\typctxtwo}$ (\refrmk{free-variables}), and so 
		$\namedtyjp{\tderivtwo}{}{\tval}{\typctxtwo, \vartwo \hastype \emptymset}{\mtypetwo}$. 
		Now, there are two subcases:
		\begin{itemize}
			\item \emph{Empty multi type}: If $n = 0$,  then $\mtypetwo = \emptymset = \mtype$ and $\dom{\typctx} = 
			\emptyset$, with $\sizem{\tderiv} = 0 = \size{\tderiv}$. 
			According to \reflemmap{typing-value-splitting}{one} applied to $\tderivtwo$, $\dom{\typctxtwo} = \emptyset$ with 
			$\sizem{\tderivtwo} = 0 = \size{\tderivtwo}$.
			We can then build the derivation 
			\begin{equation*}
			\tderivthree = 
			\begin{prooftree}
			\infer0[\footnotesize$\ruleManyVal$]{\tyjp{}{\la{\vartwo}(\tmtwo\isub{\var}{\tval})}{}{\emptymset}}
			\end{prooftree}
			\end{equation*}
			where $\sizem{\tderivthree} = 0 = \sizem{\tderiv} + \sizem{\tderivtwo}$ and $\size{\tderivthree} = 0 = 
			\size{\tderiv} + \size{\tderivtwo}$, and $\concl{\tderivthree}{\typctx \mplus 
				\typctxtwo}{\tm\isub{\var}{\tval}}{\mtype}$ since $\dom{\typctx \mplus \typctxtwo} = \emptyset$.
			
			\item\emph{Non-empty multi type}: If $n > 0$, then we can decompose $\tderivtwo$ according to the partitioning 
			$\mtypetwo = \biguplus_{i=1}^n \mtypetwo_i$ by repeatedly applying \reflemmap{typing-value-splitting}{two}, and hence 
			for all $1 \leq i \leq n$ there are context $\typctxtwo_{i}$ and a derivation 
			$\namedtyjp{\tderivtwo_i}{}{\tval}{\typctxtwo_{i} ; \vartwo \hastype \emptytype}{\mtypetwo_i}$ such that 
			$\sizem{\tderivtwo} = \sum_{i=1}^n \sizem{\tderivtwo_{i}}$ and $\size{\tderivtwo} = \sum_{i=1}^{n} 
			\size{\tderivtwo_{i}}$.
			By \ih, for all $1 \leq i \leq n$, there is a derivation 
			$\namedtyjp{\tderivthree_{i}}{}{\tmtwo\isub{\var}{\tval}}{\typctx_i \uplus \typctxtwo_i, \vartwo \hastype 
				\mtypethree_i}{\mtype_i}$ such that $\sizem{\tderivthree_{i}} = \sizem{\tderiv_{i}} + \sizem{\tderivtwo_{i}}$ and 
			$\size{\tderivthree_{i}} \leq \size{\tderiv_{i}} + \size{\tderivtwo_{i}}$.
			Since $\typctx \uplus \typctxtwo = \bigmplus_{i=1}^n (\typctx_i \uplus \typctxtwo_i)$, we can build 
			$\tderivthree$ as
			\begin{equation*}
			\begin{prooftree}[separation=1em]
			\hypo{}
			\ellipsis{$\tderivthree_i$}{\typctx_i \uplus \typctxtwo_i, \vartwo \hastype \mtypethree_i \vdash \tmtwo 
				\isub{\var}{\tval} \hastype \mtype_i}
			\infer1[\footnotesize$\ruleFun$]{\tyjp{}{\la{\vartwo}{(\tmtwo \isub{\var}{\tval})}}{\typctx_{i} \mplus 
					\typctxtwo_{i}}{\ty{\mtypethree_{i}\!}{\!\mtype_{i}}}}
			\delims{\left(}{\right)_{1\leq i \leq n}}
			\hypo{}
			\infer2[\footnotesize$\ruleManyVal$]{\typctx \uplus \typctxtwo \vdash \la{\vartwo}{(\tmtwo \isub{\var}{\tval})} 
				\hastype \mtype}
			\end{prooftree}
			\end{equation*}
			noting that $\sizem{\tderivthree} = \sum_{i=1}^n (\sizem{\tderivthree_{i}} + 1) = \sum_{i=1}^n 
			(\sizem{\tderiv_{i}} + \sizem{\tderivtwo_{i}} + 1) = \sum_{i=1}^{n} (\sizem{\tderiv_{i}} + 1) + \sum_{i=1}^{n} 
			\sizem{\tderivtwo_{i}} = \sizem{\tderiv} + \sizem{\tderivtwo} $ and that $\size{\tderivthree} = \sum_{i=1}^{n} 
			(\size{\tderivthree_{i}} + 1) \leq \sum_{i=1}^{n} (\size{\tderiv_{i}} + \size{\tderivtwo_{i}} + 1) = \sum_{i=1}^{n} 
			(\size{\tderiv_{i}} + 1) + \sum_{i=1}^{n} \size{\tderivtwo_{i}} = \size{\tderiv} + \size{\tderivtwo}$.
		\end{itemize}
		
		\item \emph{Explicit substitution}, \ie $\tm = \tmtwo \esub{\vartwo}{\tmthree}$. 
		We can suppose without loss of generality that $\vartwo \notin \fv{\tval} \cup \{\var \}$, hence $\tm 
		\isub{\var}{\tval} = \tmtwo\isub{\var}{\tval} \esub{\vartwo}{\tmthree\isub{\var}\tval}$ and necessarily
		\begin{equation*}
		\tderiv = 
		\begin{prooftree}
		\hypo{}
		\ellipsis{$\tderiv_{1}$}{\typctx_1 , \var \hastype \mtypetwo_1 , \vartwo \hastype \mtypethree \vdash \tmtwo 
			\hastype \mtype}
		\hypo{}
		\ellipsis{$\tderiv_{2}$}{\typctx_2, \var \hastype \mtypetwo_2 \vdash \tmthree \hastype \mtypethree}
		\infer2[\footnotesize$\Es$]{\typctx, \var \hastype \mtypetwo \vdash \tmtwo \esub{\vartwo}{\tmthree} \hastype \mtype}
		\end{prooftree}
		\end{equation*}
		with $\sizem{\tderiv} = \sizem{\tderiv_{1}} + \sizem{\tderiv_{2}}$, $\size{\tderiv} = \size{\tderiv_{1}} + 
		\size{\tderiv_{2}} + 1$, $\typctx = \typctx_1 \uplus \typctx_2$ and $\mtypetwo = \mtypetwo_2 \uplus \mtypetwo_2$. 
		According to \reflemmap{typing-value-splitting}{two} applied to $\tderivtwo$ and to the decomposition $\mtypetwo = 
		\mtypetwo_1 \uplus \mtypetwo_2$, there are contexts $\typctxtwo_1, \typctxtwo_2$ and derivations 
		$\namedtyjp{\tderivtwo_{1}}{}{\tval}{\typctxtwo_{1}}{\mtypetwo_{1}}$ and 
		$\namedtyjp{\tderivtwo_{2}}{}{\tval}{\typctxtwo_{2}}{\mtypetwo_{2}}$ such that $\typctxtwo = \typctxtwo_{1} \mplus 
		\typctxtwo_{2}$, $\sizem{\tderivtwo} = \sizem{\tderivtwo_{1}} + \sizem{\tderivtwo_{2}}$ and $\size{\tderivtwo} = 
		\size{\tderivtwo_{1}} + \size{\tderivtwo_{2}}$.
		
		By \ih, there are derivations $\namedtyjp{\tderivthree_{1}}{}{\tmtwo\isub{\var}{\tval}}{\typctx_{1} \mplus 
			\typctxtwo_{1}, \vartwo \hastype \mtypethree}{\mtype}$ and $\namedtyjp{\tderivthree_{2}}{}{\tmthree 
			\isub{\var}{\tval}}{\typctx_{2} \mplus \typctxtwo_{2}}{\mtypethree}$ such that $\sizem{\tderivthree_{i}} = 
		\sizem{\tderiv_{i}} + \sizem{\tderivtwo_{i}}$ and $\size{\tderivthree_{i}} \leq \size{\tderiv_{i}} + 
		\size{\tderivtwo_{i}}$ for all $i \in \{1,2\}$.
		Since $\typctx \uplus \typctxtwo = \typctx_1 \uplus \typctxtwo_1 \uplus \typctx_2 \uplus \typctxtwo_2$, we can 
		build the derivation
		\begin{equation*}
		\tderivthree = 
		\begin{prooftree}
		\hypo{}
		\ellipsis{$\tderivthree_1$}{\tyjp{}{\tmtwo\isub{\var}{\tval}}{\typctx_{1} \mplus \typctxtwo_{1}, \vartwo \hastype 
				\mtypethree}{\mtype}}
		\hypo{}
		\ellipsis{$\tderivthree_2$}{\typctx_2 \mplus \typctxtwo_2 \vdash \tmthree\isub{\var}{\tval} \hastype \mtypethree}
		\infer2[\footnotesize$\Es$]{\typctx \uplus \typctxtwo \vdash \tmtwo \isub{\var}{\tval} \esub{\vartwo} {\tmthree 
				\isub{\var}{\tval}} \hastype \mtype}
		\end{prooftree}
		\end{equation*}
		verifying that $\sizem{\tderivthree} = \sizem{\tderivthree_{1}} + \sizem{\tderivthree_{2}} = \sizem{\tderiv_{1}} + 
		\sizem{\tderivtwo_{1}} + \sizem{\tderiv_{2}} + \sizem{\tderivtwo_{2}} = \sizem{\tderiv} + \sizem{\tderivtwo}$ and 
		$\size{\tderivthree} = 1 + \size{\tderivthree_{1}} + \size{\tderivthree_{2}} \leq 1 + (\size{\tderiv_{1}} + 
		\size{\tderivtwo_{1}}) + (\size{\tderiv_{2}} + \size{\tderivtwo_{2}}) = \size{\tderiv} + \size{\tderivtwo}$.
		\qedhere
	\end{itemize}	
\end{proof}

\begin{lemma}[Typing of values: merging]
	\label{l:typing-value-complete} 
	Let $\tval$ be a theoretical value.
	\begin{enumerate}
		\item \label{p:typing-value-complete-empty} There is a derivation $\namedtyjp{\tderiv}{}{\tval}{}{\zero}$ with 
		$\sizem{\tderiv} = 0 = \size{\tderiv}$.
		
		\item \label{p:typing-value-complete-merge} For every derivations 
		$\namedtyjp{\tderiv_{1}}{}{\tval}{\typctx_{1}}{\mtype_{1}}$ and 
		$\namedtyjp{\tderiv_{2}}{}{\tval}{\typctx_{2}}{\mtype_{2}}$, there exists a  derivation 
		$\namedtyjp{\tderiv}{}{\tval}{\typctx_{1} \mplus \typctx_2}{\mtype_{1} \mplus \mtype_{2}}$ such that $\sizem{\tderiv} = 
		\sizem{\tderiv_{1}} + \sizem{\tderiv_{2}}$ and $\size{\tderiv} = \size{\tderiv_{1}} + \size{\tderiv_{2}}$.
		
	\end{enumerate}
\end{lemma}

\begin{proof}
	\begin{enumerate}
		\item 
		Let $\tderiv$ be the following type derivation (applying the rule $\ruleMany$ with $n = 0$) such that 
		$\sizem{\tderiv} = 0 = \size{\tderiv}$:
		\begin{equation*}
		\tderiv =
		\begin{prooftree}
		\hypo{}
		\infer1[\footnotesize$\ruleMany$]{\tyjp{}{\tval}{}{\emptytype}}
		\end{prooftree}\ .
		\end{equation*}
		
		\item Let 
		\begin{equation*}
		\tderiv_{1} =
		\begin{prooftree}
		\hypo{}
		\ellipsis{$\tderiv_{i}$}{\tyjp{}{\tval}{\typctx_{i}}{\ltype_{i}}}
		\delims{ \left( }{ \right)_{\iI} }
		\infer1[\footnotesize$\ruleMany$]{\tyjp{}{\tval}{\bigmplus_{\iI} \typctx_{i}}{\bigmplus_{\iI}\mult{\ltype_{i}}}}
		\end{prooftree}
		\end{equation*}
		with $\typctx_{1} = \bigmplus_{\iI} \typctx_{i}$ and $\mtype_{1} = \bigmplus_{\iI}\mult{\ltype_{i}}$, and let
		\begin{equation*}
		\tderiv_{2} =
		\begin{prooftree}
		\hypo{}
		\ellipsis{$\tderiv_{j}$}{\tyjp{}{\tval}{\typctx_{j}}{\ltype_{j}}}
		\delims{ \left( }{ \right)_{\jJ} }
		\infer1[\footnotesize$\ruleMany$]{\tyjp{}{\tval}{\bigmplus_{\jJ} \typctx_{j}}{\bigmplus_{\jJ}\mult{\ltype_{j}}}}
		\end{prooftree}
		\end{equation*}
		with $\typctx_{2} = \bigmplus_{\jJ} \typctx_{j}$ and $\mtype_{2} = \bigmplus_{\jJ}\mult{\ltype_{j}}$.
		
		We can derive $\tderiv$ by setting $K = I \cup J$ and then
		\begin{equation*}
		\tderiv =
		\begin{prooftree}
		\hypo{}
		\ellipsis{$\tderiv_{k}$}{\tyjp{}{\tval}{\typctx_{k}}{\ltype_{k}}}
		\delims{ \left( }{ \right)_{\kK} }
		\infer1[\footnotesize$\ruleMany$]{\tyjp{}{\tval}{\bigmplus_{\kK} \typctx_{k}}{\bigmplus_{\kK}\mult{\ltype_{k}}}}
		\end{prooftree}
		\end{equation*}
		trivially verifying the statement.
		\qedhere
	\end{enumerate}
\end{proof}

\begin{lemma}[Anti-substitution]
	\label{lappendix:anti-substitution}
	\NoteProof{l:anti-substitution}
	Let $\tm$ be a term, $\tval$ be a theoretical value, and 
	$\namedtyjp{\tderiv}{}{\tm\isub{\var}{\tval}}{\typctx}{\mtype}$
	be a type derivation. Then there are two  derivations $\namedtyjp{\tderivtwo}{}{\tm}{\typctxtwo, \var \hastype 
		\mtypetwo}{\mtype}$ 
	and $\namedtyjp{\tderivthree}{}{\tval}{\typctxthree}{\mtypetwo}$ such that $\typctx = \typctxtwo \mplus \typctxthree$ 
	with $\sizem{\tderiv} = \sizem{\tderivtwo} + \sizem{\tderivthree}$ and $\size{\tderiv} \leq \size{\tderivtwo} + 
	\size{\tderivthree}$.
\end{lemma}

\begin{proof}
	By induction on the term $\tm$.
	Cases:
	\begin{itemize}
		\item \emph{Variable}, then are two sub-cases (let $\mtype = \mset{\ltype_1, \dots, \ltype_n}$ for some $n \in 
		\nat$):
		\begin{enumerate}
			\item $\tm = \var$, then $\tm \isub{\var}{\tval} = \tval$. 
			Let $\typctxthree = \typctx$, let $\typctxtwo$ be the empty context (\ie $\dom{\typctxtwo} = \emptyset$), let 
			$\mtypetwo = \mtype$ and let $\tderivtwo$ be the derivation 
			\begin{equation*}
			\tderivtwo = 
			\begin{prooftree}
			\infer0[\footnotesize$\Ax$]{\tyjp{}{\var}{\var \hastype \mset{\ltype_1}}{\ltype_1}}
			\hypo{\overset{n \in \nat}{\ldots}}
			\infer0[\footnotesize$\Ax$]{\tyjp{}{\var}{\var \hastype \mset{\ltype_n}}{\ltype_n}}
			\infer3[\footnotesize$\ruleManyVar$]{\tyjp{}{\var}{\var \hastype 
					\mset{\ltype_1,\dots,\ltype_n}}{\mset{\ltype_1,\dots,\ltype_n}}}
			\end{prooftree}
			\end{equation*}
			Thus, $\concl{\tderivtwo}{\typctxtwo, \var \hastype \mtypetwo}{\tm}{\mtype}$ with $\sizem{\tderivtwo} = 0$ and 
			$\size{\tderivtwo} = n$.
			Let $\tderivthree = \tderiv$: so, $\namedtyjp{\tderivthree}{}{\tval}{\typctxthree}{\mtypetwo}$ and $\typctxthree 
			\mplus \typctxtwo = \typctx$ with $\sizem{\tderiv} = \sizem{\tderivthree} = \sizem{\tderivtwo} + \sizem{\tderivthree}$ 
			and $\size{\tderiv} = \size{\tderivthree} \leq \size{\tderivtwo} + \size{\tderivthree}$.
			
			\item $\tm = \varthree \neq \var$, then $\tm \isub{\var}{\tval} = \varthree$ and 
			the derivation $\tderiv$ has necessarily the form (for some $n~\in~\nat$)
			\begin{equation*}
			\tderiv = 
			\begin{prooftree}
			\infer0[\footnotesize$\Ax$]{\tyjp{}{\varthree}{\varthree \hastype \mset{\ltype_1}}{\ltype_1}}
			\hypo{\overset{n \in \nat}{\ldots}}
			\infer0[\footnotesize$\Ax$]{\tyjp{}{\varthree}{\varthree \hastype \mset{\ltype_n}}{\ltype_n}}
			\infer3[\footnotesize$\ruleManyVar$]{\tyjp{}{\varthree}{\varthree \hastype 
					\mset{\ltype_1,\dots,\ltype_n}}{\mset{\ltype_1,\dots,\ltype_n}}}
			\end{prooftree}
			\end{equation*}
			where $\mtype = \mset{\ltype_1, \dots, \ltype_n}$ and $\typctx = \varthree \hastype \mtype$ (while $\typctx(\var) 
			= \emptymset$).
			Thus, $\sizem{\tderiv} = 0$ and $\size{\tderiv} = n$.
			Let $\typctxthree$ be the empty context (\ie $\dom{\typctxthree} = 0$) and  $\mtypetwo = \emptymset$.
			By \reflemmap{typing-value-complete}{empty}, there is a derivation 
			$\namedtyjp{\tderivthree}{}{\tval}{}{\emptytype}$ (and hence 
			$\namedtyjp{\tderivthree}{}{\tval}{\typctxthree}{\mtypetwo}$) such that $\sizem{\tderivthree} = 0 = 
			\size{\tderivthree}$.  
			Let $\typctxtwo = \typctx$  and $\tderivtwo = \tderiv$:
			therefore, $\typctxthree \mplus \typctxtwo = \typctx$
			and $\namedtyjp{\tderivtwo}{}{\tm }{\typctxtwo, \var \hastype \mtypetwo}{\mtype}$  with $\sizem{\tderiv} = 
			\sizem{\tderivtwo} = \sizem{\tderivtwo} + \sizem{\tderivthree}$ and $\size{\tderiv} = \size{\tderivtwo} \leq 
			\size{\tderivtwo} + \size{\tderivthree}$.
		\end{enumerate}
		
		\item \emph{Application}, \ie $\tm = \tm_1\tm_2$. 
		Then $\tm \isub{\var}{\tval} = \tm_1 \isub{\var}{\tval} \tm_2 \isub{\var}{\tval}$ and necessarily
		\begin{equation*}
		\tderiv = 
		\begin{prooftree}
		\hypo{}
		\ellipsis{$\tderiv_{1}$}{\typctx_1 \vdash \tm_1\isub{\var}{\tval} \hastype \mset{\larrow{\mtypethree}{\mtype}}}
		\hypo{}
		\ellipsis{$\tderiv_{2}$}{\typctx_2 \vdash \tm_2\isub{\var}{\tval} \hastype \mtypethree}
		\infer2[\footnotesize$\ruleAp$]{\typctx \vdash \tm_1\isub{\var}{\tval} \tm_2\isub{\var}{\tval} \hastype \mtype}
		\end{prooftree}
		\end{equation*}
		with $\sizem{\tderiv} = \sizem{\tderiv_{1}} + \sizem{\tderiv_{2}} + 1$, $\size{\tderiv} = \size{\tderiv_{1}} + 
		\size{\tderiv_{2}} + 1$ and $\typctx = \typctx_1 \mplus \typctx_2$. 		
		By \ih, for all $i \in \{1,2\}$, there are derivations $\namedtyjp{\tderivtwo_i}{}{\tm_i}{\typctxtwo_i, \var 
			\hastype \mtypetwo_i}{\mult{\larrow{\mtypethree}{\mtype}}}$ and $\namedtyjp{\tderivthree_i}{}{\tval 
		}{\typctxthree_i}{\mtypetwo_i}$ such that $\typctx_i = \typctxtwo_i \mplus \typctxthree_i$ with $\sizem{\tderiv_{i}} = 
		\sizem{\tderivtwo_{i}} + \sizem{\tderivthree_{i}}$ and $\size{\tderiv_{i}} \leq \size{\tderivtwo_{i}} + 
		\size{\tderivthree_{i}}$.
		According to \reflemmap{typing-value-complete}{merge} applied to $\tderivthree_1$ and $\tderivthree_2$, there is a 
		derivation $\namedtyjp{\tderivthree}{}{\tval}{\typctxthree}{\mtypetwo}$ where $\mtypetwo = \mtypetwo_1 \mplus 
		\mtypetwo_2$ and $\typctxthree = \typctxthree_1 \mplus \typctxthree_2$, such that  $\sizem{\tderivtwo} = 
		\sizem{\tderivtwo_1} + \sizem{\tderivtwo_2}$ and $\size{\tderivtwo} = \size{\tderivtwo_{1}} + \size{\tderivtwo_{2}}$.
		We can build the derivation (where $\typctxtwo = \typctxtwo_1 \mplus \typctxtwo_2$)
		\begin{equation*}
		\tderivtwo = 
		\begin{prooftree}
		\hypo{}
		\ellipsis{$\tderivtwo_1$}{\typctxtwo_1, \var \hastype \mtypetwo_1 \vdash \tm_1 \hastype 
			\mset{\larrow{\mtypethree}{\mtype}}}
		\hypo{}
		\ellipsis{$\tderivtwo_2$}{\typctxtwo_2, \var \hastype \mtypetwo_2 \vdash \tm_2 \hastype \mtypethree}
		\infer2[\footnotesize$\ruleAp$]{\typctxtwo, \var \hastype \mtypetwo \vdash \tm \hastype \mtype}
		\end{prooftree}
		\end{equation*}
		with $\sizem{\tderiv} = \sizem{\tderiv_{1}} + \size{\tderiv_{2}} + 1 = \sizem{\tderivtwo_{1}} + 
		\sizem{\tderivthree_{1}} + \sizem{\tderivtwo_{2}} + \sizem{\tderivthree_{2}} + 1 = \sizem{\tderivtwo} + 
		\sizem{\tderivthree}$ 
		and
		$\size{\tderiv} = \size{\tderiv_{1}} + \size{\tderiv_{2}} + 1 \leq \sizem{\tderivtwo_{1}} + 
		\sizem{\tderivthree_{1}} + \sizem{\tderivtwo_{2}} + \size{\tderivthree_{2}} + 1 = \size{\tderivtwo} + 
		\size{\tderivthree}$.
		
		\item \emph{Abstraction}, \ie $\tm = \la{\vartwo}{\tmtwo}$.
		We can suppose without loss of generality that $\vartwo \notin \fv{\tval} \cup \{\var \}$, hence $\tm 
		\isub{\var}{\tval} = \la{\vartwo}{\tmtwo\isub{\var}\tval}$ and necessarily, for some $n \in \nat$,
		\begin{equation*}
		\tderiv =
		\begin{prooftree}[separation=1em]
		\hypo{}
		\ellipsis{$\tderiv_i$}{\typctx_i, \vartwo \hastype \mtypethree_i \vdash \tmtwo\isub{\var}{\tval} \hastype \mtype_i}
		
		\infer1[\footnotesize$\ruleFun$]{\tyjp{}{\la{\vartwo}{\tmtwo\isub{\var}{\tval}}}{\typctx_{i}}{\ty{\mtypethree_{i}\!}{
					\!\mtype_{1}}}}
		\delims{\left(}{\right)_{1 \leq i \leq n}}
		\hypo{}
		\infer2[\footnotesize$\ruleManyVal$]{\tyjp{}{\la{\vartwo}{\tmtwo\isub{\var}{\tval}}}{\typctx}{\mtype}}
		\end{prooftree}
		\end{equation*}
		with $\typctx = \bigmplus_{i=1}^{n} \typctx_i$, $\mtype = \bigmplus_{i=1}^{n} 
		\mset{\larrow{\mtypethree_i\!}{\!\mtype_i}}$, $\sizem{\tderiv} = n + \sum_{i=1}^{n} \sizem{\tderiv_{i}}$ and 
		$\size{\tderiv} = n + \sum_{i=1}^n \size{\tderiv_i} $.
		%
		$\namedtyjp{\tderivtwo}{}{\tval}{\typctxtwo, \vartwo \hastype \emptymset}{\mtypetwo}$. 
		There are two subcases:
		\begin{itemize}
			\item \emph{Empty multi type}: If $n = 0$,  then $\mtype = \emptymset$ and $\dom{\typctx} = \emptyset$, with 
			$\sizem{\tderiv} = 0 = \size{\tderiv}$. 
			We can  build the derivation 
			\begin{equation*}
			\tderivtwo = 
			\begin{prooftree}
			\infer0[\footnotesize$\ruleManyVal$]{\tyjp{}{\la{\vartwo}\tmtwo}{}{\emptymset}}
			\end{prooftree}
			\end{equation*}
			where $\sizem{\tderivtwo} = 0 = \size{\tderivtwo}$.
			Let $\mtypetwo = \emptymset$ and $\typctxtwo$ be the empty context (\ie $\dom{\typctxtwo} = \emptyset$): then 
			$\concl{\tderivtwo}{\typctxtwo, \var \hastype \mtypetwo}{\tm}{\mtype}$.			
			According to \reflemmap{typing-value-complete}{empty}, there is a derivation 
			$\concl{\tderivthree}{}{\tval}{\emptymset}$ with $\sizem{\tderivthree} = 0 = \size{\tderivthree}$. 
			Let $\typctxthree$ be the empty context (\ie $\dom{\typctxthree} = \emptyset$): so, 
			$\concl{\tderivthree}{\typctxthree}{\tval}{\mtypetwo}$ with $\typctx = \typctxtwo \mplus \typctxthree$ and 
			$\sizem{\tderiv} = 0 = \sizem{\tderivtwo} + \sizem{\tderivthree}$ and $\size{\tderiv} = 0 \leq \size{\tderivtwo} + 
			\size{\tderivthree}$.
			
			\item\emph{Non-empty multi type}: If $n > 0$ then by \ih, for all $1 \leq i \leq n$, there are derivations 
			$\concl{\tderivtwo_i}{\typctxtwo_i, \vartwo \hastype \mtypethree_i, \var \hastype \mtypetwo_i}{\tmtwo}{\mtype_i}$ and 
			$\concl{\tderivthree_i}{\typctxthree_i}{\tval}{\mtypetwo_i}$ such that $\typctx_i = \typctxtwo_i \mplus \typctxthree_i$ 
			with $\sizem{\tderiv_i} = \sizem{\tderivtwo_i} + \sizem{\tderivthree_i}$ and $\size{\tderiv_i} \leq \size{\tderivtwo_i} 
			+ \size{\tderivthree_i}$.
			We can build the derivation
			\begin{equation*}
			\tderivtwo = 
			\begin{prooftree}[separation=1em]
			\hypo{}
			\ellipsis{$\tderivtwo_i$}{\typctxtwo_i ; \vartwo \hastype \mtypethree_i ; \var \hastype \mtypetwo_i \vdash \tmtwo 
				\hastype \mtype_i}
			\infer1[\footnotesize$\ruleFun$]{\tyjp{}{\la{\vartwo}{\tmtwo}}{\typctxtwo_{i} ; \var \hastype 
					\mtypetwo_{i}}{\ty{\mtypethree_{i}\!}{\!\mtype_{i}}}}
			\delims{ \left( }{ \right)_{1 \leq i \leq n} }
			\infer1[\footnotesize$\ruleManyVal$]{\tyjp{}{\la{\vartwo}{\tmtwo}}{\bigmplus_{i=1}^n \typctxtwo_i ; \var \hastype 
					\mplus_{i=1}^n \mtypetwo_i}{\bigmplus_{i=1}^{n} \mult{\ty{\mtypethree_{i}\!}{\!\mtype_{i}}}}}
			\end{prooftree}
			\end{equation*}
			Thus, $\sizem{\tderivtwo} = n + \sum_{i=1}^{n} \sizem{\tderivtwo_{i}}$ and $\size{\tderivtwo} = n + \sum_{i=1}^n 
			\size{\tderivtwo_i} $.
			By repeatedly applying \reflemmap{typing-value-complete}{merge}, there is a derivation 
			$\concl{\tderivthree}{\typctxthree}{\tval}{\mtypetwo}$ with $\typctxthree = \bigmplus_{i=1}^n \typctxthree_i$ such that 
			$\sizem{\tderivthree} = \sum_{i=1}^n\sizem{\tderivthree_i}$ and $\size{\tderivthree} = 
			\sum_{i=1}^n\size{\tderivthree_i}$.
			So, $\typctx = \bigmplus_{i=1}^n \typctx_i = \bigmplus_{i=1}^n (\typctxtwo_i \mplus \typctxthree_i) = \typctxtwo 
			\mplus \typctxthree$ with $\sizem{\tderiv} = n + \sum_{i=1}^n \sizem{\tderiv_i} = n + \sum_{i=1}^n 
			(\sizem{\tderivtwo_i} + \sizem{\tderivthree_i}) = \sizem{\tderivtwo} + \sizem{\tderivthree}$ 
			and $\size{\tderiv} = n + \sum_{i=1}^n \size{\tderiv_i} \leq n + \sum_{i=1}^n (\size{\tderivtwo_i} + 
			\size{\tderivthree_i}) = \size{\tderivtwo} + \size{\tderivthree}$.
		\end{itemize}
		
		\item \emph{Explicit substitution}, \ie $\tm = \tmtwo \esub{\vartwo}{\tmthree}$. 
		We can suppose without loss of generality that $\vartwo \notin \fv{\tval} \cup \{\var \}$, hence $\tm 
		\isub{\var}{\tval} = \tmtwo\isub{\var}{\tval} \esub{\vartwo}{\tmthree\isub{\var}\tval}$ and necessarily
		\begin{equation*}
		\tderiv = 
		\begin{prooftree}
		\hypo{}
		\ellipsis{$\tderiv_1$}{\tyjp{}{\tmtwo\isub{\var}{\tval}}{\typctx_{1}, \vartwo \hastype \mtypethree}{\mtype}}
		\hypo{}
		\ellipsis{$\tderiv_2$}{\typctx_2 \vdash \tmthree\isub{\var}{\tval} \hastype \mtypethree}
		\infer2[\footnotesize$\Es$]{\typctx \vdash \tmtwo \isub{\var}{\tval} \esub{\vartwo} {\tmthree \isub{\var}{\tval}} 
			\hastype \mtype}
		\end{prooftree}
		\end{equation*}
		with $\sizem{\tderiv} = \sizem{\tderiv_{1}} + \sizem{\tderiv_{2}}$, $\size{\tderiv} = \size{\tderiv_{1}} + 
		\size{\tderiv_{2}} + 1$ and $\typctx = \typctx_1 \uplus \typctx_2$. 
		By \ih applied to $\tderiv_1$ and \refrmk{free-variables}, there are derivations 
		$\concl{\tderivtwo_1}{\typctxtwo_1, \vartwo \hastype \mtypethree, \var \hastype \mtypetwo_1}{\tmtwo}{\mtype}$ and
		$\concl{\tderivthree_1}{\typctxthree_1}{\tval}{\mtypetwo_1}$ with $\typctx_1 = \typctxtwo_1 \mplus \typctxthree_1$ 
		such that $\sizem{\tderiv_1} = \sizem{\tderivtwo_1} + \sizem{\tderivthree_1}$ and $\size{\tderiv_1} \leq 
		\size{\tderivtwo_1} + \size{\tderivthree_1}$.
		By \ih applied to $\tderiv_2$ , there are derivations $\concl{\tderivtwo_2}{\typctxtwo_2, \var \hastype 
			\mtypetwo_2}{\tmthree}{\mtype}$ and
		$\concl{\tderivthree_2}{\typctxthree_2}{\tval}{\mtypetwo_2}$ with $\typctx_2 = \typctxtwo_2 \mplus \typctxthree_2$ 
		such that $\sizem{\tderiv_2} = \sizem{\tderivtwo_2} + \sizem{\tderivthree_2}$ and $\size{\tderiv_2} \leq 
		\size{\tderivtwo_2} + \size{\tderivthree_2}$.
		According to \reflemmap{typing-value-complete}{merge}, there is a derivation 
		$\namedtyjp{\tderivthree}{}{\tval}{\typctxthree}{\mtypetwo}$ with $\typctxthree = \typctxthree_{1} \mplus 
		\typctxthree_{2}$ and $\mtypetwo = \mtypetwo_1 \mplus \mtypetwo_2$ such that $\sizem{\tderivthree} = 
		\sizem{\tderivthree_{1}} + \sizem{\tderivthree_{2}}$ and $\size{\tderivthree} = \size{\tderivthree_{1}} + 
		\size{\tderivthree_{2}}$.
		We can build the derivation (where $\typctxtwo = \typctxtwo_1 \mplus \typctxtwo_2$)
		\begin{equation*}
		\tderivtwo = 
		\begin{prooftree}
		\hypo{}
		\ellipsis{$\tderivtwo_{1}$}{\typctxtwo_1 , \var \hastype \mtypetwo_1 , \vartwo \hastype \mtypethree \vdash \tmtwo 
			\hastype \mtype}
		\hypo{}
		\ellipsis{$\tderivtwo_{2}$}{\typctxtwo_2, \var \hastype \mtypetwo_2 \vdash \tmthree \hastype \mtypethree}
		\infer2[\footnotesize$\Es$]{\typctxtwo, \var \hastype \mtypetwo \vdash \tmtwo \esub{\vartwo}{\tmthree} \hastype 
			\mtype}
		\end{prooftree}
		\end{equation*}
		verifying that $\typctx = \typctx_1 \mplus \typctx_2 = \typctxtwo_1 \mplus \typctxthree_1 \mplus \typctxtwo_2 
		\mplus \typctxthree_2 = \typctxtwo \mplus \typctxthree$ and $\sizem{\tderiv} = \sizem{\tderiv_{1}} + \sizem{\tderiv_{2}} 
		= \sizem{\tderivtwo_{1}} + \sizem{\tderivthree_{1}} + \sizem{\tderivtwo_{2}} + \sizem{\tderivthree_{2}} = 
		\sizem{\tderivtwo} + \sizem{\tderivthree}$ and $\size{\tderiv} = 1 + \size{\tderiv_{1}} + \size{\tderiv_{2}} \leq 1 + 
		(\size{\tderivtwo_{1}} + \size{\tderivthree_{1}}) + (\size{\tderivthree_{2}} + \size{\tderivthree_{2}}) = 
		\size{\tderivtwo} + \size{\tderivthree}$.
		\qedhere
	\end{itemize}	
\end{proof}

\begin{proposition}[Qualitative subjects]
	\label{propappendix:qual-subject}
	\NoteState{prop:qual-subject}
	Let $\tm \tovsub \tm'$.
	\begin{enumerate}
		\item \label{pappendix:qual-subject-reduction}
		\emph{Reduction}: 	if $\namedtyjp{\tderiv}{}{\tm}{\typctx}{\mtype}$ then there is a derivation $\namedtyjp{\tderiv'}{}{\tm'}{\typctx}{\mtype}$ such that
		$\sizem{\tderiv'} \leq \sizem{\tderiv}$ and $\size{\tderiv'} \leq \size{\tderiv} $.
		
		\item \label{pappendix:qual-subject-expansion}
		\emph{Expansion}: if $\namedtyjp{\tderiv'}{}{\tm'}{\typctx}{\mtype}$ then there is a derivation $\namedtyjp{\tderiv}{}{\tm}{\typctx}{\mtype}$ such that
		$\sizem{\tderiv'} \leq \sizem{\tderiv}$ and $\size{\tderiv'} \leq \size{\tderiv} $.
	\end{enumerate}
\end{proposition}

\begin{proof}
	\begin{enumerate}
		\item By induction on	the evaluation context $\ctx$ in the step $\tm = \ctxp{\tmtwo} \tovsub \ctxp{\tmtwo'} = \tm'$ with $\tmtwo \rtom \tmtwo'$ or $\tmtwo' \rtoe \tmtwo'$. 
		The proof is analogous to the one for open quantitative subject reduction (\Cref{prop:weak-subject-reduction}), except that now there is one more case:
		\begin{itemize}
			\item \emph{Abstraction}, \ie $\ctx = \la{\var}{\ctxtwo}$. 
			Then, $\tm = \ctxp{\tmtwo} = \la{\var}{\ctxtwop{\tmtwo}} \allowbreak\rootRew{a} \la{\var}{\ctxtwop{\tmtwo'}} = \ctxp{\tmtwo'} = \tm'$ with $\tmtwo \rootRew{a} \tmtwo'$ and $a \in \{\msym, \esym\}$.
			Then, the derivation $\tderiv$ is necessarily (for some $n \in \nat$)
			\begin{equation*}
			\tderiv = 
			\begin{prooftree}[separation=1em]
			\hypo{}
			\ellipsis{$\tderivtwo_i$}{\typctx_i, \var \hastype \mtypetwo_1 \vdash \ctxtwop{\tmtwo} \hastype \mtypethree_i}
			\infer1[\footnotesize$\lambda$]{\typctx_i \vdash \la{\var}\ctxtwop{\tmtwo} \hastype \larrow{\mtypetwo_i}{\mtypethree_i}}
			\delims{ \left( }{ \right)_{1 \leq i \leq n} }
			\infer1[\footnotesize$\ruleManyVal$]{\bigmplus_{i=1}^n \typctx_i \vdash \la{\var}\ctxtwop{\tmtwo} \hastype \bigmplus_{i=1}^{n} \mset{\larrow{\mtypetwo_i}{\mtypethree_i}} }
			\end{prooftree}
			\end{equation*}
			By \ih, for all $1 \leq i \leq n$, there is a derivation $\concl{\tderivtwo_i'}{\typctx_1, \var \hastype \mtypetwo_i}{\ctxtwop{\tmtwo'}}{\mtypethree_i}$ with: 
			$\size{\tderivtwo_i'} \leq \size{\tderivtwo_i}$ and $\sizem{\tderivtwo_i'} \leq \sizem{\tderivtwo_i}$.
			We can then build the derivation 
			\begin{equation*}
			\tderiv' = 
			\begin{prooftree}[separation=1em]
			\hypo{}
			\ellipsis{$\tderivtwo'_i$}{\typctx_i, \var \hastype \mtypetwo_1 \vdash \ctxtwop{\tmtwop} \hastype \mtypethree_i}
			\infer1[\footnotesize$\lambda$]{\typctx_i \vdash \la{\var}\ctxtwop{\tmtwop} \hastype \larrow{\mtypetwo_i}{\mtypethree_i}}
			\delims{ \left( }{ \right)_{1 \leq i \leq n} }
			\infer1[\footnotesize$\ruleManyVal$]{\bigmplus_{i=1}^n \typctx_i \vdash \la{\var}\ctxtwop{\tmtwop} \hastype \bigmplus_{i=1}^{n} \mset{\larrow{\mtypetwo_i}{\mtypethree_i}} }
			\end{prooftree}
			\end{equation*}
			where
			\begin{enumerate}
				\item $\size{\tderiv'} = 1 + \sum_{i=1}^n \size{\tderivtwo_i'} \leq 1 + \sum_{i=1}^n \size{\tderivtwo_i} = \size{\tderiv}$ (note that $\size{\tderiv'} = \size{\tderiv}$ if $n = 0$); 
				\item $\sizem{\tderiv'} = \sum_{i=1}^n \sizem{\tderivtwo_i'} \leq  \sum_{i=1}^n \sizem{\tderivtwo_i} = \sizem{\tderiv}$ (note that $\sizem{\tderiv'} = \sizem{\tderiv}$ if $n = 0$).
			\end{enumerate}
		\end{itemize}
		
		\item By induction on	the evaluation context $\ctx$ in the step $\tm = \ctxp{\tmtwo} \tovsub \ctxp{\tmtwo'} = \tm'$ with $\tmtwo \rtom \tmtwo'$ or $\tmtwo' \rtoe \tmtwo'$. 
		The proof is analogous to the one for open quantitative subject expansion (\Cref{prop:weak-subject-expansion}), except that now there is one more case:
		\begin{itemize}
			\item \emph{Abstraction}, \ie $\ctx = \la{\var}{\ctxtwo}$. 
			Then, $\tm = \ctxp{\tmtwo} = \la{\var}{\ctxtwop{\tmtwo}} \allowbreak\rootRew{a} \la{\var}{\ctxtwop{\tmtwo'}} = \ctxp{\tmtwo'} = \tm'$ with $\tmtwo \rootRew{a} \tmtwo'$ and $a \in \{\msym, \esym\}$.
			Then, the derivation $\tderiv'$ is necessarily (for some $n \in \nat$)
			\begin{equation*}
			\tderiv' = 
			\begin{prooftree}[separation=1em]
			\hypo{}
			\ellipsis{$\tderivtwo'_i$}{\typctx_i, \var \hastype \mtypetwo_1 \vdash \ctxtwop{\tmtwop} \hastype \mtypethree_i}
			\infer1[\footnotesize$\lambda$]{\typctx_i \vdash \la{\var}\ctxtwop{\tmtwop} \hastype \larrow{\mtypetwo_i}{\mtypethree_i}}
			\delims{ \left( }{ \right)_{1 \leq i \leq n} }
			\infer1[\footnotesize$\ruleManyVal$]{\bigmplus_{i=1}^n \typctx_i \vdash \la{\var}\ctxtwop{\tmtwop} \hastype \bigmplus_{i=1}^{n} \mset{\larrow{\mtypetwo_i}{\mtypethree_i}} }
			\end{prooftree}
			\end{equation*}
			By \ih, for all $1 \leq i \leq n$, there is a derivation $\concl{\tderivtwo_i}{\typctx_1, \var \hastype \mtypetwo_i}{\ctxtwop{\tmtwo}}{\mtypethree_i}$ with
			$\size{\tderivtwo_i'} \leq \size{\tderivtwo_i}$ and
			$\sizem{\tderivtwo_i'} \leq \sizem{\tderivtwo_i}$.
			We can then build the derivation 
			\begin{equation*}
			\tderiv = 
			\begin{prooftree}[separation=1em]
			\hypo{}
			\ellipsis{$\tderivtwo_i$}{\typctx_i, \var \hastype \mtypetwo_1 \vdash \ctxtwop{\tmtwo} \hastype \mtypethree_i}
			\infer1[\footnotesize$\lambda$]{\typctx_i \vdash \la{\var}\ctxtwop{\tmtwo} \hastype \larrow{\mtypetwo_i}{\mtypethree_i}}
			\delims{ \left( }{ \right)_{1 \leq i \leq n} }
			\infer1[\footnotesize$\ruleManyVal$]{\bigmplus_{i=1}^n \typctx_i \vdash \la{\var}\ctxtwop{\tmtwo} \hastype \bigmplus_{i=1}^{n} \mset{\larrow{\mtypetwo_i}{\mtypethree_i}} }
			\end{prooftree}
			\end{equation*}
			where
			\begin{enumerate}
				\item $\size{\tderiv} = 1 + \sum_{i=1}^n \size{\tderivtwo_i} \leq 1 + \sum_{i=1}^n \size{\tderivtwo_i'} = \size{\tderiv'}$ (note that $\size{\tderiv'} = \size{\tderiv}$ if $n = 0$); 
				\item $\sizem{\tderiv} = \sum_{i=1}^n \sizem{\tderivtwo_i} \leq  \sum_{i=1}^n \sizem{\tderivtwo_i'} = \sizem{\tderiv'}$ (note that $\sizem{\tderiv'} = \sizem{\tderiv}$ if $n = 0$).
						\qedhere
			\end{enumerate}
		\end{itemize} 

	\end{enumerate}
\end{proof}

\begin{proposition}[Invariance under evaluation]
	\label{propappendix:invariance}
	\NoteState{prop:invariance}
	Let $\vec{\var}$ be suitable for $\tm$. 
	If $\tm \tovsub^* \tm'$ then $\sem{\tm}_{\vec{\var}} = \sem{\tm'}_{\vec{\var}}$.
\end{proposition}

\begin{proof}
	It is enough to show that $\tm \tovsub \tm'$ implies $\sem{\tm}_{\vec{\var}} = \sem{\tm'}_{\vec{\var}}$.
	Let $\tm \tovsub \tm'$.
	For any $\concl{\tderiv}{\var_1 \hastype \mtypetwo_1, \dots, \var_n \hastype \mtypetwo_n}{\tm}{\mtype}$, there is $\concl{\tderiv'}{\var_1 \hastype \mtypetwo_1, \dots, \var_n \hastype \mtypetwo_n}{\tm'}{\mtype}$ by subject reduction (\Cref{prop:qual-subject}.\ref{p:qual-subject-reduction}), so $\sem{\tm}_{\vec{\var}} \subseteq \sem{\tm'}_{\vec{\var}}$.
	Conversely, for any $\concl{\tderiv'}{\var_1 \hastype \mtypetwo_1, \dots, \var_n \hastype \mtypetwo_n}{\tm'}{\mtype}$, there is $\concl{\tderiv}{\var_1 \hastype \mtypetwo_1, \dots, \var_n \hastype \mtypetwo_n}{\tm}{\mtype}$ by subject expansion (\Cref{prop:qual-subject}.\ref{p:qual-subject-expansion}), so $\sem{\tm'}_{\vec{\var}} \subseteq \sem{\tm}_{\vec{\var}}$.
\end{proof}

\section{Proofs of \Cref{sect:open}}

\begin{remark}[Merging and splitting inertness]
	\label{rmk:merge-split-inert}
	Let $\mtype, \mtypetwo, \mtypethree$ be multi types with $\mtype = \mtypetwo \mplus \mtypethree$; then, $\mtype$ is inert iff $\mtypetwo$ and $\mtypethree$ are inert.
	Similarly for type contexts.
\end{remark}

\begin{lemma}
	[Spreading of inertness on judgments]
	\label{lappendix:spread-inert}
	\NoteState{l:spread-inert}
	Let $\concl{\tderiv}{\typctx}{\itm}{\mtype}$ be a derivation and $\itm$ be an inert term. 
	If \,$\typctx$ is a inert type context, then $\mtype$ is a inert multi type.
\end{lemma}

\begin{proof}
	By induction on the definition of inert terms $\itm$.
	Cases:
	\begin{itemize}
		\item \emph{Variable}, \ie $\itm = \var$. Then necessarily, for some $n \in \nat$,
		\begin{equation*}
		\tderiv = 
		\begin{prooftree}
		\infer0[\footnotesize$\ruleAx$]{\tyjp{}{\var}{\var \hastype \mset{\ltype_1}}{\ltype_1}}
		\hypo{\overset{n \in \nat}{\ldots}}
		\infer0[\footnotesize$\ruleAx$]{\tyjp{}{\var}{\var \hastype \mset{\ltype_n}}{\ltype_n}}
		\infer3[\footnotesize$\ruleManyVar$]{\tyjp{}{\var}{\typctx}{\mtype}}
		\end{prooftree}
		\end{equation*}
		where $\mtype = \mset{\ltype_1, \dots, \ltype_n}$ and $\typctx = \var \hastype \mtype$. 
		As $\typctx$ is inert (by hypothesis), so is $\mtype$.
		
		\item \emph{Application}, \ie $\fire = \itm \firetwo$. Then necessarily
		\begin{equation*}
		\tderiv = 
		\begin{prooftree}
		\hypo{}
		\ellipsis{$\tderivtwo$}{\typctxtwo \vdash \itm \hastype \mult{\ty{\mtypetwo}{\mtype}}}
		\hypo{}
		\ellipsis{$\tderivthree$}{\typctxthree \vdash \firetwo \hastype\mtypetwo}
		\infer2[\footnotesize$\ruleApp$]{\tyjp{}{\itm \firetwo}{\typctxtwo \mplus \typctxthree}{\mtype}}
		\end{prooftree}
		\end{equation*}
		where $\typctx = \typctxtwo \mplus \typctxthree$.
		As $\typctx$ is inert, so is $\typctxtwo$, according to \Cref{rmk:merge-split-inert}.
		By \ih applied to $\tderivtwo$, $\mult{\ty{\mtypetwo}{\mtype}}$ is inert and hence $\mtype$ is inert.
		
		\item \emph{Explicit substitution on inert}, \ie $\itm = \itmthree \esub{\var}{\itmtwo}$. Then necessarily
		\begin{equation*}
		\tderiv = 
		\begin{prooftree}
		\hypo{}
		\ellipsis{$\tderivtwo$}{\typctxtwo, \var \hastype \mtypetwo \vdash \itmthree \hastype \mtype}
		\hypo{}
		\ellipsis{$\tderivthree$}{\typctxthree \vdash \itmtwo \hastype \mtypetwo}
		\infer2[\footnotesize$\ruleES$]{\tyjp{}{\itmthree \esub{\var}{\itmtwo}}{\typctxtwo \mplus \typctxthree}{\mtype}}
		\end{prooftree}
		\end{equation*}
		where $\typctx = \typctxtwo \mplus \typctxthree$.
		As $\typctx$ is inert, so are $\typctxtwo$ and $\typctxthree$, according to \Cref{rmk:merge-split-inert}.
		By \ih applied to $\tderivthree$, $\mtypetwo$ is inert. 
		Therefore, $\typctxtwo, \var \hastype \mtypetwo$ is a inert type context.
		By \ih applied to $\tderivtwo$, the multi type $\mtype$ is inert.
		\qedhere
	\end{itemize}
\end{proof}

\paragraph{Correctness}

\begin{lemma}[Size of fireballs]
	\label{lappendix:size-fireballs}
	\NoteState{l:size-fireballs}
	Let $\fire$ be a fireball. 
	If $\namedtyjp{\tderiv}{}{\fire}{\typctx}{\mtype}$
	then $\sizem{\tderiv} \geq \sizeo{\fire}$.
	If, moreover, $\typctx$ is inert and ($\mtype$ is ground inert or $\fire$ is inert), then 
$\sizem{\tderiv} = \sizeo{\fire}$.
\end{lemma}

\begin{proof}
	By mutual induction on the definition of fireballs $\fire$ and inert terms $\itm$.
	Cases for inert terms:
	\begin{itemize}
		\item \emph{Variable}, \ie $\fire = \var$. Then necessarily, for some $n \in \nat$,
		\begin{equation*}
		\tderiv = 
		\begin{prooftree}
		\infer0[\footnotesize$\ruleAx$]{\tyjp{}{\var}{\var \hastype \mset{\ltype_1}}{\ltype_1}}
		\hypo{\overset{n \in \nat}{\ldots}}
		\infer0[\footnotesize$\ruleAx$]{\tyjp{}{\var}{\var \hastype \mset{\ltype_n}}{\ltype_n}}
		\infer3[\footnotesize$\ruleManyVar$]{\tyjp{}{\var}{\typctx}{\mtype}}
		\end{prooftree}
		\end{equation*}
		where $\mtype = \mset{\ltype_1, \dots, \ltype_n}$ and $\typctx = \var \hastype \mtype$. 
		Therefore, $\sizeo{\fire} = 0 = \sizem{\tderiv}$.
		
		\item \emph{Application}, \ie $\fire = \itm \firetwo$. Then necessarily
		\begin{equation*}
		\tderiv = 
		\begin{prooftree}
		\hypo{}
		\ellipsis{$\tderivtwo$}{\typctxtwo \vdash \itm \hastype \mult{\ty{\mtypetwo}{\mtype}}}
		\hypo{}
		\ellipsis{$\tderivthree$}{\typctxthree \vdash \firetwo \hastype\mtypetwo}
		\infer2[\footnotesize$\ruleApp$]{\tyjp{}{\itm \firetwo}{\typctxtwo \mplus \typctxthree}{\mtype}}
		\end{prooftree}
		\end{equation*}
		where $\typctx = \typctxtwo \mplus \typctxthree$.
		By \ih applied to both premises, $\sizem{\tderivtwo} \geq \size{\itm}$ and $\sizem{\tderivthree} \geq 
\size{\firetwo}$.
		Therefore, $\sizeo{\fire} = \sizeo{\itm} + \sizeo{\firetwo} + 1 \leq \sizem{\tderivtwo} + \sizem{\tderivthree} + 1 = 
\sizem{\tderiv}$.
		
		If, moreover, $\typctx$ is inert then so are $\typctxtwo$ and $\typctxthree$, according to 
\Cref{rmk:merge-split-inert}.
		By \ih for inert terms applied to $\tderivtwo$, $\sizem{\tderivtwo} = \sizeo{\itm}$.
		By spreading of inertness (\Cref{l:spread-inert}), $\mset{\larrow{\mtypetwo}{\mtype}}$ is inert, hence $\mtypetwo = \emptytype$.
		By \ih for fireballs applied to $\tderivthree$, $\sizem{\tderivthree} = \sizeo{\firetwo}$.
		Therefore, $\size{\fire} = \size{\itm} + \size{\firetwo} + 1 = \sizem{\tderivtwo} + \sizem{\tderivthree} + 1 = 
\sizem{\tderiv}$.
		
		\item \emph{Explicit substitution on inert}, \ie $\fire = \itm \esub{\var}{\itmtwo}$. Then necessarily
		\begin{equation*}
		\tderiv = 
		\begin{prooftree}
		\hypo{}
		\ellipsis{$\tderivtwo$}{\typctxtwo, \var \hastype \mtypetwo \vdash \itm \hastype \mtype}
		\hypo{}
		\ellipsis{$\tderivthree$}{\typctxthree \vdash \itmtwo \hastype \mtypetwo}
		\infer2[\footnotesize$\ruleES$]{\tyjp{}{\itm \esub{\var}{\itmtwo}}{\typctxtwo \mplus \typctxthree}{\mtype}}
		\end{prooftree}
		\end{equation*}
		where $\typctx = \typctxtwo \mplus \typctxthree$.
		We can then apply \ih to both premises: $\sizem{\tderivtwo} \geq \sizeo{\itm}$ and $\sizem{\tderivthree} \geq \sizeo{\itmtwo}$. 
		Therefore, $\sizeo{\sfire} = \sizeo{\sitm} + \sizeo{\itmtwo} \leq \sizem{\tderivtwo} + \sizem{\tderivthree} = \sizem{\tderiv}$.	
		
		If, moreover, $\typctx$ is inert then so are $\typctxtwo$ and $\typctxthree$, according to 
\Cref{rmk:merge-split-inert}.
		By spreading of inertness (\Cref{l:spread-inert}), $\mtypetwo$ is inert, hence $\typctxtwo, \var \hastype \mtypetwo$ is a inert type context.
		By \ih for inert terms applied to $\tderivtwo$ and $\tderivthree$, we have $\sizem{\tderivtwo} = \sizeo{\itm}$.
		and $\sizem{\tderivthree} = \sizeo{\itmtwo}$.
	 	So, $\sizeo{\fire} = \sizeo{\itm} + \sizeo{\itmtwo} = \sizem{\tderivtwo} + \sizem{\tderivthree} = \sizem{\tderiv}$.
		
	\end{itemize}

	Cases for fireballs that are not inert terms:
	\begin{itemize}
		\item \emph{Explicit substitution on fireball}, \ie $\fire = \firetwo \esub{\var}{\itm}$. Then necessarily
		\begin{equation*}
		\tderiv = 
		\begin{prooftree}
		\hypo{}
		\ellipsis{$\tderivtwo$}{\typctxtwo, \var \hastype \mtypetwo \vdash \firetwo \hastype \mtype}
		\hypo{}
		\ellipsis{$\tderivthree$}{\typctxthree \vdash \itm \hastype \mtypetwo}
		\infer2[\footnotesize$\ruleES$]{\tyjp{}{\firetwo \esub{\var}{\itm}}{\typctxtwo \mplus \typctxthree}{\mtype}}
		\end{prooftree}
		\end{equation*}
		where $\typctx = \typctxtwo \mplus \typctxthree$.
		We can then apply \ih to both premises: $\sizem{\tderivtwo} \geq \sizeo{\firetwo}$ and $\sizem{\tderivthree} \geq 
\sizeo{\itm}$. 
		Therefore, $\sizeo{\fire} = \sizeo{\firetwo} + \sizeo{\sitm} \leq \sizem{\tderivtwo} + \sizem{\tderivthree} = 
\sizem{\tderiv}$ 
		
		If, moreover, $\typctx$ is inert and $\mtype$ is ground inert, then $\typctxtwo$ and $\typctxthree$ are inert, according to \Cref{rmk:merge-split-inert}.
		By \ih for inert terms applied to $\tderivthree$, $\sizem{\tderivthree} = \sizeo{\itm}$.
		By spreading of inertness (\Cref{l:spread-inert}), $\mtypetwo$ is inert, hence $\typctxtwo, \var \hastype \mtypetwo$ is a inert type context.
		By \ih for fireballs applied to $\tderivtwo$, $\sizem{\tderivtwo} = \sizeo{\firetwo}$.
		Therefore, $\sizeo{\fire} = \sizeo{\firetwo} + \sizeo{\itm} = \sizem{\tderivtwo} + \sizem{\tderivthree} = 
\sizem{\tderiv}$.
		
		\item \emph{Abstraction}, \ie $\fire = \la{\var}{\tm}$. 
		Then necessarily, for some $n \in \nat$,
		\begin{equation*}
		\tderiv = 
		\begin{prooftree}[separation = 1em]
		\hypo{}
		\ellipsis{$\tderivtwo_1$}{\typctx_1, \var \hastype \mtypethree_1 \vdash \tm \hastype \mtypetwo_1}
		\infer1[\footnotesize$\ruleFun$]{\tyjp{}{\la{\var}{\tm}}{\typctx_1}{\ty{\mtypethree_1}{\mtypetwo_1}}}
		\hypo{\overset{n \in \nat}{\ldots}}
		\hypo{}
		\ellipsis{$\tderivtwo_n$}{\typctx_n, \var \hastype \mtypethree_n \vdash \tm \hastype \mtypetwo_n}
		\infer1[\footnotesize$\ruleFun$]{\tyjp{}{\la{\var}{\tm}}{\typctx_n}{\ty{\mtypethree_n}{\mtypetwo_n}}}
		\infer3[\footnotesize$\ruleManyVal$]{\tyjp{}{\la{\var}{\tm}}{\typctx}{\mtype}}
		\end{prooftree}
		\end{equation*}
		where $\mtype = \bigmplus_{i=1}^n\mset{\larrow{\mtypethree_i}{\mtypetwo_i}}$ and $\typctx = 
\bigmplus_{i=1}^n\typctx_i$. 
		Therefore, $\sizeo{\fire} = 0 \leq \sum_{i=1}^n(\sizem{\tderivtwo_i} + 1) = \sizem{\tderiv}$.
		
		If, moreover, $\mtype$ is ground inert, then necessarily $\mtype = \emptytype$ and $n = 0$, hence $\tderiv$ consist of the rule $\ruleManyVal$ with $0$ premises.
		Therefore, $\sizeo{\fire}= 0 = \sizem{\tderiv}$. 
		\qedhere
	\end{itemize}
\end{proof}

\begin{proposition}[Open quantitative subject reduction]
	\label{propappendix:weak-subject-reduction}
	\NoteState{prop:weak-subject-reduction}
	Let $\namedtyjp{\tderiv}{}{\tm}{\typctx}{\mtype}$ be a derivation.
	\begin{enumerate}
		\item If $\tm \towm \tm'$ then there exists a derivation $\namedtyjp{\tderiv'}{}{\tm'}{\typctx}{\mtype}$ such that
		$\sizem{\tderiv'} = \sizem{\tderiv} - 2$ and $\size{\tderiv'} = \size{\tderiv} - 1$; 
		\item If $\tm \towe \tm'$ then there exists a derivation $\namedtyjp{\tderiv'}{}{\tm'}{\typctx}{\mtype}$ such that
		$\sizem{\tderiv'} = \sizem{\tderiv}$ and $\size{\tderiv'} < \size{\tderiv}$.
	\end{enumerate}
\end{proposition}

\begin{proof}
	By induction on the open evaluation context $\weakctx$ such that $\tm = \weakctxp{\tmtwo} \tow \weakctxp{\tmtwo'} = \tm'$ with $\tmtwo \rtom \tmtwo'$ or $\tmtwo' \rtoe \tmtwo'$. 
	Cases for $\weakctx$:
	\begin{itemize}
		\item \emph{Hole}, \ie $\weakctx = \ctxhole$.
		Then there are two sub-cases:
		\begin{enumerate}
			\item \emph{Multiplicative}, \ie $\tm = \subctxp{\la\var\tmtwo}\tmthree \rtom  \subctxp{\tmtwo \esub{\var}{\tmthree}} = \tm'$.
			Then $\tderiv$ has necessarily the form:
			\begin{equation*}
			\tderiv = 
			\begin{prooftree}[separation = 1em]
			\hypo{}
			\ellipsis{$\tderivtwo$}{\typctx', \var \hastype \mtypetwo \vdash \tmtwo \hastype \mtype}
			\infer1[\footnotesize$\ruleFun$]{\typctx' \vdash \la\var\tmtwo \hastype \larrow{\mtypetwo}{\mtype}}
			\infer1[\footnotesize$\ruleManyVal$]{\typctx' \vdash \la\var\tmtwo \hastype \mset{\larrow{\mtypetwo}{\mtype}}}
			\hypo{}
			\ellipsis{$\tderiv_1$}{\quad}
			\infer2[\footnotesize$\ruleES$]{}
			\ellipsis{}{\quad}
			\hypo{}
			\ellipsis{$\tderiv_n$}{\quad}
			\infer2[\footnotesize$\ruleES$]{\typctx \vdash \subctxp{\la\var\tmtwo} \hastype \mset{\larrow{\mtypetwo}{\mtype}}}
			\hypo{}
			\ellipsis{$\tderivthree$}{\typctxtwo\vdash\tmthree \hastype \mtypetwo}
			\infer2[\footnotesize$\ruleApp$]{\typctx \uplus \typctxtwo \vdash \subctxp{\la\var\tmtwo}\tmthree \hastype \mtype}
			\end{prooftree}
			\end{equation*}
			
			with $\sizem{\tderiv} = 2 + \sizem{\tderivtwo} + \sizem{\tderivthree} + \sum_{i=1}^{n} \sizem{\tderiv_{i}}$ and $\size{\tderiv} = 2 + n + \size{\tderivtwo} + \size{\tderivthree} + \sum_{i=1}^{n} \size{\tderiv_{i}}$.
			We can then build $\tderiv'$ as follows:
			\begin{equation*}
			\tderiv' = 
			\begin{prooftree}
			\hypo{}
			\ellipsis{$\tderivtwo$}{\tyjp{}{\tmtwo}{\typctx' ; \var \hastype \mtypetwo}{\mtype}}
			\hypo{}
			\ellipsis{$\tderivthree$}{\tyjp{}{\tmthree}{\typctxtwo}{\mtypetwo}}
			\infer2[\footnotesize$\ruleES$]{\typctx'\uplus\typctxtwo \vdash\tmtwo\esub\var\tmthree \hastype \mtype}
			\hypo{}
			\ellipsis{$\tderiv_1$}{\quad}
			\infer2[\footnotesize$\ruleES$]{}
			\ellipsis{}{\quad}
			\hypo{}
			\ellipsis{$\tderiv_n$}{\quad}
			\infer2[\footnotesize$\Es$]{\typctx\uplus\typctxtwo \vdash \subctxp{\tmtwo\esub\var\tmthree} \hastype \mtype}
			\end{prooftree}
			\end{equation*}
			where $\sizem{\tderiv'} = \sizem{\tderivtwo} + \sizem{\tderivthree} + \sum_{i=1}^{n} \sizem{\tderiv_{i}} = \sizem{\tderiv} - 2$ and $\size{\tderiv'} = 1+ n + \size{\tderivtwo} + \size{\tderivthree} + \sum_{i=1}^{n} \size{\tderiv_{i}} = \size{\tderiv} - 1$.
			
			\item \emph{Exponential}, \ie $\tm = \tmtwo\esub\var{\subctxp{\tval}} \rtoe \subctxp{\tmtwo \isub{\var}{\tval}} = \tmp$.
			Then the derivation $\tderiv$ has necessarily the form:
			\begin{equation*}
			\tderiv = 
			\begin{prooftree}
			\hypo{}
			\ellipsis{$\tderivtwo$}{\tyjp{}{\tmtwo}{\typctxtwo, \var \hastype \mtypetwo}{\mtype}}
			\hypo{}
			\ellipsis{$\tderivthree$}{\tyjp{}{\tval}{\typctxthree'}{\mtypetwo}}
			\hypo{}
			\ellipsis{$\tderiv_1$}{\quad}
			\infer2[\footnotesize$\Es$]{}
			\ellipsis{}{\quad}
			\hypo{}
			\ellipsis{$\tderiv_n$}{\quad}
			\infer2[\footnotesize$\Es$]{\typctxthree \vdash \subctxp{\tval} \hastype \mtypetwo}
			\infer2[\footnotesize$\Es$]{\typctxtwo\uplus\typctxthree \vdash\tmtwo\esub\var{\subctxp{\tval}}\hastype \mtype}
			\end{prooftree}
			\end{equation*}
			where $\typctx = \typctxtwo \uplus \typctxthree$, $\sizem{\tderiv} = \sizem{\tderivtwo} + \sizem{\tderivthree} + \sum_{i=1}^{n} \sizem{\tderiv_{i}}$ and $\size{\tderiv} = 1 + n + \size{\tderivtwo} + \size{\tderivthree} + \sum_{i=1}^{n} \size{\tderiv_{i}}$.
			By the substitution lemma (\reflemma{substitution}), there is a derivation $\namedtyjp{\tderiv''}{}{\tmtwo\isub{\var}{\tval}}{\typctxtwo \mplus \typctxthree'}{\mtype}$
			such that $\sizem{\tderiv''} = \sizem{\tderivtwo} + \sizem{\tderivthree}$ and $\size{\tderiv''} \leq \size{\tderivtwo} + \size{\tderivthree}$.
			We can then build the following derivation $\tderiv'$:
			\begin{equation*}
			\tderiv' = 
			\begin{prooftree}
			\hypo{}
			\ellipsis{$\tderiv''$}{\typctxtwo \mplus\typctxthree' \vdash \tmtwo\isub\var \tval \hastype \mtype}
			\hypo{}
			\ellipsis{$\tderiv_1$}{\quad}
			\infer2[\footnotesize{$\Es$}]{}
			\ellipsis{}{}
			\hypo{}
			\ellipsis{$\tderiv_n$}{\quad}
			\infer2[\footnotesize$\Es$]{\typctxtwo\mplus\typctxthree \vdash \subctxp{\tmtwo\isub\var \val} \hastype \mtype}
			\end{prooftree}
			\end{equation*}
			where $\sizem{\tderiv'} = \sizem{\tderiv''} + \sum_{i=1}^{n} \sizem{\tderiv_{i}} = \sizem{\tderivtwo} + \sizem{\tderivthree} + \sum_{i=1}^{n} \sizem{\tderiv_{i}} = \sizem{\tderiv}$ and $\size{\tderiv'} = n + \size{\tderiv''} + \sum_{i=1}^{n} \size{\tderiv_{i}} \leq n + \size{\tderivtwo} + \size{\tderivthree} + \sum_{i=1}^{n} \size{\tderiv_{i}} < 1 + n + \size{\tderivtwo} + \size{\tderivthree} + \sum_{i=1}^{n} \size{\tderiv_{i}} = \size{\tderiv}$ ($\tderiv'$ contains at least one rule $\Es$ less than $\tderiv$).
		\end{enumerate}
		
		\item \emph{Application left}, \ie $\weakctx = \weakctxtwo\tmthree$.
		Then, $\tm = \weakctxp{\tmtwo} = \weakctxtwop{\tmtwo} \tmthree \rootRew{a} \weakctxtwop{\tmtwo'} \tmthree = \weakctxp{\tmtwo'} = \tm'$ with $\tmtwo \rootRew{a} \tmtwo'$ and $a \in \{\msym, \esym\}$.
		The derivation $\tderiv$ is necessarily
		\begin{equation*}
		\tderiv = 
		\begin{prooftree}
		\hypo{}
		\ellipsis{$\tderivtwo$}{\tyjp{}{\weakctxtwop{\tmtwo}}{\typctxtwo}{\mult{\ty{\mtypetwo}{\mtype}}}}
		\hypo{}
		\ellipsis{$\tderivthree$}{\tyjp{}{\tmthree}{\typctxthree}{\mtypetwo}}
		\infer2[\footnotesize$\ruleAp$]{\tyjp{}{\weakctxtwop{\tmtwo} \tmthree}{\typctxtwo \mplus \typctxthree}{\mtype}}
		\end{prooftree}
		\end{equation*}
		where $\typctx = \typctxtwo \uplus \typctxthree$, $\sizem{\tderiv} = 1 + \sizem{\tderivtwo} + \sizem{\tderivthree}$ and $\size{\tderiv} = 1 + \size{\tderivtwo} + \size{\tderivthree}$.
		By \ih, there is a derivation $\namedtyjp{\tderivtwo'}{}{\weakctxtwop{\tmtwo'}}{\typctxtwo}{\mult{\ty{\mtypetwo}{\mtype}}}$ with: 
		\begin{enumerate}
			\item $\sizem{\tderivtwo'} = \sizem{\tderivtwo} - 2$ and $\size{\tderivtwo'} = \size{\tderivtwo} - 1$ if $\tmtwo \rtom \tmtwo'$; 
			\item $\sizem{\tderivtwo'} = \sizem{\tderivtwo}$ and $\size{\tderivtwo'} < \size{\tderivtwo}$ if $\tmtwo \rtoe \tmtwo'$.
		\end{enumerate}
		We can then build the derivation 
		\begin{equation*}
		\tderiv' = 
		\begin{prooftree}
		\hypo{}
		\ellipsis{$\tderivtwo'$}{\tyjp{}{\weakctxtwop{\tmtwo'}}{\typctxtwo}{\mult{\ty{\mtypetwo}{\mtype}}}}
		\hypo{}
		\ellipsis{$\tderivthree$}{\tyjp{}{\tmthree}{\typctxthree}{\mtypetwo}}
		\infer2[\footnotesize$\ruleAp$]{\tyjp{}{\weakctxtwop{\tmtwo'} \tmthree}{\typctxtwo \mplus \typctxthree}{\mtype}}
		\end{prooftree}
		\end{equation*}
		noting that
		\begin{enumerate}
			\item If $\tmtwo \rtom \tmtwo'$ then $\sizem{\tderiv'} = 1 + \sizem{\tderivtwo'} + \sizem{\tderivthree} = 1 + (\sizem{\tderivtwo} - 2) + \sizem{\tderivthree} = \sizem{\tderiv} - 2$ and $\size{\tderiv'} = 1 + \size{\tderivtwo'} + \size{\tderivthree} = 1 + (\size{\tderivtwo} - 1) + \size{\tderivthree} = \size{\tderiv} - 1$; 
			\item If $\tmtwo \rtoe \tmtwo'$ then $\sizem{\tderiv'} = 1 + \sizem{\tderivtwo'} + \sizem{\tderivthree} = 1 + \sizem{\tderivtwo} + \sizem{\tderivthree} = \sizem{\tderiv}$ and $\size{\tderiv'} = 1 + \size{\tderivtwo'} + \size{\tderivthree} < 1 + \size{\tderivtwo} + \size{\tderivthree} = \size{\tderiv}$.
		\end{enumerate}
		
		\item \emph{Application right}, \ie $\weakctx = \tmthree \weakctxtwo$.
		Analogous to the previous case.
		
		\item \emph{Explicit substitution left}, \ie $\weakctx = \weakctxtwo\esub{\var}{\tmthree}$. 
		Then, $\tm = \weakctxp{\tmtwo} = \weakctxtwop{\tmtwo} \esub{\var}{\tmthree} \rootRew{a} \weakctxtwop{\tmtwo'}\esub{\var}{\tmthree} = \weakctxp{\tmtwo'} = \tm'$ with $\tmtwo \rootRew{a} \tmtwo'$ and $a \in \{\msym, \esym\}$.
		The derivation $\tderiv$ is necessarily
		\begin{equation*}
		\tderiv = 
		\begin{prooftree}
		\hypo{}
		\ellipsis{$\tderivtwo$}{\tyjp{}{\weakctxtwop{\tmtwo}}{\typctxtwo ; \var \hastype \mtypetwo}{\mtype}}
		\hypo{}
		\ellipsis{$\tderivthree$}{\tyjp{}{\tmthree}{\typctxthree}{\mtypetwo}}
		\infer2[\footnotesize$\Es$]{\tyjp{}{\weakctxtwop{\tmtwo} \esub{\var}{\tmthree}}{\typctxtwo \mplus \typctxthree}{\mtype}}
		\end{prooftree}
		\end{equation*}
		where $\typctx = \typctxtwo \uplus \typctxthree$, $\sizem{\tderiv} = \sizem{\tderivtwo} + \sizem{\tderivthree}$ and $\size{\tderiv} = 1 + \size{\tderivtwo} + \size{\tderivthree}$.
		By \ih, there is a derivation $\namedtyjp{\tderivtwo'}{}{\weakctxtwop{\tmtwo'}}{\typctxtwo, \var \hastype \mtypetwo}{\mtype}$ with: 
		\begin{enumerate}
			\item $\sizem{\tderivtwo'} = \sizem{\tderivtwo} - 2$ and $\size{\tderivtwo'} = \size{\tderivtwo} - 1$ if $\tmtwo \rtom \tmtwo'$; 
			\item $\sizem{\tderivtwo'} = \sizem{\tderivtwo}$ and $\size{\tderivtwo'} < \size{\tderivtwo}$ if $\tmtwo \rtoe \tmtwo'$.
		\end{enumerate}
		We can then build the derivation 
		\begin{equation*}
		\tderiv' = 
		\begin{prooftree}
		\hypo{}
		\ellipsis{$\tderivtwo'$}{\tyjp{}{\weakctxtwop{\tmtwo'}}{\typctxtwo ; \var \hastype \mtypetwo}{\mtype}}
		\hypo{}
		\ellipsis{$\tderivthree$}{\tyjp{}{\tmthree}{\typctxthree}{\mtypetwo}}
		\infer2[\footnotesize$\Es$]{\typctxtwo \uplus \typctxthree \vdash \weakctxtwop{\tmtwo'} \esub{\var}{\tmthree} \hastype \mtype}
		\end{prooftree}
		\end{equation*}
		noting that 
		\begin{enumerate}
			\item If $\tmtwo \rtom \tmtwo'$ then $\sizem{\tderiv'} = \sizem{\tderivtwo'} + \sizem{\tderivthree} = (\sizem{\tderivtwo} - 2) + \sizem{\tderivthree} = \sizem{\tderiv} - 2$ and $\size{\tderiv'} = \size{\tderivtwo'} + \size{\tderivthree} = (\size{\tderivtwo} - 1) + \size{\tderivthree} = \size{\tderiv} - 1$; 
			\item If $\tmtwo \rtoe \tmtwo'$ then $\sizem{\tderiv'} = \sizem{\tderivtwo'} + \sizem{\tderivthree} = \sizem{\tderivtwo} + \sizem{\tderivthree} = \sizem{\tderiv}$ and $\size{\tderiv'} = \size{\tderivtwo'} + \size{\tderivthree} < \size{\tderivtwo} + \size{\tderivthree} = \size{\tderiv}$.
		\end{enumerate}
		
		\item \emph{Explicit substitution right}, \ie $\weakctx = \tmthree \esub{\var}{\weakctxtwo}$. 
		Analogous to the previous case.
		\qedhere
	\end{itemize}
\end{proof}

\begin{theorem}[Open correctness]
	\label{thmappendix:open-correctness}
	\NoteState{thm:open-correctness}
	Let $\derive{\tderiv}{\tm}$ be a derivation.
	Then there is an $\osym$-normalizing evaluation $\deriv \colon \tm \tovsubo^* \tmtwo$ with $2\sizem{\deriv} + \sizeo{\tmtwo} \leq \sizem{\tderiv}$.
And if $\tderiv$ is tight, then $2\sizem{\deriv} + \sizeo{\fire} = \sizem{\tderiv}$.
\end{theorem}

\begin{proof}
	Given the derivation (\resp tight derivation) $\concl{\tderiv}{\typctx}{\tm}{\mtype}$, we proceed by induction on the general size $\size{\tderiv}$ of $\tderiv$.
	
	If $\tm$ is normal for $\tovsubo$, then $\tm = \fire$ is a fireball.
	Let $\deriv$ be the empty evaluation (so $\sizem{\deriv} = 0$), thus $\sizem{\tderiv} \geq \sizeo{\fire} = \sizeo{\fire} + 2\sizem{\deriv}$ 
	(\resp $\sizem{\tderiv} = \sizeo{\fire} = \sizeo{\fire} + 2\sizem{\deriv}$) by \reflemma{size-fireballs}.
	
	Otherwise, $\tm$ is not normal for $\tomo$ and so $\tm \tovsubo \tmtwo$.
	According to open subject reduction (\Cref{prop:weak-subject-reduction}), there is a derivation $\concl{\tderivtwo}{\typctx}{\tmtwo}{\mtype}$ such that $\size{\tderivtwo} < \size{\tderiv}$  and 
	\begin{itemize}
		\item $\sizem{\tderivtwo} \leq \sizem{\tderiv} - 2$ (\resp $\sizem{\tderivtwo} = \sizem{\tderiv} - 2$) if $\tm \tomo \tmtwo$,
		\item $\sizem{\tderivtwo} = \sizem{\tderiv}$ if $\tm \toeo \tmtwo$.
	\end{itemize}
	By \ih, there exists a fireball $\fire$ and a reduction sequence $\deriv' \colon \tmtwo \tovsubo^* \fire$ with 
	$2\sizem{\deriv'} + \sizeo{\fire} \leq \sizem{\tderivtwo}$ (\resp $2\sizem{\deriv'} + \sizeo{\fire} = \sizem{\tderivtwo}$).
	Let $\deriv$ be the $\osym$-evaluation obtained by concatenating the first step $\tm \tovsubo \tmtwo$ and $\deriv'$.
	There are two cases:
	\begin{itemize}
		\item \emph{Multiplicative:} if $\tm \tomo \tmtwo$ then $\sizem{\tderiv} \geq \sizem{\tderivtwo} + 2 \geq \sizeo{\fire} 
		+ 2\sizem{\deriv'} + 2 = \sizeo{\fire} + 2\sizem{\deriv}$ (\resp $\sizem{\tderiv} = \sizem{\tderivtwo} + 2 = \sizeo{\fire} + 
		2\sizem{\deriv'} + 2 = \sizeo{\fire} + 2\size{\deriv}$), since $\sizem{\deriv} = \sizem{\deriv'} + 1$.
		\item \emph{Exponential:} if $\tm \toeo \tmtwo$ then $\sizem{\tderiv} = \sizem{\tderivtwo} \geq \sizeo{\fire} + 2\sizem{\deriv'} = \sizeo{\fire} + 2\sizem{\deriv}$ (\resp $\sizem{\tderiv} = \sizem{\tderivtwo} = \sizeo{\fire} + 2\sizem{\deriv'} = \sizeo{\fire} + 2\sizem{\deriv}$),  since $\sizem{\deriv} = \sizem{\deriv'}$.
		\qedhere
	\end{itemize} 
\end{proof}

\paragraph{Completeness}

\begin{proposition}[Tight typability of open normal forms]
	\label{propappendix:precise-open-typability-nf}
	\NoteState{prop:precise-open-typability-nf}
	\begin{enumerate}
		\item \emph{Inert:}\label{pappendix:precise-open-typability-nf-inert} if $\tm$ is an inert term then, for any inert multi type $\mtype$, there is an inert
		type derivation $\concl{\tderiv}{\typctx}{\tm}{\mtype}$.
		\item \emph{Fireball:}\label{pappendix:precise-open-typability-nf-fireball} if $\tm$ is a fireball then there is a tight derivation $\concl{\tderiv}{\typctx}{\tm}{\emptytype}$.		
	\end{enumerate}
\end{proposition}

\begin{proof}
	We prove simultaneously \Cref{p:precise-open-typability-nf-fireball,p:precise-open-typability-nf-inert} by 
		mutual induction on the definition of fireball and inert term.
		Note that \Cref{p:precise-open-typability-nf-inert} is stronger than \Cref{p:precise-open-typability-nf-fireball}.
		Cases for inert terms:
		\begin{itemize}
			\item \emph{Variable}, \ie $\tm = \var$, which is an inert term. 
			Let $\imtype$ be an inert multi type: hence, $\imtype = \mset{\ltype_1, \dots, \ltype_n}$ for some $n \in \nat$ and some $\ltype_1, \dots, \ltype_n$ inert linear types.
			We can then build  the derivation 
			\begin{equation*}
			\tderiv = 
			\begin{prooftree}[separation = 1em]
			\infer0[\footnotesize$\ruleAx$]{\var \hastype \mset{\ltype_1} \vdash \var \hastype \ltype_1}
			\hypo{\dots}
			\infer0[\footnotesize$\ruleAx$]{\var \hastype \mset{\ltype_n} \vdash \var \hastype \ltype_n}
			\infer3[\footnotesize$\ruleManyVar$]{\var \hastype \mset{\ltype_1, \dots, \ltype_n} \vdash \var \hastype \mset{\ltype_1, \dots, \ltype_n}}
			\end{prooftree}
			\end{equation*}
			where $\Gamma = \var \hastype \mset{\ltype_1, \dots, \ltype_n}$ is an inert type context.

			\item \emph{Inert application}, \ie $\tm = \itm \fire$ for some inert term $\itm$ and fireball $\fire$.
			Let $\imtype$ be a multi type.
			By \ih for fireballs, there is a derivation $\concl{\tderivthree}{\typctxthree}{\fire}{\emptytype}$  for some inert type context $\typctxthree$.
			By \ih for inert terms, since $\mset{\larrow{\emptytype}{\imtype}}$ is an inert multi type, there is a derivation $\concl{\tderivtwo}{\typctxtwo}{\itm}{\mset{\larrow{\emptytype}{\imtype}}}$ for some inert type context $\typctxtwo$.
			We can then build the derivation 
			\begin{equation*}
			\tderiv  =
			\begin{prooftree}
			\hypo{}
			\ellipsis{$\tderivtwo$}{\typctxtwo \vdash \itm \hastype \mset{\larrow{\emptytype}{\imtype}}}
			\hypo{}
			\ellipsis{$\tderivthree$}{\typctxthree \vdash \fire \hastype \emptytype}
			\infer2[\footnotesize$\ruleApp$]{\typctxtwo \mplus \typctxthree \vdash \itm \fire \hastype \imtype}				
			\end{prooftree}
			\end{equation*}
			where $\typctx = \typctxtwo \mplus \typctxthree$ is an inert type context, by \Cref{rmk:merge-split-inert}.
			
			\item \emph{Explicit substitution on inert}, \ie $\tm = \itm \esub{\var}{\itmtwo}$ for some inert terms $\itm$ and $\itmtwo$.
			Let $\imtype$ be an inert multi type.
			By \ih for inert terms applied to $\itm$, there is a derivation $\concl{\tderivtwo}{\typctxtwo, \var \hastype \mtypetwo}{\itm}{\imtype}$ for some inert multi type $\mtypetwo$ and inert type context $\typctxtwo$.
			By \ih for inert terms applied to $\itmtwo$, there is a derivation $\concl{\tderivthree}{\typctxthree}{\itmtwo}{\mtypetwo}$  for some inert type context $\typctxthree$.
			We can then build the derivation 
			\begin{equation*}
				\tderiv =
				\begin{prooftree}
				\hypo{}
				\ellipsis{$\tderivtwo$}{\typctxtwo, \var \hastype \mtypetwo \vdash \itm \hastype \imtype}
				\hypo{}
				\ellipsis{$\tderivthree$}{\typctxthree \vdash \itmtwo \hastype \mtypetwo}
				\infer2[\footnotesize$\ruleES$]{\typctxtwo \mplus \typctxthree \vdash \itm \esub{\var} {\itmtwo} \hastype \imtype}				
				\end{prooftree}
			\end{equation*}
			where $\typctx = \typctxtwo \mplus \typctxthree$ is an inert type context, by \Cref{rmk:merge-split-inert}.
		\end{itemize}
	
		Cases for fireballs that may not be inert terms:
		\begin{itemize}
			\item \emph{Abstraction}, \ie $\tm  = \la{\var}{\tmtwo}$.
			We can then build the derivation
			\begin{equation*}
				\tderiv = 
				\begin{prooftree}
				\infer0[\footnotesize$\ruleManyVal$]{\vdash \la{\var}{\tmtwo} \hastype \emptytype}
			\end{prooftree}
			\end{equation*}
			where the type context $\typctx$ is empty and hence inert, thus $\tderiv$ is tight.
			
			\item \emph{Explicit substitution on fireball}, \ie $\tm = \fire \esub{\var}{\itm}$ for some fireball $\fire$ and inert term $\itm$.
			By \ih for fireballs applied to $\fire$, there is a derivation $\concl{\tderivtwo}{\typctxtwo, \var \hastype \imtypetwo}{\fire}{\emptytype}$ for some inert multi type $\imtypetwo$ and inert type context $\typctxtwo$.
			By \ih for inert terms applied to $\itm$, there is a derivation $\concl{\tderivthree}{\typctxthree}{\itm}{\imtypetwo}$  for some inert type context $\typctxthree$.
			We can then build the derivation 
			\begin{equation*}
			\tderiv = 
			\begin{prooftree}
			\hypo{}
			\ellipsis{$\tderivtwo$}{\typctxtwo, \var \hastype \imtypetwo \vdash \fire \hastype \emptytype}
			\hypo{}
			\ellipsis{$\tderivthree$}{\typctxthree \vdash \itm \hastype \imtypetwo}
			\infer2[\footnotesize$\ruleES$]{\typctxtwo \mplus \typctxthree \vdash \fire \esub{\var} {\itm} \hastype \emptytype}				
			\end{prooftree}
			\end{equation*}
			where $\typctx = \typctxtwo \mplus \typctxthree$ is an inert type context, by \Cref{rmk:merge-split-inert}.
			Therefore, $\tderiv$ is tight.
			\qedhere
		\end{itemize}
\end{proof}

\begin{proposition}[Open quantitative subject expansion]
	\label{propappendix:weak-subject-expansion}
	\NoteState{prop:weak-subject-expansion}
	Let $\namedtyjp{\tderiv'}{}{\tm'}{\typctx}{\mtype}$ be a derivation.
	\begin{enumerate}
		\item \emph{Multiplicative step:} if $\tm \towm \tm'$ then there is a derivation $\namedtyjp{\tderiv}{}{\tm}{\typctx}{\mtype}$ with
		$\sizem{\tderiv'} = \sizem{\tderiv} - 2$ and $\size{\tderiv'} = \size{\tderiv} - 1$; 
		\item \emph{Exponential step:} if $\tm \towe \tm'$ then there is a derivation $\namedtyjp{\tderiv}{}{\tm}{\typctx}{\mtype}$ such that
		$\sizem{\tderiv'} = \sizem{\tderiv}$ and $\size{\tderiv'} < \size{\tderiv}$.
	\end{enumerate}
\end{proposition}

\begin{proof}
	By induction on 
	the evaluation context $\weakctx$ in the step $\tm = \weakctxp{\tmtwo} \tow \weakctxp{\tmtwo'} = \tm'$ with $\tmtwo \rtom \tmtwo'$ or $\tmtwo \rtoe \tmtwo'$. 
	Cases for $\weakctx$:
	\begin{itemize}
		\item \emph{Hole}, \ie $\weakctx = \ctxhole$.
		Then there are two sub-cases:
		\begin{enumerate}
			\item \emph{Multiplicative}, \ie $\tm = \subctxp{\la\var\tmtwo}\tmthree \rtom  \subctxp{\tmtwo \esub{\var}{\tmthree}} = \tm'$.
			Then $\tderiv'$ has necessarily the form:
			\begin{equation*}
			\tderiv' = 
			\begin{prooftree}
			\hypo{}
			\ellipsis{$\tderivtwo$}{\tyjp{}{\tmtwo}{\typctx' ; \var \hastype \mtypetwo}{\mtype}}
			\hypo{}
			\ellipsis{$\tderivthree$}{\tyjp{}{\tmthree}{\typctxtwo}{\mtypetwo}}
			\infer2[\footnotesize$\ruleES$]{\typctx'\uplus\typctxtwo \vdash\tmtwo\esub\var\tmthree \hastype \mtype}
			\hypo{}
			\ellipsis{$\tderiv_1$}{\quad}
			\infer2[\footnotesize$\ruleES$]{}
			\ellipsis{}{\quad}
			\hypo{}
			\ellipsis{$\tderiv_n$}{\quad}
			\infer2[\footnotesize$\Es$]{\typctx\uplus\typctxtwo \vdash \subctxp{\tmtwo\esub\var\tmthree} \hastype \mtype}
			\end{prooftree}
			\end{equation*}
			with $\sizem{\tderiv'} = \sizem{\tderivtwo} + \sizem{\tderivthree} + \sum_{i=1}^{n} \sizem{\tderiv_{i}}$ and $\size{\tderiv' } = 1 + n + \size{\tderivtwo} + \size{\tderivthree} + \sum_{i=1}^{n} \size{\tderiv_{i}}$.
			
			We can build $\tderiv$ as follows:
			\begin{equation*}
			\tderiv = 
			\begin{prooftree}[separation = 1em]
			\hypo{}
			\ellipsis{$\tderivtwo$}{\typctx', \var \hastype \mtypetwo \vdash \tmtwo \hastype \mtype}
			\infer1[\footnotesize$\ruleFun$]{\typctx' \vdash \la\var\tmtwo \hastype \larrow{\mtypetwo}{\mtype}}
			\infer1[\footnotesize$\ruleManyVal$]{\typctx' \vdash \la\var\tmtwo \hastype \mset{\larrow{\mtypetwo}{\mtype}}}
			\hypo{}
			\ellipsis{$\tderiv_1$}{\quad}
			\infer2[\footnotesize$\ruleES$]{}
			\ellipsis{}{\quad}
			\hypo{}
			\ellipsis{$\tderiv_n$}{\quad}
			\infer2[\footnotesize$\ruleES$]{\typctx \vdash \subctxp{\la\var\tmtwo} \hastype \mset{\larrow{\mtypetwo}{\mtype}}}
			\hypo{}
			\ellipsis{$\tderivthree$}{\typctxtwo\vdash\tmthree \hastype \mtypetwo}
			\infer2[\footnotesize$\ruleApp$]{\typctx \uplus \typctxtwo \vdash \subctxp{\la\var\tmtwo}\tmthree \hastype \mtype}
			\end{prooftree}
			\end{equation*}
			where $\sizem{\tderiv} = 2 + \sizem{\tderivtwo} + \sizem{\tderivthree} + \sum_{i=1}^{n} \sizem{\tderiv_{i}} = \sizem{\tderiv'} + 2$ and $\size{\tderiv} = 2 + n + \size{\tderivtwo} + \size{\tderivthree} + \sum_{i=1}^{n} \size{\tderiv_{i}} = \size{\tderiv'} + 1$.
			
			\item \emph{Exponential}, \ie $\tm = \tmtwo\esub\var{\subctxp{\tval}} \rtoe \subctxp{\tmtwo \isub{\var}{\tval}} = \tmp$.
			Then the derivation $\tderiv$ has necessarily the form:
			\begin{equation*}
			\tderiv' = 
			\begin{prooftree}
			\hypo{}
			\ellipsis{$\tderiv''$}{\typctxtwo \mplus\typctxthree' \vdash \tmtwo\isub\var \tval \hastype \mtype}
			\hypo{}
			\ellipsis{$\tderiv_1$}{\quad}
			\infer2[\footnotesize{$\Es$}]{}
			\ellipsis{}{}
			\hypo{}
			\ellipsis{$\tderiv_n$}{\quad}
			\infer2[\footnotesize$\Es$]{\typctxtwo\uplus\typctxthree \vdash \subctxp{\tmtwo\isub\var \tval} \hastype \mtype}
			\end{prooftree}
			\end{equation*}
			where $\typctx = \typctxtwo \uplus \typctxthree$, $\sizem{\tderiv'} = \sizem{\tderiv''} + \sum_{i=1}^{n} \sizem{\tderiv_{i}}$ and $\size{\tderiv'} =  n + \size{\tderiv''} + \sum_{i=1}^{n} \size{\tderiv_{i}}$.
			By the anti-substitution lemma (\reflemma{anti-substitution}), there are  derivations $\namedtyjp{\tderivtwo}{}{\tmtwo}{\typctxtwo, \var\hastype\mtypetwo}{\mtype}$ and $\namedtyjp{\tderivthree}{}{\tval}{\typctxthree'}{\mtypetwo}$
			such that $\sizem{\tderiv''} = \sizem{\tderivtwo} + \sizem{\tderivthree}$ and $\size{\tderiv''} \leq \size{\tderivtwo} + \size{\tderivthree}$.
			We can then build the following derivation $\tderiv$:
			\begin{equation*}
			\tderiv = 
			\begin{prooftree}
			\hypo{}
			\ellipsis{$\tderivtwo$}{\tyjp{}{\tmtwo}{\typctxtwo, \var \hastype \mtypetwo}{\mtype}}
			\hypo{}
			\ellipsis{$\tderivthree$}{\tyjp{}{\tval}{\typctxthree'}{\mtypetwo}}
			\hypo{}
			\ellipsis{$\tderiv_1$}{\quad}
			\infer2[\footnotesize$\Es$]{}
			\ellipsis{}{\quad}
			\hypo{}
			\ellipsis{$\tderiv_n$}{\quad}
			\infer2[\footnotesize$\Es$]{\typctxthree \vdash \subctxp{\tval} \hastype \mtypetwo}
			\infer2[\footnotesize$\Es$]{\typctxtwo \mplus \typctxthree \vdash\tmtwo\esub\var{\subctxp{\tval}}\hastype \mtype}
			\end{prooftree}
			\end{equation*}
			where $\sizem{\tderiv} = \sizem{\tderivtwo} + \sizem{\tderivthree} + \sum_{i=1}^{n} \sizem{\tderiv_{i}} = \sizem{\tderiv''} + \sum_{i=1}^{n} \sizem{\tderiv_{i}} = \sizem{\tderiv'}$ and $\size{\tderiv} = 1 + n + \size{\tderivtwo} + \size{\tderivthree} + \sum_{i=1}^{n} \size{\tderiv_{i}} > n + \size{\tderiv''} +  \sum_{i=1}^{n} \size{\tderiv_{i}} = \size{\tderiv'}$.
		\end{enumerate}
		
		\item \emph{Application left}, \ie $\weakctx = \weakctxtwo \tmthree$.
		Then, $\tm = \weakctxp{\tmtwo} = \weakctxtwop{\tmtwo} \tmthree \rootRew{a} \weakctxtwop{\tmtwo'} \tmthree = \weakctxp{\tmtwo'} = \tm'$ with $\tmtwo \rootRew{a} \tmtwo'$ and $a \in \{\msym, \esym\}$.
		The derivation $\tderiv'$ is necessarily
		\begin{equation*}
		\tderiv' = 
		\begin{prooftree}
		\hypo{}
		\ellipsis{$\tderivtwo'$}{\tyjp{}{\weakctxtwop{\tmtwo'}}{\typctxtwo}{\mult{\ty{\mtypetwo}{\mtype}}}}
		\hypo{}
		\ellipsis{$\tderivthree$}{\tyjp{}{\tmthree}{\typctxthree}{\mtypetwo}}
		\infer2[\footnotesize$\ruleAp$]{\tyjp{}{\weakctxtwop{\tmtwo'} \tmthree}{\typctxtwo \mplus \typctxthree}{\mtype}}
		\end{prooftree}
		\end{equation*}
		where $\typctx = \typctxtwo \uplus \typctxthree$, $\sizem{\tderiv'} = 1 +  \sizem{\tderivtwo'} + \sizem{\tderivthree}$ and $\size{\tderiv'} = 1 +  \size{\tderivtwo'} + \size{\tderivthree}$.
		By \ih, there is a derivation $\namedtyjp{\tderivtwo}{}{\weakctxtwop{\tmtwo}}{\typctxtwo}{\mult{\ty{\mtypetwo}{\mtype}}}$ with:
		\begin{enumerate}
			\item $\sizem{\tderivtwo} = \sizem{\tderivtwo'} + 2$ and $\size{\tderivtwo} = \size{\tderivtwo'} + 1$ if $\tmtwo \rtom \tmtwo'$; 
			\item $\sizem{\tderivtwo} = \sizem{\tderivtwo'}$ and $\size{\tderivtwo} < \size{\tderivtwo'}$ if $\tmtwo \rtoe \tmtwo'$.
		\end{enumerate}
		We can then build the derivation 
		\begin{equation*}
		\tderiv = 
		\begin{prooftree}
		\hypo{}
		\ellipsis{$\tderivtwo$}{\tyjp{}{\weakctxtwop{\tmtwo}}{\typctxtwo}{\mult{\ty{\mtypetwo}{\mtype}}}}
		\hypo{}
		\ellipsis{$\tderivthree$}{\tyjp{}{\tmthree}{\typctxthree}{\mtypetwo}}
		\infer2[\footnotesize$\ruleAp$]{\tyjp{}{\weakctxtwop{\tmtwo} \tmthree}{\typctxtwo \mplus \typctxthree}{\mtype}}
		\end{prooftree}
		\end{equation*}
		noting that
		\begin{enumerate}
			\item If $\tmtwo \rtom \tmtwo'$ then $\sizem{\tderiv} = 1 + \sizem{\tderivtwo} + \sizem{\tderivthree} = 1 + (\sizem{\tderivtwo'} + 2) + \sizem{\tderivthree} = \sizem{\tderiv'} + 2$ and $\size{\tderiv} = 1 + \size{\tderivtwo} + \size{\tderivthree} = 1 + 1 +\size{\tderivtwo'} + \size{\tderivthree} = \size{\tderiv'} + 1$; 
			\item If $\tmtwo \rtoe \tmtwo'$ then $\sizem{\tderiv} = 1 + \sizem{\tderivtwo} + \sizem{\tderivthree} = 1 + \sizem{\tderivtwo'} + \sizem{\tderivthree} = \sizem{\tderiv'}$ and $\size{\tderiv} = 1 + \size{\tderivtwo} + \size{\tderivthree} > 1 + \size{\tderivtwo'} + \size{\tderivthree} = \size{\tderiv'}$.
		\end{enumerate}
		
		\item \emph{Application right}, \ie $\weakctx = \tmthree \weakctxtwo$.
		Analogous to the previous case.
		
		\item \emph{Explicit substitution left}, \ie $\weakctx = \weakctxtwo\esub{\var}{\tmthree}$. 
		Then, $\tm = \weakctxp{\tmtwo} = \weakctxtwop{\tmtwo} \esub{\var}{\tmthree} \rootRew{a} \weakctxtwop{\tmtwo'}\esub{\var}{\tmthree} = \weakctxp{\tmtwo'} = \tm'$ with $\tmtwo \rootRew{a} \tmtwo'$ and $a \in \{\msym, \esym\}$.
		The derivation $\tderiv'$ is necessarily
		\begin{equation*}
		\tderiv' = 
		\begin{prooftree}
		\hypo{}
		\ellipsis{$\tderivtwo'$}{\tyjp{}{\weakctxtwop{\tmtwo'}}{\typctxtwo ; \var \hastype \mtypetwo}{\mtype}}
		\hypo{}
		\ellipsis{$\tderivthree$}{\tyjp{}{\tmthree}{\typctxthree}{\mtypetwo}}
		\infer2[\footnotesize$\Es$]{\typctxtwo \mplus \typctxthree \vdash \weakctxtwop{\tmtwo'} \esub{\var}{\tmthree} \hastype \mtype}
		\end{prooftree}
		\end{equation*}
		where $\typctx = \typctxtwo \mplus \typctxthree$, $\sizem{\tderiv} = \sizem{\tderivtwo} + \sizem{\tderivthree}$ and $\size{\tderiv} = 1 + \size{\tderivtwo} + \size{\tderivthree}$.
		By \ih, there is a derivation $\namedtyjp{\tderivtwo}{}{\weakctxtwop{\tmtwo}}{\typctxtwo, \var \hastype \mtypetwo}{\mtype}$ with: 
		\begin{enumerate}
			\item $\sizem{\tderivtwo} = \sizem{\tderivtwo'} + 2$ and $\size{\tderivtwo} =  \size{\tderivtwo'} + 1$ if $\tmtwo \rtom \tmtwo'$; 
			\item $\sizem{\tderivtwo} = \sizem{\tderivtwo'}$ and $\size{\tderivtwo} > \size{\tderivtwo'}$ if $\tmtwo \rtoe \tmtwo'$.
		\end{enumerate}
		We can then build the derivation 
		\begin{equation*}
		\tderiv = 
		\begin{prooftree}
		\hypo{}
		\ellipsis{$\tderivtwo$}{\tyjp{}{\weakctxtwop{\tmtwo}}{\typctxtwo ; \var \hastype \mtypetwo}{\mtype}}
		\hypo{}
		\ellipsis{$\tderivthree$}{\tyjp{}{\tmthree}{\typctxthree}{\mtypetwo}}
		\infer2[\footnotesize$\Es$]{\tyjp{}{\weakctxtwop{\tmtwo} \esub{\var}{\tmthree}}{\typctxtwo \mplus \typctxthree}{\mtype}}
		\end{prooftree}
		\end{equation*}
		noting that 
		\begin{enumerate}
			\item If $\tmtwo \rtom \tmtwo'$ then $\sizem{\tderiv} = \sizem{\tderivtwo} + \sizem{\tderivthree} = (\sizem{\tderivtwo'} + 2) + \sizem{\tderivthree} = \sizem{\tderiv'} + 2$ and $\size{\tderiv} = 1 + \size{\tderivtwo} + \size{\tderivthree} = 1 + 1 +  \size{\tderivtwo'} + \size{\tderivthree} = \size{\tderiv'}$; 
			\item If $\tmtwo \rtoe \tmtwo'$ then $\sizem{\tderiv} = \sizem{\tderivtwo} + \sizem{\tderivthree} = \sizem{\tderivtwo'} + \sizem{\tderthree} = \sizem{\tderiv'}$ and $\size{\tderiv} = \size{\tderivtwo} + \size{\tderivthree} > \size{\tderivtwo'} + \size{\tderivthree} = \size{\tderiv'}$.
		\end{enumerate}
		
		\item \emph{Explicit substitution right}, \ie $\weakctx = \tmthree \esub{\var}{\weakctxtwo}$. 
		Analogous to the previous case.
		\qedhere
	\end{itemize}
\end{proof}

\begin{theorem}[Open completeness]
	\label{thmappendix:open-completeness}
	\NoteState{thm:open-completeness}
	Let $\deriv \colon \tm \tovsubo^* \tmtwo$ be an $\osym$-normalizing evaluation.
	Then there is a tight derivation $\concl{\tderiv}{\typctx}{\tm}{\emptytype}$ such that $2\sizem{\deriv} + \sizes{\fire} = \sizem{\tderiv}$.
\end{theorem}

\begin{proof}
	By induction on the length $\size{\deriv}$ of the $\osym$-evaluation $\deriv$.
	
	If $\size{\deriv} = 0$ then $\sizem{\deriv} = 0$ and $\tm = \fire$ is $\osym$-normal and hence a fireball.
	According to tight typability of fireballs (\Cref{prop:precise-open-typability-nf}), there is a tight derivation $\concl{\tderiv}{\typctx}{\tm}{\emptytype}$.
	Therefore, $\sizem{\tderiv} = \sizeo{\fire} = \sizeo{\fire} + 2\sizem{\deriv}$ by \Cref{l:size-fireballs}.
	
	Otherwise, $\size{\deriv} > 0$ and $\deriv$ is the concatenation of a first step $\tm \tovsubo \tmtwo$ and an evaluation $\deriv' \colon \tmtwo \tovsubo^* \fire$, with $\size{\deriv} = 1 + \size{\deriv'}$.
	By \ih, there is a tight derivation $\concl{\tderivtwo}{\typctx}{\tmtwo}{\emptytype}$  such that $\sizem{\tderivtwo} = \sizeo{\fire} + 2\sizem{\deriv'}$. 
	According to open subject expansion (\Cref{prop:weak-subject-expansion}), there is a derivation 
	$\concl{\tderiv}{\typctx}{\tm}{\emptytype}$ with 
	\begin{itemize}
		\item $\sizem{\tderiv} = \sizem{\tderivtwo} + 2 = \sizeo{\fire} + 2\sizem{\deriv'} + 2 = \sizeo{\fire} + 2\sizem{\deriv}$ 
		if $\tm \tomo \tmtwo$, since $\sizem{\deriv} = \sizem{\deriv'} + 1$;
		\item $\sizem{\tderiv} = \sizem{\tderivtwo} = \sizeo{\fire} + 2\sizem{\deriv'} = \sizeo{\fire} + 2\sizem{\deriv}$ if $\tm \toeo \tmtwo$, since $\sizem{\deriv} = \sizem{\deriv'}$.
		\qedhere
	\end{itemize} 
\end{proof}

\begin{theorem}[Open adequacy]
	\label{thmappendix:open-adequacy}
	\NoteState{thm:open-adequacy}
	Let $\tm$ be a term, with $\vec{\var}$ suitable for $\tm$.
	Then, $\tm$ is $\osym$-normalizing if and only if $\sem{\tm}_{\vec{\var}} \neq \emptyset$. 
\end{theorem} 

\begin{proof}
	According to open completeness (\Cref{thm:open-completeness}), if $\tm$ is $\osym$-normalizing, then  there is a derivation $\concl{\tderiv}{\typctx}{\tm}{\mtype}$ for some type context $\typctx$ and multi type $\mtype$, thus $\sem{\tm}_{\vec{\var}} \neq \emptyset$.
	
	Conversely, if $\sem{\tm}_{\vec{\var}} \neq \emptyset$, then  there is a derivation $\concl{\tderiv}{\typctx}{\tm}{\mtype}$ for some type context $\typctx$ and multi type $\mtype$, thus by open correctness (\Cref{thm:open-correctness}), $\tm$ is $\osym$-normalizing.
\end{proof}

\section{Proofs of \Cref{sect:shrinking}}

\Cref{rmk:merge-split-coshrinking} and \Cref{l:spread-shrinking} below are  used to prove both correctness (\Cref{thm:correctness}) and completeness (\Cref{thm:completeness}) 

\begin{remark}[Merging and splitting \leftsh shrinkingess]
	\label{rmk:merge-split-coshrinking}
	Let $\mtype, \mtypetwo, \mtypethree$ be multi types with $\mtype = \mtypetwo \mplus \mtypethree$; then $\mtype$ is \leftsh (\resp unitary \leftsh) iff $\mtypetwo$ and $\mtypethree$ are \leftsh (\resp unitary \leftsh).
	Similarly for type contexts.
\end{remark}

The property above fails for unitary \rightsh shrinkingness. 
Indeed, $\mset{\ground}$ is unitary \rightsh but $\mset{\ground} \mplus \mset{\ground} = \mset{\ground, \ground}$ is \rightsh and unitary \leftsh, but not unitary  \rightsh.

%
%
%

\begin{lemma}
	[Spreading of \leftsh shrinkingness]
	\label{lappendix:spread-shrinking}
	\NoteState{l:spread-shrinking}
	Let $\concl{\tderiv}{\typctx}{\ptm}{\mtype}$ be a derivation and $\ptm$ a \pointed term. 
	If $\typctx$ is a \leftsh (\resp~unitary \leftsh) type context, then $\mtype$ is a \leftsh (\resp~unitary \leftsh) multi type.
\end{lemma}

\begin{proof}
	By induction on the definition of the \pointed term $\ptm$. 
	Cases:
	\begin{itemize}
		\item \emph{Variable}, \ie $\ptm = \var$. 
		Then necessarily, for some $n \in \nat$, 
		\begin{equation*}
		\tderiv = 
		\begin{prooftree}
		\infer0[\footnotesize$\Ax$]{\var \hastype \mset{\ltype_1} \vdash \var \hastype \ltype_1}
		\hypo{\overset{n \in \nat}{\ldots}}
		\infer0[\footnotesize{$\Ax$}]{\var \hastype \mset{\ltype_n} \vdash \var \hastype \ltype_n}
		\infer3[\footnotesize$\ruleManyVar$]{\typctx \vdash \var \hastype \mtype}
		\end{prooftree}
		\end{equation*}
		with $\mtype = \mset{\ltype_1, \dots, \ltype_n}$ and $\typctx = \var \hastype \mtype$.
		Since $\typctx$ is a \leftsh (\resp~unitary \leftsh) type context, $\mtype$ is a \leftsh (\resp~unitary \leftsh) multi type.
		
		\item \emph{Application}, \ie $\ptm = \ptmtwo \tm$ where $\ptmtwo$ is a \pointed term.
		The derivation $\tderiv$ is necessarily (with $\typctx = \typctxtwo \mplus \typctxthree$)
		\begin{equation*}
		\tderiv = 
		\begin{prooftree}
		\hypo{}
		\ellipsis{$\tderivtwo$}{\typctxtwo \vdash \ptmtwo \hastype \mset{\larrow{\mtypetwo}{\mtype}}}
		\hypo{}
		\ellipsis{$\tderivthree$}{\typctxthree \vdash \tm \hastype \mtypetwo}
		\infer2[\footnotesize$\ruleAp$]{\typctxtwo \uplus \typctxthree \vdash \ptmtwo \tm \hastype \mtype}
		\end{prooftree}
		\end{equation*}
		where $\typctxtwo$ and $\typctxthree$ are \leftsh (\resp~unitary \leftsh) type contexts by 
		\refrmk{merge-split-coshrinking}. 
		By \ih, $\mset{\larrow{\mtypetwo}{\mtype}}$ is a \leftsh (\resp~unitary \leftsh) multi type, and hence $\mtype$ is a \leftsh (\resp~unitary \leftsh) multi type.
		
		\item \emph{Explicit substitution}, \ie $\ptm = \ptmtwo \esub{\var}{\ptmthree}$ where $\ptmtwo$ and $\ptmthree$ are 
		\pointed terms.
		The derivation $\tderiv$ is necessarily (with $\typctx = \typctxtwo \mplus \typctxthree$)
		\begin{equation*}
		\tderiv = 
		\begin{prooftree}
		\hypo{}
		\ellipsis{$\tderivtwo$}{\typctxtwo, \var \hastype \mtypetwo \vdash \ptmtwo \hastype \mtype}
		\hypo{}
		\ellipsis{$\tderivthree$}{\typctxthree \vdash \ptmthree \hastype \mtypetwo}
		\infer2[\footnotesize$\Es$]{\typctxtwo \uplus \typctxthree \vdash \ptm\esub{\var}{\ptmthree} \hastype \mtype}
		\end{prooftree}
		\end{equation*}
		where $\typctxtwo$ and $\typctxthree$ are \leftsh (\resp~unitary \leftsh) contexts by 
		\refrmk{merge-split-coshrinking}. 
		By \ih applied to $\tderivthree$ (as $\ptmthree$ is a \pointed term), $\mtypetwo$ is a \leftsh (\resp~unitary \leftsh) multi type, and hence $\typctxtwo, \var \hastype \mtypetwo$ is a \leftsh (\resp~unitary \leftsh) type context.
		By \ih applied to $\tderivtwo$ (as $\ptmtwo$ is a \pointed term), $\mtype$ is a \leftsh (\resp~unitary \leftsh) multi type.
		\qedhere
	\end{itemize}
\end{proof}

\subsection{Correctness}

We aim to prove correctness (\Cref{thm:correctness}), refined with quantitative information: 
if a term is shrinking typable then it terminates, and the type derivation provides bounds for both the number of 
multiplicative steps to normal form and the size of the normal form.

The are two crucial ingredients:
\begin{enumerate}
	\item shrinking type derivations bound the size of normal forms (\Cref{l:size-strong-fireballs}), and unitary 
	shrinking type derivations give the exact size of normal forms;
	
	\item quantitative subject reduction (\Cref{prop:shrinking-subject-reduction}) says that not only types are preserved 
	after a $\tovsubs$ step but also that if the type derivation is shrinking then the size of the derivation strictly 
	decreases.
\end{enumerate}

In both proofs \Cref{l:spread-shrinking} is used.

\begin{lemma}[Size of \full fireballs]
	\label{lappendix:size-strong-fireballs}
	\NoteState{l:size-strong-fireballs}
	Let $\typctx$ be a \leftsh (\resp~unitary \leftsh) type context, $\sfire$ be a \full fireball and $\namedtyjp{\tderiv}{}{\sfire}{\typctx}{\mtype}$ be a derivation such that if $\sfire$ is an \valES then $\mtype$ is \rightsh (\resp~unitary \rightsh).
	Then $\sizefu{\sfire} \leq \sizem{\tderiv}$ (\resp~$\sizefu{\sfire} = \sizem{\tderiv}$).
\end{lemma}

\begin{proof}
	By structural induction on $\sfire$.
	Cases:
	\begin{itemize}
		\item \emph{Variable}, \ie $\sfire = \var$. Then necessarily, for some $n \in \nat$,
		\begin{equation*}
		\tderiv = 
		\begin{prooftree}
		\infer0[\footnotesize$\ruleAx$]{\tyjp{}{\var}{\var \hastype \mset{\ltype_1}}{\ltype_1}}
		\hypo{\overset{n \in \nat}{\ldots}}
		\infer0[\footnotesize$\ruleAx$]{\tyjp{}{\var}{\var \hastype \mset{\ltype_n}}{\ltype_n}}
		\infer3[\footnotesize$\ruleManyVar$]{\tyjp{}{\var}{\typctx}{\mtype}}
		\end{prooftree}
		\end{equation*}
		where $\mtype = \mset{\ltype_1, \dots, \ltype_n}$ and $\typctx = \var \hastype \mtype$. 
		Therefore, $\sizefu{\sfire} = 0 = \sizem{\tderiv}$.
		
		\item \emph{Application}, \ie $\sfire = \sitm \sfiretwo$. Then necessarily
		\begin{equation*}
		\tderiv = 
		\begin{prooftree}
		\hypo{}
		\ellipsis{$\tderivtwo$}{\typctxtwo \vdash \sitm \hastype \mult{\ty{\mtypetwo}{\mtype}}}
		\hypo{}
		\ellipsis{$\tderivthree$}{\typctxthree \vdash \sfiretwo \hastype\mtypetwo}
		\infer2[\footnotesize$\ruleApp$]{\tyjp{}{\sitm \sfiretwo}{\typctxtwo \mplus \typctxthree}{\mtype}}
		\end{prooftree}
		\end{equation*}
		where $\typctx = \typctxtwo \mplus \typctxthree$.
		Since $\typctx$ is \leftsh (\resp~unitary \leftsh), then so are $\typctxtwo$ and $\typctxthree$ by 
		\Cref{rmk:merge-split-coshrinking}.
		As $\sitm$ is a \pointed term (\Cref{rmk:rigid}), we can apply the spreading of left shrinkingness (\Cref{l:spread-shrinking}) to $\tderivtwo$ and so 
		$\mset{\larrow{\mtypetwo}{\mtype}}$ is \leftsh (\resp~unitary \leftsh), which entails that $\mtypetwo$ is \rightsh (\resp~unitary \rightsh).
		As $\sitm$ is not a \valES, we can apply \ih to both premises and hence $\sizem{\tderivtwo} \geq \sizefu{\sitm}$ and 
		$\sizem{\tderivthree} \geq \sizefu{\sfiretwo}$ (\resp~$\sizem{\tderivtwo} = \sizefu{\sitm}$ and $\sizem{\tderivthree} = \sizefu{\sfiretwo}$). 
		So, $\sizefu{\sfire} = \sizefu{\sitm} + \sizefu{\sfiretwo} + 1 \leq \sizem{\tderivtwo} + \sizem{\tderivthree} + 1 = \sizem{\tderiv}$
		(\resp~$\sizefu{\sfire} = \sizefu{\sitm} + \sizefu{\sfiretwo} + 1 = \sizem{\tderivtwo} + \sizem{\tderivthree} + 1 = \sizem{\tderiv}$).

		\item \emph{Explicit substitution on inert}, \ie $\sfire = \sitm \esub{\var}{\sitmtwo}$. Then necessarily
		\begin{equation*}
		\tderiv = 
		\begin{prooftree}
		\hypo{}
		\ellipsis{$\tderivtwo$}{\typctxtwo, \var \hastype \mtypetwo \vdash \sitm \hastype \mtype}
		\hypo{}
		\ellipsis{$\tderivthree$}{\typctxthree \vdash \sitmtwo \hastype \mtypetwo}
		\infer2[\footnotesize$\ruleES$]{\tyjp{}{\sitm \esub{\var}{\sitmtwo}}{\typctxtwo \mplus \typctxthree}{\mtype}}
		\end{prooftree}
		\end{equation*}
		where $\typctx = \typctxtwo \mplus \typctxthree$.
		Since $\typctx$ is \leftsh (\resp unitary \leftsh), then so are $\typctxtwo$ and $\typctxthree$ by 
		\Cref{rmk:merge-split-coshrinking}.
		As $\sitmtwo$ is a \pointed term (\Cref{rmk:rigid}), we can apply the spreading of left shrinkingess (\Cref{l:spread-shrinking}) to $\tderivthree$ and 
		so $\mtypetwo$ is \leftsh (resp.~unitary \leftsh), which entails that $\typctxtwo, \var \hastype \mtypetwo$ 
		is \leftsh (\resp~unitary \leftsh).
		Neither $\sitm$ nor $\sitmtwo$ are \valES.
		We can then apply \ih to both premises: $\sizem{\tderivtwo} \geq \sizefu{\sitm}$ and $\sizem{\tderivthree} \geq 
		\sizefu{\sitmtwo}$ (\resp~$\sizem{\tderivtwo} = \sizefu{\sitm}$ and $\sizem{\tderivthree} = \sizefu{\sitmtwo}$). 
		Therefore, $\sizefu{\sfire} = \sizefu{\sitm} + \sizefu{\sitmtwo} \leq \sizem{\tderivtwo} + \sizem{\tderivthree} = 
		\sizem{\tderiv}$ (\resp~$\sizefu{\sfire} = \sizefu{\sitm} + \size{\sitmtwo} = \sizem{\tderivtwo} + \sizem{\tderivthree} = 
		\sizem{\tderiv}$).

		\item \emph{Explicit substitution on fireball}, \ie $\sfire = \sfiretwo \esub{\var}{\sitm}$. Then necessarily
		\begin{equation*}
		\tderiv = 
		\begin{prooftree}
		\hypo{}
		\ellipsis{$\tderivtwo$}{\typctxtwo, \var \hastype \mtypetwo \vdash \sfiretwo \hastype \mtype}
		\hypo{}
		\ellipsis{$\tderivthree$}{\typctxthree \vdash \sitm \hastype \mtypetwo}
		\infer2[\footnotesize$\ruleES$]{\tyjp{}{\sfiretwo \esub{\var}{\sitm}}{\typctxtwo \mplus \typctxthree}{\mtype}}
		\end{prooftree}
		\end{equation*}
		where $\typctx = \typctxtwo \mplus \typctxthree$.
		Since $\typctx$ is \leftsh (\resp~unitary \leftsh), then so are $\typctxtwo$ and $\typctxthree$ by 
		\Cref{rmk:merge-split-coshrinking}.
		As $\sitm$ is a \pointed term (\Cref{rmk:rigid}), we can apply the spreading of left shrinkingness (\Cref{l:spread-shrinking}) to $\tderivthree$ and so 
		$\mtypetwo$ is \leftsh (\resp~unitary \leftsh), which entails that $\typctxtwo, \var \hastype \mtypetwo$ is \leftsh (\resp~unitary \leftsh).
		Note that $\sitm$ is not a \valES, while $\sfiretwo$ is a \valES if and only if so is $\sfire = \sfiretwo\esub{\var}{\itm}$, hence if $\sfiretwo$ is a \valES then $\mtype$ is \rightsh (\resp~unitary \rightsh).
		We can then apply \ih to both premises: $\sizem{\tderivtwo} \geq \sizefu{\sfiretwo}$ and $\sizem{\tderivthree} \geq \sizefu{\sitm}$ (\resp~$\sizem{\tderivtwo} = \sizefu{\sfiretwo}$ and $\sizem{\tderivthree} = \sizefu{\sitm}$). 
		Therefore, $\sizefu{\sfire} = \sizefu{\sfiretwo} + \sizefu{\sitm} \leq \sizem{\tderivtwo} + \sizem{\tderivthree} = \sizem{\tderiv}$ (\resp~$\sizefu{\sfire} = \sizefu{\sfiretwo} + \sizefu{\sitm} = \sizem{\tderivtwo} + \sizem{\tderivthree} = \sizem{\tderiv}$).

		\item \emph{Abstraction}, \ie $\sfire = \la{\var}{\sfiretwo}$. 
		Then necessarily, for some $n \in \nat$,
		\begin{equation*}
		\tderiv = 
		\begin{prooftree}[separation = 0.8em, label separation = .2em]
		\hypo{}
		\ellipsis{$\tderivtwo_1$}{\typctx_1, \var \hastype \mtypethree_1 \vdash \sfiretwo \hastype \mtypetwo_1}
		\infer1[\footnotesize$\ruleFun$]{\tyjp{}{\la{\var}{\sfiretwo}}{\typctx_1}{\ty{\mtypethree_1}{\mtypetwo_1}}}
		\hypo{\overset{n \in \nat}{\ldots}}
		\hypo{}
		\ellipsis{$\tderivtwo_n$}{\typctx_n, \var \hastype \mtypethree_n \vdash \sfiretwo \hastype \mtypetwo_n}
		\infer1[\footnotesize$\ruleFun$]{\tyjp{}{\la{\var}{\sfiretwo}}{\typctx_n}{\ty{\mtypethree_n}{\mtypetwo_n}}}
		\infer3[\footnotesize$\ruleManyVal$]{\tyjp{}{\la{\var}{\sfiretwo}}{\typctx}{\mtype}}
		\end{prooftree}
		\end{equation*}
		where $\mtype = \bigmplus_{i=1}^n\mset{\larrow{\mtypethree_i}{\mtypetwo_i}}$ and $\typctx = 
		\bigmplus_{i=1}^n\typctx_i$. 
		As $\mtype$ is \rightsh (\resp~unitary \rightsh),  then $n > 0$ (\resp~$n = 1$) and $\mtypethree$ is \leftsh (\resp~unitary \leftsh) and $\mtypetwo$ is \rightsh (\resp~unitary \rightsh), which entails that $\typctx, \var 
		\hastype \mtypethree$ is a \leftsh (\resp~unitary \leftsh) context. 
		We can then apply the \ih to $\tderivtwo_i$ for all $1 \leq i \leq n$, so $\sizefu{\sfiretwo} \leq \sizem{\tderivtwo_i}$ (\resp~$\sizefu{\sfiretwo} = \sizem{\tderivtwo_i}$). 
		So, $\sizefu{\sfire} = \sizefu{\sfiretwo} + 1 \leq n(\sizefu{\sfiretwo} + 1) \leq 
		\sum_{i=1}^n(\sizem{\tderivtwo_i} + 1) = \sizem{\tderiv}$ where the first inequality holds because $n > 0$ 
		(\resp~$\sizefu{\sfire}= \sizefu{\sfiretwo} + 1 = \sizem{\tderivtwo_1} + 1 = \sizem{\tderiv}$ because $n = 1$). 
		\qedhere
	\end{itemize}
\end{proof}

We split the proof of shrinking quantitative subject reduction in two cases:
\begin{enumerate}
	\item \emph{Open:} after a $\tovsubo$ step any type derivation strictly shrinks its size 
	(\Cref{prop:weak-subject-reduction});
	\item \emph{Strong:} after a $\tovsubs$ step the derivation strictly shrinks its size if it is shrinking 
	(\Cref{prop:shrinking-subject-reduction}). 
\end{enumerate}

More precisely,  two sizes measure how much a derivation shrinks. 
They play two different roles:
\begin{enumerate}
	\item \emph{Qualitative:} To have a combinatorial proof of the characterization of normalizable terms, we need a 
	measure
	that strictly decreases for every (suitable) step to the normal form. 
	This role is played by the general size. 
	\item \emph{Quantitative:} To measure the number of steps to reach the normal forms, we need a measure that decreases 
	by a constant number for each step. This is the role
	played by the multiplicative size to count the number of multiplicative steps (it does not count the number of 
	exponential steps and, indeed, it does not change 	after an exponential step).
\end{enumerate}

The open case (\Cref{prop:weak-subject-reduction}) represents the base case for the strong case 
(\Cref{prop:shrinking-subject-reduction}).  
To prove the base case, we need a substitution lemma (\Cref{l:substitution}) to handle the exponential step, which in turn relies on \Cref{l:typing-value-splitting}.



\begin{proposition}[Shrinking quantitative subject reduction]
	\label{propappendix:shrinking-subject-reduction}
	\NoteState{prop:shrinking-subject-reduction}
	Let $\typctx$ be a \leftsh (\resp unitary \leftsh) context.
	Suppose that $\concl{\tderiv}{\typctx}{\tm}{\mtype}$ and if $\tm$ is a \valES, then $\mtype$ is \rightsh (\resp unitary \rightsh).
	\begin{enumerate}
		\item \emph{Multiplicative step:} If $\tm \toessm \tm'$ then there is a derivation 
		$\concl{\tderiv'}{\typctx}{\tm'}{\mtype}$ such that $\sizem{\tderiv'} \leq \sizem{\tderiv}-2$ and $\size{\tderiv'} < 
		\size{\tderiv}$
		(\resp $\sizem{\tderiv'} = \sizem{\tderiv}-2$ and $\size{\tderiv'} = \size{\tderiv}-1$);
		\item \emph{Exponential step:} if $\tm \toesse \tm'$ then there is a derivation 
		$\concl{\tderiv'}{\typctx}{\tm'}{\mtype}$ such that
		$\sizem{\tderiv'} = \sizem{\tderiv}$ and $\size{\tderiv'} < \size{\tderiv}$.
	\end{enumerate}
\end{proposition}

\begin{proof}
	First, we prove the unitary version of the statement (\ie under the hypothesis that $\typctx$ is unitary \leftsh and if $\tm$ is a \valES then $\mtype$ is unitary \rightsh).
	The proof is by induction on the evaluation strong context $\strongctx$ such that $\tm = \strongctxp{\tmtwo} \toess 
	\strongctxp{\tmtwo'} = \tm'$ with $\tmtwo \tomo \tmtwo'$ or $\tmtwo' \toeo \tmtwo'$. 
	Cases for $\strongctx$:
	\begin{itemize}
		\item \emph{Hole}, \ie{} $\strongctx = \ctxhole$ and $\tm \Rew{\wsym a}  \tm'$ with $a \in \{\msym, \esym\}$.
		According to open subject reduction (\Cref{prop:weak-subject-reduction}),
		\begin{itemize}
			\item if $\tm \towm \tm'$ then there exists a derivation $\concl{\tderiv'}{\typctx}{\tm'}{\mtype}$ such that
			$\sizem{\tderiv'} = \sizem{\tderiv}-2$ and $\size{\tderiv'} = \size{\tderiv}-1$;
			\item if $\tm \towe \tm'$ then there is a derivation $\concl{\tderiv'}{\typctx}{\tm'}{\mtype}$ such that 
			$\sizem{\tderiv'} = \sizem{\tderiv}$ and
			$\size{\tderiv'} < \size{\tderiv}$.
		\end{itemize}
		Note that in this case the hypothesis that $\typctx$ is a (unitary) \leftsh context and that $\mtype$ is a (unitary) \rightsh multi type if $\tm$ is an \valES are not used.
		
		\item \emph{Abstraction}, \ie $\strongctx = \la{\var}{\strongctxtwo}$. 
		So, $\tm = \strongctxp{\tmtwo} = \la{\var}{\strongctxtwop{\tmtwo}} \Rew{\esssym a} 
		\la{\var}{\strongctxtwop{\tmtwo'}} = \strongctxp{\tmtwo'} = \tm'$ with $\tmtwo \Rew{\wsym a} \tmtwo'$ and $a \in \{\msym, \esym\}$.
		Since $\tm$ is an \valES, $\mtype$ is a unitary \rightsh multi type by hypothesis and hence it has the form 
		$\mtype = \mset{\larrow{\mtypethree}{\mtypetwo}}$ where $\mtypethree$ is unitary \leftsh and $\mtypetwo$ is unitary \rightsh.
		Thus, the derivation $\tderiv$ is necessarily
		\begin{equation*}
		\tderiv = 
		\begin{prooftree}
		\hypo{}
		\ellipsis{$\tderivtwo$}{\typctx, \var \hastype \mtypethree \vdash \strongctxtwop{\tmtwo} \hastype \mtypetwo}
		\infer1[\footnotesize$\lambda$]{\typctx \vdash \la{\var}\strongctxtwop{\tmtwo} \hastype 
			\larrow{\mtypethree}{\mtypetwo}}
		\infer1[\footnotesize$\ruleManyVal$]{\typctx \vdash \la{\var}\strongctxtwop{\tmtwo} \hastype 
			\mset{\larrow{\mtypethree}{\mtypetwo}}}
		\end{prooftree}
		\end{equation*}
		By \ih (as $\typctx, \var \hastype \mtypethree$ is a unitary \leftsh type context and $\mtypetwo$ is a unitary \rightsh multi type), there is a derivation $\concl{\tderivtwo'}{\typctx, \var \hastype 
			\mtypethree}{\strongctxtwop{\tmtwo'}}{\mtypetwo}$ with: 
		\begin{enumerate}
			\item $\sizem{\tderivtwo'} = \sizem{\tderivtwo} - 2$ and $\size{\tderivtwo'} = \size{\tderivtwo}-1$ if $\tmtwo 
			\tomo \tmtwo'$ ; 
			\item $\sizem{\tderivtwo'} = \sizem{\tderivtwo}$ and $\size{\tderivtwo'} < \size{\tderivtwo}$ if $\tmtwo \toeo 
			\tmtwo'$.
		\end{enumerate}
		We can then build the derivation 
		\begin{equation*}
		\tderiv' = 
		\begin{prooftree}
		\hypo{}
		\ellipsis{$\tderivtwo'$}{\typctx, \var \hastype \mtypethree \vdash \strongctxtwop{\tmtwo'} \hastype \mtypetwo}
		\infer1[\footnotesize$\lambda$]{\typctx \vdash \la{\var}\strongctxtwop{\tmtwo'} \hastype 
			\larrow{\mtypethree}{\mtypetwo}}
		\infer1[\footnotesize$\ruleManyVal$]{\typctx \vdash \la{\var}\strongctxtwop{\tmtwo'} \hastype 
			\mset{\larrow{\mtypethree}{\mtypetwo}}}
		\end{prooftree}
		\end{equation*}
		where
		\begin{enumerate}
			\item $\sizem{\tderiv'} = \sizem{\tderivtwo'} +1 = \sizem{\tderivtwo} + 1 - 2 = \sizem{\tderiv} - 2$ and 
			$\size{\tderiv'} = \size{\tderivtwo'} +1 = \size{\tderivtwo} + 1 - 1 = \size{\tderiv} - 1$ if $\tmtwo \tomo \tmtwo'$ (\ie if $\tm \toms \tm'$); 
			\item $\sizem{\tderiv'} = \sizem{\tderivtwo'} +1 = \sizem{\tderivtwo} + 1 = \sizem{\tderiv}$ and $\size{\tderiv'} 
			= \size{\tderivtwo'} +1 < \size{\tderivtwo} + 1 = \size{\tderiv}$ if $\tmtwo \toeo \tmtwo'$ (\ie if $\tm \toes \tm'$).
		\end{enumerate}

		\item \emph{Explicit substitution of rigid context}, \ie $\strongctx = \tmthree \esub{\var}{\ictx}$. 
		So, $\tm = \strongctxp{\tmtwo} = \tmthree\esub{\var}{\ictxp{\tmtwo}} \Rew{\esssym a} 
		\tmthree\esub{\var}{\ictxp{\tmtwo'}} = \strongctxp{\tmtwo'} = \tm'$ with $\tmtwo \Rew{\wsym a} \tmtwo'$ and $a \in 
		\{\msym, \esym\}$.
		Then, necessarily
		\begin{equation*}
		\tderiv = 
		\begin{prooftree}
		\hypo{}
		\ellipsis{$\tderivtwo$}{\typctxtwo, \var \hastype \mtypetwo \vdash \tmthree \hastype \mtype}
		\hypo{}
		\ellipsis{$\tderivthree$}{\typctxthree \vdash \ictxp{\tmtwo} \hastype \mtypetwo}
		\infer2[\footnotesize$\Es$]{\typctxtwo \uplus \typctxthree \vdash \tmthree \esub{\var}{\ictxp{\tmtwo}} \hastype 
			\mtype}
		\end{prooftree}
		\end{equation*}
		with $\typctx = \typctxtwo \uplus \typctxthree$ unitary \leftsh by hypothesis, and then so is 		$\typctxthree$ by \Cref{rmk:merge-split-coshrinking}.
		By \ih applied to $\tderivthree$ (as $\ictxp{\tmtwo}$ is not an \valES), there is a derivation 
		$\concl{\tderivthree'}{\typctxthree}{\ictxp{\tmtwo'}}{\mtypetwo}$~~with: 
		\begin{enumerate}
			\item $\sizem{\tderivthree'} = \sizem{\tderivthree} - 2$ and $\size{\tderivthree'} = \size{\tderivthree} - 1$ if 
			$\tmtwo \tomo \tmtwo'$; 
			\item $\sizem{\tderivthree'} = \sizem{\tderivthree}$ and $\size{\tderivthree'} < \size{\tderivthree}$ if $\tmtwo 
			\toeo \tmtwo'$.
		\end{enumerate}
		We can then build the derivation 
		\begin{equation*}
		\tderiv' = 
		\begin{prooftree}
		\hypo{}
		\ellipsis{$\tderivtwo$}{\typctxtwo, \var \hastype \mtypetwo \vdash \tmthree \hastype \mtype}
		\hypo{}
		\ellipsis{$\tderivthree'$}{\typctxthree \vdash \ictxp{\tmtwo'} \hastype \mtypetwo}
		\infer2[\footnotesize$\Es$]{\typctxtwo \uplus \typctxthree \vdash \tmthree \esub{\var}{\ictxp{\tmtwo'} } \hastype 
			\mtype}
		\end{prooftree}
		\end{equation*}
		where $\typctx = \typctxtwo \uplus \typctxthree$ and
		\begin{enumerate}
			\item $\sizem{\tderiv'} = \sizem{\tderivtwo} + \sizem{\tderivthree'} = \sizem{\tderivtwo} + \sizem{\tderivthree} - 2 = \sizem{\tderiv} - 2$ and $\size{\tderiv'} = \size{\tderivtwo} + \size{\tderivthree'} +1 = \size{\tderivtwo} + \size{\tderivthree} - 1 +1 = \size{\tderiv} - 1$ if $\tmtwo \tomo \tmtwo'$ (\ie if $\tm \toms \tm'$); 
			\item $\size{\tderiv'} = \size{\tderivtwo} + \size{\tderivthree'} = \size{\tderivtwo} + \size{\tderivthree} = \size{\tderiv}$ and $\size{\tderiv'} = \size{\tderivtwo} + \size{\tderivthree'} +1 < \size{\tderivtwo} + \size{\tderivthree} +1 = \size{\tderiv}$ if $\tmtwo \toeo \tmtwo'$ (\ie if $\tm \toes \tm'$).
		\end{enumerate}
		
		\item \emph{Strong context with explicit substitution of \pointed term}, \ie $\strongctx = 
		\strongctxtwo\esub{\var}{\ptm}$. 
		Then, $\tm = \strongctxp{\tmtwo} = \strongctxtwop{\tmtwo} \esub{\var}{\ptm} \Rew{\esssym a} 
		\strongctxtwop{\tmtwo'}\esub{\var}{\ptm} = \strongctxp{\tmtwo'} = \tm'$ with $\tmtwo \Rew{\wsym a} \tmtwo'$ and $a \in 
		\{\msym, \esym\}$.
		The derivation $\tderiv$ is necessarily
		\begin{equation*}
		\tderiv = 
		\begin{prooftree}
		\hypo{}
		\ellipsis{$\tderivtwo$}{\typctxtwo, \var \hastype \mtypetwo \vdash \strongctxtwop{\tmtwo} \hastype \mtype}
		\hypo{}
		\ellipsis{$\tderivthree$}{\typctxthree \vdash \ptm \hastype \mtypetwo}
		\infer2[\footnotesize$\Es$]{\typctxtwo \uplus \typctxthree \vdash \strongctxtwop{\tmtwo} \esub{\var}{\ptm} \hastype 
			\mtype}
		\end{prooftree}
		\end{equation*}
		where $\typctx = \typctxtwo \uplus \typctxthree$ is unitary \leftsh by hypothesis, and then so are 		$\typctxtwo$ and $\typctxthree$ by \Cref{rmk:merge-split-coshrinking}.
		According to spreading of \leftsh shrinkingness (\Cref{l:spread-shrinking}, as $\ptm$ is a \pointed term), $\mtypetwo$ is unitary \leftsh.
		Note that $\strongctxtwop{\tmtwo}\esub{\var}{\ptm}$ is an \valES if and only if $\strongctxtwop{\tmtwo}$ is an \valES.
		So, the \ih can be applied to $\tderivtwo$ (since $\typctxtwo, \var \hastype \mtypetwo$ is a unitary \leftsh	context) and hence there is a derivation $\concl{\tderivtwo'}{\typctxtwo, \var \hastype 
			\mtypetwo}{\strongctxtwop{\tmtwo'}}{\mtype}$ with: 
		\begin{enumerate}
			\item $\sizem{\tderivtwo'} = \sizem{\tderivtwo} - 2$ and $\size{\tderivtwo'} = \size{\tderivtwo} - 1$ if $\tmtwo \tomo \tmtwo'$; 
			\item $\sizem{\tderivtwo'} = \sizem{\tderivtwo}$ and $\size{\tderivtwo'} < \size{\tderivtwo}$ if $\tmtwo \toeo \tmtwo'$.
		\end{enumerate}
		We can then build the derivation 
		\begin{equation*}
		\tderiv' = 
		\begin{prooftree}
		\hypo{}
		\ellipsis{$\tderivtwo'$}{\typctxtwo, \var \hastype \mtypetwo \vdash \strongctxtwop{\tmtwo'} \hastype \mtype}
		\hypo{}
		\ellipsis{$\tderivthree$}{\typctxthree \vdash \ptm \hastype \mtypetwo}
		\infer2[\footnotesize$\Es$]{\typctxtwo \uplus \typctxthree \vdash \strongctxtwop{\tmtwo'} \esub{\var}{\ptm} 
			\hastype \mtype}
		\end{prooftree}
		\end{equation*}
		where $\typctx = \typctxtwo \uplus \typctxthree$ and
		\begin{enumerate}
			\item $\sizem{\tderiv'} = \sizem{\tderivtwo'} + \sizem{\tderivthree} = \sizem{\tderivtwo} + \sizem{\tderivthree} 
			- 2 = \sizem{\tderiv} - 2$ and $\size{\tderiv'} = \size{\tderivtwo'} + \size{\tderivthree} +1 = \size{\tderivtwo} -1 + 
			\size{\tderivthree} +1 = \size{\tderiv} - 1$ if $\tmtwo \tomo \tmtwo'$ (\ie if $\tm \toms \tm'$); 
			\item $\sizem{\tderiv'} = \sizem{\tderivtwo'} + \sizem{\tderivthree} = \sizem{\tderivtwo} + \sizem{\tderivthree} 
			= \sizem{\tderiv}$ and $\size{\tderiv'} = \size{\tderivtwo'} + \size{\tderivthree} +1 < \size{\tderivtwo} + 
			\size{\tderivthree} +1 = \size{\tderiv}$ if $\tmtwo \toeo \tmtwo'$ (\ie if $\tm \toes \tm'$).
		\end{enumerate}
		
		\item \emph{Rigid term applied to strong context}, \ie $\strongctx = \ptm \strongctxtwo$. 
		Then, $\tm = \strongctxp{\tmtwo} = \ptm \strongctxtwop{\tmtwo} \Rew{\esssym a} \ptm \strongctxtwop{\tmtwo'} = 
		\strongctxp{\tmtwo'} = \tm'$ with $\tmtwo \Rew{\wsym a} \tmtwo'$ and $a \in \{\msym, \esym\}$.
		The derivation $\tderiv$ is necessarily
		\begin{equation*}
		\tderiv = 
		\begin{prooftree}
		\hypo{}
		\ellipsis{$\tderivtwo$}{\typctxtwo \vdash \ptm \hastype \mset{\larrow{\mtypetwo}{\mtype}}}
		\hypo{}
		\ellipsis{$\tderivthree$}{\typctxthree \vdash \strongctxtwop{\tmtwo} \hastype \mtypetwo}
		\infer2[\footnotesize$\ruleAp$]{\typctxtwo \uplus \typctxthree \vdash \ptm \strongctxtwop{\tmtwo} \hastype \mtype}
		\end{prooftree}
		\end{equation*}
		where $\typctx = \typctxtwo \mplus \typctxthree$ is unitary \leftsh by hypothesis, and then so are $\typctxtwo$ and $\typctxthree$ by \Cref{rmk:merge-split-coshrinking}.
		According to spreading of \leftsh shrinkingness (\Cref{l:spread-shrinking}, as $\ptm$ is a \pointed term), $\mset{\larrow{\mtypetwo}{\mtype}}$ is a unitary \leftsh multi type and hence $\mtypetwo$ is a unitary \rightsh multi type.
		Thus, we can apply the \ih to $\tderivthree$ and get a derivation 
		$\concl{\tderivthree'}{\typctxthree}{\strongctxtwop{\tmtwo'}}{\mtypetwo}$ such that: 
		\begin{enumerate}
			\item $\sizem{\tderivthree'} = \sizem{\tderivthree} - 2$ and $\size{\tderivthree'} = \size{\tderivthree} - 1$ if 
			$\tmtwo \tomo \tmtwo'$; 
			\item $\sizem{\tderivthree'} = \sizem{\tderivthree}$ and $\size{\tderivthree'} < \size{\tderivthree}$ if $\tmtwo 
			\toeo \tmtwo'$.
		\end{enumerate}
		We can then build the derivation 
		\begin{equation*}
		\tderiv' = 
		\begin{prooftree}
		\hypo{}
		\ellipsis{$\tderivtwo$}{\typctxtwo \vdash \ptm \hastype \mset{\larrow{\mtypetwo}{\mtype}}}
		\hypo{}
		\ellipsis{$\tderivthree'$}{\typctxthree \vdash \strongctxtwop{\tmtwo'} \hastype \mtypetwo}
		\infer2[\footnotesize$\ruleAp$]{\typctxtwo \uplus \typctxthree \vdash \ptm \strongctxtwop{\tmtwo'}  \hastype \mtype}
		\end{prooftree}
		\end{equation*}
		where $\typctx = \typctxtwo \uplus \typctxthree$ and
		\begin{enumerate}
			\item $\sizem{\tderiv'} = \sizem{\tderivtwo} + \sizem{\tderivthree'} +1 = \sizem{\tderivtwo} + 
			\sizem{\tderivthree} - 2 +1 = \sizem{\tderiv} - 2$ and $\size{\tderiv'} = \size{\tderivtwo} + \size{\tderivthree'} +1 = 
			\size{\tderivtwo} + \size{\tderivthree} - 1 +1 = \size{\tderiv} - 1$ if $\tmtwo \tomo \tmtwo'$ (\ie if $\tm \toms \tm'$); 
			\item $\sizem{\tderiv'} = \sizem{\tderivtwo} + \sizem{\tderivthree'} +1 = \sizem{\tderivtwo} + 
			\sizem{\tderivthree} +1 = \sizem{\tderiv}$ and $\size{\tderiv'} = \size{\tderivtwo} + \size{\tderivthree'} +1 < 
			\size{\tderivtwo} + \size{\tderivthree} +1 = \size{\tderiv}$ if $\tmtwo \toeo \tmtwo'$ (\ie if $\tm \toes \tm'$).
		\end{enumerate}
		
		\item \emph{Rigid context applied to term}, \ie $\strongctx = \ictx\tmthree$. 
		Then, $\tm = \strongctxp{\tmtwo} = \ictxp{\tmtwo} \tmthree \Rew{\esssym a} \ictxp{\tmtwo'} \tmthree = \strongctxp{\tmtwo'} = \tm'$ with $\tmtwo \Rew{\wsym a} \tmtwo'$ and $a \in \{\msym, \esym\}$.
		The derivation $\tderiv$ is necessarily
		\begin{equation*}
		\tderiv = 
		\begin{prooftree}
		\hypo{}
		\ellipsis{$\tderivtwo$}{\typctxtwo \vdash \ictxp{\tmtwo} \hastype \mset{\larrow{\mtypetwo}{\mtype}}}
		\hypo{}
		\ellipsis{$\tderivthree$}{\typctxthree \vdash \tmthree \hastype \mtypetwo}
		\infer2[\footnotesize$\ruleAp$]{\typctxtwo \uplus \typctxthree \vdash \ictxp{\tmtwo} \tmthree \hastype \mtype}
		\end{prooftree}
		\end{equation*}
		where $\typctx = \typctxtwo \mplus \typctxthree$ is unitary \leftsh by hypothesis, and then so is $\typctxtwo$ 	by \Cref{rmk:merge-split-coshrinking}.
		By \ih (as $\ictxp{\tmtwo}$ is not an \valES), there is a derivation 
		$\concl{\tderivtwo'}{\typctxtwo}{\ictxp{\tmtwo'}}{\mset{\larrow{\mtypetwo}{\mtype}}}$ with: 
		\begin{enumerate}
			\item $\sizem{\tderivtwo'} = \sizem{\tderivtwo} - 2$ and $\size{\tderivtwo'} = \size{\tderivtwo} - 1$ if $\tmtwo 
			\tomo \tmtwo'$; 
			\item $\sizem{\tderivtwo'} = \sizem{\tderivtwo}$ and $\size{\tderivtwo'} < \size{\tderivtwo}$ if $\tmtwo \toeo 
			\tmtwo'$.
		\end{enumerate}
		We can then build the derivation 
		\begin{equation*}
		\tderiv' = 
		\begin{prooftree}
		\hypo{}
		\ellipsis{$\tderivtwo'$}{\typctxtwo \vdash \ictxp{\tmtwo'} \hastype \mset{\larrow{\mtypetwo}{\mtype}}}
		\hypo{}
		\ellipsis{$\tderivthree$}{\typctxthree \vdash \tmthree \hastype \mtypetwo}
		\infer2[\footnotesize$\ruleAp$]{\typctxtwo \uplus \typctxthree \vdash \ictxp{\tmtwo'} \tmthree \hastype \mtype}
		\end{prooftree}
		\end{equation*}
		where $\typctx = \typctxtwo \uplus \typctxthree$ and
		\begin{enumerate}
			\item $\sizem{\tderiv'} = \sizem{\tderivtwo'} + \sizem{\tderivthree} +1 = \sizem{\tderivtwo} - 2 + 
			\sizem{\tderivthree} +1 = \sizem{\tderiv} - 2$ and $\size{\tderiv'} = \size{\tderivtwo'} + \size{\tderivthree} +1 = 
			\size{\tderivtwo} - 1 + \size{\tderivthree} +1 = \size{\tderiv} - 1$ if $\tmtwo \tomo \tmtwo'$ (\ie if $\tm \toms \tm'$); 
			\item $\sizem{\tderiv'} = \sizem{\tderivtwo'} + \sizem{\tderivthree} +1 = \sizem{\tderivtwo} + 
			\sizem{\tderivthree} +1 = \sizem{\tderiv}$ and $\size{\tderiv'} = \size{\tderivtwo'} + \size{\tderivthree} +1 < 
			\size{\tderivtwo} + \size{\tderivthree} +1 = \size{\tderiv}$ if $\tmtwo \toeo \tmtwo'$ (\ie if $\tm \toes \tm'$).
		\end{enumerate}
		
		\item \emph{Rigid context with explicit substitution of rigid term}, \ie $\strongctx = \ictx\esub{\var}{\ptm}$. 
		Then, $\tm = \strongctxp{\tmtwo} = \ictxp{\tmtwo} \esub{\var}{\ptm} \Rew{\esssym a} 
		\ictxp{\tmtwo'}\esub{\var}{\ptm} = \strongctxp{\tmtwo'} = \tm'$ with $\tmtwo \Rew{\wsym a} \tmtwo'$ and $a \in \{\msym, \esym\}$.
		The derivation $\tderiv$ is necessarily
		\begin{equation*}
		\tderiv = 
		\begin{prooftree}
		\hypo{}
		\ellipsis{$\tderivtwo$}{\typctxtwo, \var \hastype \mtypetwo \vdash \ictxp{\tmtwo} \hastype \mtype}
		\hypo{}
		\ellipsis{$\tderivthree$}{\typctxthree \vdash \ptm \hastype \mtypetwo}
		\infer2[\footnotesize$\Es$]{\typctxtwo \uplus \typctxthree \vdash \ictxp{\tmtwo} \esub{\var}{\ptm} \hastype \mtype}
		\end{prooftree}
		\end{equation*}
		where $\typctx = \typctxtwo \mplus \typctxthree$ is unitary \leftsh by hypothesis, and then so are $\typctxtwo$ and $\typctxthree$ by \Cref{rmk:merge-split-coshrinking}.
		According to spreading of \leftsh shrinkingness (\Cref{l:spread-shrinking}, as $\ptm$ is a \pointed term), $\mtypetwo$ is unitary \leftsh.
		Thus, the \ih can be applied to $\tderivtwo$ (since $\typctxtwo, \var \hastype \mtypetwo$ is unitary \leftsh and $\ictxp{\tmtwo}$ is not an \valES) to obtain a derivation $\concl{\tderivtwo'}{\typctxtwo, \var \hastype \mtypetwo}{\ictxp{\tmtwo'}}{\mtype}$ such that: 
		\begin{enumerate}
			\item $\sizem{\tderivtwo'} = \sizem{\tderivtwo} - 2$ and $\size{\tderivtwo'} = \size{\tderivtwo} - 1$ if $\tmtwo 
			\tomo \tmtwo'$; 
			\item $\sizem{\tderivtwo'} = \sizem{\tderivtwo}$ and $\size{\tderivtwo'} < \size{\tderivtwo}$ if $\tmtwo \toeo 
			\tmtwo'$.
		\end{enumerate}
		We can then build the derivation 
		\begin{equation*}
		\tderiv' = 
		\begin{prooftree}
		\hypo{}
		\ellipsis{$\tderivtwo'$}{\typctxtwo, \var \hastype \mtypetwo \vdash \ictxp{\tmtwo'} \hastype \mtype}
		\hypo{}
		\ellipsis{$\tderivthree$}{\typctxthree \vdash \ptm \hastype \mtypetwo}
		\infer2[\footnotesize$\Es$]{\typctxtwo \uplus \typctxthree \vdash \ictxp{\tmtwo'} \esub{\var}{\ptm} \hastype \mtype}
		\end{prooftree}
		\end{equation*}
		where $\typctx = \typctxtwo \uplus \typctxthree$ and
		\begin{enumerate}
			\item $\sizem{\tderiv'} = \sizem{\tderivtwo'} + \sizem{\tderivthree} = \sizem{\tderivtwo} - 2 + 
			\sizem{\tderivthree} = \sizem{\tderiv} - 2$ and $\size{\tderiv'} = \size{\tderivtwo'} + \size{\tderivthree} = 
			\size{\tderivtwo} - 1 + \size{\tderivthree} = \size{\tderiv} - 1$ if $\tmtwo \tomo \tmtwo'$ (\ie if $\tm \toms \tm'$); 
			\item $\sizem{\tderiv'} = \sizem{\tderivtwo'} + \sizem{\tderivthree} = \sizem{\tderivtwo} + \sizem{\tderivthree} 
			= \sizem{\tderiv}$ and $\size{\tderiv'} = \size{\tderivtwo'} + \size{\tderivthree} < \size{\tderivtwo} + 
			\size{\tderivthree} = \size{\tderiv}$ if $\tmtwo \toeo \tmtwo'$ (\ie if $\tm \toes \tm'$).
		\end{enumerate}
		
		\item \emph{Rigid term with explicit substitution of rigid context}, \ie $\strongctx = \ptm\esub{\var}{\ictx}$. 
		Then, $\tm = \strongctxp{\tmtwo} = \ptm\esub{\var}{\ictxp{\tmtwo}} \Rew{\esssym a} \ptm\esub{\var}{\ictxp{\tmtwo'}} 
		= \strongctxp{\tmtwo'} = \tm'$ with $\tmtwo \Rew{\wsym a} \tmtwo'$ and $a \in \{\msym, \esym\}$.
		The derivation $\tderiv$ is necessarily
		\begin{equation*}
		\tderiv = 
		\begin{prooftree}
		\hypo{}
		\ellipsis{$\tderivtwo$}{\typctxtwo, \var \hastype \mtypetwo \vdash \ptm \hastype \mtype}
		\hypo{}
		\ellipsis{$\tderivthree$}{\typctxthree \vdash \ictxp{\tmtwo} \hastype \mtypetwo}
		\infer2[\footnotesize$\Es$]{\typctxtwo \uplus \typctxthree \vdash \ptm \esub{\var}{\ictxp{\tmtwo}} \hastype \mtype}
		\end{prooftree}
		\end{equation*}
		with $\typctx = \typctxtwo \mplus \typctxthree$ unitary \leftsh by hypothesis, and then so is 	$\typctxthree$ by \Cref{rmk:merge-split-coshrinking}.
		By \ih applied to $\tderivthree$ (as $\ictxp{\tmtwo}$ is not an \valES), there is a derivation 
		$\concl{\tderivthree'}{\typctxthree}{\ictxp{\tmtwo'}}{\mtypetwo}$~with: 
		\begin{enumerate}
			\item $\sizem{\tderivthree'} = \sizem{\tderivthree} - 2$ and $\size{\tderivthree'} = \size{\tderivthree} - 1$ if 
			$\tmtwo \tomo \tmtwo'$; 
			\item $\sizem{\tderivthree'} = \sizem{\tderivthree}$ and $\size{\tderivthree'} < \size{\tderivthree}$ if $\tmtwo 
			\toeo \tmtwo'$.
		\end{enumerate}
		We can then build the derivation 
		\begin{equation*}
		\tderiv' = 
		\begin{prooftree}
		\hypo{}
		\ellipsis{$\tderivtwo$}{\typctxtwo, \var \hastype \mtypetwo \vdash \ptm \hastype \mtype}
		\hypo{}
		\ellipsis{$\tderivthree'$}{\typctxthree \vdash \ictxp{\tmtwo'} \hastype \mtypetwo}
		\infer2[\footnotesize$\Es$]{\typctxtwo \uplus \typctxthree \vdash \ptm \esub{\var}{\ictxp{\tmtwo'} } \hastype 
			\mtype}
		\end{prooftree}
		\end{equation*}
		where $\typctx = \typctxtwo \uplus \typctxthree$ and
		\begin{enumerate}
			\item $\sizem{\tderiv'} = \sizem{\tderivtwo} + \sizem{\tderivthree'} = \sizem{\tderivtwo} + \sizem{\tderivthree} 
			- 2 = \sizem{\tderiv} - 2$ and $\size{\tderiv'} = \size{\tderivtwo} + \size{\tderivthree'} +1 = \size{\tderivtwo} + 
			\size{\tderivthree} - 1 +1 = \size{\tderiv} - 1$ if $\tmtwo \tomo \tmtwo'$; 
			\item $\sizem{\tderiv'} = \sizem{\tderivtwo} + \sizem{\tderivthree'} = \sizem{\tderivtwo} + \sizem{\tderivthree} 
			= \sizem{\tderiv}$ and $\size{\tderiv'} = \size{\tderivtwo} + \size{\tderivthree'} +1 < \size{\tderivtwo} + 
			\size{\tderivthree} +1 = \size{\tderiv}$ if $\tmtwo \toeo \tmtwo'$.
		\end{enumerate}
	\end{itemize}
	
	This completes the proof for the unitary case.
	
	In the non-unitary statement (\ie under the weaker hypothesis that $\typctx$ is a \leftsh type context and if $\tm$ is an \valES then $\mtype$ is \rightsh), the proof is analogous to the unitary statement, except for the \emph{Abstraction} case.
	Indeed, in the base case (\emph{Hole context}) unitarity does not play any role, and the other cases follow from the \ih analogously to the unitary statement.
	Let us see the only substantially different case: 
	\begin{itemize}
		\item \emph{Abstraction}, \ie $\strongctx = \la{\var}{\strongctxtwo}$. 
		So, $\tm = \strongctxp{\tmtwo} = \la{\var}{\strongctxtwop{\tmtwo}} \Rew{\esssym a} 
		\la{\var}{\strongctxtwop{\tmtwo'}} = \strongctxp{\tmtwo'} = \tm'$ with $\tmtwo \Rew{\wsym a} \tmtwo'$ and $a \in 
		\{\msym, \esym\}$.
		Since $\tm$ is an \valES, $\mtype$ is a \rightsh multi type by hypothesis and hence it has the form $\mtype = 
		\mset{\larrow{\mtypethree_1}{\mtypetwo_1}, \dots, \larrow{\mtypethree_n}{\mtypetwo_n}}$ for some $n > 0$, where 
		$\mtypethree_i$ is \leftsh and $\mtypetwo_i$ is \rightsh for all $1 \leq i \leq n$.
		Thus, the derivation $\tderiv$ is necessarily
		\begin{equation*}
		\tderiv = 
		\begin{prooftree}[separation=1em]
		\hypo{}
		\ellipsis{$\tderivtwo_i$}{\typctx_i, \var \hastype \mtypethree_i \vdash \strongctxtwop{\tmtwo} \hastype \mtypetwo_i}
		\infer1[\footnotesize$\lambda$]{\typctx_i \vdash \la{\var}\strongctxtwop{\tmtwo} \hastype 
			\larrow{\mtypethree_i}{\mtypetwo_i}}
		\delims{ \left( }{ \right)_{1 \leq i \leq n} }
		\infer1[\footnotesize$\ruleManyVal$]{ \bigmplus_{i=1}^n \typctx_{i} \vdash \la{\var}\strongctxtwop{\tmtwo} \hastype 
			\bigmplus_{i=1}^n \mset{\larrow{\mtypethree_i}{\mtypetwo_i}}}
		\end{prooftree}
		\end{equation*}
		For all $1 \leq i \leq n$, by \ih (as $\typctx_i, \var \hastype \mtypethree_i$ is a \leftsh type context and $\mtypetwo_i$ is a \rightsh multi type), there is a derivation $\concl{\tderivtwo_i'}{\typctx_i, \var \hastype 
			\mtypethree_i}{\strongctxtwop{\tmtwo'}}{\mtypetwo_i}$ with: 
		\begin{enumerate}
			\item $\sizem{\tderivtwo_i'} \leq \sizem{\tderivtwo_i} - 2$ and $\size{\tderivtwo_i'} \leq \size{\tderivtwo_i}-1$ if $\tmtwo \tomo \tmtwo'$; 
			\item $\sizem{\tderivtwo_i'} = \sizem{\tderivtwo_i}$ and $\size{\tderivtwo_i'} < \size{\tderivtwo_i}$ if $\tmtwo \toeo \tmtwo'$.
		\end{enumerate}
		We can then build the derivation 
		\begin{equation*}
		\tderiv' = 
		\begin{prooftree}[separation=1em]
		\hypo{}
		\ellipsis{$\tderivtwop_i$}{\typctx_i, \var \hastype \mtypethree_i \vdash \strongctxtwop{\tmtwo'} \hastype 
			\mtypetwo_i}
		\infer1[\footnotesize$\lambda$]{\typctx_i \vdash \la{\var}\strongctxtwop{\tmtwo'} \hastype 
			\larrow{\mtypethree_i}{\mtypetwo_i}}
		\delims{ \left( }{ \right)_{1 \leq i \leq n} }
		\infer1[\footnotesize$\ruleManyVal$]{ \bigmplus_{i=1}^n \typctx_{i} \vdash \la{\var}\strongctxtwop{\tmtwop} 
			\hastype  \bigmplus_{i=1}^n \mset{\larrow{\mtypethree_i}{\mtypetwo_i}}}
		\end{prooftree}
		\end{equation*}
		where
		\begin{enumerate}
			\item $\sizem{\tderiv'} = \sum_{i=1}^n(\sizem{\tderivtwo_i'} +1) \leq \sum_{i=1}^n(\sizem{\tderivtwo_i} + 1 - 2) 
			= \sizem{\tderiv} - 2n \leq \sizem{\tderiv} - 2$ (where the last inequality holds because $n >0$) and $\size{\tderiv'} 
			=  \sum_{i=1}^n(\size{\tderivtwo_i'} +1) =  \sum_{i=1}^n(\size{\tderivtwo_i} + 1 - 1) \leq \size{\tderiv} - n < 	\size{\tderiv}$ (where the last inequality holds because $n >0$) if $\tmtwo \tomo \tmtwo'$; 
			\item $\sizem{\tderiv'} = \sum_{i=1}^n(\sizem{\tderivtwo_i'} +1) = \sum_{i=1}^n(\sizem{\tderivtwo_i} + 1) = 
			\sizem{\tderiv}$ and $\size{\tderiv'} = \sum_{i=1}^n(\size{\tderivtwo_i'} +1) < \sum_{i=1}^n(\size{\tderivtwo_i} + 1) = 
			\size{\tderiv}$ if $\tmtwo \toeo \tmtwo'$ (the inequality holds because $n > 0$).
			\qedhere
		\end{enumerate}
	\end{itemize}	
\end{proof}

\begin{theorem}[Shrinking correctness]
	\label{thmappendix:correctness}
	\NoteState{thm:correctness}
	Let $\concl{\tderiv}{\typctx}{\tm}{\mtype}$ be a shrinking (\resp unitary shrinking) derivation.
	Then there are a $\vsub$-normal form $\tm'$ and an evaluation $\deriv \colon \tm \tovsubs^* \tm'$ with $2\sizem{\deriv} + 
	\sizefu{\tm'} \leq \sizem{\tderiv}$ (\resp $2\sizem{\deriv} + \sizefu{\tm'} = \sizem{\tderiv}$).
\end{theorem}

\begin{proof}
	By induction on the general size $\size{\tderiv}$ of $\tderiv$.
	
	If $\tm$ is $\esssym$-normal, then $\tm$ is $\vsub$-normal by \Cref{prop:external-properties}.\ref{p:external-properties-fullness}.
	Let $\tm' = \tm$ and $\deriv$ be the empty evaluation (so $\sizem{\deriv} = 0$), thus $\sizem{\tderiv} \geq \sizefu{\tm'} = \sizefu{\tm'} + 2\sizem{\deriv}$ 
	(\resp $\sizem{\tderiv} = \sizefu{\tm'} = \sizefu{\tm'} + 2\sizem{\deriv}$) by \Cref{l:size-strong-fireballs}.
	
	Otherwise, $\tm$ is not $\esssym$-normal and so $\tm \tovsubs \tmtwo$.
	According to shrinking subject reduction (\Cref{prop:shrinking-subject-reduction}), there is a derivation 
	$\concl{\tderivtwo}{\typctx}{\tmtwo}{\mtype}$ such that $\size{\tderivtwo} < \size{\tderiv}$  and 
	\begin{itemize}
		\item $\sizem{\tderivtwo} \leq \sizem{\tderiv} - 2$ (\resp $\sizem{\tderivtwo} = \sizem{\tderiv} - 2$) if $\tm \toms \tmtwo$,
		\item $\sizem{\tderivtwo} = \sizem{\tderiv}$ if $\tm \toes \tmtwo$.
	\end{itemize}
	By \ih, there exists a normal form $\tm'$ and a reduction sequence $\deriv' \colon \tmtwo \tovsubs^* \tm'$ with $2\size{\deriv'} + \sizefu{\tm'} \leq \sizem{\tderivtwo}$ (\resp $2\size{\deriv'} + \sizefu{\tm'} = \sizem{\tderivtwo}$).
	Let $\deriv$ be the $\esssym$-evaluation obtained by concatenating the first step $\tm \tovsubs \tmtwo$ and $\deriv'$.
	There are two cases:
	\begin{itemize}
		\item \emph{Multiplicative:} if $\tm \toms \tmtwo$ then $\sizem{\tderiv} \geq \sizem{\tderivtwo} + 2 \geq \sizefu{\tm'} 
		+ 2\sizem{\deriv'} + 2 = \sizefu{\tm'} + 2\sizem{\deriv}$ (\resp $\sizem{\tderiv} = \sizem{\tderivtwo} + 2 = \sizefu{\tm'} + 
		2\sizem{\deriv'} + 2 = \sizefu{\tm'} + 2\sizem{\deriv}$), since $\sizem{\deriv} = \sizem{\deriv'} + 1$.
		\item \emph{Exponential:} if $\tm \toes \tmtwo$ then $\sizem{\tderiv} = \sizem{\tderivtwo} = \sizefu{\tm'} + 		2\sizem{\deriv'} = \sizefu{\tm'} + 2\sizem{\deriv}$,  since $\sizem{\deriv} = \sizem{\deriv'}$.
		\qedhere
	\end{itemize} 
\end{proof}

\subsection{Completeness}

The proof of completeness (\Cref{thm:correctness}) is dual to the proof of (\Cref{thm:correctness}), refined as usual 
with quantitative information: 
if a term noralizes then it is shrinking typable, and the derivation provides bounds for both the number of 
multiplicative steps to normal form and the size of the normal form.

The are two crucial ingredients:
\begin{enumerate}
	\item normal forms are typable with a unitary shrinking type derivations  (\Cref{prop:typability-normal});
	
	\item quantitative subject expansion (\Cref{prop:shrinking-subject-expansion}) says that not only types are preserved 
	going back though a $\tovsubs$ step but also that if the type derivation is shrinking then the size of the derivation 
	strictly increases.
\end{enumerate}        


\begin{lemma}[Shrinking typability of normal forms]\hfill
	\label{propappendix:typability-normal}
	\NoteState{prop:typability-normal}
	\begin{enumerate}
		\item\label{papppendix:typability-normal-inert}
		\emph{Inert:}
		For every \full inert term $\sitm$ and \leftsh multi type $\mtype$,
		there exist a \leftsh type context $\typctx$ 
		and a derivation $\concl{\tderiv}{\typctx}{\sitm}{\mtype}$.
		If moreover $\mtype$ is unitary, then so is $\typctx$.
		\item\label{pappendix:typability-normal-fireball}
		\emph{Fireball:}
		For every \full fireball $\sfire$ 
		there exists a unitary shrinking derivation $\concl{\tderiv}{\typctx}{\sfire}{\mtype}$.
	\end{enumerate}
\end{lemma}

\begin{proof}
	Both points are proved by mutual induction on the definition of \full inert terms $\sitm$ and \full fireballs 
	$\sfire$.
	Cases:
	\begin{itemize}
		\item \emph{Variable}, \ie $\sitm = \var = \sfire$. 
		Let $\mtype = \mset{\ltype_1, \dots, \ltype_n}$ be a \leftsh (\resp unitary \leftsh) multi type for some 	$n \in \nat$.
		We can build the following derivation
		\begin{equation*}
		\tderiv = 
		\begin{prooftree}
		\infer0[\footnotesize$\Ax$]{\var \hastype \mset{\ltype_1} \vdash \var \hastype \ltype_1}
		\hypo{\overset{n \in \nat}{\ldots}}
		\infer0[\footnotesize$\Ax$]{\var \hastype \mset{\ltype_n} \vdash \var \hastype \ltype_n}
		\infer3[\footnotesize$\ruleManyVar$]{\var \hastype \mtype \vdash \var \hastype \mtype}
		\end{prooftree}
		\end{equation*}
		where $\typctx = \var \hastype \mtype$ is a \leftsh (\resp~unitary \leftsh) type context.
		Moreover, if $\mtype = \mset{\ground}$ then $\mtype$ is both a unitary \rightsh and a unitary \leftsh multi type, hence $\typctx$ is a unitary \leftsh type context and so also \refpoint{typability-normal-fireball} holds.
		
		\item \emph{Inert application}, \ie $\sitm = \sitmtwo \sfire$.
		Let $\mtype$ be a \leftsh (\resp unitary \leftsh) multi type.
		By \ih there is a derivation $\concl{\tderivtwo}{\typctxtwo}{\sfire}{\mtypetwo}$ where $\typctxtwo$ is unitary \leftsh and $\mtypetwo$ is unitary \rightsh,
		thus $\mset{\larrow{\mtypetwo}{\mtype}}$ is \leftsh (\resp unitary \leftsh).
		By \ih there is a derivation $\concl{\tderivthree}{\typctxthree}{\sitm}{\mset{\larrow{\mtypetwo}{\mtype}}}$ where 
		$\typctxthree$ is \leftsh (\resp unitary \leftsh).
		We have the following derivation:
		\begin{equation*}
		\tderiv = 
		\begin{prooftree}
		\hypo{}
		\ellipsis{$\tderivthree$}{\typctxthree \vdash \sitm \hastype \mset{\larrow{\mtypetwo}{\mtype}}}
		\hypo{}
		\ellipsis{$\tderivtwo$}{\typctxtwo \vdash \sfire \hastype \mtypetwo}
		\infer2[\footnotesize$\ruleAp$]{\typctx \vdash \sitm\sfire \hastype \mtype}
		\end{prooftree}
		\end{equation*}
		where $\typctx = \typctxtwo \uplus \typctxthree$ is \leftsh (\resp unitary \leftsh) by \Cref{rmk:merge-split-coshrinking}.
		Moreover, if $\mtype = \mset{\ground}$ then $\mtype$ is both a unitary \rightsh and a unitary \leftsh multi type, hence $\typctx$ is a unitary \leftsh type context (by \ih) and so also \refpoint{typability-normal-fireball} holds.
		
		\item \emph{Explicit substitution on inert}, \ie $\sitm = \sitmtwo \esub{\var}{\sitmthree}$.
		Let $\mtype$ be a \leftsh (\resp unitary \leftsh) multi type.
		By \ih, there is a derivation $\concl{\tderivtwo}{\typctxtwo, \var \hastype \mtypetwo}{\sitmtwo}{\mtype}$ where $\typctxtwo, \var \hastype \mtypetwo$ is a \leftsh (\resp unitary \leftsh) type context;  
		in particular, $\mtypetwo$ is a \leftsh (\resp unitary \leftsh) multi type.
		By \ih, there is a derivation $\concl{\tderivthree}{\typctxthree}{\sitmthree}{\mtypetwo}$ where $\typctxthree$ is a \leftsh (\resp unitary \leftsh) context and $\mtypetwo$ is a \leftsh (resp.~unitary \leftsh) multi type.
		We have the following derivation:
		\begin{equation*}
		\tderiv = 
		\begin{prooftree}
		\hypo{}
		\ellipsis{$\tderivtwo$}{\typctxtwo, \var \hastype \mtypetwo \vdash \sitmtwo \hastype \mtype}
		\hypo{}
		\ellipsis{$\tderivthree$}{\typctxthree \vdash \sitmthree \hastype \mtypetwo}
		\infer2[\footnotesize$\ruleAp$]{\typctx \vdash \sitmtwo \esub{\var}{\sitmthree} \hastype \mtype}
		\end{prooftree}
		\end{equation*}
		where $\typctx = \typctxtwo \uplus \typctxthree$ is \leftsh (resp.~unitary \leftsh) by 	\Cref{rmk:merge-split-coshrinking}.
		Moreover, if $\mtype = \mset{\ground}$ then $\mtype$ is both a unitary \rightsh and a unitary \leftsh multi type, hence $\typctx$ is a unitary \leftsh type context (by \ih of \refpoint{typability-normal-inert}) and so also \refpoint{typability-normal-fireball}	holds.
		
		\item \emph{Abstraction}, \ie $\sfire = \la{\var}{\sfiretwo}$. 
		By \ih, there exists a derivation $\concl{\tderivtwo}{\typctx, \var \hastype \mtypethree}{\sfire}{\mtypetwo}$ where 
		$\typctx, \var \hastype \mtypethree$ is a unitary \leftsh type context and $\mtypetwo$ is a unitary \rightsh multi 		type.
		Hence, $\mtype = \mset{\larrow{\mtypethree}{\mtypetwo}}$ is a unitary \leftsh multi type.
		We have the following derivation:
		\begin{equation*}
		\tderiv = 
		\begin{prooftree}
		\hypo{}
		\ellipsis{$\tderivtwo$}{\typctx, \var \hastype \mtypethree \vdash \sfire \hastype \mtypetwo}
		\infer1[\footnotesize$\lambda$]{\typctx \vdash \la{\var}\sfire \hastype \larrow{\mtypethree}{\mtypetwo}}
		\infer1[\footnotesize$\ruleMany$]{\typctx \vdash \la{\var}\sfire \hastype \mtype}
		\end{prooftree}
		\end{equation*}
		where $\typctx$ is a unitary \leftsh type context.
		
		\item \emph{Explicit substitution on fireball}, \ie $\sfire = \sfiretwo \esub{\var}{\sitm}$.
		By \ih, there is a derivation $\concl{\tderivtwo}{\typctxtwo, \var \hastype \mtypetwo}{\sfiretwo}{\mtype}$ where 
		$\mtype$ is a unitary \rightsh multi type and $\typctxtwo, \var \hastype \mtypetwo$ is a unitary \leftsh type context;  
		in particular, $\mtypetwo$ is a unitary \leftsh multi type.
		By \ih, there is a derivation $\concl{\tderivthree}{\typctxthree}{\sitm}{\mtypetwo}$ where $\typctxthree$ is a unitary \leftsh type context and $\mtypetwo$ is a unitary \leftsh multi type.
		We have the following derivation:
		\begin{equation*}
		\tderiv = 
		\begin{prooftree}
		\hypo{}
		\ellipsis{$\tderivtwo$}{\typctxtwo, \var \hastype \mtypetwo \vdash \sfiretwo \hastype \mtype}
		\hypo{}
		\ellipsis{$\tderivthree$}{\typctxthree \vdash \sitm \hastype \mtypetwo}
		\infer2[\footnotesize$\ruleAp$]{\typctx \vdash \sfiretwo \esub{\var}{\sitm} \hastype \mtype}
		\end{prooftree}
		\end{equation*}
		where $\typctx = \typctxtwo \uplus \typctxthree$ is a unitary \leftsh by 
		\Cref{rmk:merge-split-coshrinking}.
		\qedhere
	\end{itemize}	
\end{proof}

As for quantitative subject reduction, we split the proof of quantitative subject expansion in two cases:
\begin{enumerate}
	\item \emph{Open:} going back though a $\tovsubo$ step any type derivation strictly increases its size 
	(\Cref{prop:weak-subject-expansion});
	\item \emph{Strong:} going back though a $\tovsubs$ step the derivation strictly increases its size if it is 
	shrinking (\Cref{prop:shrinking-subject-expansion}). 
\end{enumerate}

The open case (\Cref{prop:weak-subject-expansion}) represents the base case for the strong case 
(\Cref{prop:shrinking-subject-expansion}).  
To prove the base case, we need an anti-substitution lemma (\Cref{l:anti-substitution}) to handle the exponential step, which in turn relies on \Cref{l:typing-value-complete}.


\begin{proposition}[Shrinking quantitative subject expansion]
	\label{propappendix:shrinking-subject-expansion}
	\NoteState{prop:shrinking-subject-expansion}
	Let $\typctx$ be a \leftsh (\resp unitary \leftsh) type context 
	and $\mtype$ be a multi type.
	Suppose that $\concl{\tderiv'}{\typctx}{\tm'}{\mtype}$ and if $\tm'$ is an \valES, then $\mtype$ is \rightsh (\resp unitary \rightsh).
	\begin{enumerate}
		\item \emph{Multiplicative step:} If $\tm \toessm \tm'$ then there is a derivation 
		$\concl{\tderiv}{\typctx}{\tm}{\mtype}$ such that
		$\sizem{\tderiv'} \leq \sizem{\tderiv}-2$ and $\size{\tderiv'} < \size{\tderiv}$
		(\resp $\sizem{\tderiv'} = \sizem{\tderiv}-2$ and $\size{\tderiv'} = \size{\tderiv} - 1$);
		\item \emph{Exponential step:} if $\tm \toesse \tm'$ then there is a derivation 
		$\concl{\tderiv}{\typctx}{\tm}{\mtype}$ such that
		$\sizem{\tderiv'} = \sizem{\tderiv}$ and $\size{\tderiv'} < \size{\tderiv}$.
	\end{enumerate}
\end{proposition}

\begin{proof}
	First we prove the unitary version of the statement (\ie under the hypothesis that $\typctx$ is \leftsh and if $\tm'$ is a \valES then $\mtype$ is unitary \rightsh).
	The proof is by induction on 
	the evaluation context $\strongctx$ in the step $\tm = \strongctxp{\tmtwo} \toess \strongctxp{\tmtwo'} = \tm'$ with 
	$\tmtwo \tomo \tmtwo'$ or $\tmtwo \toeo \tmtwo'$. 
	Cases for $\strongctx$:
	\begin{itemize}
		\item \emph{Hole context}, \ie{} $\strongctx = \ctxhole$ and $\tm \Rew{\wsym a} \tm'$ with $a \in \{\msym, \esym\}$.
		According to open subject expansion (\Cref{prop:weak-subject-expansion}),
		\begin{itemize}
			\item if $\tm \towm \tm'$ then there is a derivation $\concl{\tderiv}{\typctx}{\tm}{\mtype}$ such that
			$\sizem{\tderiv'} = \sizem{\tderiv}-2$ and $\size{\tderiv'} = \size{\tderiv} - 1$;
			\item if $\tm \towe \tm'$ then there is a derivation $\concl{\tderiv}{\typctx}{\tm}{\mtype}$ such that
			$\sizem{\tderiv'} = \sizem{\tderiv}$ and $\size{\tderiv'} < \size{\tderiv}$.
		\end{itemize}
		Note that in this case the hypotheses that $\typctx$ is unitary \leftsh and $\mtype$ is unitary \rightsh if $\tm'$ is an \valES are not used.
		
		\item \emph{Abstraction}, \ie $\strongctx = \la{\var}{\strongctxtwo}$. 
		Then, $\tm = \strongctxp{\tmtwo} = \la{\var}{\strongctxtwop{\tmtwo}} \Rew{\esssym a} 
		\la{\var}{\strongctxtwop{\tmtwo'}} = \strongctxp{\tmtwo'} = \tm'$ with $\tmtwo \Rew{\wsym a} \tmtwo'$ and $a \in 
		\{\msym, \esym\}$.
		Since $\tm'$ is an \valES, $\mtype$ is a unitary \rightsh multi type by hypothesis and hence it has the form $\mtype = \mset{\larrow{\mtypetwo_1}{\mtypetwo_2}}$ where $\mtypetwo_1$ is unitary \leftsh and $\mtypetwo_2$ is 	unitary \rightsh.
		Therefore, the derivation $\tderiv'$ is necessarily
		\begin{equation*}
		\tderiv' = 
		\begin{prooftree}
		\hypo{}
		\ellipsis{$\tderivtwo'$}{\typctx, \var \hastype \mtypetwo_1 \vdash \strongctxtwop{\tmtwo'} \hastype \mtypetwo_2}
		\infer1[\footnotesize$\lambda$]{\typctx \vdash \la{\var}\strongctxtwop{\tmtwo'} \hastype 
			\larrow{\mtypetwo_1}{\mtypetwo_2}}
		\infer1[\footnotesize$\ruleManyVal$]{\typctx \vdash \la{\var}\strongctxtwop{\tmtwo'} \hastype 
			\mset{\larrow{\mtypetwo_1}{\mtypetwo_2}}}
		\end{prooftree}
		\end{equation*}
		By \ih (as $\typctx, \var \hastype \mtypetwo_1$ is a unitary \leftsh type context and $\mtypetwo_{2}$ is a unitary \rightsh multi type), there is a derivation $\concl{\tderivtwo}{\typctx, \var \hastype 
			\mtypetwo_1}{\strongctxtwop{\tmtwo}}{\mtypetwo_2}$ with: 
		\begin{enumerate}
			\item $\sizem{\tderivtwo'} = \sizem{\tderivtwo} - 2$ and $\size{\tderivtwo'} = \size{\tderivtwo} - 1$ if $\tmtwo \tomo \tmtwo'$; 
			\item $\size{\tderivtwo'} = \size{\tderivtwo}$ and $\size{\tderivtwo'} < \size{\tderivtwo}$ if $\tmtwo \toeo \tmtwo'$.
		\end{enumerate}
		We can then build the derivation 
		\begin{equation*}
		\tderiv = 
		\begin{prooftree}
		\hypo{}
		\ellipsis{$\tderivtwo$}{\typctx, \var \hastype \mtypetwo_1 \vdash \strongctxtwop{\tmtwo} \hastype \mtypetwo_2}
		\infer1[\footnotesize$\lambda$]{\typctx \vdash \la{\var}\strongctxtwop{\tmtwo} \hastype 
			\larrow{\mtypetwo_1}{\mtypetwo_2}}
		\infer1[\footnotesize$\ruleManyVal$]{\typctx \vdash \la{\var}\strongctxtwop{\tmtwo} \hastype 
			\mset{\larrow{\mtypetwo_1}{\mtypetwo_2}}}
		\end{prooftree}
		\end{equation*}
		where
		\begin{enumerate}
			\item $\sizem{\tderiv'} = \sizem{\tderivtwo'} +1 = \sizem{\tderivtwo} - 2 +1 = \sizem{\tderiv} - 2$ and 
			$\size{\tderiv'} = \size{\tderivtwo'} +1 = \size{\tderivtwo} - 1 +1 = \size{\tderiv} - 1$ if $\tmtwo \tomo \tmtwo'$ (\ie if $\tm \toms \tm'$); 
			\item $\sizem{\tderiv'} = \sizem{\tderivtwo'} +1 = \sizem{\tderivtwo} + 1 = \sizem{\tderiv}$ and $\size{\tderiv'} 
			= \size{\tderivtwo'} +1 < \size{\tderivtwo} + 1 = \size{\tderiv}$ if $\tmtwo \toeo \tmtwo'$  (\ie if $\tm \toes \tm'$).
		\end{enumerate}
		
		\item \emph{Explicit substitution of rigid context}, \ie $\strongctx = \tmthree\esub{\var}{\ictx}$. 
		So, $\tm = \strongctxp{\tmtwo} = \tmthree\esub{\var}{\ictxp{\tmtwo}} \Rew{\esssym a} 
		\tmthree\esub{\var}{\ictxp{\tmtwo'}} = \strongctxp{\tmtwo'} = \tm'$ with $\tmtwo \Rew{\wsym a} \tmtwo'$ and $a \in \{\msym, \esym\}$.
		Then, necessarily
		\begin{equation*}
			\tderiv' = 
			\begin{prooftree}
			\hypo{}
			\ellipsis{$\tderivtwo$}{\typctxtwo, \var \hastype \mtypetwo \vdash \tmthree \hastype \mtype}
			\hypo{}
			\ellipsis{$\tderivthree'$}{\typctxthree \vdash \ictxp{\tmtwo'} \hastype \mtypetwo}
			\infer2[\footnotesize$\Es$]{\typctxtwo \uplus \typctxthree \vdash \tmthree \esub{\var}{\ictxp{\tmtwo'} } \hastype 
				\mtype}
			\end{prooftree}
		\end{equation*}
		where $\typctx = \typctxtwo \uplus \typctxthree$, with $\typctxtwo$ and $\typctxthree$ unitary \leftsh by \Cref{rmk:merge-split-coshrinking}.
		By \ih applied to $\tderivthree'$ (as $\ictxp{\tmtwo'}$ is not a \valES), there is a derivation $\concl{\tderivthree}{\typctxthree}{\ictxp{\tmtwo}}{\mtypetwo}$ with: 
		\begin{enumerate}
			\item $\sizem{\tderivthree'} = \sizem{\tderivthree} - 2$ and $\size{\tderivthree'} = \size{\tderivthree} - 1$ if 	$\tmtwo \tomo \tmtwo'$; 
			\item $\size{\tderivthree'} = \size{\tderivthree}$ and $\size{\tderivthree'} < \size{\tderivthree}$ if $\tmtwo 	\toeo \tmtwo'$.
		\end{enumerate}
		We can then build the derivation
		\begin{equation*}
		\tderiv = 
		\begin{prooftree}
		\hypo{}
		\ellipsis{$\tderivtwo$}{\typctxtwo, \var \hastype \mtypetwo \vdash \tmthree \hastype \mtype}
		\hypo{}
		\ellipsis{$\tderivthree$}{\typctxthree \vdash \ictxp{\tmtwo} \hastype \mtypetwo}
		\infer2[\footnotesize$\Es$]{\typctxtwo \uplus \typctxthree \vdash \tmthree \esub{\var}{\ictxp{\tmtwo}} \hastype 
			\mtype}
		\end{prooftree}
		\end{equation*}
		where $\typctx = \typctxtwo \uplus \typctxthree$ and
		\begin{enumerate}
			\item $\sizem{\tderiv'} = \sizem{\tderivtwo} + \sizem{\tderivthree'} = \sizem{\tderivtwo} + \sizem{\tderivthree} 
			- 2 = \sizem{\tderiv} - 2$ and
			$\size{\tderiv'} = \size{\tderivtwo} + \size{\tderivthree'} +1 = \size{\tderivtwo} + \size{\tderivthree} -1 +1 = \size{\tderiv} - 1$ if $\tmtwo \tomo \tmtwo'$ (\ie if $\tm \toms \tm'$); 
			\item $\sizem{\tderiv'} = \sizem{\tderivtwo} + \sizem{\tderivthree'} = \sizem{\tderivtwo} + \sizem{\tderivthree} = \sizem{\tderiv}$ and $\size{\tderiv'} = \size{\tderivtwo} + \size{\tderivthree'} +1 < \size{\tderivtwo} + \size{\tderivthree} +1 = \size{\tderiv}$ if $\tmtwo \toeo \tmtwo'$ (\ie if $\tm \toes \tm'$).
		\end{enumerate}
		
		\item \emph{Explicit substitution on strong}, \ie $\strongctx =  \strongctxtwo\esub{\var}{\ptm}$. 
		So, $\tm = \strongctxp{\tmtwo} = \strongctxtwop{\tmtwo} \esub{\var}{\ptm} \Rew{\esssym a} 
		\strongctxtwop{\tmtwo'}\esub{\var}{\ptm} = \strongctxp{\tmtwo'} = \tm'$ with $\tmtwo \Rew{\wsym a} \tmtwo'$ and $a \in \{\msym, \esym\}$.
		Then, necessarily
		\begin{equation*}
		\tderiv' = 
		\begin{prooftree}
		\hypo{}
		\ellipsis{$\tderivtwo'$}{\typctxtwo, \var \hastype \mtypetwo \vdash \strongctxtwop{\tmtwo'} \hastype \mtype}
		\hypo{}
		\ellipsis{$\tderivthree$}{\typctxthree \vdash \ptm \hastype \mtypetwo}
		\infer2[\footnotesize$\Es$]{\typctxtwo \uplus \typctxthree \vdash \strongctxtwop{\tmtwo'} \esub{\var}{\ptm} 
			\hastype \mtype}
		\end{prooftree}
		\end{equation*}
		with $\typctx = \typctxtwo \uplus \typctxthree$, and $\typctxtwo$ and $\typctxthree$ unitary \leftsh by 	\refrmk{merge-split-coshrinking}.
		By spreading of \leftsh shrinkingness (\Cref{l:spread-shrinking}, as $\ptm$ is a \pointed term), $\mtypetwo$ is unitary \leftsh.
		Note that $\strongctxtwop{\tmtwo'}\esub{\var}{\ptm}$ is an \valES iff $\strongctxtwop{\tmtwo}$ is an \valES.
		So, the \ih can be applied to $\tderivtwo'$ (since $\typctxtwo, \var \hastype \mtypetwo$ is a unitary \leftsh type context) and hence there is a derivation $\concl{\tderivtwo}{\typctxtwo, \var \hastype 
			\mtypetwo}{\strongctxtwop{\tmtwo}}{\mtype}$ with: 
		\begin{enumerate}
			\item $\sizem{\tderivtwo'} = \sizem{\tderivtwo} - 2$ and $\size{\tderivtwo'} = \size{\tderivtwo} - 1$ if $\tmtwo 
			\tomo \tmtwo'$; 
			\item $\sizem{\tderivtwo'} = \sizem{\tderivtwo}$ and $\size{\tderivtwo'} < \size{\tderivtwo}$ if $\tmtwo \toeo 
			\tmtwo'$.
		\end{enumerate}
		We can then build the derivation 
		\begin{equation*}
		\tderiv = 
		\begin{prooftree}
		\hypo{}
		\ellipsis{$\tderivtwo$}{\typctxtwo, \var \hastype \mtypetwo \vdash \strongctxtwop{\tmtwo} \hastype \mtype}
		\hypo{}
		\ellipsis{$\tderivthree$}{\typctxthree \vdash \ptm \hastype \mtypetwo}
		\infer2[\footnotesize$\Es$]{\typctxtwo \uplus \typctxthree \vdash \strongctxtwop{\tmtwo} \esub{\var}{\ptm} \hastype 
			\mtype}
		\end{prooftree}
		\end{equation*}
		where $\typctx = \typctxtwo \uplus \typctxthree$ and
		\begin{enumerate}
			\item $\sizem{\tderiv'} = \sizem{\tderivtwo'} + \sizem{\tderivthree} = \sizem{\tderivtwo} + \sizem{\tderivthree} 
			- 2 = \sizem{\tderiv} - 2$ and $\size{\tderiv'} = \size{\tderivtwo'} + \size{\tderivthree} +1 = \size{\tderivtwo} - 1 + \size{\tderivthree} +1 = \size{\tderiv} - 1$ if $\tmtwo \tomo \tmtwo'$ (\ie if $\tm \toms \tm'$); 
			\item $\sizem{\tderiv'} = \sizem{\tderivtwo'} + \sizem{\tderivthree} = \sizem{\tderivtwo} + \sizem{\tderivthree} 
			= \sizem{\tderiv}$ and $\size{\tderiv'} = \size{\tderivtwo'} + \size{\tderivthree} +1 < \size{\tderivtwo} + 
			\size{\tderivthree} +1 = \size{\tderiv}$ if $\tmtwo \toeo \tmtwo'$ (\ie if $\tm \toes \tm'$).
		\end{enumerate}
		
		\item \emph{Application to strong}, \ie $\strongctx = \ptm \strongctxtwo$. 
		Then, $\tm = \strongctxp{\tmtwo} = \ptm \strongctxtwop{\tmtwo} \Rew{\esssym a} \ptm \strongctxtwop{\tmtwo'} = \strongctxp{\tmtwo'} = \tm'$ with $\tmtwo \Rew{\wsym a} \tmtwo'$ and $a \in \{\msym, \esym\}$.
		The derivation $\tderiv'$ is necessarily
		\begin{equation*}
		\tderiv' = 
		\begin{prooftree}
		\hypo{}
		\ellipsis{$\tderivtwo$}{\typctxtwo \vdash \ptm \hastype \mset{\larrow{\mtypetwo}{\mtype}}}
		\hypo{}
		\ellipsis{$\tderivthree'$}{\typctxthree \vdash \strongctxtwop{\tmtwo'} \hastype \mtypetwo}
		\infer2[\footnotesize$\ruleAp$]{\typctxtwo \uplus \typctxthree \vdash \ptm \strongctxtwop{\tmtwo'} \hastype \mtype}
		\end{prooftree}
		\end{equation*}
		where $\typctx = \typctxtwo \uplus \typctxthree$, with $\typctxtwo$ and $\typctxthree$ unitary \leftsh by \Cref{rmk:merge-split-coshrinking}.
		According to spreading of \leftsh shrinkingness (\Cref{l:spread-shrinking}, as $\ptm$ is a \pointed term), $\mset{\larrow{\mtypetwo}{\mtype}}$ is unitary \leftsh and so $\mtypetwo$ is  unitary \rightsh.
		We can then apply the \ih to $\tderivthree'$, thus there is a derivation 
		$\concl{\tderivthree}{\typctxthree}{\strongctxtwop{\tmtwo}}{\mtypetwo}$ with: 
		\begin{enumerate}
			\item $\sizem{\tderivthree'} = \sizem{\tderivthree} - 2$ and $\size{\tderivthree'} = \size{\tderivthree} - 1$ if $\tmtwo \tomo \tmtwo'$; 
			\item $\sizem{\tderivthree'} = \sizem{\tderivthree}$ and $\size{\tderivthree'} < \size{\tderivthree}$ if $\tmtwo 
			\toeo \tmtwo'$.
		\end{enumerate}
		We can then build the derivation 
		\begin{equation*}
		\tderiv = 
		\begin{prooftree}
		\hypo{}
		\ellipsis{$\tderivtwo$}{\typctxtwo \vdash \ptm \hastype \mset{\larrow{\mtypetwo}{\mtype}}}
		\hypo{}
		\ellipsis{$\tderivthree$}{\typctxthree \vdash \strongctxtwop{\tmtwo} \hastype \mtypetwo}
		\infer2[\footnotesize$\ruleAp$]{\typctxtwo \uplus \typctxthree \vdash \ptm \strongctxtwop{\tmtwo} \hastype \mtype}
		\end{prooftree}
		\end{equation*}
		where $\typctx = \typctxtwo \uplus \typctxthree$ and
		\begin{enumerate}
			\item $\sizem{\tderiv'} = \sizem{\tderivtwo} + \sizem{\tderivthree'} +1 = \sizem{\tderivtwo} + \sizem{\tderivthree} - 2 +1 = \sizem{\tderiv} - 2$ and $\size{\tderiv'} = \size{\tderivtwo} + \size{\tderivthree'} +1 = \size{\tderivtwo} + \size{\tderivthree} -1 +1 = \size{\tderiv} - 1$ if $\tmtwo \tomo \tmtwo'$ (\ie if $\tm \toms \tm'$); 
			\item $\sizem{\tderiv'} = \sizem{\tderivtwo} + \sizem{\tderivthree'} +1 = \sizem{\tderivtwo} + 
			\sizem{\tderivthree} +1 = \sizem{\tderiv}$ and $\size{\tderiv'} = \size{\tderivtwo} + \size{\tderivthree'} +1 < 
			\size{\tderivtwo} + \size{\tderivthree} +1 = \size{\tderiv}$ if $\tmtwo \toeo \tmtwo'$ (\ie if $\tm \toes \tm'$).
		\end{enumerate}
		
		\item \emph{Application of inert}, \ie $\strongctx = \ictx\tmthree$. 
		Then, $\tm = \strongctxp{\tmtwo} = \ictxp{\tmtwo} \tmthree \Rew{\esssym a} \ictxp{\tmtwo'} \tmthree = 	\strongctxp{\tmtwo'} = \tm'$ with $\tmtwo \Rew{\wsym a} \tmtwo'$ and $a \in \{\msym, \esym\}$.
		The derivation $\tderiv'$ is necessarily
		\begin{equation*}
		\tderiv' = 
		\begin{prooftree}
		\hypo{}
		\ellipsis{$\tderivtwo'$}{\typctxtwo \vdash \ictxp{\tmtwo'} \hastype \mset{\larrow{\mtypetwo}{\mtype}}}
		\hypo{}
		\ellipsis{$\tderivthree$}{\typctxthree \vdash \tmthree \hastype \mtypetwo}
		\infer2[\footnotesize$\ruleAp$]{\typctxtwo \uplus \typctxthree \vdash \ictxp{\tmtwo'} \tmthree \hastype \mtype}
		\end{prooftree}
		\end{equation*}
		where $\typctx = \typctxtwo \uplus \typctxthree$, with $\typctxtwo$ and $\typctxthree$ unitary \leftsh by \Cref{rmk:merge-split-coshrinking}.
		By \ih (as $\ictxp{\tmtwo'}$ is not an \valES), there is a derivation 
		$\concl{\tderivtwo}{\typctxtwo}{\ictxp{\tmtwo}}{\mset{\larrow{\mtypetwo}{\mtype}}}$ with: 
		\begin{enumerate}
			\item $\sizem{\tderivtwo'} = \sizem{\tderivtwo} - 2$ and $\size{\tderivtwo'} = \size{\tderivtwo} - 1$ if $\tmtwo \tomo \tmtwo'$; 
			\item $\sizem{\tderivtwo'} = \sizem{\tderivtwo}$ and $\size{\tderivtwo'} < \size{\tderivtwo}$ if $\tmtwo \toeo \tmtwo'$.
		\end{enumerate}
		We can then build the derivation 
		\begin{equation*}
		\tderiv = 
		\begin{prooftree}
		\hypo{}
		\ellipsis{$\tderivtwo$}{\typctxtwo \vdash \ictxp{\tmtwo} \hastype \mset{\larrow{\mtypetwo}{\mtype}}}
		\hypo{}
		\ellipsis{$\tderivthree$}{\typctxthree \vdash \tmthree \hastype \mtypetwo}
		\infer2[\footnotesize$\ruleAp$]{\typctxtwo \uplus \typctxthree \vdash \ictxp{\tmtwo} \tmthree \hastype \mtype}
		\end{prooftree}
		\end{equation*}
		where $\typctx = \typctxtwo \uplus \typctxthree$ and
		\begin{enumerate}
			\item $\sizem{\tderiv'} = \sizem{\tderivtwo'} + \sizem{\tderivthree} +1 = \sizem{\tderivtwo} -2 + 
			\sizem{\tderivthree} +1 = \sizem{\tderiv} - 2$ and $\size{\tderiv'} = \size{\tderivtwo'} + \size{\tderivthree} +1 = 	\size{\tderivtwo} -1 + \size{\tderivthree} +1 = \size{\tderiv} - 1$ if $\tmtwo \tomo \tmtwo'$ (\ie if $\tm \toms \tm'$); 
			\item $\size{\tderiv'} = \size{\tderivtwo'} + \size{\tderivthree} +1 = \size{\tderivtwo} + \size{\tderivthree} +1 
			= \size{\tderiv}$ and $\size{\tderiv'} = \size{\tderivtwo'} + \size{\tderivthree} +1 < \size{\tderivtwo} + 
			\size{\tderivthree} +1 = \size{\tderiv}$ if $\tmtwo \toeo \tmtwo'$ (\ie if $\tm \toes \tm'$).
		\end{enumerate}
		
		\item \emph{Explicit substitution on inert}, \ie $\strongctx = \ictx\esub{\var}{\ptm}$. 
		Then, $\tm = \strongctxp{\tmtwo} = \ictxp{\tmtwo} \esub{\var}{\ptm} \Rew{\esssym a} \ictxp{\tmtwo'}\esub{\var}{\ptm} 
		= \strongctxp{\tmtwo'} = \tm'$ with $\tmtwo \Rew{\wsym a} \tmtwo'$ and $a \in \{\msym, \esym\}$.
		So, necessarily
		\begin{equation*}
		\tderiv' = 
		\begin{prooftree}
		\hypo{}
		\ellipsis{$\tderivtwo'$}{\typctxtwo, \var \hastype \mtypetwo \vdash \ictxp{\tmtwo'} \hastype \mtype}
		\hypo{}
		\ellipsis{$\tderivthree$}{\typctxthree \vdash \ptm \hastype \mtypetwo}
		\infer2[\footnotesize$\Es$]{\typctxtwo \uplus \typctxthree \vdash \ictxp{\tmtwo'} \esub{\var}{\ptm} \hastype \mtype}
		\end{prooftree}
		\end{equation*}
		with $\typctx = \typctxtwo \uplus \typctxthree$, and $\typctxtwo$ and $\typctxthree$ unitary \leftsh by \Cref{rmk:merge-split-coshrinking}.
		By spreading of \leftsh shrinkingness (\reflemma{spread-shrinking}, as $\ptm$ is a \pointed term), $\mtypetwo$ is a unitary \leftsh multi type.
		Thus, the \ih can be applied to $\tderivtwo'$ (since $\typctxtwo, \var \hastype \mtypetwo$ is a unitary \leftsh type	context and $\ictxp{\tmtwo}$ is not a \valES) and so there is a derivation $\concl{\tderivtwo}{\typctxtwo, \var \hastype 
			\mtypetwo}{\ictxp{\tmtwo}}{\mtype}$ with: 
		\begin{enumerate}
			\item $\sizem{\tderivtwo'} = \sizem{\tderivtwo} - 2$ and $\size{\tderivtwo'} = \size{\tderivtwo} - 1$ if $\tmtwo \tomo \tmtwo'$; 
			\item $\sizem{\tderivtwo'} = \sizem{\tderivtwo}$ and $\size{\tderivtwo'} < \size{\tderivtwo}$ if $\tmtwo \toeo \tmtwo'$.
		\end{enumerate}
		We can then build the derivation 
		\begin{equation*}
		\tderiv = 
		\begin{prooftree}
		\hypo{}
		\ellipsis{$\tderivtwo$}{\typctxtwo, \var \hastype \mtypetwo \vdash \ictxp{\tmtwo} \hastype \mtype}
		\hypo{}
		\ellipsis{$\tderivthree$}{\typctxthree \vdash \ptm \hastype \mtypetwo}
		\infer2[\footnotesize$\Es$]{\typctxtwo \uplus \typctxthree \vdash \ictxp{\tmtwo} \esub{\var}{\ptm} \hastype \mtype}
		\end{prooftree}
		\end{equation*}
		where $\typctx = \typctxtwo \uplus \typctxthree$ and
		\begin{enumerate}
			\item $\sizem{\tderiv'} = \sizem{\tderivtwo'} + \sizem{\tderivthree} = \sizem{\tderivtwo} -2 + \sizem{\tderivthree} = \sizem{\tderiv} - 2$ 
			and $\size{\tderiv'} = \size{\tderivtwo'} + \size{\tderivthree} +1 = \size{\tderivtwo} -1 + \size{\tderivthree} +1 = \size{\tderiv} - 1$ if $\tmtwo \tomo \tmtwo'$ (\ie if $\tm \toms \tm'$); 
			\item $\sizem{\tderiv'} = \sizem{\tderivtwo'} + \sizem{\tderivthree} = \sizem{\tderivtwo} + \sizem{\tderivthree} = 
			\sizem{\tderiv}$ and $\size{\tderiv'} = \size{\tderivtwo'} + \size{\tderivthree} +1 = \size{\tderivtwo} + 
			\size{\tderivthree} +1 = \size{\tderiv}$ if $\tmtwo \toeo \tmtwo'$ (\ie if $\tm \toes \tm'$).
		\end{enumerate}
		
		\item \emph{Explicit substitution of inert on rigid}, \ie $\strongctx = \ptm\esub{\var}{\ictx}$. 
		So, $\tm = \strongctxp{\tmtwo} = \ptm\esub{\var}{\ictxp{\tmtwo}} \Rew{\esssym a} \ptm\esub{\var}{\ictxp{\tmtwo'}} = \strongctxp{\tmtwo'} = \tm'$ with $\tmtwo \Rew{\wsym a} \tmtwo'$ and $a \in \{\msym, \esym\}$.
		Thus, necessarily
		\begin{equation*}
		\tderiv' = 
		\begin{prooftree}
		\hypo{}
		\ellipsis{$\tderivtwo$}{\typctxtwo, \var \hastype \mtypetwo \vdash \ptm \hastype \mtype}
		\hypo{}
		\ellipsis{$\tderivthree'$}{\typctxthree \vdash \ictxp{\tmtwo'} \hastype \mtypetwo}
		\infer2[\footnotesize$\Es$]{\typctxtwo \uplus \typctxthree \vdash \ptm \esub{\var}{\ictxp{\tmtwo'} } \hastype 
			\mtype}
		\end{prooftree}
		\end{equation*}
		where $\typctx = \typctxtwo \uplus \typctxthree$, with $\typctxtwo$ and $\typctxthree$ unitary \leftsh type contexts by \Cref{rmk:merge-split-coshrinking}.
		By \ih applied to $\tderivthree'$ (as $\ictxp{\tmtwo}$ is not a \valES), there is a derivation 
		$\concl{\tderivthree}{\typctxthree}{\ictxp{\tmtwo}}{\mtypetwo}$ with: 
		\begin{enumerate}
			\item $\sizem{\tderivthree'} = \sizem{\tderivthree} - 2$ and $\size{\tderivthree'} = \size{\tderivthree} - 1$ if $\tmtwo \tomo \tmtwo'$; 
			\item $\sizem{\tderivthree'} = \sizem{\tderivthree}$ and $\size{\tderivthree'} < \size{\tderivthree}$ if $\tmtwo \toeo \tmtwo'$.
		\end{enumerate}
		We can then build the derivation 
		\begin{equation*}
		\tderiv = 
		\begin{prooftree}
		\hypo{}
		\ellipsis{$\tderivtwo$}{\typctxtwo, \var \hastype \mtypetwo \vdash \ptm \hastype \mtype}
		\hypo{}
		\ellipsis{$\tderivthree$}{\typctxthree \vdash \ictxp{\tmtwo} \hastype \mtypetwo}
		\infer2[\footnotesize$\Es$]{\typctxtwo \uplus \typctxthree \vdash \ptm \esub{\var}{\ictxp{\tmtwo}} \hastype \mtype}
		\end{prooftree}
		\end{equation*}
		where $\typctx = \typctxtwo \uplus \typctxthree$ and
		\begin{enumerate}
			\item $\sizem{\tderiv'} = \sizem{\tderivtwo} + \sizem{\tderivthree'} = \sizem{\tderivtwo} + \sizem{\tderivthree} - 2 = \sizem{\tderiv} - 2$ and $\size{\tderiv'} = 1 +\size{\tderivtwo} + \size{\tderivthree'} = 1+ \size{\tderivtwo} + \size{\tderivthree} - 1 = \size{\tderiv} - 1$ if $\tmtwo \tomo \tmtwo'$ (\ie if $\tm \toms \tm'$); 
			\item $\sizem{\tderiv'} = \sizem{\tderivtwo} + \sizem{\tderivthree'} = \sizem{\tderivtwo} + \sizem{\tderivthree} = 
			\sizem{\tderiv}$ and $\size{\tderiv'} = \size{\tderivtwo} + \size{\tderivthree'} +1 < \size{\tderivtwo} + 
			\size{\tderivthree} +1 = \size{\tderiv}$ if $\tmtwo \toeo \tmtwo'$ (\ie if $\tm \toes \tm'$).
		\end{enumerate}
		
	\end{itemize}

	This completes the proof for the unitary case.
	
	In the non-unitary statement (\ie under the weaker hypothesis that $\typctx$ is a \leftsh type context and if $\tm'$ is an \valES then $\mtype$ is \rightsh), the proof is analogous to the unitary statement, except for the \emph{Abstraction} case.
	Indeed, in the base case (\emph{Hole context}) unitarity does not play any role, and the other cases follow from the \ih analogously to the unitary statement.
	Let us see the only substantially different case: 
	\begin{itemize}
		\item \emph{Abstraction}, \ie $\strongctx = \la{\var}{\strongctxtwo}$. 
		So, $\tm = \strongctxp{\tmtwo} = \la{\var}{\strongctxtwop{\tmtwo}} \Rew{\esssym a} 
		\la{\var}{\strongctxtwop{\tmtwo'}} = \strongctxp{\tmtwo'} = \tm'$ with $\tmtwo \Rew{\wsym a} \tmtwo'$ and $a \in 
		\{\msym, \esym\}$.
		Since $\tm'$ is an \valES, $\mtype$ is a \rightsh multi type by hypothesis and hence it has the form $\mtype = 
		\mset{\larrow{\mtypethree_1}{\mtypetwo_1}, \dots, \larrow{\mtypethree_n}{\mtypetwo_n}}$ for some $n > 0$, where 
		$\mtypethree_i$ is \leftsh and $\mtypetwo_i$ is \rightsh for all $1 \leq i \leq n$.
		Thus, the derivation $\tderiv'$ is necessarily
		\begin{equation*}
			\tderiv' = 
			\begin{prooftree}[separation=1em]
			\hypo{}
			\ellipsis{$\tderivtwo_i'$}{\typctx_i, \var \hastype \mtypethree_i \vdash \strongctxtwop{\tmtwo} \hastype \mtypetwo_i}
			\infer1[\footnotesize$\lambda$]{\typctx_i \vdash \la{\var}\strongctxtwop{\tmtwo} \hastype 
				\larrow{\mtypethree_i}{\mtypetwo_i}}
			\delims{ \left( }{ \right)_{1 \leq i \leq n} }
			\infer1[\footnotesize$\ruleManyVal$]{ \bigmplus_{i=1}^n \typctx_{i} \vdash \la{\var}\strongctxtwop{\tmtwo} \hastype 
				\bigmplus_{i=1}^n \mset{\larrow{\mtypethree_i}{\mtypetwo_i}}}
			\end{prooftree}
		\end{equation*}
		For all $1 \leq i \leq n$, by \ih (as $\typctx_i, \var \hastype \mtypethree_i$ is a \leftsh type context and $\mtypetwo_i$ is a \rightsh multi type), there is a derivation $\concl{\tderivtwo_i}{\typctx_i, \var \hastype \mtypethree_i}{\strongctxtwop{\tmtwo'}}{\mtypetwo_i}$ with: 
		\begin{enumerate}
			\item $\sizem{\tderivtwo_i'} \leq \sizem{\tderivtwo_i} - 2$ and $\size{\tderivtwo_i'} \leq \size{\tderivtwo_i}-1$ if $\tmtwo \tomo \tmtwo'$; 
			\item $\sizem{\tderivtwo_i'} = \sizem{\tderivtwo_i}$ and $\size{\tderivtwo_i'} < \size{\tderivtwo_i}$ if $\tmtwo \toeo \tmtwo'$.
		\end{enumerate}
		We can then build the derivation 
		\begin{equation*}
		\tderiv = 
		\begin{prooftree}[separation=1em]
		\hypo{}
		\ellipsis{$\tderivtwo_i$}{\typctx_i, \var \hastype \mtypethree_i \vdash \strongctxtwop{\tmtwo'} \hastype \mtypetwo_i}
		\infer1[\footnotesize$\lambda$]{\typctx_i \vdash \la{\var}\strongctxtwop{\tmtwo'} \hastype 
			\larrow{\mtypethree_i}{\mtypetwo_i}}
		\delims{ \left( }{ \right)_{1 \leq i \leq n} }
		\infer1[\footnotesize$\ruleManyVal$]{ \bigmplus_{i=1}^n \typctx_{i} \vdash \la{\var}\strongctxtwop{\tmtwop} 
			\hastype  \bigmplus_{i=1}^n \mset{\larrow{\mtypethree_i}{\mtypetwo_i}}}
		\end{prooftree}
		\end{equation*}
		where
		\begin{enumerate}
			\item $\sizem{\tderiv'} = \sum_{i=1}^n(\sizem{\tderivtwo_i'} +1) \leq \sum_{i=1}^n(\sizem{\tderivtwo_i} + 1 - 2) 
			= \sizem{\tderiv} - 2n \leq \sizem{\tderiv} - 2$ (where the last inequality holds because $n >0$) and $\size{\tderiv'} 
			=  \sum_{i=1}^n(\size{\tderivtwo_i'} +1) =  \sum_{i=1}^n(\size{\tderivtwo_i} + 1 - 1) \leq \size{\tderiv} - n < 	\size{\tderiv}$ (where the last inequality holds because $n >0$) if $\tmtwo \tomo \tmtwo'$; 
			\item $\sizem{\tderiv'} = \sum_{i=1}^n(\sizem{\tderivtwo_i'} +1) = \sum_{i=1}^n(\sizem{\tderivtwo_i} + 1) = 
			\sizem{\tderiv}$ and $\size{\tderiv'} = \sum_{i=1}^n(\size{\tderivtwo_i'} +1) < \sum_{i=1}^n(\size{\tderivtwo_i} + 1) = 
			\size{\tderiv}$ if $\tmtwo \toeo \tmtwo'$ (the inequality holds because $n > 0$).
			\qedhere
		\end{enumerate}
	\end{itemize}	
\end{proof}

\begin{theorem}[Shrinking completeness]
	\label{thmappendix:completeness}
	\NoteState{thm:completeness}
	Let $\deriv \colon \tm \tovsubs^* \tm'$ be an evaluation with $\tm'$ normal. 
	Then there is a unitary shrinking derivation $\concl{\tderiv}{\typctx}{\tm}{\mtype}$ 
	such that \mbox{$2\sizem{\deriv} + \sizefu{\tm'} = \sizem{\tderiv}$}.
\end{theorem}

\begin{proof}
	By induction on the length $\size{\deriv}$ of $\deriv$.
	
	If $\size{\deriv} = 0$ then $\tm = \tm'$ is $\esssym$-normal and hence a \full fireball (\Cref{prop:external-properties}.\ref{p:external-properties-fullness}).
	According to typability of \full fireballs (\Cref{prop:typability-normal}), there is unitary shrinking derivation 
	$\concl{\tderiv}{\typctx}{\tm}{\mtype}$. 
	Therefore, $\sizem{\tderiv} = \sizefu{\tm'} = \sizefu{\tm'} + 2\sizem{\deriv}$ by \reflemma{size-strong-fireballs}.
	
	Otherwise, $\size{\deriv} > 0$ and $\deriv$ is the concatenation of a first step $\tm \tovsubs \tmtwo$ and an evaluation $\deriv' \colon \tmtwo \tovsubs^* \tm'$, with $\size{\deriv} = 1 + \size{\deriv'}$.
	By \ih, there is a unitary shrinking derivation $\concl{\tderivtwo}{\typctx}{\tmtwo}{\mtype}$ 
	such that $\sizem{\tderivtwo} = \sizefu{\tm'} + 2\sizem{\deriv'}$. 
	By shrinking subject expansion (\Cref{prop:shrinking-subject-expansion}), there is a derivation $\concl{\tderiv}{\typctx}{\tm}{\mtype}$~with 
	\begin{itemize}
		\item $\sizem{\tderiv} = \sizem{\tderivtwo} + 2 = \sizefu{\tm'} + 2\sizem{\deriv'} + 2 = \sizefu{\tm'} + 2\sizem{\deriv}$ 
		if $\tm \toessm \tmtwo$, since $\sizem{\deriv} = \sizem{\deriv'} + 1$;
		\item $\sizem{\tderiv} = \sizem{\tderivtwo} = \sizefu{\tm'} + 2\sizem{\deriv'} = \sizefu{\tm'} + 2\sizem{\deriv}$ if $\tm \toesse \tmtwo$, since $\sizem{\deriv} = \sizem{\deriv'}$.
		\qedhere
	\end{itemize} 
\end{proof}

\section{Proofs of \Cref{sect:normalization}}

\begin{theorem}[Operational and semantic characterizations of normalization]
	\label{thmappendix:normalization}
	\NoteState{thm:normalization}
	Let $\tm$ be a term. The following are equivalent:
	\begin{enumerate}
		\item there is a $\vsub$-normal form $\tmtwo$ such that $\tm \tovsub^* \tmtwo$;
		\item $\tm$ is $\esssym$-normalizing;
		\item $\semfull{\tm}_{\vec{\var}} \neq \emptyset$, where $\vec{\var} = (\var_1, \dots, \var_n)$ is suitable for $\tm$ and 
	\end{enumerate}
	\begin{equation*}
	\begin{split}
	\semfull{\tm}_{\vec{\var}} \defeq \{((\mtypetwo_1&,\dots, \mtypetwo_n),\mtype) \mid \exists 
	\, \concl{\tderiv}{\var_1 \hastype \mtypetwo_1, \dots, \var_n \hastype \mtypetwo_n}{\tm}{\mtype} 
	\\
	&\mbox{ such that $\mtypetwo_1, \dots, \mtypetwo_n$ are \leftsh and $\mtype$ is \rightsh} \} .
	\end{split}
	\end{equation*}
\end{theorem}

\begin{proof}
	\begin{description}
		\item
		[1)$\Rightarrow$3)] \ By \Cref{prop:properties-full-reduction}, $\tmtwo$ is a \full fireball, thus there is a shrinking derivation $\concl{\tderiv}{\var_1 \hastype \mtypetwo_1, \dots, \var_n \hastype \mtypetwo_n}{\tmtwo}{\mtype}$ by \Cref{prop:typability-normal}.
		Subject expansion (\refpropp{qual-subject}{expansion}) iterated along $\tm \tovsub^* 
		\tmtwo$ gives a shrinking derivation $\concl{\tderivtwo}{\var_1 \hastype \mtypetwo_1, \dots, \var_n \hastype \mtypetwo_n}{\tm}{\mtype}$ and so $\semfull{\tm}_{\vec{\var}} \neq \emptyset$.
		\item
		[3)$\Rightarrow$2)] \ Since $\semfull{\tm}_{\vec{\var}} \neq \emptyset$, there is a shrinking derivation $\derive{\tderiv}{\tm}$, so $\tm \tovsubs^* \tmthree$ with $\tmthree$ $\vsub$-normal (and hence $\esssym$-normal, as $\tovsubs \,\subseteq\, \tovsub$) by shrinking correctness (\Cref{thm:correctness}).
		\item
		[2)$\Rightarrow$1)] \ Trivial, as $\tovsubs \, \subseteq \, \tovsub$ and by \Cref{prop:external-properties}.\ref{p:external-properties-fullness}.
		\qedhere
	\end{description}
\end{proof}

\section{Proofs of \Cref{sect:semantic-bounds}}
\label{sect:size-types}

The following proposition has a stronger and more involved statement with respect to the one in the body of the paper in order to have the right \ih to prove it.

\begin{proposition}[Types bound the size of derivations for normal forms]
	\label{propappendix:types-bound-normal-derivations}
	\NoteState{prop:types-bound-normal-derivations}
	Let $\concl{\tderiv}{\typctx}{\tm}{\mtype}$ be a derivation.
	\begin{enumerate}
		\item\label{pappendix:types-bound-normal-derivations-inert}\emph{Inert:} if $\tm$ is a \full inert term, then $ 
\sizem{\tderiv} \leq \sizectx{\typctx} -\sizetyp{\mtype}$.
		\item\label{pappendix:types-bound-normal-derivations-fireball}\emph{Fireball:} if $\tm$ is a \full fireball then 
$\sizem{\tderiv} \leq \sizectx{\typctx} + \sizetyp{\mtype}$.
	\end{enumerate}
\end{proposition}

\begin{proof}
	By mutual induction on the definition of \full inert term and \full fireball.
	Note that the statement for the former is stronger than the one for the latter, hence for \full inert terms it is enough to prove Point \ref{pappendix:types-bound-normal-derivations-inert}.
	Cases: 
	\begin{itemize}
		\item \emph{Variable}, \ie~$\tm = \var$ (which is \full inert). So, necessarily, for some finite set of indices $I$,
		\begin{equation*}
		\tderiv = 
		\begin{prooftree}
		\infer0[\footnotesize$\ruleAx$]{\tyjp{}{\var}{\var \hastype \mset{\ltype_i}}{\ltype_i}}
		\delims{\left[}{\right]_{\iI}}
		\infer1[\footnotesize$\ruleManyVar$]{\tyjp{}{\var}{\var \hastype \mtype}{\mtype}}
		\end{prooftree}
		\end{equation*}
		where $\mtype = \mset{\ltype_i}_{\iI}$ and $\typctx = \var \hastype \mtype$.
		Since $\sizem{\tderiv} = 0$ and $\sizetyp{\mtype} = \sizectx{\typctx}$,
		then $\sizem{\tderiv} = 0 = \sizectx{\typctx}- \sizetyp{\mtype}$.

		\item \emph{Application}, \ie~$\tm = \sitm \sfire$ (which is \full inert). 
		Then necessarily
		\begin{equation*}
		\tderiv = 
		\begin{prooftree}
		\hypo{}
		\ellipsis{$\tderiv_{\sitm}$}{\typctxtwo \vdash \sitm \hastype \mult{\ty{\mtypetwo}{\mtype}}}
		\hypo{}
		\ellipsis{$\tderiv_{\sfire}$}{\typctxthree \vdash \sfire \hastype\mtypetwo}
		\infer2[\footnotesize$\ruleApp$]{\tyjp{}{\sitm \sfire}{\typctxtwo \mplus \typctxthree}{\mtype}}
		\end{prooftree}
		\end{equation*}
		where $\typctx = \typctxtwo \mplus \typctxthree$.
		By \ih, $\sizem{\tderiv_{\sitm}} \leq \sizectx{\typctxtwo} - \sizetyp{\mset{\larrow{\mtypetwo}{\mtype}}} = 
\sizectx{\typctxtwo} - \sizetyp{\mtypetwo} - \sizetyp{\mtype} -1$  and $\sizem{\tderiv_{\sfire}} \leq 
\sizectx{\typctxthree} + \sizetyp{\mtypetwo}$. 
		Therefore, 
		\[\begin{array}{rcl}
		\sizem{\tderiv} & = & \sizem{\tderiv_{\sitm}} + \sizem{\tderiv_{\sfire}} + 1 
		\\
		& \leq_{\ih} &\sizectx{\typctxtwo} - \sizetyp{\mtypetwo} - \sizetyp{\mtype} -1+ \sizem{\tderiv_{\sfire}} + 1  
		\\
		& = &\sizectx{\typctxtwo} - \sizetyp{\mtypetwo} - \sizetyp{\mtype}+ \sizem{\tderiv_{\sfire}}   
		\\
		& \leq_{\ih} & \sizectx{\typctxtwo} - \sizetyp{\mtypetwo} - \sizetyp{\mtype}+ \sizectx{\typctxthree} + 
\sizetyp{\mtypetwo}
		\\
		& = & \sizectx{\typctxtwo} + \sizectx{\typctxthree}- \sizetyp{\mtype}
		\\
		& = & \sizectx{\typctx} - \sizetyp{\mtype}.
		\end{array}\]
		
		\item \emph{Explicit substitution on inert}, \ie $\tm = \sitm \esub{\var}{\sitmtwo}$ (which is \full inert).
		Then necessarily
		\begin{equation*}
		\tderiv = 
		\begin{prooftree}
		\hypo{}
		\ellipsis{$\tderiv_{\sitm}$}{\typctxtwo, \var \hastype \mtypetwo \vdash \sitm \hastype \mtype}
		\hypo{}
		\ellipsis{$\tderiv_{\sitmtwo}$}{\typctxthree \vdash \sitmtwo \hastype \mtypetwo}
		\infer2[\footnotesize$\ruleES$]{\tyjp{}{\sitm \esub{\var}{\sitmtwo}}{\typctxtwo \mplus \typctxthree}{\mtype}}
		\end{prooftree}
		\end{equation*}
		where $\typctx = \typctxtwo \mplus \typctxthree$.
		By \ih, $\sizem{\tderiv_{\sitm}} \leq \sizectx{\typctxtwo} + \sizetyp{\mtypetwo} - \sizetyp{\mtype}$ and $ 
\sizem{\tderiv_{\sitmtwo}} \leq \sizectx{\typctxthree} - \sizetyp{\mtypetwo}$. 
		So, 
		\[\begin{array}{rcl}
		\sizem{\tderiv} & = & \sizem{\tderiv_{\sitm}} + \sizem{\tderiv_{\sitmtwo}}
		\\
		& \leq_{\ih} & \sizectx{\typctxtwo} + \sizetyp{\mtypetwo} - \sizetyp{\mtype} + \sizem{\tderiv_{\sitmtwo}}
		\\
		& \leq_{\ih} & \sizectx{\typctxtwo} + \sizetyp{\mtypetwo} - \sizetyp{\mtype} + \sizectx{\typctxthree} - 
\sizetyp{\mtypetwo}
		\\
		& = & \sizectx{\typctxtwo} + \sizectx{\typctxthree} - \sizetyp{\mtype}
		\\
		& = & \sizectx{\typctx} - \sizetyp{\mtype}.
		\end{array}\]
		
		\item \emph{Explicit substitution on fireball}, \ie $\tm = \sfire \esub{\var}{\sitm}$. Then necessarily
		\begin{equation*}
		\tderiv = 
		\begin{prooftree}
		\hypo{}
		\ellipsis{$\tderiv_{\sfire}$}{\typctxtwo, \var \hastype \mtypetwo \vdash \sfire \hastype \mtype}
		\hypo{}
		\ellipsis{$\tderiv_{\sitm}$}{\typctxthree \vdash \sitm \hastype \mtypetwo}
		\infer2[\footnotesize$\ruleES$]{\tyjp{}{\sfire \esub{\var}{\sitm}}{\typctxtwo \mplus \typctxthree}{\mtype}}
		\end{prooftree}
		\end{equation*}
		where $\typctx = \typctxtwo \mplus \typctxthree$.
		By \ih, $\sizem{\tderiv_{\sfire}} \leq \sizectx{\typctxtwo} + \sizetyp{\mtypetwo} + \sizetyp{\mtype}$  and $ 
\sizem{\tderiv_{\sitm}} \leq \sizectx{\typctxthree}-\sizetyp{\mtypetwo}$. 
		Therefore, 
		\[\begin{array}{rcl}
		\size{\tderiv} & = & \sizem{\tderiv_{\sfire}} + \sizem{\tderiv_{\sitm}}
		\\
		& \leq_{\ih} & \sizectx{\typctxtwo} + \sizetyp{\mtypetwo} + \sizetyp{\mtype} + \sizem{\tderiv_{\sitm}}
		\\
		& \leq_{\ih} & 
		\sizectx{\typctxtwo} + \sizetyp{\mtypetwo} + \sizetyp{\mtype} + \sizectx{\typctxthree} - \sizetyp{\mtypetwo}
		\\
		& = & \sizectx{\typctxtwo} + \sizectx{\typctxthree} + \sizetyp{\mtype}.
		\\
		& = & \sizectx{\typctx} + \sizetyp{\mtype}.
		\end{array}\]
		
		\item \emph{Abstraction}, \ie $\tm = \la{\var}{\sfire}$. 
		$\tderiv$ is necessarily of the form (for some set of indices $I$)
		\begin{equation*}
		\begin{prooftree}
		\hypo{}
		\ellipsis{$\tderiv^i$}{\typctx_i, \var \hastype \mtypethree_i \vdash \sfire \hastype \mtypetwo_i}
		
		\infer1[\footnotesize$\lambda$]{\typctx_i \vdash \la{\var}\sfire \hastype \larrow{\mtypethree_i}{\mtypetwo_i}}
		\delims{\left[}{\right]_{\iI}}
		\infer1[\footnotesize$\ruleManyVal$]{\uplus_{\iI}\typctx_i \vdash \la{\var}\sfire \hastype 
\mset{\larrow{\mtypethree_i}{\mtypetwo_i}}_{\iI}}		\end{prooftree}
		\end{equation*}
		where $\mtype = \bigmplus_{\iI}\mult{\ty{\mtypethree_i}{\mtypetwo_i}}$, $\typctx = \bigmplus_{\iI} \typctx_i$.
		By \ih, $\sizem{\tderiv^i} \leq \sizectx{\typctx_i} + \sizetyp{\mtypethree_i} + \sizetyp{\mtypetwo_i}$. 
		Therefore, 
		\[\begin{array}{rcl}
		\sizem{\tderiv} & = & \sum_{\iI} (\sizem{\tderiv^i} + 1)
		\\
		& \leq_{\ih} & \sum_{\iI}(\sizectx{\typctx_i} + \sizetyp{\mtypethree_i} + \sizetyp{\mtypetwo_i} + 1)
		\\
		& = & \sum_{\iI}\sizectx{\typctx_i} + \sum_{\iI}(\sizetyp{\mtypethree_i} + \sizetyp{\mtypetwo_i} +1)
		\\
		& = & \sizectx{\typctx} + \sum_{\iI}\sizetyp{\larrow{\mtypetwo_i}{\mtypethree_i}}
		\\
		& = & \sizectx{\typctx} + \sizetyp{\mtype}.
		\end{array}\]
		\qedhere
	\end{itemize}
\end{proof}

\begin{corollary}[Types bound the size of \full fireballs]
	\label{corappendix:type-bound-size-normal}
	\NoteState{cor:type-bound-size-normal}
	Let $\tm$ be a \full fireball and $\concl{\tderiv}{\typctx}{\tm}{\mtype}$ be a shrinking derivation. Then $\size{\tm} 
	\leq \sizetyp{\mtype} + \sizectx{\typctx}$.
\end{corollary}

\begin{proof}
	By \Cref{l:size-strong-fireballs}, $\size\tm \leq \sizem\tderiv$, and by 
	\Cref{prop:types-bound-normal-derivations}.\ref{pappendix:types-bound-normal-derivations-fireball} $\sizem\tderiv \leq \sizetyp{\mtype} + \sizectx{\typctx}$.
\end{proof}

\begin{theorem}[Lax bounds of kind 3]
	\label{thmappendix:lax-bounds-3}
	\NoteState{thm:lax-bounds-3}
Let $\tm$ and $\tmtwo$ be closed normalizing terms. Then $\tm\tmtwo$ normalizes if and only if $U(\tm,\tmtwo) \neq \emptyset$. Moreover, if $\tm$ and $\tmtwo$ are normal, $\deriv:\tm\tmtwo \tovsubs^* \tmthree$, and $\tmthree$ is normal then $2\sizem{\deriv} + \size\tmthree \leq \size\mtype+\size\mtypetwo+1$ for every composable pair $(\mtype, \mtypetwo)\in U(\tm,\tmtwo)$.
\end{theorem}

\begin{proof} If $\tm\tmtwo$ normalizes then by shrinking completeness (\refthm{completeness}) there exists a shrinking 
derivation $\tderiv$ for $\tm\tmtwo$. The last rule of $\tderiv$ is necessarily $\ruleAp$, and the premises of that rule 
provide a composable pair in $U(\tm,\tmtwo)$. Vice versa, if $U(\tm,\tmtwo)\neq\emptyset$ then every composable pair 
induces a shrinking derivation for $\tm\tmtwo$, which is then typable. By shrinking correctness (\refthm{correctness}), $\tm\tmtwo$ is $\tovsubs$-normalizing.

Now, the \emph{moreover} part. Let $(\larrow{\mtype}{\mtypetwo}, \mtype) \in U(\tm,\tmtwo)$ be a composable pair. Then there are two derivations $\concl{\tderiv_\tm}{}{\tm}{\larrow{\mtype}{\mtypetwo}}$ and $\concl{\tderiv_\tmtwo}{}{\tmtwo}{\mtype}$. We compose them via a $\ruleAp$ rule into a derivation $\tderiv$ for $\tm\tmtwo$:
\begin{equation*}
		\tderiv \defeq 
		\begin{prooftree}
		\hypo{}
		\ellipsis{$\tderiv_{\tm}$}{\vdash \tm \hastype \mult{\larrow{\mtype}{\mtypetwo}}}
		\hypo{}
		\ellipsis{$\tderiv_{\tmtwo}$}{\vdash \tmtwo \hastype\mtype}
		\infer2[\footnotesize$\ruleApp$]{\tyjp{}{\tm\tmtwo}{}{\mtypetwo}}
		\end{prooftree}
		\end{equation*}
		Note that the definition of composable pair guarantees that $\mtypetwo$ is right shrinking, so that $\tderiv$ is shrinking.
By shrinking correctness (\refthm{correctness}), $2\sizem{\deriv} + \size\tmthree \leq \sizem\tderiv = \sizem{\tderiv_\tm} + \sizem{\tderiv_\tmtwo} + 1$. Now, since for normal terms types bound the size of the derivation (\refprop{types-bound-normal-derivations}),  we obtain $\sizem{\tderiv_\tm} \leq \size{\mult{\larrow{\mtype}{\mtypetwo}}}$ and $\sizem{\tderiv_\tmtwo} \leq \size\mtype$. Therefore,  $2\sizem{\deriv} + \size\tmthree \leq \sizem\tderiv  \leq \size{\mult{\larrow{\mtype}{\mtypetwo}}}+\size\mtype+1$, as required.	
\end{proof}

\begin{lemma}[$\tderiveq$-Invariants]
	\label{lappendix:skel-equiv-properties}
	\NoteProof{l:skel-equiv-properties}
	Let $\concl{\tderiv}{\typctx}{\tm}{\mtype}$ and $\concl{\tderivtwo}{\typctxtwo}{\tm}{\mtypetwo}$ be two derivations such 
	that $\tderiv \tderiveq \tderivtwo$. Then $\size\tderiv = \size\tderivtwo$, $\sizem\tderiv = \sizem\tderivtwo$, 
	$\card\mtype = \card\mtypetwo$, $\dom\typctx = \dom\typctxtwo$, and $\card{(\typctx(\var))} = \card{(\typctxtwo(\var))}$ 
	for every variable $\var$. Moreover, $\tderiv$ is shrinking (resp. unitary shrinking) if and only if $\tderivtwo$ is.
\end{lemma}

\begin{proof}
	By straightforward induction on the derivation $\tderiv$.
\end{proof}

The following proposition has a stronger and more involved statement with respect to the one in the body of the paper in order to have the right \ih to prove it.
\begin{proposition}[Size representation]
	\label{propappendix:size-representation} 
	\NoteState{prop:size-representation}
	Let $\concl{\tderiv}{\typctx}{\tm}{\mtype}$ be a derivation. 
	\begin{enumerate}
		\item\label{pappendix:size-representation-inert}
		\emph{Inert:} if $\tm$ is a \full inert term, then for every multi type $\mtypetwo$ such that $\card\mtypetwo = 
\card\mtype$ there is a derivation $\concl{\tderivtwo}{\typctxtwo}{\tm}{\mtypetwo}$ such that $\tderivtwo \tderiveq 
\tderiv$ and 
		$\sizem{\tderiv} = \sizectx{\typctxtwo} - \sizetyp{\mtypetwo}$.
		\item\label{p:size-representation-fireball}
		\emph{Fireball:} if $\tm$ is a \full fireball, then there is a derivation 
$\concl{\tderivtwo}{\typctxtwo}{\tm}{\mtypetwo}$ such that $\tderivtwo \tderiveq \tderiv$ and 
		$\sizem{\tderiv} = \sizectx{\typctxtwo} + \sizetyp{\mtypetwo}$.
	\end{enumerate}
\end{proposition}

\begin{proof}
	By mutual induction on the definition of \full inert term and \full fireball. Cases: 
	
	\begin{itemize}
		\item \emph{Variable}, \ie $\tm = \var$. Then $\tderiv$ is a $\ruleManyVar$ rule with $n$ premises $\tderiv_i$ for 
$i\in\set{1,\ldots n}$, which are all axioms, that is, $\sizem{\tderiv_i} = 0$. Therefore $\sizem\tderiv = 0$. If 
$\mtypetwo = \mset{\type_1,\ldots, \type_n}$ then the derivation $\concl{\tderivtwo}{\typctxtwo}{\tm}{\mtypetwo}$ for 
the statement is made out of $n$ axioms of types $\type_1,\ldots, \type_n$ over a $\ruleManyVar$ rule and $\typctxtwo$ 
is simply $\mtypetwo$, so that $\sizectx{\typctxtwo} = \sizetyp{\mtypetwo}$. Then $\tderiv \tderiveq \tderivtwo$ and 
$\sizem{\tderivtwo} = 0 = \sizectx{\typctxtwo} - \sizetyp{\mtypetwo}$.

		\item \emph{Inert application}, \ie $\tm = \sitm \sfire$. Then $\tderiv$ ends with rule $\ruleAp$. Let 
$\tderiv_{\sitm}$ and $\tderiv_{\sfire}$ be its two premises. 		
		By \ih, there is a derivation $\concl{\tderivtwo_{\sfire}}{\typctxtwo_{\sfire}}{\sfire}{\mtypethree}$ such that 
$\tderivtwo_{\sfire} \tderiveq \tderiv_{\sfire}$ and $\sizem{\tderiv_{\sfire}} =  
\sizectx{\typctxtwo_{\sfire}} + \sizetyp{\mtypethree}$.
		Let $\mtypetwo$ be a multi type such that $\card\mtypetwo = \card\mtype$. Note that rule $\ruleAp$ forces the multi 
type assigned by $\tderiv_{\sitm}$ to $\sitm$ to be a singleton. By \ih, there is a derivation 
$\concl{\tderivtwo_{\sitm}}{\typctxtwo_{\sitm}}{\sitm}{\mset{\larrow{\mtypethree}{\mtypetwo}}}$ such that 
$\tderivtwo_{\sitm} \tderiveq \tderiv_{\sitm}$ and $ 
\sizem{\tderiv_{\sitm}} = \sizectx{\typctxtwo_{\sitm}} - \sizetyp{\mset{\larrow{\mtypethree}{\mtypetwo}}} = 
\sizectx{\typctxtwo_{\sitm}} - 
\sizetyp{\mtypethree} - \sizetyp{\mtypetwo} - 1$.
		We have the following derivation $\tderivtwo$:
		\begin{equation*}
		\tderivtwo \defeq 
		\begin{prooftree}
		\hypo{}
		\ellipsis{$\tderivtwo_{\sitm}$}{\typctxtwo_{\sitm} \vdash \sitm \hastype \mset{\larrow{\mtypethree}{\mtypetwo}}}
		\hypo{}
		\ellipsis{$\tderivtwo_{\sfire}$}{\typctxtwo_{\sfire} \vdash \sfire \hastype \mtypethree}
		\infer2[\footnotesize$\ruleAp$]{\typctxtwo_{\sitm}\mplus\typctxtwo_{\sfire} \vdash \sitm\sfire \hastype \mtypetwo}
		\end{prooftree}
		\end{equation*}
		Clearly, $\tderiv \tderiveq \tderivtwo$. Let  $\typctxtwo \defeq \typctxtwo_{\sitm}\mplus\typctxtwo_{\sfire}$. Then 
		\[\small\begin{array}{rcl}
			\sizem\tderiv &= &\sizem{\tderiv_{\sitm}}+ \sizem{\tderiv_{\sfire}} 
+ 1 
			\\
			& =_{\ih} & \sizem{\tderiv_{\sitm}} + \sizectx{\typctxtwo_{\sfire}} + \sizetyp{\mtypethree} +1
			\\
			& =_{\ih} & \sizectx{\typctxtwo_{\sitm}} - 
\sizetyp{\mtypethree} - \sizetyp{\mtypetwo} - 1 + \sizectx{\typctxtwo_{\sfire}} + \sizetyp{\mtypethree} +1
			\\
			& = & \sizectx{\typctxtwo_{\sitm}} + \sizectx{\typctxtwo_{\sfire}} - \sizetyp{\mtypetwo}
			\\
			& = &
			\sizectx\typctxtwo - \sizetyp{\mtypetwo}
		\end{array}\]
		
		For point 2, we use point 1 with $\mtypetwo \defeq \mset{\ground_1, \ldots, \ground_{\card\mtype}}$, for which 
$\sizetyp{\mtypetwo} = 0$, so that equation $\sizem{\tderiv} = \sizectx{\typctxtwo} + \sizetyp{\mtypetwo}$ also 
holds true. 
				
		\item \emph{Explicit substitution on inert}, \ie $\tm = \sitm \esub{\var}{\sitmtwo}$.
		Then $\tderiv$ ends with rule $\ruleES$. Let $\tderiv_{\sitm}$ and $\tderiv_{\sitmtwo}$ be its two premises and Let 
$\mtypetwo$ be a multi type such that $\card\mtypetwo = \card\mtype$.
		By \ih, there is a derivation $\concl{\tderivtwo_{\sitm}}{\typctxtwo_{\sitm}, 
\var\hastype\mtypethree}{\sitm}{\mtypetwo}$ such that $\tderivtwo_{\sitm} \tderiveq \tderiv_{\sitm}$ and 
$\sizem{\tderiv_{\sitm}} = \sizectx{\typctxtwo_{\sitm}, \var\hastype\mtypethree } - \sizetyp{\mtypetwo} = 
\sizectx{\typctxtwo_{\sitm}} + \sizetyp{\mtypethree} - \sizetyp{\mtypetwo}$. By 
\reflemma{skel-equiv-properties}, $\mtypethree$ has the same cardinality of the type $\mtypethree'$ given by 
$\tderiv_{\sitm}$ to $\var$, which is also the same type given by $\tderiv_{\sitmtwo}$ to $\sitmtwo$. We can then apply 
the \ih and obtain a derivation $\concl{\tderivtwo_{\sitmtwo}}{\typctxtwo_{\sitmtwo}}{\sitmtwo}{\mtypethree}$ such that 
$\tderivtwo_{\sitmtwo} \tderiveq \tderiv_{\sitmtwo}$ and $\sizem{\tderiv_{\sitmtwo}} = 
\sizectx{\typctxtwo_{\sitmtwo}} - \sizetyp{\mtypethree}$.
		We have the following derivation $\tderivtwo$:
				\begin{equation*}
		\tderivtwo \defeq 
		\begin{prooftree}
		\hypo{}
		\ellipsis{$\tderivtwo_{\sitm}$}{\typctxtwo_{\sitm}, \var\hastype\mtypethree \vdash \sitm \hastype \mtypetwo}
		\hypo{}
		\ellipsis{$\tderivtwo_{\sitmtwo}$}{\typctxtwo_{\sitmtwo} \vdash \sitmtwo \hastype \mtypethree}
		\infer2[\footnotesize$\ruleES$]{\typctxtwo_{\sitm} \mplus \typctxtwo_{\sitmtwo} \vdash \sitm \esub{\var}{\sitmtwo} 
\hastype \mtypetwo}
		\end{prooftree}
		\end{equation*}
Clearly, $\tderiv \tderiveq \tderivtwo$. Let  $\typctxtwo \defeq \typctxtwo_{\sitmtwo}\mplus\typctxtwo_{\sitm}$. Then 
		\[\small\begin{array}{rcl}
			\sizem\tderiv &=&\sizem{\tderiv_{\sitm}}+ \sizem{\tderiv_{\sitmtwo}} 
			\\
			& =_{\ih} & \sizectx{\typctxtwo_{\sitm}} + \sizetyp{\mtypethree} - \sizetyp{\mtypetwo} + 
\sizem{\tderiv_{\sitmtwo}} 
			\\
			& =_{\ih} & \sizectx{\typctxtwo_{\sitm}} + \sizetyp{\mtypethree} - \sizetyp{\mtypetwo} + 
\sizectx{\typctxtwo_{\sitmtwo}} - \sizetyp{\mtypethree} 
			\\
			& = & \sizectx{\typctxtwo_{\sitmtwo}} + \sizectx{\typctxtwo_{\sitm}} - \sizetyp{\mtypetwo}
			\\
			& = &
			\sizectx\typctxtwo - \sizetyp{\mtypetwo}
		\end{array}\]		
				For point 2, we use point 1 with $\mtypetwo \defeq \mset{\ground_1, \ldots, \ground_{\card\mtype}}$, for which 
$\sizetyp{\mtypetwo} = 0$, so that equation $\sizem{\tderiv} = \sizectx{\typctxtwo} + \sizetyp{\mtypetwo}$ also 
holds true.

		\item \emph{Abstraction}, \ie $\sfire = \la{\var}{\sfiretwo}$. 
		Then $\tderiv$ ends with a rule $\ruleManyVal$ with $n$ premises, each one ending with a $\ruleFun$ rule. Let 
$\tderiv^i$ be the premise of the $\ruleFun$ rule of the $i$-th derivation above the $\ruleManyVal$ rule, for 
$i\in\set{1,\ldots,n} \eqdef I$.
		By \ih, there exists a derivation $\concl{\tderivtwo^i}{\typctxtwo_i, \var \hastype 
\mtypethree_i'}{\sfiretwo}{\mtypethree_i}$ such that $\tderivtwo^i \tderiveq \tderiv^i$ and $\sizem{\tderiv^i} = 
\sizectx{\typctxtwo_i} + \sizetyp{\mtypethree_i} + \sizetyp{\mtypethree_i'}$.
		We have the following derivation:
		\begin{equation*}
		\tderivtwo \defeq 
		\begin{prooftree}
		\hypo{}
		\ellipsis{$\tderivtwo^i$}{\typctxtwo_i, \var \hastype \mtypethree_i' \vdash \sfiretwo \hastype \mtypethree_i}
		
		\infer1[\footnotesize$\lambda$]{\typctxtwo_i \vdash \la{\var}\sfiretwo \hastype 
\larrow{\mtypethree'_i}{\mtypethree_i}}
		\delims{\left[}{\right]_{\iI}}
		\infer1[\footnotesize$\ruleManyVal$]{\uplus_{\iI}\typctxtwo_i \vdash \la{\var}\sfiretwo \hastype 
\mset{\larrow{\mtypethree'_i}{\mtypethree_i}}_{\iI}}		\end{prooftree}
		\end{equation*}
		Clearly, $\tderiv \tderiveq \tderivtwo$. Let $\typctxtwo \defeq \uplus_{\iI}\typctxtwo_i$ and $ \mtypetwo \defeq 
\mset{\larrow{\mtypethree'_i}{\mtypethree_i}}_{\iI}$. We have:
		\[\small\begin{array}{rcl}
			\sizem\tderiv &= &n+\sum_{\iI} \sizem{\tderiv^i}
			\\
			& =_{\ih} & n+\sum_{\iI}(\sizectx{\typctxtwo_i} + \sizetyp{\mtypethree_i} + \sizetyp{\mtypethree_i'}) 
			\\
			& = & \sum_{\iI}\sizectx{\typctxtwo_i} + (n+\sum_{\iI}(\sizetyp{\mtypethree_i} + \sizetyp{\mtypethree_i'}))
			\\
			& = & \sizectx\typctxtwo + \sizetyp{\mset{\larrow{\mtypethree'_i}{\mtypethree_i}}_{\iI}}
			\\
			& = & \sizectx\typctxtwo + \sizetyp\mtypetwo
		\end{array}\]

		\item \emph{Explicit substitution on fireball}, \ie $\tm = \sfire \esub{\var}{\sitm}$.
		Then $\tderiv$ ends with rule $\ruleES$. Let $\tderiv_{\sfire}$ and $\tderiv_{\sitm}$ be its two premises. 
		By \ih, there is a derivation $\concl{\tderivtwo_{\sfire}}{\typctxtwo_{\sfire}, 
\var\hastype\mtypethree}{\sfire}{\mtypetwo}$ such that $\tderivtwo_{\sfire} \tderiveq \tderiv_{\sfire}$ and 
$\sizem{\tderiv_{\sfire}} = \sizectx{\typctxtwo_{\sfire}} + \size{\mtypethree} + \sizetyp{\mtypetwo}$. By 
\reflemma{skel-equiv-properties}, $\mtypethree$ has the same cardinality of the type $\mtypethree'$ given by 
$\tderiv_{\sfire}$ to $\var$, which is also the same type given by $\tderiv_{\sitm}$ to $\sitm$.
		By \ih, there is a derivation $\concl{\tderivtwo_{\sitm}}{\typctxtwo_{\sitm}}{\sitm}{\mtypethree}$ such that 
$\tderivtwo_{\sitm} \tderiveq \tderiv_{\sitm}$ and $\sizem{\tderiv_{\sitm}} = 
\sizectx{\typctxtwo_{\sitm}} - \sizetyp{\mtypethree}$.
		We have the following derivation $\tderivtwo$:
				\begin{equation*}
		\tderivtwo \defeq 
		\begin{prooftree}
		\hypo{}
		\ellipsis{$\tderivtwo_{\sfire}$}{\typctxtwo_{\sfire}, \var\hastype\mtypethree \vdash \sfire \hastype \mtypetwo}
		\hypo{}
		\ellipsis{$\tderivtwo_{\sitm}$}{\typctxtwo_{\sitm} \vdash \sitm \hastype \mtypethree}
		\infer2[\footnotesize$\ruleES$]{\typctxtwo_{\sfire} \mplus \typctxtwo_{\sitm} \vdash \sfire \esub{\var}{\sitm} 
\hastype \mtypetwo}
		\end{prooftree}
		\end{equation*}
	Clearly, $\tderiv \tderiveq \tderivtwo$. Let  $\typctxtwo \defeq \typctxtwo_{\sitm}\mplus\typctxtwo_{\sfire}$. Then 
		\[\begin{array}{rcl}
			\sizem\tderiv &= & \sizem{\tderiv_{\sfire}}+ \sizem{\tderiv_{\sitm}} 
			\\
			& =_{\ih} & \sizectx{\typctxtwo_{\sfire}} + \size{\mtypethree} + \sizetyp{\mtypetwo} + \sizem{\tderiv_{\sitm}} 
			\\
			& =_{\ih} & \sizectx{\typctxtwo_{\sfire}} + \size{\mtypethree} + \sizetyp{\mtypetwo} + 
\sizectx{\typctxtwo_{\sitm}} 
-\sizetyp{\mtypethree}  
			\\
			& = & \sizectx{\typctxtwo_{\sitm}} + \sizectx{\typctxtwo_{\sfire}} + \sizetyp{\mtypetwo}
			\\
			& = &
			\sizectx\typctxtwo + \sizetyp{\mtypetwo}
		\end{array}\]
\qedhere
	\end{itemize}	
\end{proof}


\begin{theorem}[Weak exact bounds of kind 3]
	\label{thmappendix:weak-exact-bounds-3}
	\NoteState{thm:weak-exact-bounds-3}
	Let $\tm$ and $\tmtwo$ be normal. 
	If $\deriv:\tm\tmtwo \tovsubs^* \tmthree$ and $\tmthree$ is normal then there 
	exist $\mtype\in \sem{\tm}$ and $\mtypetwo \in \sem{\tmtwo}$ such that $2\sizem{\deriv} + \sizefu\tmthree = 
	\size\mtype+\size\mtypetwo+1$. 
\end{theorem}


\begin{proof}
	
	By shrinking completeness (\Cref{thm:completeness}), there is a derivation $\derive{\tderivtwo}{\tm\tmtwo}$ with $2\sizem{\deriv} + \sizefu{\tmthree} = \sizem{\tderivtwo}$. 
	The last rule of $\tderivtwo$ is an $\ruleAp$ rule between two derivations 
	$\tderivtwo_\tm$ and $\tderivtwo_\tmtwo$. Note that $\sizem{\tderivtwo} = \sizem{\tderivtwo_\tm} + \sizem{\tderivtwo_\tmtwo} + 1$. 
		According to size representation (\Cref{prop:size-representation}.\ref{p:size-representation-fireball}) and  $\tderiveq$-invariants (\Cref{l:skel-equiv-properties}), there are derivations $\concl{\tderivtwo_\tm'}{\,}{\tm}{\mtype}$ and $\concl{\tderivtwo_\tmtwo'}{\,}{\tmtwo}{\mtypetwo}$ such that $\sizem{\tderivtwo_\tm} = \sizetyp{\mtype}$ and $\sizem{\tderivtwo_\tmtwo} = \sizetyp{\mtypetwo}$, that is, $2\sizem{\deriv} + \sizefu{\tmthree} = \sizem{\tderivtwo_\tm} + \sizem{\tderivtwo_\tmtwo} + 1 = \size\mtype+\size\mtypetwo+1$ with $\mtype\in \sem\tm$ and $\mtypetwo\in\sem\tmtwo$.	
\end{proof}

The next lemma shows that every derivation for a \full fireball admits a size dissection. For the induction to go 
through, we need to reinforce the disjoint names requirement and to have a stronger statement for inert terms, including 
a different size requirement. \emph{Terminology}: a substitution $\sigma$ \emph{extends} a substitution $\tau$ if 
$\dom\sigma \supseteq \dom\tau$ and $\sigma(\var) = \tau(\var)$ for every $\var \in \dom\tau$. We also use $\gt\type$, $\gt\mtype$, $\gt\typctx$, and $\gt\tderiv$ for 
the set of ground types occurring in $\type$, $\mtype$, $\typctx$, and $\tderiv$.

\begin{lemma}[Size dissection]
	\label{l:refined-size-representation} 
	\label{lappendix:size-dissection}
	\NoteState{l:size-dissection}
	Let $\conclin{\tderiv}{\typctx}{\tm}{\mtype}$ be a derivation and $\groundset$ be a set of ground types containing 
$\gt\tderiv$. 
	\begin{enumerate}
		\item\label{p:refined-size-representation-inert}
		\emph{Inert:} if $\tm$ is a \full inert term, then for every multi type $\mtype'$ and for every substitution $\tau$ 
such that $\dom\tau = \gt{\mtype'}$ and $\tau(\mtype') = \mtype$ there exist a substitution $\sigma$ extending $\tau$ 
and a derivation $\conclin{\tderivtwo}{\typctx'}{\tm}{\mtype'}$ forming a dissection of $\tderiv$ such that 
		$\sizem{\tderiv} = \sizectxm{\typctx'} -\sizetypm{\mtype'}$ and $\gt{\tderivtwo}\cap \groundset = \emptyset$.
		\item\label{p:refined-size-representation-fireball}
		\emph{Fireball:} if $\tm$ is a \full fireball, then there exist a substitution $\sigma$ and a derivation 
$\conclin{\tderivtwo}{\typctx'}{\tm}{\mtype'}$ forming a size dissection of $\tderiv$ such that $\gt{\tderivtwo}\cap 
\groundset = \emptyset$.
	\end{enumerate}
\end{lemma}

\begin{proof}
	By mutual induction on the definition of \full inert term and \full fireball. Cases: 
	
	\begin{itemize}
		\item \emph{Variable}, \ie $\tm = \var$. Then $\tderiv$ has the following shape:
		\begin{equation*}
		\tderiv =
		\begin{prooftree}
		\hypo{\left[\var \hastype \mset{\type_i} \vdash^\infty \var \hastype \type_i\right]_{i\in\set{1,\ldots, n}}}
		\infer1[\footnotesize$\ruleManyVar$]{\var\hastype\mtype \vdash^\infty \var \hastype \mtype}		\end{prooftree}
		\end{equation*}
		with $\mtype = \mset{\type_1,\ldots, \type_n}$ and $\typctx$ being simply $\var\hastype \mtype$. Note that 
$\sizem{\tderiv} = 0$. Take $\mtype' \defeq \mset{\ground_{i_1},\ldots, \ground_{i_n}}$ with $\ground_{i_1},\ldots, 
\ground_{i_n}$ fresh variables with respect to $\groundset$ (thus in particular not appearing in $\type_1,\ldots, 
\type_n$). Then the derivation $\conclin{\tderivtwo}{\typctx'}{\tm}{\mtype'}$ for the statement is:
		\begin{equation*}
		\tderivtwo \defeq 
		\begin{prooftree}
		\hypo{\left[\var \hastype \mset{\ground_i} \vdash^\infty \var \hastype \ground_i\right]_{i\in\set{1,\ldots, n}}}
		\infer1[\footnotesize$\ruleManyVar$]{\var\hastype\mtype' \vdash^\infty \var \hastype \mtype'}		\end{prooftree}
		\end{equation*}
		with $\typctx'$ being simply $\var\hastype\mtype'$. The substitution is $\sigma \defeq 
\isub{\ground_{i_1}}{\type_1}\ldots \isub{\ground_{i_n}}{\type_n}$. 
		\begin{itemize}
			\item \emph{Skeleton}: clearly, $\tderiv \tderiveq \tderivtwo$. 
			\item \emph{Representation}: clearly, $\sigma(\mtypetwo) = \mtype$, that implies $\sigma(\tderivtwo) = \tderiv$.
			\item \emph{Reinforced disjoint names}: by construction the ground types in $\gt{\tderivtwo}= 
\set{\ground_{i_1},\ldots, \ground_{i_n}}$ are not in $\groundset$.
			\item \emph{Types size}: $\sizem{\tderivtwo} = 0 = \sizetyp{\mtype'} - \sizetyp{\mtype'} = \sizectx{\typctx'} - 
\sizetyp{\mtype'}$
		\end{itemize}
		
		For point 2, simply note that $\sizem{\tderivtwo} = 0 = \sizetyp{\mtype'} + \sizetyp{\mtype'} = \sizectx{\typctx'} + 
\sizetyp{\mtype'}$.

				\item \emph{Inert application}, \ie $\tm = \sitm \sfire$. Then $\tderiv$ ends with rule $\ruleAp$:
		\begin{equation*}
		\tderiv =
		\begin{prooftree}
		\hypo{}
		\ellipsis{$\tderiv_{\sitm}$}{\typctx_{\sitm} \vdash^\infty \sitm \hastype \mset{\larrow{\mtypetwo}{\mtype}}}
		\hypo{}
		\ellipsis{$\tderiv_{\sfire}$}{\typctx_{\sfire} \vdash^\infty \sfire \hastype \mtypetwo}
		\infer2[\footnotesize$\ruleAp$]{\typctx_{\sitm}\mplus\typctx_{\sfire} \vdash^\infty \sitm\sfire \hastype \mtype}
		\end{prooftree}
		\end{equation*}
		with $\typctx = \typctx_{\sitm}\mplus\typctx_{\sfire}$.
		By \ih (point 2) with respect to $\groundsettwo \defeq \groundset \cup \dom\tau$, there exist a substitution 
$\sigma_{\sfire}$ and derivation $\conclin{\tderivtwo_{\sfire}}{\typctx'_{\sfire}}{\sfire}{\mtypetwo'}$ such that they 
form a size dissection of $\tderiv_{\sfire}$ such that $\gt{\tderivtwo_{\sfire}}\cap \groundsettwo = \emptyset$.
		
		Consider now the type $\mset{\larrow{\mtypetwo'}{\mtype'}}$. Because of the properties given by the \ih above, we 
have that $\sigma_{\sfire}(\tau(\mset{\larrow{\mtypetwo'}{\mtype'}})) = 
\mset{\larrow{\sigma_{\sfire}(\mtypetwo')}{\tau(\mtype')}} = \mset{\larrow{\mtypetwo}{\mtype}}$. Let $\tau'$ be 
$\sigma_{\sfire}\comp\tau |_{\gt{\mset{\larrow{\mtypetwo'}{\mtype'}}}}$, that is the restriction of 
$\sigma_{\sfire}\comp\tau$ to the ground types in $\mset{\larrow{\mtypetwo'}{\mtype'}}$. By \ih (point 1) with respect 
to $\groundsetthree \defeq \groundset \cup (\dom{\sigma_{\sfire}} \setminus \gt{\mtypetwo'})$, there exist a 
substitution $\sigma_{\sitm}$ extending $\tau'$ and a derivation 
$\conclin{\tderivtwo_{\sitm}}{\typctx'_{\sitm}}{\sitm}{\mset{\larrow{\mtypetwo'}{\mtype'}}}$ forming a dissection of 
$\tderiv_{\sitm}$ such that $\sizem{\tderiv_{\sitm}} = \sizectx{\typctx'_{\sitm}} - 
\sizetyp{\mset{\larrow{\mtypetwo'}{\mtype'}}} = \sizectx{\typctx'_{\sitm}} - \sizetyp{\mtypetwo'} - \sizetyp{\mtype'} 
-1$ and $\gt{\tderivtwo_{\sitm}}\cap \groundsetthree = \emptyset$.
		We have the following derivation $\tderivtwo$:
		\begin{equation*}
		\tderivtwo \defeq 
		\begin{prooftree}
		\hypo{}
		\ellipsis{$\tderivtwo_{\sitm}$}{\typctx'_{\sitm} \vdash^\infty \sitm \hastype \mset{\larrow{\mtypetwo'}{\mtype'}}}
		\hypo{}
		\ellipsis{$\tderivtwo_{\sfire}$}{\typctx'_{\sfire} \vdash^\infty \sfire \hastype \mtypetwo'}
		\infer2[\footnotesize$\ruleAp$]{\typctx'_{\sitm}\mplus\typctx'_{\sfire} \vdash^\infty \sitm\sfire \hastype \mtype'}
		\end{prooftree}
		\end{equation*}
		Let $\sigma \defeq \sigma_{\sfire} \comp \sigma_{\sitm}$. We have to show that $(\tderivtwo, \sigma)$ is a 
dissection satisfying the statement of point 1.
		\begin{itemize}
			\item \emph{Skeleton}: the two \ih give $\tderiv_{\sfire} \tderiveq \tderivtwo_{\sfire}$ and $\tderiv_{\sitm} 
\tderiveq \tderivtwo_{\sitm}$, that imply $\tderiv \tderiveq \tderivtwo$. 

			\item \emph{Representation}: by choice of $\groundsettwo$ and $\groundsetthree$, $\dom{\sigma_{\sfire}} 
\cap\dom{\sigma_{\sitm}} = \gt{\mtypetwo'}$, on which they agree because $\sigma_{\sitm}$ extends $\tau'$, that is, 
$\sigma = \sigma_{\sfire}$ on $\dom{\sigma_{\sfire}}$ and similarly for $\sigma_{\sitm}$. Then 
$\sigma(\tderivtwo_{\sfire}) = \sigma_{\sfire}(\tderivtwo_{\sfire}) =_{\ih} \tderiv_{\sfire}$ and similarly 
$\sigma(\tderivtwo_{\sfire}) = \tderiv_{\sfire}$, that is, $\sigma(\tderivtwo) = \tderiv$.
			\item \emph{Reinforced disjoint names}: we have that $\gt{\tderivtwo_{\sfire}}\cap \groundsettwo = \emptyset$ and 
$\gt{\tderivtwo_{\sitm}}\cap \groundsetthree = \emptyset$. Since both $\groundsettwo$ and $\groundsetthree$ contain 
$\groundset$, we have that the reinforced disjoint names requirement holds: $\gt{\tderivtwo}\cap \groundset = 
(\gt{\tderivtwo_{\sfire}}\cup \gt{\tderivtwo_{\sitm}} )\cap \groundset = (\gt{\tderivtwo_{\sfire}}\cap \groundset)\cup 
(\gt{\tderivtwo_{\sitm}} \cap \groundset) = \emptyset \cup \emptyset = \emptyset$.
			\item \emph{Types size}: let  $\typctx' \defeq \typctx'_{\sitm}\mplus\typctx'_{\sfire}$. Then 
		\[\small\begin{array}{rcl}
			\sizem\tderiv &= & \sizem{\tderiv_{\sitm}}+ \sizem{\tderiv_{\sfire}} + 1
			\\
			& =_{\ih} & \sizem{\tderiv_{\sitm}} + \sizectx{\typctx'_{\sfire}} + \sizetyp{\mtypetwo'} + 1
			\\
			& =_{\ih} & \sizectx{\typctx'_{\sitm}} - \sizetyp{\mtypetwo'} - \sizetyp{\mtype'} -1 + \sizectx{\typctx'_{\sfire}} 
+ \sizetyp{\mtypetwo'} + 1
			\\			
			& = & \sizectx{\typctx'_{\sitm}} + \sizectx{\typctx'_{\sfire}} -\sizetyp{\mtype'}
			\\
			& = & \sizectx\typctx' -\sizetyp{\mtype'}.
		\end{array}\]
		\end{itemize}
		
		For point 2, if $\mtype = \mset{\type_1, \ldots, \type_{\card\mtype}}$ we use point 1 with $\mtype' \defeq 
\mset{\ground_{i_1}, \ldots, \ground_{i_{\card\mtype}}}$ where $\ground_{i_1}, \ldots, \ground_{i_{\card\mtype}}$ are 
not in $\gt\tderiv$, and $\tau \defeq \isub{\ground_{i_1}}{\type_1}\ldots 
\isub{\ground_{i_{\card\mtype}}}{\type_{\card\mtype}}$---clearly $\tau$ is a substitution such that $\tau(\mtype') = 
\mtype$ and $\dom\tau=\gt{\mtype'}$. Note that $\sizetyp{\mtype'} = 0$, so that equation $\sizem{\tderiv} =  
\sizectx{\typctx'} + \sizetyp{\mtype'}$ also holds true. 
				
		\item \emph{Explicit substitution on inert}, \ie $\tm = \sitm \esub{\var}{\sitmtwo}$.
		Then $\tderiv$ ends with rule $\ruleES$:
		\begin{equation*}
		\tderiv =
		\begin{prooftree}
		\hypo{}
		\ellipsis{$\tderiv_{\sitm}$}{\typctx_{\sitm}, \var\hastype\mtypetwo \vdash^\infty \sitm \hastype \mtype}
		\hypo{}
		\ellipsis{$\tderiv_{\sitmtwo}$}{\typctx_{\sitmtwo} \vdash^\infty \sitmtwo \hastype \mtypetwo}
		\infer2[\footnotesize$\ruleES$]{\typctx_{\sitm} \mplus \typctx_{\sitmtwo} \vdash^\infty \sitm \esub{\var}{\sitmtwo} 
\hastype \mtype}
		\end{prooftree}
		\end{equation*}
		Let $\mtype'$ and $\tau$ be such that $\tau(\mtype') = \mtype$ and $\dom\tau = \gt{\mtype'}$.
By \ih (point 1), there exist a substitution $\sigma_{\sitm}$ extending $\tau$ and a derivation 
$\conclin{\tderivtwo_{\sitm}}{\typctx'_{\sitm}, \var\hastype\mtypetwo'}{\sitm}{\mtype'}$ forming a dissection of 
$\tderiv_{\sitm}$ such that $\sizem{\tderiv_{\sitm}} = \sizectx{\typctx'_{\sitm}} +\sizetyp{\mtypetwo'} - 
\sizetyp{\mtype'}$ and $\gt{\tderivtwo_{\sitm}}\cap \groundset = \emptyset$.

		Let $\sigma_{\mtypetwo'}$ be the restriction of $\sigma_{\sitm}$ to $\gt{\mtypetwo'}$. By \ih (point 1) with respect 
to $\groundsettwo \defeq \groundset \cup (\dom{\sigma_{\sitm}} \setminus \gt{\mtypetwo'})$, there exist a substitution 
$\sigma_{\sitmtwo}$ extending $\sigma_{\mtypetwo'}$ and a derivation 
$\conclin{\tderivtwo_{\sitmtwo}}{\typctx'_{\sitmtwo}}{\sitmtwo}{\mtypetwo'}$ forming a dissection of 
$\tderiv_{\sitmtwo}$ such that $\sizem{\tderiv_{\sitmtwo}} = \sizectx{\typctx'_{\sitmtwo}} - \sizetyp{\mtypetwo'}$.
		We have the following derivation $\tderivtwo$:
		\begin{equation*}
		\tderivtwo \defeq 
		\begin{prooftree}
		\hypo{}
		\ellipsis{$\tderivtwo_{\sitm}$}{\typctx'_{\sitm}, \var\hastype\mtypetwo' \vdash^\infty \sitm \hastype \mtype'}
		\hypo{}
		\ellipsis{$\tderivtwo_{\sitmtwo}$}{\typctx'_{\sitmtwo} \vdash^\infty \sitmtwo \hastype \mtypetwo'}
		\infer2[\footnotesize$\ruleES$]{\typctx'_{\sitm} \mplus \typctx'_{\sitmtwo} \vdash^\infty \sitm 
\esub{\var}{\sitmtwo} \hastype \mtype'}
		\end{prooftree}
		\end{equation*}
		\begin{itemize}
			\item \emph{Skeleton}: the two \ih give $\tderiv_{\sitm} \tderiveq \tderivtwo_{\sitm}$ and $\tderiv_{\sitmtwo} 
\tderiveq \tderivtwo_{\sitmtwo}$, that imply $\tderiv \tderiveq \tderivtwo$. 
			\item \emph{Representation}: by choice of $\groundsettwo$, $\dom{\sigma_{\sitm}} \cap\dom{\sigma_{\sitmtwo}} = 
\gt{\mtypetwo'}$, on which they agree because $\sigma_{\sitm}$ extends $\sigma_{\mtypetwo'}$, that is, $\sigma = 
\sigma_{\sitm}$ on $\dom{\sigma_{\sitm}}$ and similarly for $\sigma_{\sitmtwo}$. Then $\sigma(\tderivtwo_{\sitm}) = 
\sigma_{\sitm}(\tderivtwo_{\sitm}) =_{\ih} \tderiv_{\sitm}$ and similarly $\sigma(\tderivtwo_{\sitmtwo}) = 
\tderiv_{\sitmtwo}$, that is, $\sigma(\tderivtwo) = \tderiv$.
			\item \emph{Reinforced disjoint names}: we have that $\gt{\tderivtwo_{\sitm}}\cap \groundset = \emptyset$ and 
$\gt{\tderivtwo_{\sitmtwo}}\cap \groundsettwo = \emptyset$. Since $\groundsettwo$ contains $\groundset$, we have that 
the reinforced disjoint names requirement holds: $\gt{\tderivtwo}\cap \groundset = (\gt{\tderivtwo_{\sitm}}\cup 
\gt{\tderivtwo_{\sitmtwo}} )\cap \groundset = (\gt{\tderivtwo_{\sitm}}\cap \groundset)\cup (\gt{\tderivtwo_{\sitmtwo}} 
\cap \groundset) = \emptyset \cup \emptyset = \emptyset$.
			\item \emph{Types size}: let  $\typctx' \defeq \typctx'_{\sitm}\mplus\typctx'_{\sitmtwo}$. Then 
		\[\small\begin{array}{rcl}
			\sizem\tderiv &= &\sizem{\tderiv_{\sitm}}+ \sizem{\tderiv_{\sitmtwo}} 
			\\
			& =_{\ih} &\sizem{\tderiv_{\sitm}} + \sizectx{\typctx'_{\sitmtwo}} - \sizetyp{\mtypetwo'}
			\\
			& =_{\ih} & \sizectx{\typctx'_{\sitm}} +\sizetyp{\mtypetwo'} - \sizetyp{\mtype'} + \sizectx{\typctx'_{\sitmtwo}} - 
\sizetyp{\mtypetwo'}
			\\
			& = & \sizectx{\typctx'_{\sitmtwo}} + \sizectx{\typctx'_{\sitm}} - \sizetyp{\mtype'}
			\\
			& = & \sizectx{\typctx'}- \sizetyp{\mtype'}
		\end{array}\]		
		\end{itemize}
		
		For point 2, if $\mtype = \mset{\type_1, \ldots, \type_{\card\mtype}}$ we use point 1 with $\mtype' \defeq 
\mset{\ground_{i_1}, \ldots, \ground_{i_{\card\mtype}}}$ where $\ground_{i_1}, \ldots, \ground_{i_{\card\mtype}}$ are 
not in $\gt\tderiv$, and $\tau \defeq \isub{\ground_{i_1}}{\type_1}\ldots 
\isub{\ground_{i_{\card\mtype}}}{\type_{\card\mtype}}$---clearly $\tau$ is a substitution such that $\tau(\mtype') = 
\mtype$ and $\dom\tau=\gt{\mtype'}$. Note that $\sizetyp{\mtype'} = 0$, so that equation $\sizem{\tderiv} =  
\sizectx{\typctx'} + \sizetyp{\mtype'}$ also holds true.

		\item \emph{Abstraction}, \ie $\sfire = \la{\var}{\sfiretwo}$. 
		Then $\tderiv$ has the following shape, for some $n\in\nat$. 
		\begin{equation*}
		\tderiv =
		\begin{prooftree}
		\hypo{}
		\ellipsis{$\tderiv^i$}{\typctx_i, \var \hastype \mtypethree_i \vdash^\infty \sfiretwo \hastype \mtypetwo_i}
		
		\infer1[\footnotesize$\lambda$]{\typctx_i \vdash^\infty \la{\var}\sfiretwo \hastype 
\larrow{\mtypethree_i}{\mtypetwo_i}}
		\delims{\left[}{\right]_\iN}
		\infer1[\footnotesize$\ruleManyVal$]{\uplus_\iN\typctx_i \vdash^\infty \la{\var}\sfiretwo \hastype 
\mset{\larrow{\mtypethree_i}{\mtypetwo_i}}_\iN}		
		\end{prooftree}
		\end{equation*}
		with $\typctx \defeq \uplus_\iN\typctx_i$ and $ \mtype \defeq \mset{\larrow{\mtypethree_i}{\mtypetwo_i}}_\iN$. 
		
		By \ih (point 2) with respect to $\groundset_1 \defeq \groundset$, there exists a substitution 
$\sigma_{\tderivtwo^1}$ and a derivation $\conclin{\tderivtwo^1}{\typctx'_1, \var \hastype 
\mtypethree_1'}{\sfiretwo}{\mtypetwo_1'}$ forming a size dissection for $\tderiv^1$ such that $\sizem{\tderivtwo^1} = 
\sizectx{\typctx'_1,\var \hastype \mtypethree_1'} + \sizetyp{\mtypetwo_1'} = \sizectx{\typctx'_1} + 
\sizetyp{\mtypethree_1'} + \sizetyp{\mtypetwo_1'}$. Applying again the \ih (point 2) with respect to $\groundset_2 
\defeq \groundset \cup \dom{\sigma_{\tderivtwo^1}}$ we obtain a similar size dissection for $\tderiv^2$. Iterating with 
respect to $\groundset_i \defeq \groundset \cup \dom{\sigma_{\tderivtwo^1}}\cup\mydots \cup 
\dom{\sigma_{\tderivtwo^{i-1}}}$, we obtain size dissections for all $\tderiv^i$.
		We then have the following derivation:
		\begin{equation*}
		\tderivtwo \defeq 
		\begin{prooftree}
		\hypo{}
		\ellipsis{$\tderivtwo^i$}{\typctx'_i, \var \hastype \mtypethree'_i \vdash^\infty \sfiretwo \hastype \mtypetwo'_i}
		
		\infer1[\footnotesize$\lambda$]{\typctx'_i \vdash^\infty \la{\var}\sfiretwo \hastype 
\larrow{\mtypethree'_i}{\mtypetwo'_i}}
		\delims{\left[}{\right]_\iN}
		\infer1[\footnotesize$\ruleManyVal$]{\uplus_\iN\typctx'_i \vdash^\infty \la{\var}\sfiretwo \hastype 
\mset{\larrow{\mtypethree'_i}{\mtypetwo'_i}}_\iN}		
		\end{prooftree}
		\end{equation*}
Let $\typctx' \defeq \uplus_\iN\typctx'_i$ and $ \mtype' \defeq \mset{\larrow{\mtypethree'_i}{\mtypetwo'_i}}_\iN$. 		
\begin{itemize}
			\item \emph{Skeleton}: by \ih $\tderiv^i \tderiveq \tderivtwo^i$. Then, $\tderiv \tderiveq \tderivtwo$. 
			\item \emph{Representation}: by choice of $\groundset_i$, the substitutions $\sigma_{\tderiv^i}$ have pairwise 
disjoint domains. Then $\sigma = \sigma_{\tderiv^i}$ on $\dom{\sigma_{\tderiv^i}}$. Then $\sigma(\tderivtwo^i) = 
\sigma_{\tderivtwo^i}(\tderivtwo^i) =_{\ih} \tderiv^i$, that is, $\sigma(\tderivtwo) = \tderiv$.
			\item \emph{Reinforced disjoint names}: by choice of $\groundset_i$, we have that $\gt{\tderivtwo^i}\cap 
\groundset = \emptyset$ for all $i$. Then, $\gt{\tderivtwo}\cap \groundset = (\cup_\iN\gt{\tderivtwo^i})\cap \groundset 
= \cup_\iN(\gt{\tderivtwo^i}\cap \groundset) = \cup_\iN\emptyset = \emptyset$.
			\item \emph{Types size}: we have:
		\[\small\begin{array}{rcl}
			\sizem\tderiv &= &n+\sum_\iN \sizem{\tderiv^i}
			\\
			& =_{\ih} & n+\sum_\iN(\sizectx{\typctx'_i} + \sizetyp{\mtypethree_i'} + \sizetyp{\mtypetwo'_i}) 
			\\
			& = & \sum_\iN\sizectx{\typctx'_i} + (\sum_\iN(\sizetyp{\mtypethree_i'} + \sizetyp{\mtypetwo'_i} +1))
			\\
			& = & \sizectx\typctx' + \sizetyp{\mset{\larrow{\mtypethree'_i}{\mtypetwo'_i}}_\iN}
			\\
			& = & \sizectx\typctx' + \sizetyp{\mtype'}
		\end{array}\]
		\end{itemize}

		\item \emph{Explicit substitution on fireball}, \ie $\tm = \sfire \esub{\var}{\sitm}$.
		Then $\tderiv$ ends with rule $\ruleES$:
		\begin{equation*}
		\tderiv = 
		\begin{prooftree}
		\hypo{}
		\ellipsis{$\tderiv_{\sfire}$}{\typctx_{\sfire}, \var\hastype\mtypetwo \vdash^\infty \sfire \hastype \mtype}
		\hypo{}
		\ellipsis{$\tderiv_{\sitm}$}{\typctx_{\sitm} \vdash^\infty \sitm \hastype \mtypetwo}
		\infer2[\footnotesize$\ruleES$]{\typctx_{\sfire} \mplus \typctx_{\sitm} \vdash^\infty \sfire \esub{\var}{\sitm} 
\hastype \mtype}
		\end{prooftree}
		\end{equation*}
		with $\typctx \defeq \typctx_{\sfire}\mplus\typctx_{\sitm}$.
		By \ih (point 2), there is a substitution $\sigma_{\sfire}$ and a derivation 
$\conclin{\tderivtwo_{\sfire}}{\typctx'_{\sfire}, \var\hastype\mtypetwo'}{\sfire}{\mtype'}$ forming a size dissection of 
$\tderiv_{\sfire}$ such that $\gt{\tderivtwo_{\sfire}} \cap \groundset = \emptyset$. 
		Let $\sigma_{\mtypetwo'}$ the restriction of $\sigma_{\sfire}$ to $\gt{\mtypetwo'}$.
		By \ih (point 1) with respect to $\groundsettwo \defeq \groundset \cup (\dom{\sigma_{\sfire}} \setminus 
\gt{\mtypetwo'})$, there is a substitution $\sigma_{\sitm}$ extending $\sigma_{\mtypetwo'}$ and a derivation 
$\conclin{\tderivtwo_{\sitm}}{\typctx'_{\sitm}}{\sitm}{\mtypetwo'}$ forming a dissection of $\tderiv_{\sitm}$ such that 
$\sizem{\tderiv_{\sitm}} = \sizectx{\typctx'_{\sitm}}-\sizetyp{\mtypetwo'}$ and $\gt{\tderivtwo_{\sitm}} \cap 
\groundsettwo = \emptyset$.
		We have the following derivation $\tderivtwo$:
		\begin{equation*}
		\tderivtwo \defeq 
		\begin{prooftree}
		\hypo{}
		\ellipsis{$\tderivtwo_{\sfire}$}{\typctx'_{\sfire}, \var\hastype\mtypetwo' \vdash^\infty \sfire \hastype \mtype'}
		\hypo{}
		\ellipsis{$\tderivtwo_{\sitm}$}{\typctx'_{\sitm} \vdash^\infty \sitm \hastype \mtypetwo'}
		\infer2[\footnotesize$\ruleES$]{\typctx'_{\sfire} \mplus \typctx'_{\sitm} \vdash^\infty \sfire \esub{\var}{\sitm} 
\hastype \mtype'}
		\end{prooftree}
		\end{equation*}
		\begin{itemize}
			\item \emph{Skeleton}: the two \ih give $\tderiv_{\sfire} \tderiveq \tderivtwo_{\sfire}$ and $\tderiv_{\sitm} 
\tderiveq \tderivtwo_{\sitm}$, that imply $\tderiv \tderiveq \tderivtwo$. 
			\item \emph{Representation}: by choice of $\groundsettwo$, $\dom{\sigma_{\sfire}} \cap\dom{\sigma_{\sitm}} = 
\gt{\mtypetwo'}$, on which they agree because $\sigma_{\sitm}$ extends $\sigma_{\mtypetwo'}$, that is, $\sigma = 
\sigma_{\sfire}$ on $\dom{\sigma_{\sfire}}$ and similarly for $\sigma_{\sitm}$. Then $\sigma(\tderivtwo_{\sfire}) = 
\sigma_{\sfire}(\tderivtwo_{\sfire}) =_{\ih} \tderiv_{\sfire}$ and similarly $\sigma(\tderivtwo_{\sitm}) = 
\tderiv_{\sitm}$, that is, $\sigma(\tderivtwo) = \tderiv$.

			\item \emph{Reinforced disjoint names}: we have that $\gt{\tderivtwo_{\sfire}}\cap \groundset = \emptyset$ and 
$\gt{\tderivtwo_{\sitm}}\cap \groundsettwo = \emptyset$. Since $\groundsettwo$ contains $\groundset$, we obtain: 
$\gt{\tderivtwo}\cap \groundset = (\gt{\tderivtwo_{\sfire}}\cup \gt{\tderivtwo_{\sitm}} )\cap \groundset = 
(\gt{\tderivtwo_{\sfire}}\cap \groundset)\cup (\gt{\tderivtwo_{\sitm}} \cap \groundset) = \emptyset \cup \emptyset = 
\emptyset$.

			\item \emph{Types size}: let  $\typctx' \defeq \typctx'_{\sfire}\mplus\typctx'_{\sitm}$. Then 
		\[\begin{array}{rcl}
			\sizem\tderiv &= & \sizem{\tderiv_{\sfire}}+ \sizem{\tderiv_{\sitm}} 
			\\
			& =_{\ih} & \sizectx{\typctx'_{\sfire}} + \size{\mtypetwo'} + \sizetyp{\mtype'} + \sizem{\tderiv_{\sitm}} 
			\\
			& =_{\ih} & \sizectx{\typctx'_{\sfire}} + \size{\mtypetwo'} + \sizetyp{\mtype'} + \sizectx{\typctx'_{\sitm}} 
-\sizetyp{\mtypetwo'}  
			\\
			& = & \sizectx{\typctx'_{\sitm}} + \sizectx{\typctx'_{\sfire}} + \sizetyp{\mtype'}
			\\
			& = &
			\sizectx{\typctx'} + \sizetyp{\mtype'}
		\end{array}\]
		\end{itemize}
		
\qedhere
	\end{itemize}	
\end{proof}

\begin{theorem}[Exact bounds of kind 3]
	\label{thmappendix:exact-bounds-3}
	\NoteState{thm:exact-bounds-3}
Let $\tm$ and $\tmtwo$ be closed \full fireballs. If $\deriv:\tm\tmtwo \tovsubs^* \tmthree$ and $\tmthree$ is normal 
then $2\sizem{\deriv} + \sizefu{\tmthree} =  \inf\set{\size{\larrow{\mtypetwo}{\mtype}}+\size\mtypethree+1\ |\ \exists 
\sigma\mbox{ such that }(\larrow{\sigma(\mtypetwo)}{\sigma(\mtype)}, \sigma(\mtypethree))\in U(\tm,\tmtwo)}$.
\end{theorem}

\begin{proof}
We show a construction that from any shrinking derivation for $\tmthree$ gives types $\mtype$, $\mtypetwo$, and 
$\mtypethree$ as in the set of the statement and such that $\size{\larrow{\mtypetwo}{\mtype}}+\size\mtypethree+1 \geq 
2\sizem{\deriv} + \sizefu{\tmthree}$. The \emph{inf} of that set is then obtained by taking a minimal shrinking derivation, 
that is a unitary shrinking one, for which then the inequality is an equality.

By shrinking completeness (\Cref{thm:completeness}), there exists a derivation $\concl{\tderiv}{}{\tm\tmtwo}{\mtype'}$ such that $2\sizem{\deriv} + \sizefu{\tmthree} \leq \sizem\tderiv$. The 
last rule of $\tderiv$ is an $\ruleAp$ rule between two derivations 
$\conclin{\tderiv_\tm}{}{\tm}{\mset{\larrow{\mtypetwo'}{\mtype'}}}$ and 
$\conclin{\tderiv_\tmtwo}{}{\tmtwo}{\mtypetwo'}$, for some $\mtypetwo'$. Note that $\size\tderiv = \size{\tderiv_\tm} + \size{\tderiv_\tmtwo} + 1$. By size 
dissection (\reflemma{size-dissection}) applied to $\tderiv_\tm$ (with respect to $\groundset 
\defeq \gt{\tderiv}$), there exist a substitution $\sigma_\tm$ and a derivation 
$\conclin{\tderivthree_\tm}{}{\tm}{\mset{\larrow\mtypetwo\mtype}}$ forming a size dissection of $\tderiv_\tm$. In 
particular, we have $\sigma_\tm(\mset{\larrow\mtypetwo\mtype}) = \mset{\larrow{\mtypetwo'}{\mtype'}}$ and 
$\sizem{\tderiv_\tm} = \sizetyp{\mset{\larrow\mtypetwo\mtype}}$.  By size dissection 
(\Cref{l:size-dissection}) applied to $\tderiv_\tmtwo$ (with respect to $\groundsettwo \defeq 
\gt{\tderiv} \cup \dom{\sigma_\tm}$) there exist a substitution $\sigma_\tmtwo$ and a derivation 
$\conclin{\tderivthree_\tmtwo}{}{\tmtwo}{\mtypethree}$ forming a size dissection of $\tderiv_\tmtwo$. In particular, 
we have $\sigma_\tmtwo(\mtypethree) = \mtypetwo'$ and $\sizem{\tderiv_\tmtwo} = \sizetyp{\mtypethree}$. Note that 
$\sigma_\tm(\mset{\larrow\mtypetwo\mtype})$ and $\sigma_\tmtwo(\mtypethree)$ form a composable pair because $\mtype'$ is right (because $\tderiv$ is shrinking). By construction, 
the domains of $\sigma_\tm$ and $\sigma_\tmtwo$ are disjoint, so that we can define $\sigma$ as $\sigma_\tm \cup 
\sigma_\tmtwo$. Also, 
 $\size\tderiv = \size{\tderiv_\tm} + \size{\tderiv_\tmtwo} + 1 = \sizetyp{\mset{\larrow\mtypetwo\mtype}} + 
\sizetyp{\mtypethree} + 1$, that is, $2\sizem{\deriv} + \sizefu{\tmthree} \leq \size\tderiv = 
\sizetyp{\mset{\larrow\mtypetwo\mtype}} + \sizetyp{\mtypethree} + 1$.	Now, by starting with a unitary shrinking 
derivation $\tderiv$ for $\tm\tmtwo$, by correctness, 
$2\sizem{\deriv} + \sizefu{\tmthree} = \size\tderiv$, that is, $2\sizem{\deriv} + \sizefu{\tmthree} 
=\sizetyp{\mset{\larrow\mtypetwo\mtype}} + \sizetyp{\mtypethree} + 1$.
\end{proof}

We conclude with the proof of a corollary of the size representation property  (\Cref{prop:size-representation}.\ref{p:size-representation-fireball}), and not of the previous theorem.

\begin{corollary}[Minimal types catch the size of normal forms]
	\label{corappendix:minimal-type-normal}
	\NoteState{cor:minimal-type-normal}
	Let $\tm$ be a \full fireball. There exists a unitary shrinking derivation $\concl{\tderiv}{\typctx}{\tm}{\mtype}$ such 
	that $\sizefu{\tm} = \sizetyp{\mtype} + \sizectx{\typctx}$.
\end{corollary}

\begin{proof}
	By \refpropp{typability-normal}{fireball}, there exists a unitary shrinking derivation $\derive{\tderivtwo}{\tm}$. 
	By \Cref{l:size-strong-fireballs}, $\sizefu{\tm} = \sizem\tderivtwo$. 
	By size representation (\Cref{prop:size-representation}.\ref{p:size-representation-fireball}), there 
	exists $\concl{\tderiv}{\typctx}{\tm}{\mtype}$ such that $\tderiv \tderiveq \tderivtwo$ and $\sizem\tderivtwo = \sizetyp{\mtype} + \sizectx{\typctx}$ (thus $\sizefu{\tm} = \sizetyp{\mtype} + \sizectx{\typctx}$).
	According to $\tderiveq$-invariants (\Cref{l:skel-equiv-properties}), $\tderiv$ is unitary shrinking.
\end{proof}





\end{document}